\documentclass[11pt]{article}
\usepackage[a4paper, total={6.5in, 9.75in}]{geometry}

\usepackage{siunitx}
\usepackage{fancyhdr}
\usepackage{extramarks}
\usepackage{amsmath}
\usepackage{amsthm}
\usepackage{arydshln}
\usepackage{mathtools}
\usepackage{amsfonts}
\usepackage{tikz}
\usepackage{enumerate}
\usepackage{graphicx}
 \graphicspath{ {images/} }
\usepackage{algpseudocode}
\usepackage{textcomp}
\usepackage{import}
\usepackage{natbib}
\setcitestyle{round}
\usepackage[font={footnotesize, it}]{caption}
\usepackage{bbm}
\usepackage{subcaption}
\usepackage[ruled,vlined]{algorithm2e} 
\usepackage[font={footnotesize, it}]{caption}
\usepackage{xr}
 \usepackage{hyperref}
 \hypersetup{
     linktoc=all,     %set to all if you want both sections and subsections linked
     linkcolor=black,  %choose some color if you want links to stand out
 }
 \usepackage{geometry}
 \usepackage[T1]{fontenc}
\usepackage[utf8]{inputenc}
\usepackage{authblk}
\usepackage{blindtext}
\usepackage{booktabs}
\usepackage{array, makecell}

\newcommand{\E}{\mathbb{E}}
\newcommand{\Var}{\mathrm{Var}}
\newcommand{\Cov}{\mathrm{Cov}}

\newcommand\indep{\protect\mathpalette{\protect\independenT}{\perp}}
\def\independenT#1#2{\mathrel{\rlap{$#1#2$}\mkern2mu{#1#2}}}

\newtheorem{theorem}{Theorem}
\newtheorem{remark}{Remark}

\newtheorem{lemma}[theorem]{Lemma}
\newtheorem{proposition}[theorem]{Proposition}

\numberwithin{theorem}{section}
\numberwithin{remark}{section}
\numberwithin{equation}{section}

\begin{document}

\title{Fused Extended Two-Way Fixed Effects for Difference-in-Differences With Staggered Adoptions}
\date{\today}

\author{Gregory Faletto\thanks{gregory.faletto@marshall.usc.edu} }
\affil{University of Southern California}

\maketitle

\begin{abstract}

To address the bias of the canonical two-way fixed effects estimator for difference-in-differences under staggered adoptions, \citet{wooldridge2021two} proposed the \textit{extended two-way fixed effects} estimator, which adds many parameters. However, this reduces efficiency. Restricting some of these parameters to be equal (for example, subsequent treatment effects within a cohort) helps, but \textit{ad hoc} restrictions may reintroduce bias. We propose a machine learning estimator with a single tuning parameter, \textit{fused extended two-way fixed effects} (FETWFE), that enables automatic data-driven selection of these restrictions. We prove that under an appropriate sparsity assumption FETWFE identifies the correct restrictions with probability tending to one, which improves efficiency. We also prove the consistency, oracle property, and asymptotic normality of FETWFE for several classes of heterogeneous marginal treatment effect estimators under either conditional or marginal parallel trends, and we prove the same results for conditional average treatment effects under conditional parallel trends. We provide an R package implementing fused extended two-way fixed effects, and we demonstrate FETWFE in simulation studies and an empirical application.

\end{abstract}


\section{Introduction}

We investigate difference-in-differences estimation under staggered adoptions. Suppose \(N\) units are observed at \(T \geq 2\) times, where our theoretical results will assume \(N\) grows asymptotically while \(T\) is fixed. For convenience we assume a balanced panel where all units are observed at all times, so we have \(NT\) observations.

Suppose \(R \leq T - 1\) cohorts start receiving a treatment at times \(r \in \mathcal{R} =\{r_1, r_2, \ldots, r_R\} \subseteq \{2, \ldots, T\}\). For example, we might want to investigate the effect of a law or policy that is implemented state-by-state across time \citep{callaway2021difference, goodman2021difference, borusyak2024revisiting}, or the effect of a product or service that is released across different regions over time \citep{de2020two}. We consider settings where the treatment is an absorbing state, so that once a unit starts being treated, they continue receiving treatment thereafter, but this can be relaxed with little complication. Since no cohort enters at time \(t = 1\) and we assume the probability of never selecting into treatment is strictly positive for all units (see Assumption F2 in Section \ref{setup.sec}), for sufficiently large \(N\) we have untreated units available as a baseline comparison for all treated cohorts at all times.

 For \(i \in \{1, \ldots, N\} =: [N]\) and \(t \in [T]\), we encode the treatment status for each unit in a random variable \(W_i \in \{0\} \cup \mathcal{R}\), where \(W_i  = r\) if unit \(i\) is in cohort \(r \in \mathcal{R}\) or \(W_i = 0\) if unit \(i\) is never treated. Let \(\tilde{y}_{(it)}(r)\) be the potential outcome for unit \(i\) at time \(t\) if they were in cohort \(r\). We will be interested in estimating treatment effects of the form
 \begin{align} 
\tau_{\text{ATT}} (r, t)  :=    \E \left[  \tilde{y}_{(i t)}(r) - \tilde{y}_{(i t)}(0) \mid W_i = r  \right]  \label{att.cohort.time} , \qquad  r \in \mathcal{R}, t \geq r,
\end{align}
as well as conditional treatment effects that depend on the value of unit covariates; see \eqref{catt.t.r.def}, to come. It is well-known that in the \(T =2\), \(\mathcal{R} = \{2\}\) case, under no anticipation and parallel trends assumptions
\begin{equation}\label{did.estimand}
\tau_{\text{ATT}}(2, 2) = \E \left[ \tilde{y}_{i2}(2) - \tilde{y}_{i1}(2) \mid W_i = 2 \right] - \E \left[  \tilde{y}_{i2}(0) - \tilde{y}_{i1}(0)  \mid W_i = 0 \right]
,
\end{equation}
and this treatment effect can be estimated consistently by the canonical two-way fixed effects regression
\begin{equation}\label{bad.reg}
\tilde{y}_{it} = \alpha_i^* + \gamma_t^* + \tau^* \cdot \sum_{r \in \mathcal{R}} \mathbbm{1}\{W_i = r \} \mathbbm{1}\{t \geq r\}  + \epsilon_{it}
,
\end{equation}
where \(\alpha_i^*\) are fixed effects corresponding to units, \(\gamma_t^*\) are fixed effects corresponding to time, \(\mathbbm{1}\{ \cdot \} \) is an indicator variable that equals 1 on the event \(\cdot\) and 0 otherwise, \(\epsilon_{it} \) is a noise term, and for conciseness we denote \(\tilde{y}_{it} := \tilde{y}_{it}(W_i)\), the observed response for unit \(i\) at time \(t\). 

In the more general \(T \geq 2\) and arbitrary \(\mathcal{R}\) setting, one might hope that regression \eqref{bad.reg} would aggregate the treatment effects in a sensible way to estimate some form of an average treatment effect, but it is now well-known that this isn't the case, and many alternative estimators have been proposed \citep{borusyak2018revisiting, de2020two, sun2021estimating,  goodman2021difference, callaway2021difference}.

\subsection{The Extended Two-Way Fixed Effects Estimator}

Most of the alternative estimators proposed in these and other works depart from ordinary least squares (OLS) estimation, but \citet{wooldridge2021two}\footnote{Along with \citet{sun2021estimating}.} shows that linear regression can still be unbiased provided that enough parameters are estimated. He first shows that \eqref{did.estimand} is also the estimand for \(\tau\) in the regression
\begin{equation}\label{wooldridge.smaller.model}
\tilde{y}_{it} =  \eta^* + \nu_2^* \cdot \mathbbm{1}\{W_i = 2\}  + \gamma_2^* \cdot \mathbbm{1}\{t = 2\}  + \tau^* \cdot \mathbbm{1}\{W_i = 2 \} \mathbbm{1}\{t  = 2\}  + \epsilon_{it}
,
\end{equation}
where \( \eta \) can be interpreted as the expected outcome for the never-treated units at \(t = 1\), \(\nu_2\) is the selection bias for the treated units at \(t = 1\), and \(\gamma_2\) is the trend in the untreated potential outcomes. \citeauthor{wooldridge2021two} proposes generalizing this to the arbitrary \(T\) case and allowing for each parameter to be linear in covariates rather than constant. This leads to a linear model on cohort dummy variables, time dummies, covariates, treatment indicators, and interactions of those terms: for all \(i \in [N]\), \(t \in [T]\), and \(r \in \mathcal{R}\), 
\begin{equation}\label{wooldridge.6.33.model}
\tilde{y}_{it} = \eta^* +  \gamma_t^* + \boldsymbol{X}_{i}^\top(\boldsymbol{\kappa}^* + \boldsymbol{\xi}_t^*) +\sum_{r \in \mathcal{R}}  \mathbbm{1} \{W_i = r \} \left( \nu_r^* + \boldsymbol{X}_{i}^\top \boldsymbol{\zeta}_r^*  + \mathbbm{1}\{ t \geq r\}  (\tau_{rt}^*  + \boldsymbol{\dot{X}}_{(ir)}^\top  \boldsymbol{\rho}_{rt}^*  ) \right) + \epsilon_{it}
.
\end{equation}
Here \(\boldsymbol{X}_{i} \in \mathbb{R}^d\) is a set of time-invariant (pre-treatment) covariates, \( \eta^* + \boldsymbol{X}_{i}^\top \boldsymbol{\kappa}^* \) is the conditionally expected outcome of the never-treated units at \(t = 1\), the \( \gamma_t^*  + \boldsymbol{X}_{i}^\top  \boldsymbol{\xi}_t^*  \) terms are the conditional trends in the untreated potential outcomes, \(\nu_r^* + \boldsymbol{X}_{i}^\top \boldsymbol{\zeta}_r^*\) is the conditional selection bias for cohort \(r\), and \(\tau_{rt}^*  + \boldsymbol{\dot{X}}_{(ir)}^\top  \boldsymbol{\rho}_{rt}^* \) is the conditional treatment effect for cohort \(r\) at time \(t\).  (See Assumption (LINS) in Section \ref{lins.sec} for more details.)

%\(\eta^*\) has the same interpretation as in \eqref{wooldridge.smaller.model}; \(\gamma_t^*\) are fixed effect coefficients for times \(\{2, \ldots, T\}\); ; \(\boldsymbol{\kappa}^* \in \mathbb{R}^d\) are covariate coefficients corresponding to the conditionally expected outcome of the never-treated units at \(t =1\); \(\boldsymbol{\xi}_t^* \in \mathbb{R}^d\) and \(\boldsymbol{\zeta}_r^* \in \mathbb{R}^d\) are time-specific and cohort-specific coefficients (coefficients on the interactions between the covariates and the time and cohort dummies) that correspond to the conditionally expected trends and selection biases for each cohort, respectively; \(\nu_r^*\) are cohort fixed effect coefficients; \(\tau_{rt}^*\) are interactions between cohort indicators and time dummies that correspond to \(\tau_{\text{ATT}} (r, t)\) under certain assumptions; \(\boldsymbol{\rho}_{rt}^* \in \mathbb{R}^d\) are coefficients governing the interactions between treatment dummies and covariates (corresponding to the conditional treatment effects); and \(\boldsymbol{\dot{X}}_{(ir)} := \boldsymbol{X}_{i} - \E[ \boldsymbol{X}_{i} \mid W_i = r]\) are covariate vectors that are centered with respect to their cohort means.

\citeauthor{wooldridge2021two} calls \eqref{wooldridge.6.33.model} the \textit{extended two-way fixed effects} (ETWFE) estimator \citep[Equation 6.33]{wooldridge2021two} and shows that in this regression \(\tau_{rt}^*\) is unbiased for  \(\tau_{\text{ATT}} (r, t)\) under conditional no-anticipation and parallel trends assumptions; see Assumptions (CNAS) and (CCTS) in Section \ref{sec.did.assumptions}. The ETWFE estimator can also be used to estimate conditional average treatment effects on the treated units
\begin{equation}\label{catt.t.r.def}
\tau_{\text{CATT}} (r, t, \boldsymbol{x}) := \E [ \tilde{y}_{(i t)}(r) - \tilde{y}_{(i t)}(0) \mid W_i = r, \boldsymbol{X}_{i} = \boldsymbol{x} ]
, \qquad i \in [N], r \in \mathcal{R},  t \in \{r, \ldots, T\}
.
\end{equation}

\subsection{Restrictions in the ETWFE Model}

As \citet[Section 6.5]{wooldridge2021two} points out, the ETWFE estimator \eqref{wooldridge.6.33.model} may contain a large number of parameters to estimate if \(R\), \(T\), or \(d\) is moderately large relative to \(N\). In this case, ordinary least squares estimation of \eqref{wooldridge.6.33.model} may lead to very imprecise estimates of at least some of the parameters, trading the bias problem that model \eqref{bad.reg} has for a low-efficiency (or high-variance) problem. 

\citet{wooldridge2021two} suggests some possible ways to mitigate this problem by imposing restrictions in order to reduce the number of parameters to be estimated. For example, one might assume that the treatment effects only vary across cohorts and not with time, so that \(\tau_{rt} = \tau_r\) for all \(r \in \mathcal{R}\), \(t \geq r\). (Other works that consider similar restrictions include \citet{borusyak2018revisiting} and \citet{borusyak2024revisiting}.) If many of these restrictions hold, then imposing them results in more efficient estimation of the treatment effects than regression \eqref{wooldridge.6.33.model}. However, imposing these restrictions in an \textit{ad hoc} way could reintroduce bias if the wrong restrictions are imposed, or leave unnecessary inefficiency if not enough restrictions are imposed.

The first goal of this paper is to propose a novel estimator that imposes these restrictions in an automatic, data-driven way. We propose the \textit{fused extended two-way fixed effects} (FETWFE) estimator, a novel machine learning method for difference-in-differences estimation. Before we describe FETWFE, we discuss sparse fusion penalties in linear regression, which play a key role in our estimator.

\subsection{Sparse Fusion Penalties}

Given a \(N \times p\) real-valued design matrix \(\boldsymbol{X}\) and a response \(\boldsymbol{y} \in \mathbb{R}^N\), the well-known lasso \citep{Tibshirani1996} estimator is defined by the optimization problem
\begin{align*}
\boldsymbol{\hat{\beta}}(\lambda_N) \in ~ & \underset{\boldsymbol{\beta} \in \mathbb{R}^p}{\arg \min} \left\{  \sum_{i=1}^N  \left( y_i - \boldsymbol{\beta}^\top \boldsymbol{X}_i \right)^2 + \lambda_N \sum_{j=1}^p | \beta_j| \right\}
\\  = ~ &  \underset{\boldsymbol{\beta} \in \mathbb{R}^p}{\arg \min} \left\{\left \lVert \boldsymbol{y}  -  \boldsymbol{X} \boldsymbol{\beta} \right \rVert_2^2 + \lambda_N \lVert \boldsymbol{\beta} \rVert_1 \right\}
,
\end{align*}
where \(\lambda_N\) is a tuning parameter that can be chosen from a set of candidate values by cross-validation or a criterion like BIC, and recall that the \(\ell_q\) norm\footnote{It is a norm for finite \(q \geq 1\), though we also consider \(q \in (0, 1)\), in which case \(\lVert \cdot \rVert_q \) does not satisfy the triangle inequality.} of an \(n\)-vector \(\boldsymbol{v} = (v_1, v_2, \ldots, v_n)^\top \) is defined as \(\lVert \boldsymbol{v} \rVert_q := \left( \sum_{i=1}^n |v_i|^q \right)^{1/q}\) . The \(\ell_1\) penalty term leads to \textit{sparse} estimated coefficients \(\boldsymbol{\hat{\beta}}(\lambda_N)\) with entries set exactly equal to 0.

The \textit{fused lasso} \citep{tibshirani2005sparsity} adds a second tuning parameter and \(\ell_1\) penalty term:
\begin{align*}
\boldsymbol{\hat{\beta}}(\lambda_N^{(1)}, \lambda_N^{(2)}) \in ~ &  \underset{\boldsymbol{\beta} \in \mathbb{R}^p}{\arg \min} \left\{ \sum_{i=1}^N  \left( y_i - \boldsymbol{\beta}^\top \boldsymbol{X}_i \right)^2 + \lambda_N^{(1)} \sum_{j=1}^p | \beta_j|  + \lambda_N^{(2)} \sum_{j=2}^p | \beta_j - \beta_{j-1}| \right\}
.
\end{align*}
The fused lasso tends to yield solutions that both set some individual coefficients exactly equal to 0 (due to the first \(\ell_1\) penalty term) and also set the differences between some adjacent pairs of coefficients exactly equal to 0 (due to the second \(\ell_1\) penalty term). That is, some adjacent pairs of coefficients are set exactly equal to each other, or \textit{fused} together.

FETWFE uses \(\ell_q\) fusion penalties for \(q \in (0, 2]\) (\textit{bridge} regularization) to shrink certain sets of parameters in \eqref{wooldridge.6.33.model} towards each other in a way that aligns with intuition and previously proposed restrictions. See Figure \ref{fig:fusion-penalties} for a visual depiction of which of the marginal average treatment effects \(\hat{\tau}_{rt}\) we propose penalizing towards each other, and see further details about the FETWFE penalty term \eqref{fetwfe.penalty} in Section \ref{sec.meth}. A lasso estimator is within our framework if one chooses \(q = 1\). Choices of \(q \in (0, 1]\) will result in sparse solutions, and choices of \(q \in (0,1)\) will allow for convergence at \(1/\sqrt{N}\) rates and asymptotic normality of the estimated coefficients. These fusion penalties allow the restrictions to be selected in a hands-off, automatic way that comes with theoretical guarantees.

\begin{figure}[ht]
    \centering
\begin{tikzpicture}[every node/.style={font=\small}]
    % Define T and R (Adjust \T as needed for your document)
    \def\T{6} % Number of time periods
    \def\R{5} % Number of cohorts, set equal to T here
    
    % Draw axes
    \draw[->,thick] (0,0) -- (0,\R+1) node[above] {\(r\)};
    \draw[->,thick] (0,0) -- (\T+1,0) node[right] {\(t\)};
    % Draw tick marks and labels on axes
    \foreach \x in {1,...,\T} {
        \draw (\x,-0.2) -- (\x,0.2);
        \node at (\x,0.5) {\x};
    }
    \foreach \y in {2,...,\R} {
        \draw (-0.2,\y) -- (0.2,\y);
        \node at (-0.5,\y) {\y};
    }

    % Draw grid with labeled nodes
    \foreach \r in {2,...,\R} {
        \foreach \t in {\r,...,\T} {
            % Determine label position based on whether t equals r or is greater than r
            \ifnum\t=\r
                \node[draw, circle, inner sep=2pt, fill=black, label=left:\(\hat{\tau}_{\r\t}\)] (node-\r-\t) at (\t,\r) {};
            \else
                \node[draw, circle, inner sep=2pt, fill=black, label=above:\(\hat{\tau}_{\r\t}\)] (node-\r-\t) at (\t,\r) {};
            \fi
        }
    }

    % Draw penalty lines
    % Penalties within the same cohort
    \foreach \r in {2,...,\R} {
        \foreach \t in {\r,...,\T} {
            \ifnum\t>\r
                \draw (node-\r-\t) -- (node-\r-\the\numexpr\t-1\relax);
            \fi
        }
    }

    % Penalties between initial treatment effects of each cohort
    \foreach \r in {3,...,\R} {
        \draw (node-\r-\r) -- (node-\the\numexpr\r-1\relax-\the\numexpr\r-1\relax);
    }
\end{tikzpicture}

    \caption{Visualization of which of the estimated marginal average treatment effect terms \(\hat{\tau}_{rt}\) from regression \eqref{wooldridge.6.33.model} (which estimate the average treatment effects \(\tau_{\text{ATT}} (r, t) \) from Equation \ref{att.cohort.time}) we penalize towards each other in the FETWFE penalty \eqref{fetwfe.penalty}. In this setting, \(T = 6\) and \(\mathcal{R} = \{2, \ldots, 5\}\). The horizontal axis depicts time and the vertical axis depicts cohorts. FETWFE works well under an assumption that the linked treatment effects tend to be close together, and at least some of them are exactly equal. See further details in Section \ref{sec.meth}.}
    \label{fig:fusion-penalties}
\end{figure}
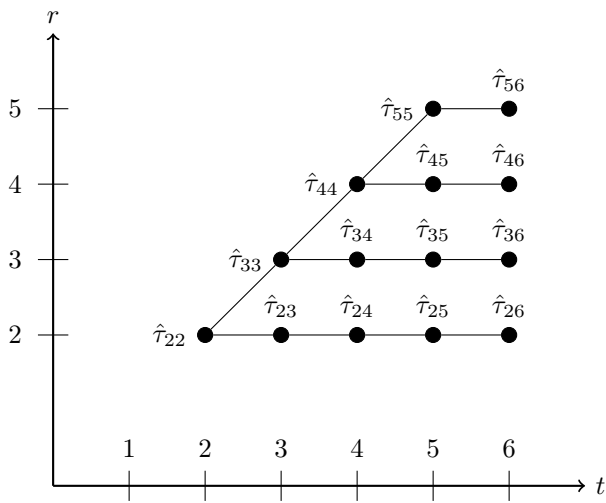

\subsection{Demonstrating the Efficacy of FETWFE}

The second goal of this paper is to provide theoretical guarantees for the performance of FETWFE. We prove that under a suitable sparsity assumption on the true coefficients and regularity conditions, FETWFE learns the correct restrictions with probability tending to one (Theorem \ref{te.sel.cons.thm}). FETWFE then estimates only the needed parameters, improving efficiency. We also prove that for \(q \in (0, 2]\) our proposed estimators of two classes of heterogeneous marginal average treatment effects are consistent under either conditional or marginal parallel trends, and two classes of conditional average treatment effect estimators are consistent under conditional parallel trends (Theorem \ref{main.te.cons.thm}). To the best of our knowledge, this is the first work establishing the consistency of a difference-in-differences estimator for marginal average treatment effects under either form of parallel trends. For estimating propensity-weighted conditional average treatment effects, the generalized propensity scores can be estimated by an arbitrary nonparametric estimator with theoretical convergence guarantees, so our theory applies to propensity score estimators as simple as multinomial logit or as flexible as deep neural networks \citep{farrell2021deep} or random forests \citep{gao2022towards}.

We also show that our treatment effect estimators are asymptotically normal for \(q \in (0, 1)\) and we provide feasible, provably consistent finite-sample variance estimators (Theorems \ref{te.asym.norm.thm} and \ref{te.asym.norm.thm.gen.cond}). For probability-weighted treatment effect estimates, we propose both an asymptotically normal test statistic if the generalized propensity scores are estimated on a separate data set from FETWFE (Theorem \ref{te.asym.norm.thm}b) and an asymptotically subgaussian test statistic if both models are estimated on the same data set (Theorem \ref{te.asym.norm.thm}c). The data-driven restrictions FETWFE learns improve parameter estimation while preserving asymptotic unbiasedness and consistency, allowing for the construction of asymptotically valid confidence intervals.

Under the assumptions of Theorem \ref{te.oracle.thm} we show that FETWFE has an oracle property \citep{fan2001variable, fan2004nonconcave}: the asymptotic variance of its estimates depend only on the parameters of the smaller restricted model, and it converges at the same \(1/\sqrt{N}\) rate as an OLS-estimated model on the true restricted model even if the number of covariates grows to infinity with \(N\). Our theoretical guarantees leverage extensions of the bridge regression theory of \citet{kock2013oracle} (Theorems \ref{prop.2i} and \ref{prop.ext.2} in the appendix) that may be of independent interest.

Finally, a third goal of this paper is to show demonstrate and facilitate the practical use of FETWFE. An open-source R package implementing FEWFE is available on the Comprehensive R Archive Network \citep{Faletto:2025aa}. We demonstrate that FETWFE is useful in practice through simulation studies that we conduct in Section \ref{synth.exps.sec} comparing the FETWFE estimator to competitor estimators. We also illustrate FETWFE in an empirical application in Section \ref{sec.data.app}.

\subsection{Related Literature}

There has been an abundance of recent research on difference-in-differences and it is not possible to cite all of the work in this stream; see \citet{de2023tworeview, sun2022linear, roth2023s}, and \citet{zimmermann2020handbookcallaway} for recent reviews. \citet{lechner2011estimation} reviews earlier literature. \citet{chiu2023and} review the literature from a political science perspective. \citet{de2023two} and \citet{arkhangelsky2023causal} review causal inference in panel data more broadly.  Beyond those reviews and the papers we have already cited, throughout the paper we cite papers in the difference-in-differences literature and relate them to FETWFE; see Sections \ref{sec.did.assumptions} and \ref{sec.caus.estimands} and Appendix \ref{par.trend.ciuu.app}, in particular.

Besides \citet{wooldridge2021two}, there are many other estimators that target similar estimands to ours, and we discuss several of these in Section \ref{sec.caus.estimands}, though none of these are machine learning estimators. It is also natural to compare FETWFE to other machine learning and non- and semi-parametric difference-in-differences estimators. Generally these methods are more flexible in the kinds of models than they can estimate than FETWFE (which is a regularized linear model, though our theoretical results allow for nonparametric estimation of the generalized propensity scores), but much less flexible in the estimands they target. Further, these methods tend to require assumptions that an estimator is available that achieves consistency at particular rates, so applying these results would involve further checking the required assumptions for a particular estimator. In contrast, we provide ``end-to-end" assumptions for the consistency, asymptotic normality, and oracle property of three classes of estimators \eqref{catt.estimand.fixed.est}, \eqref{att.estimator.fixed}, and \eqref{att.estimator.weighted}.

The parametric and semi-parametric estimators proposed by \citet{zimmert2020efficient}, \citet{sant2020doubly}, and \citet{chang2020double} yield theoretical results for a wide variety of machine learning estimators with theoretical guarantees under conditional parallel trends assumptions, but only in the \(T = 2\) setting and only for a single average treatment effect estimand (specifically, our estimand in Equation \ref{att.def}). \citet{blackwell2024estimating} accomplish similar goals with different estimands under different assumptions, focusing mainly on a setting with \(T  = 2\) and a single mediator of interest. \citet{nie2021nonparametric} propose related non-parametric estimators that achieve \(1/\sqrt{N}\) convergence and asymptotic normality in the \(T = 2\) repeated cross-sections setting only for an average treatment effect, and consistency for conditional average treatment effects at much slower rates, in contrast to the \(1/\sqrt{N}\) rate for conditional average treatment effects estimated by FETWFE (see Theorem  \ref{te.asym.norm.thm.gen.cond}). \citet{hatamyar2023machine} provide similar semi-parametric estimators for both marginal and conditional average treatment effects in the staggered adoptions setting with arbitrary \(T\), but they do not provide any theoretical guarantees.

\citet{gavrilova2023dynamic} propose an approach using random forests that is somewhat more flexible in its estimands. Their approach allows for arbitrary \(T\), though they focus on the case with only one treated cohort that starts an absorbing treatment at time \(t= 2\), providing a brief sketch as to how their method might be extended to allow for staggered adoptions. Further, their method requires splitting the data and estimating completely separate models to estimate the treatment effect of the treated cohort at each time \(2, 3, \ldots, T\), while FETWFE allows for the full sample to be used to jointly estimate all treatment effects simultaneously. Indeed, a central idea of our method is that the fusion penalties we use allow FETWFE to borrow strength across cohorts and time in order to improve estimation efficiency. (None of the above works mentions fusion penalties.)

Our work also contributes to the much broader topic of estimating conditional average treatment effects. In Section \ref{sec.caus.estimands} we mention some prominent works in this stream of literature.

Outside of causal inference, our approach is clearly closely related to both the fused lasso \citep{tibshirani2005sparsity} and the generalized lasso \citep{tibshirani2011solution}; see \eqref{opt.prob}. Our work also draws much from the bridge regression theory of \citet{kock2013oracle}, and, in turn, \citet{Huang2008}. \citet{kaddoura2022estimation} propose a linear panel data model with a fused lasso-like penalty, though not for causal inference. \citet{park2023change} propose a related Bayesian approach. \citet{heiler2021shrinkage} also propose a penalized panel data model that combines categories, but not in the causal inference setting. Similarly, \citet{kwon2021optimal} proposes a shrinkage estimator for linear panel data models (though not for causal inference) and cites more related work in this stream.

FETWFE also relates more specifically to the generalized fused lasso \citep{hofling2010coordinate, xin2014efficient}. Like FETWFE, the generalized fused lasso involves applying \(\ell_q\) penalties to the absolute differences between a variety of pairs of coefficients, though the generalized fused lasso only uses the \(q =1\) case. Some works that prove theoretical results for the generalized fused lasso include \citet{viallon2013adaptive} and \citet{viallon2016robustness}. \citet{viallon2013adaptive} prove consistency for the generalized fused lasso under different assumptions from ours, including a restriction to the fixed \(p\) case. They prove selection consistency and asymptotic normality only for an \textit{adaptive} \citep{adalasso} version of their method. \citet{viallon2016robustness} prove similar results (with similar limitations) for a broader class of generalized linear models.

\subsection{Organization of the Paper}

After we introduce some notation, we establish our setup and basic assumptions in Section \ref{setup.sec} and present our difference-in-differences assumptions in Section \ref{sec.did.assumptions}. In particular, we provide a testable sufficient condition (CIUN) for several kinds of estimators, including FETWFE, to be consistent under either marginal or conditional parallel trends. In Section \ref{sec.caus.estimands} we introduce causal estimands that FETWFE can be used to estimate, and in Section \ref{sec.meth} we propose the fused extended two-way fixed effects model, as well as our estimators for our proposed causal estimands. We present our theoretical guarantees for FETWFE in Section \ref{sec.theory}. In Section \ref{synth.exps.sec} we demonstrate the superior performance of FETWFE in our setting against extended two-way fixed effects (ETWFE), bridge-penalized ETWFE, and a slightly more flexible model than \eqref{bad.reg} in simulation studies. We illustrate the usefulness of FETWFE in an empirical application using data from \citet{stevenson2006bargaining} in Section \ref{sec.data.app}. Finally, Section \ref{conc.sec} concludes.

\subsection{Notation}

Throughout we will use subscripts like \(\tilde{y}_{(it)}\) to refer to the \(i^\text{th}\) unit at the \(t^{\text{th}}\) time, \(\boldsymbol{\tilde{y}}_{(i \cdot)} \in \mathbb{R}^T\) to refer to all response observations corresponding to unit \(i\), and so on. For matrices, \(\boldsymbol{Z}_{(it) \cdot}\) refers to the row corresponding to unit \(i\) at time \(t\), \(\boldsymbol{Z}_{(i\cdot) j}\) refers to the \(T\) observations of the \(j^\text{th}\) feature in unit \(i\), and so on. We will assume that unit-specific covariates are time invariant, and it will be more convenient to refer to the covariates for unit \(i\) with the notation \(\boldsymbol{X}_i\).

Recall that for a sequence of random variables \(\{X_N\}\) and a sequence of positive real numbers \(\{a_N\}\), we say that \(X_N = \mathcal{O}_{\mathbb{P}}(a_N)\) if for every \(\epsilon > 0\) there exists a finite \(M_\epsilon > 0\) and \(N_\epsilon > 0\) such that for all \(N > N_\epsilon\),
\[
\mathbb{P} \left( \left| \frac{X_N}{a_N} \right| > M_\epsilon \right) < \epsilon
.
\]
If \(a_N \to 0\), \(X_N = \mathcal{O}_{\mathbb{P}}(a_N)\) is sufficient for \(X_N \xrightarrow{p} 0\), and \(a_N\) can be interpreted as the rate of convergence \citep[Chapter 1]{dasgupta2008asymptotic}.

%by the following argument: it is enough to show that for every \(\delta > 0\) and \(\epsilon > 0\), there exists an \(N^*\) such that \(N > N^* \) implies \( \mathbb{P} \left( |X_N | >  \delta  \right) < \epsilon\). Fix \(\epsilon > 0\). Since \(a_N \to 0\), for any \(\delta > 0\) there exists an \(N_\delta\) such that \(|a_N| <  \delta /M_\epsilon\) for all \(N > N_\delta\). Then for all \(N > \max\{N_\delta, N_\epsilon\}\),
%\[
%\mathbb{P} \left( |X_N | > \delta \right)  = \mathbb{P} \left( \left| \frac{X_n}{a_N} \right| > \frac{\delta}{ |a_N|} \right) <  \mathbb{P} \left( \left| \frac{X_n}{a_N} \right| > M_\epsilon \right) < \epsilon 
%.
%\]

\section{Basic Setup and Assumptions}\label{setup.sec}

We will expand on the setup from the introduction. For each \(i \in [N]\), assume
\[
\boldsymbol{\tilde{y}}_{(i\cdot)} =  \mathbb{E} [ \boldsymbol{\tilde{y}}_{(i\cdot)} \mid W_i, \boldsymbol{X}_{i} ] + c_i \boldsymbol{1}_T + \boldsymbol{u}_{(i\cdot)}
,
\]
where \(c_i \in \mathbb{R}\) and \(\boldsymbol{u}_{(i\cdot)} \in \mathbb{R}^T\) are both random variables and \(\boldsymbol{1}_T\) is a \(T\)-vector with 1 in every entry.

\textbf{Assumption (F1):} \(\E [ c_i \mid W_i, \boldsymbol{X}_{i} ] = 0\), \( \E [\boldsymbol{u}_{(i\cdot)}   \mid c_i, W_i, \boldsymbol{X}_{i} ] =  \boldsymbol{0},\) \( \Var ( c_i \mid W_i, \boldsymbol{X}_{i} ) =  \sigma_c^2\), and \( \Var ( \boldsymbol{u}_{(i\cdot)}   \mid c_i, W_i, \boldsymbol{X}_{i} ) =  \sigma^2 \boldsymbol{I}_T\) almost surely, where \(\sigma_c^2\) and \(\sigma^2\) are assumed to be known and finite and \(\boldsymbol{I}_T \in \mathbb{R}^{T \times T}\) is the identity matrix.

With Assumption (F1) we can use a generalized least squares transform to create a data set with conditionally uncorrelated rows. Let \(\boldsymbol{\epsilon}_{(i\cdot)} := c_i \boldsymbol{1}_T + \boldsymbol{u}_{(i\cdot)}\), and observe that for every \(i \in [N]\) almost surely \(\Var ( \boldsymbol{\epsilon}_{(i\cdot)}  \mid W_i, \boldsymbol{X}_{i} )= \sigma^2 \boldsymbol{I}_T + \sigma_c^2 \boldsymbol{1}_T \boldsymbol{1}_T^\top  =:  \boldsymbol{\Omega} \in \mathbb{R}^{T \times T}\). Under (F1), for every \(i \in [N]\) almost surely \(\E [ \sigma \boldsymbol{\Omega}^{-1/2} \boldsymbol{\epsilon}_{(i\cdot)}  \mid W_i, \boldsymbol{X}_{i} ] = \boldsymbol{0}\) and
\begin{equation}\label{gls.transform.id}
\Var ( \sigma \boldsymbol{\Omega}^{-1/2} \boldsymbol{\epsilon}_{(i\cdot)}  \mid W_i, \boldsymbol{X}_{i} ) =   \sigma^2 \boldsymbol{\Omega}^{-1/2}   \Var ( \boldsymbol{\epsilon}_{(i\cdot)}  \mid W_i, \boldsymbol{X}_{i} ) \boldsymbol{\Omega}^{-1/2}  = \sigma^2 \boldsymbol{I}_{T}
.
\end{equation}
We can carry out this transformation for any known finite and symmetric positive definite \(\boldsymbol{\Omega} = \Var ( \boldsymbol{\epsilon}_{(i\cdot)}  \mid W_i, \boldsymbol{X}_{i} ) \). All of our results can be easily adapted to apply for any such \(\boldsymbol{\Omega}\), though we focus on the specified setting for concreteness.

Next we make an assumption on the joint distribution of the potential outcomes, treatment, and covariates. Let
\begin{equation}\label{cond.prob.cohort.def}
\pi_r(\boldsymbol{x}) := \mathbb{P} ( W_i = r \mid  \boldsymbol{X}_{i} = \boldsymbol{x} ),  \qquad r \in \{0\} \cup \mathcal{R}
\end{equation}
be the conditional probability of a unit's treatment status, a propensity score.

\textbf{Assumption (F2):} The \(N\) joint random variables \((W_i, \boldsymbol{X}_{i}, c_i, \boldsymbol{u}_{(i\cdot)})_{i=1}^N\) are independent and identically distributed (iid). Further, assume that for all \(r \in  \{0\} \cup \mathcal{R}\), almost surely \(0 < \pi_r(\boldsymbol{X}_i) < 1\).

Notice that for each \(\boldsymbol{X}_i\) all of the potential outcomes \(\tilde{y}_{it}(r)\) are well-defined under Assumption (F2) since treatment probabilities are always strictly between 0 and 1. So far we have made no assumptions on the relationship between the treatment assignments and the potential outcomes for a given unit, but we do so in the next section.

\section{Difference-in-Differences Assumptions}\label{sec.did.assumptions}

We will use the generalizations of the no anticipation and common trends assumptions described in the introduction---Assumptions (CNAS) and (CCTS), respectively---that were proposed by \citet{wooldridge2021two}, in addition to some variations and other assumptions. 

\textbf{Assumption (CNAS)} (``conditional no-anticipation with staggered interventions"): For treatment cohorts \(r \in \mathcal{R} \), almost surely
\[
\E [  \tilde{y}_{(it)}(r) - \tilde{y}_{(it)}(0) \mid W_i = r, \boldsymbol{X}_{i} ] = 0 \qquad \forall i \in [N],  t < r.
\]
This requires that before treatment, in conditional expectation each individual's expected outcome is identical to what it would have been if the unit were never treated. Some works (two recent examples are \citet{caetano2022difference} and \citealt{ghanem2023selection}) define the potential outcomes of the units at each time in terms of their treatment status at that time rather than in terms of their cohort membership. In our notation, this amounts to assuming that \(\tilde{y}_{(i t)}(r) = \tilde{y}_{(i t)}(0)\) almost surely for \( t < r\). As \citet[Section 6.3]{wooldridge2021two} points out, this stronger assumption is sufficient for (CNAS). Such a framework may not be troubling, but we only need the weaker assumption (CNAS). See further discussion in \citet[Sections 5.3, 6.1, 6.3]{wooldridge2021two}.

\subsection{Parallel Trends Assumptions}\label{par.trends.sec}

We will discuss and use a variety of parallel or common trends assumptions.

\textbf{Assumption (CCTS)} (``conditional common trends with staggered interventions"): For all \(r \in \mathcal{R}\), \( i \in [N], t \in \{2, \ldots, T\}\), almost surely
\begin{equation}\label{ccts.def}
\E [ \tilde{y}_{(it)} (0) - \tilde{y}_{(i1)}(0) \mid  W_i = r, \boldsymbol{X}_{i} ]  = \E [ \tilde{y}_{(it)} (0) - \tilde{y}_{(i1)}(0) \mid W_i = 0, \boldsymbol{X}_{i} ] 
.
\end{equation}

Assumption (CCTS) requires that each unit's untreated potential outcomes \(\tilde{y}_{(it)} (0)\) follow the same trend over time in conditional expectation regardless of when or if they are treated. Assumption (CCTS) is used by \citet[Assumption 5]{zimmermann2020handbookcallaway}, and also by \citet[Assumption 1]{gavrilova2023dynamic} in the \(\mathcal{R} = \{2\}\) case. Assumption (CCTS) has become widely adopted in the \(T = 2\) setting; early examples include \citet[Assumption A-5]{heckman1997matching}, \citet[Equation 12]{heckman1998characterizing}, and \citet[Assumption 3.1]{abadie2005semiparametric}, and some more recent examples include \citet[Assumption 2]{sant2020doubly},  \citet[Assumption 6(iv)]{zimmert2020efficient}, \citet[Assumption 2.1]{chang2020double}, \citet[Assumption 6]{roth2023s}, and \citet[Assumption PT-X]{ghanem2023selection}.

We also consider versions of (CCTS) that only apply before or after treatment begins for each cohort.

\textbf{Assumption (CCTSB)} (``conditional common trends before treatment"): \eqref{ccts.def} holds almost surely for all \(r \in \mathcal{R} \setminus \{2\}\), \( i \in [N]\), and \(t \in \{2, \ldots, r - 1\}\).

\textbf{Assumption (CCTSA)} (``conditional common trends after treatment"): \eqref{ccts.def} holds almost surely for all \(r \in \mathcal{R}\), \( i \in [N]\), and \(t \in \{r, \ldots, T\}\).

Assumption (CCTSB) requires that (CCTS) holds only until the period before each unit is treated, while (CCTSA) requires (CCTS) to hold after treatment begins. Many previous works have distinguished between the parallel trends before treatment versus after treatment, noting that (CCTSB) can be directly tested under assumption (CNAS) while (CCTSA) cannot, since untreated potential outcomes are not observed for units after treatment, and noting that parallel trends holding before treatment is no guarantee of holding after treatment. See \citet{bilinski2018seeking, kahn2020promise,  dette2020difference, sun2021estimating, ban2022generalized, callaway2022pre, roth2022pretest}, \citet[Equation 22]{henderson2022complete}, \citet{borusyak2024revisiting,  rambachan2023more}, and \citet[Section 3.3]{zimmermann2020handbookcallaway}. See \citet[Section 7.1, Equation 7.4]{wooldridge2021two} for an example of how (CCTSB) can be tested under (CNAS).

We can contrast (CCTS) with an unconditional or marginal version, which \citet{wooldridge2021two} calls Assumption (CTS).

\textbf{Assumption (CTS)} (``common trends with staggered interventions"): Almost surely for all \(r \in \mathcal{R}\)
\begin{equation}\label{cts}
\E [ \tilde{y}_{(it)} (0) - \tilde{y}_{(i1)}(0) \mid  W_i = r  ]  = \E [ \tilde{y}_{(it)} (0) - \tilde{y}_{(i1)}(0) \mid W_i = 0 ] ,  \qquad  i \in [N], t \in \{2, \ldots, T\}
.
\end{equation}

We will also refer to a version of (CTS) that holds only after treatment:

\textbf{Assumption (CTSA)} (``common trends with staggered interventions after treatment"): \eqref{cts} holds almost surely for all \( i \in [N]\), \(r \in \mathcal{R}\), and \( t \in \{r, \ldots, T\}\).

Because (CTS) does not account for covariates, (CCTS) is often thought of as more plausible than (CTS). We generally agree, though in Appendix \ref{thm.discus.sec} we add some nuance by pointing out that (CTS) can hold when (CCTS) does not hold (Theorem \ref{unconf.ccts.cts.thm}). In such a setting it is not possible to estimate conditional average treatment effects using regression \eqref{wooldridge.6.33.model}; we show this formally in Theorem \ref{te.interp.prop}(a) in the appendix. Further, it turns out that even marginal average treatment effect estimates from models like FETWFE that rely on (CCTS) will in general be inconsistent. 

However, (CTS) holding is enough to estimate marginal average treatment effects consistently using many estimators, so it is reasonable to wonder whether regression \eqref{wooldridge.6.33.model} might be able to estimate treatment effects consistently under (CTS) and some additional assumptions. The following assumption will turn out to allow FETWFE to consistently estimate marginal average treatment effects if (CTS) holds, even if (CCTS) does not.

\textbf{Assumption (CIUN)} (``conditional independence of trends of the untreated potential outcomes for never-treated units"): Almost surely for all \(i \in [N]\) and any \(t \in \{2, \ldots, T\}\),
\[
 \E \left[ \tilde{y}_{(i t)}(0) - \tilde{y}_{(i1)}(0) \mid W_i = 0 \right]     =   \E \left[ \tilde{y}_{(i t)}(0) - \tilde{y}_{(i1)}(0) \mid W_i = 0, \boldsymbol{X}_{i}   \right]  
 .
\]

Assumption (CIUN) requires that the trends in the untreated potential outcomes are mean-independent of the covariates conditional on treatment status for units that are untreated. Since (CIUN) depends only on the observed potential outcomes, it can be tested, as we discuss in the next section.

\subsection{Linearity Assumption (Assumption LINS)}\label{lins.sec}

We present an assumption that is slightly different than one introduced by \citet{wooldridge2021two} on the functional forms of the conditional distributions of the untreated potential outcomes, the trends \(\tilde{y}_{(i t)}(0) - \tilde{y}_{(i1)}(0) \), and the difference-in-differences estimands.

\textbf{Assumption (LINS):} For each \(N \in \mathbb{N}\), assume we observe \(d_N\) covariates, where \(d_N\) may be fixed or increasing in \(N\). There exist fixed parameters \(\eta^* \in \mathbb{R}, \{\nu_r^*\}_{r \in \mathcal{R}} \in \mathbb{R}\), \(\{\gamma_t^*\}_{t \in \{2, \ldots, T\}} \in \mathbb{R}\), and \(\{ \tau_{rt}^* \}_{r \in \mathcal{R}, t \geq r}\in \mathbb{R} \), along with fixed sequences of parameters \(\{\boldsymbol{\kappa}_N^*\}_{N \in \mathbb{N}} \in \mathbb{R}^{d_N}, \{ \boldsymbol{\zeta}_{rN}^*\}_{r \in \mathcal{R}, N \in \mathbb{N}} \in \mathbb{R}^{d_N}, \{\boldsymbol{\xi}_{tN}^*\}_{t \in \{2, \ldots, T\}, N \in \mathbb{N}} \in \mathbb{R}^{d_N}\), and \( \{ \boldsymbol{\rho}_{rtN}^* \}_{r \in \mathcal{R}, t \geq r, N \in \mathbb{N}} \in \mathbb{R}^{d_N}\) such that for each \(N \in \mathbb{N}\), for all \(i \in [N]\) and any fixed \(\boldsymbol{x} \in \mathbb{R}^{d_N}\) in the support of the random vector \(\boldsymbol{X}_i\) it holds that
\begin{align}
\E \left[\tilde{y}_{(i1)}(0) \mid W_i = 0, \boldsymbol{X}_{i} = \boldsymbol{x} \right] =  ~ & \eta^*  + \boldsymbol{x}^\top \boldsymbol{\kappa}_N^*  ,  \label{t.1.untreated}
\\ \E \left[\tilde{y}_{(i1)}(0) \mid W_i = r, \boldsymbol{X}_{i} = \boldsymbol{x} \right]   
\qquad & \nonumber
\\
 - \E \left[\tilde{y}_{(i1)}(0) \mid W_i = 0, \boldsymbol{X}_{i} = \boldsymbol{x} \right]  =  ~ &  \nu_r^* + \boldsymbol{x}^\top \boldsymbol{\zeta}_{rN}^*
\qquad \forall r \in \mathcal{R},  \label{t.1.treated}
\\ \E \left[ \tilde{y}_{(i t)}(0) - \tilde{y}_{(i1)}(0) \mid W_i = 0, \boldsymbol{X}_{i} = \boldsymbol{x} \right] = ~ & \gamma_t^* + \boldsymbol{x}^\top \boldsymbol{\xi}_{tN}^*, \qquad  t \in \{2, \ldots, T\}, \qquad \text{and} \label{trend.params}
\\ 
 \E\left[ \tilde{y}_{(i t)}(r)  - \tilde{y}_{(i 1)}(r) \mid W_i = r, \boldsymbol{X}_i  = \boldsymbol{x}\right]  \qquad &   \nonumber
\\ - \E \left[  \tilde{y}_{(it)}(0) - \tilde{y}_{(i1)}(0)  \mid  W_i = 0, \boldsymbol{X}_{i} = \boldsymbol{x} \right]   = ~ & \tau_{rt}^* + \boldsymbol{\dot{x}}_r^\top \boldsymbol{\rho}_{rtN}^*, \qquad r \in \mathcal{R}, t \geq r  \label{treat.eff.def.covs}
,
\end{align}
where
\[
\boldsymbol{\dot{x}}_r := \boldsymbol{x} - \E \left[ \boldsymbol{X}_{i} \mid  W_i = r \right]
\]
are covariate vectors that are centered with respect to their cohort means. (For conciseness, hereafter we omit \(N\) from the subscripts of the parameters defined here that may vary with \(N\).)

The version of this assumption presented by \citet{wooldridge2021two} uses a slightly different expression for \eqref{treat.eff.def.covs} and requires \(d_N = d\) to be fixed, since ETWFE is inconsistent in the high-dimensional setting we consider. Notice the similarity of the quantity on the left side of \eqref{treat.eff.def.covs} to the canonical two-period difference-in-differences estimand \eqref{did.estimand}. \citet{wooldridge2021two} shows that Assumption (LINS), along with Assumptions (CNAS) and (CCTS), make \eqref{wooldridge.6.33.model} correctly specified for the potential outcomes and the the left side of \eqref{treat.eff.def.covs} equal to the conditional average treatment effects. See also our derivation \eqref{stagger.tau.deriv} and Theorem \ref{first.thm.fetwfe}(a) in the appendix.

\begin{remark}\label{eff.rmk}

%We also point out the connection between \eqref{t.1.treated} and (CCTS). 
We can gain insight on (CCTS) by re-arranging and using the conditional independence from \(W_i\) to arrive at the following equivalent condition:
\begin{align}
& \E [ \tilde{y}_{(it)} (0)  \mid  W_i = r, \boldsymbol{X}_{i} ] -  \E [ \tilde{y}_{(it)} (0) \mid  W_i = 0, \boldsymbol{X}_{i} ]   \nonumber
\\ = ~ &  \E [\tilde{y}_{(i1)}(0) \mid  W_i = r, \boldsymbol{X}_{i} ]  - \E [  \tilde{y}_{(i1)}(0) \mid W_i = 0,  \boldsymbol{X}_{i} ]  ,  \qquad r \in \mathcal{R}, i \in [N], t \in \{2, \ldots, T\}. \label{ccts.unconf}
\end{align}
(This is similar to a framing used by, for example, \citet[Equations 10 and 12]{heckman1998characterizing}, \citet[Equation 3.2]{lechner2011estimation}, \citet[Assumption 1]{henderson2022complete}, \citet[Section 2.1]{ban2022generalized},  and \citet[Section 2.2]{park2023universal} in the \(T = 2\) setting.) Under (CCTS) and (LINS), \eqref{t.1.treated} quantifies the time-invariant difference in conditionally expected potential outcomes between units in each cohort and units in the never-treated cohort that appears on each side of \eqref{ccts.unconf}. The inclusion of the terms in \eqref{t.1.treated} provides robustness against time-invariant confounding (see Proposition \ref{unconf.ccts.cts.thm.app}(a) in the appendix).

On the other hand, under unconfoundedness of the untreated potential outcomes (which we define formally and discuss in more detail in \eqref{unconfound.def} in the appendix), \eqref{t.1.treated} (and each side of Equation \ref{ccts.unconf}) is identically 0. Although including these covariates in a regression does not harm unbiasedness under unconfoundedness, it reduces efficiency. Fortunately, we will see later that under the assumptions of Theorem \ref{te.oracle.thm}, if the estimated model contains any coefficients that equal 0 exactly, FETWFE successfully screens out these covariates and achieves the same asymptotic efficiency as a model that omitted them. In this way FETWFE provides robustness without compromising asymptotic efficiency under unconfoundedness, avoiding the kind of trade-off between robustness and asymptotic efficiency that \citet{zimmert2020efficient} discusses.

Many previous works have also noted that difference-in-differences estimation allows for selection bias and can effectively control for some kinds of unobserved time-invariant confounding. \citet{imai2019should} explore this idea in depth, and also raise interesting issues surrounding the plausibility of an assumption that time-varying covariates do not affect the observed response and the implications when this assumption is violated. \citet{imai2023matching} propose a matching estimator in a setting with time-varying covariates. For more in-depth investigations into the parallel trends assumption, see \citet{ghanem2023selection}, \citet{roth2023parallel}, and \citet{marx2024parallel}. See also our further discussion in Appendix \ref{par.trend.ciuu.app}.

\end{remark}

Under Assumption (CIUN) and (LINS), the covariates \(\boldsymbol{\xi}_t^*\) all equal \(\boldsymbol{0}\). So a test of this null hypothesis (for example, a Wald test) serves as a test for the null hypothesis that both (CIUN) and (LINS) hold.

\section{Causal Estimands}\label{sec.caus.estimands}

Next we will discuss several classes of causal estimands we hope to estimate using FETWFE. The building block of our estimands will be the conditional average treatment effects on the treated units \(\tau_{\text{CATT}} (r, t, \boldsymbol{x}) \) from \eqref{catt.t.r.def}, which matches definition (6.23) in \citet{wooldridge2021two}. Estimating \( \tau_{\text{CATT}} (r, t, \boldsymbol{x}) \) is useful if we are interested in identifying which units are most affected by the treatment at a specific time. Besides economics \citep{florens2008identification, djebbari2008heterogeneous}, domains where conditional average treatment effect estimation have been found useful include personalized medicine \citep{feinstein1972estimating, lesko2007personalized,frankovich2011evidence,  foster2011subgroup}, education \citep{morgan2001counterfactuals, brand2010benefits, murphy2016handbook}, targeted marketing \citep{bottou2013counterfactual, athey-yelp-2018, hitsch2018heterogeneous, ascarza2018retention, yang2020targeting, ellickson2022estimating}, psychology \citep{bolger2019causal, winkelbeiner2019evaluation}, sociology \citep{xie2012estimating, breen2015heterogeneous}, and political science \citep{imai2011estimation, grimmer2017estimating}.

Many methods for estimating conditional average (or \textit{heterogeneous}\footnote{Within the differences-in-differences literature, the descriptor ``heterogeneous" has often been used in reference to treatment effects that vary across cohorts or across time; e.g., \citet{de2020two, sun2021estimating, wooldridge2021two}. Here we refer to heterogeneity in treatment effects across the distribution of \(\boldsymbol{X}_i\), though our model also accommodates heterogeneity in treatment effects across cohorts and time. To avoid ambiguity, we will generally use the term ``conditional average treatment effects."}) treatment effects  have been developed within econometrics over the last decade; see \citet{abrevaya2015estimating,athey_imbens_2016,  Wager2018, nie2021quasi, semenova2021debiased}, and \citet{fan2022estimation} for a few examples.

\subsection{Conditional Average Treatment Effect Estimands}

We can use \eqref{catt.t.r.def} to construct a variety of causal estimands. For any set of finite constants \(\{\psi_{rt}\}_{r \in \mathcal{R}, t \in \{r, \ldots, T\}}\), consider the estimand
\begin{equation}\label{catt.estimand.fixed}
\sum_{r \in \mathcal{R}} \sum_{t=r}^T \psi_{rt}  \tau_{\text{CATT}} (r, t, \boldsymbol{x}) 
,
\end{equation}
a weighted average of the cohort- and time-specific conditional average treatment effects. This is similar to the framework of estimation targets of \citet{borusyak2024revisiting} and \citet[Equation 3.1]{callaway2021difference}, though they study average treatment effects (like our estimand \eqref{att.estimand.fixed}, to come) rather than conditional average treatment effects. Refer to these papers for many examples of choices of \(\{\psi_{rt}\}\). Similar aggregations are also discussed by \citet[Remark 3.2]{RePEc:qed:wpaper:1495} and \citet[Section 6.5]{wooldridge2021two}.

To give one example among many possible estimands of interest using specific \(\{\psi_{rt}\}\), we could choose for some \(r' \in \mathcal{R}\)
\begin{equation}\label{cohort.avg.psi.r.t}
\psi_{rt} = \begin{cases}
  \frac{1}{T-r'+1} , &  r = r',
\\  0, & \text{otherwise}
  \end{cases}
\end{equation}
to define the conditional average treatment effect across time for a unit in cohort \(r'\) with covariates \(\boldsymbol{x}\):
\begin{align}
\tau_{\text{CATT}} (r', \boldsymbol{x}) :=   ~ & \frac{1}{T-r'+1} \sum_{t=r'}^T \tau_{\text{CATT}} (r', \boldsymbol{x}, t) 
 \label{catt.r.def}
\\ =  ~ &   \frac{1}{T-r'+1} \sum_{t=r'}^T \E [ \tilde{y}_{(i t)}(r') - \tilde{y}_{(i t)}(0) \mid W_i = r', \boldsymbol{X}_{i} = \boldsymbol{x} ]  \nonumber
.
\end{align}
A version of \eqref{catt.r.def} for average (rather than conditional average) treatment effects was considered by \citet[Equation 6.41]{wooldridge2021two}, and a slightly more general form (also accommodated by our framework) was considered by \citet[Equation 12]{goodman2021difference}.

\subsubsection{Propensity-Weighted Conditional Average Treatment Effects}
A second class of causal estimands can be generated by weighting each \(\tau_{\text{CATT}} (r, t, \boldsymbol{x}) \) by the probability a unit is in cohort \(r\) given that they are treated conditional on \(\boldsymbol{x}\),
\begin{equation}\label{cond.prob.cohort.def.treat}
\tilde{\pi}_r(\boldsymbol{x}) := \mathbb{P} ( W_i = r \mid W_i \neq 0, \boldsymbol{X}_{i} = \boldsymbol{x} ) =  \left( \sum_{r' \in \mathcal{R}} \pi_{r'}(\boldsymbol{x})  \right)^{-1} \pi_r(\boldsymbol{x}) \qquad \forall r \in \mathcal{R}
.
\end{equation}
In particular, for any set of finite constants \(\{\psi_{rt}\}\), consider the estimands
\begin{equation}\label{catt.estimand.weighted}
 \sum_{r \in \mathcal{R}} \sum_{t =r}^T \psi_{rt}     \tilde{\pi}_r (\boldsymbol{x}) \tau_{\text{CATT}} (r, t, \boldsymbol{x}) 
 .
\end{equation}
One example of an estimand we can form in this class arises from choosing
\begin{equation}\label{avg.cohorts.psi.r.t}
\psi_{rt} :=    \begin{cases}
0, & r  = 0,
\\ \frac{1}{T-r+1} , & \text{otherwise,}
\end{cases}
\end{equation}
which yields a conditional average treatment effect by marginalizing \eqref{catt.r.def} across treatment status:
%One other conditional average treatment effect we can estimate is 
\begin{align}
\tau_{\text{CATT}} ( \boldsymbol{x} ) :=    ~ & \sum_{r \in \mathcal{R}}  \sum_{t=r}^T  \frac{1}{T-r+1}  \tilde{\pi}_r(\boldsymbol{x})  \tau_{\text{CATT}} (r, t, \boldsymbol{x})   \label{catt.def}
\\  = ~ &     \sum_{r \in \mathcal{R}} \tilde{\pi}_r(\boldsymbol{x}) \tau_{\text{CATT}} (r, \boldsymbol{x})  \nonumber
\\ = ~ &  \E \left[  \tau_{\text{CATT}} ( W_i, \boldsymbol{X}_{i})  \mid W_i \neq 0, \boldsymbol{X}_{i}  =  \boldsymbol{x} \right] \nonumber
.
\end{align}
This estimand is closely related to many previously considered heterogenous treatment effects; one of many works that considers a similar estimand (outside of the specific context of difference-in-differences) is \citet[Equation 1]{Wager2018}. 

\subsection{Marginal Average Treatment Effects}

Outside of conditional average treatment effects, we can also marginalize these conditional treatment effects over the distribution of \(\boldsymbol{X}_i\) to get average treatment effects. Since these estimands are "coarser," it may be possible to get more precise estimates of these estimands given a fixed sample size.

For \(\tau_{\text{ATT}} (r, t)  =  \mathbb{E} \left[ \tau_{\text{CATT}} ( r, t, \boldsymbol{X}_i)  \mid W_i = r \right]  \) from \eqref{att.cohort.time}, consider the classes of estimands
\begin{equation}\label{att.estimand.fixed}
\sum_{r \in \mathcal{R}} \sum_{t=r}^T \psi_{rt}  \tau_{\text{ATT}} (r,  t ) ) 
\end{equation}
and
\begin{equation}\label{att.estimand.weighted}
 \sum_{r \in \mathcal{R}} \sum_{t =r}^T \psi_{rt}     \tilde{\pi}_r \tau_{\text{ATT}} (r, t) 
 ,
\end{equation}
where \(\tilde{\pi}_r := \mathbb{P}(W_i = r \mid W_i \neq 0) = \E[ \tilde{\pi}_r (\boldsymbol{X}_{i}) \mid W_i \neq 0 ] \) is the marginal probability of a unit being in cohort \(r\) given that it is treated. 

Estimand \eqref{att.cohort.time} was used by \citet{callaway2021difference} and \citet[Equation 1]{hatamyar2023machine}, and it matches the estimand from \citet[Equation 6.2]{wooldridge2021two} and \citet[Equation 3]{RePEc:qed:wpaper:1495}. It is also closely related to the \textit{cohort-specific average treatment effect on the treated} (3) from \citet{sun2021estimating} (though we use the acronym CATT to refer to conditional, not cohort, effects) and estimands from \citet{goodman2021difference}. 

As we mentioned earlier, the class of estimands \eqref{att.estimand.fixed} is the same framework of estimation targets considered by \citet{borusyak2024revisiting} and is similar to aggregations considered by \citet[Equation 3.1]{callaway2021difference}, \citet[Section 6.5]{wooldridge2021two}, \citet[Remark 3.2]{RePEc:qed:wpaper:1495}, and \citet[Equation 18]{hatamyar2023machine}.

The probability-weighted class of estimands \eqref{att.estimand.weighted} bears a resemblance to the probability-weighted estimands proposed by \citet{callaway2021difference}. Theorem \ref{main.te.cons.thm.gen} in the appendix extends our consistency result to accommodate classes of estimands with much broader forms of estimated probability weights than \eqref{catt.estimand.weighted} and \eqref{att.estimand.weighted}, like estimand (26) of \citet{sun2021estimating} or the broad class of estimands used by \citet{callaway2021difference}. Theorems \ref{te.asym.norm.thm.gen} and \ref{te.asym.norm.thm.gen.cond} extend our asymptotic normality results to these more general estimators as well.

As a particular choice of \(\{\psi_{rt}\}\), we will again highlight \eqref{cohort.avg.psi.r.t} for any \(r' \in \mathcal{R}\), which for estimand \eqref{att.estimand.fixed} marginalizes \(\tau_{\text{CATT}} ( r', \boldsymbol{X}_{i}) \) over the distribution of \(\boldsymbol{X}_i\):
\begin{align}
 \tau_{\text{ATT}} (r')  := ~ &   \frac{1}{T-r'+1} \sum_{t=r'}^T \tau_{\text{ATT}} (r', t)  \label{att.cohort}
\\ = ~ &   \frac{1}{T-r'+1} \sum_{t=r'}^T  \E \left[  \tilde{y}_{(i t)}(r') - \tilde{y}_{(i t)}(0) \mid W_i = r'  \right]  \nonumber
\\ = ~ & \E \left[ \tau_{\text{CATT}} (r', \boldsymbol{X}_{i}) \mid W_i = r'  \right] \nonumber
.
\end{align}
Estimand \eqref{att.cohort} is identical to estimand (3.7) in \citet{callaway2021difference}, and is also the estimand of estimator (6.41) in \citet{wooldridge2021two}.

\subsubsection{The Average Effect of the Treatment on Treated Units}\label{att.subsec}

Lastly, one of the most basic causal parameters of interest is the average effect of treatment on the treated group (ATT). By using the \(\psi_{rt}\) from \eqref{avg.cohorts.psi.r.t} with \eqref{att.estimand.weighted}, we can express this estimand as
\begin{align}
\tau_{\text{ATT}}  := ~ &  
% \mathbb{E} \left[  \tau_{\text{CATT}} ( W_i, \boldsymbol{X}_{i}) \mid W_i \neq 0  \right] =   
\mathbb{E} \left[\tau_{\text{CATT}} (\boldsymbol{X}_{i})  \mid W_i \neq 0  \right]
=   \mathbb{E} \left[ \sum_{r \in \mathcal{R}} \tilde{\pi}_r(\boldsymbol{X}_{i}) \tau_{\text{CATT}} (r, \boldsymbol{X}_{i})  \mid W_i \neq 0   \right]   \label{att.def}
.
\end{align}
Estimand \eqref{att.def} is similar to estimand (6.40) proposed by \citet{wooldridge2021two} (a similar estimand is also mentioned by \citealt{borusyak2024revisiting}), though we weight the cohort average treatment effects \(\tau_{\text{CATT}} (r, \boldsymbol{X}_{i}) \) by the probabilities \(\tilde{\pi}_r(\boldsymbol{X}_{i}) \) rather than taking an average where each average treatment effect in each time period is weighted equally (we could also take an equal average with our estimand from Equation \ref{att.estimand.fixed}). This probability-weighted average is closely related to estimand (26) from \citet{sun2021estimating}. 

The virtues of estimand \(\tau_{\text{ATT}}\) include summarizing the effect of the treatment in a single number and potentially being a less noisy quantity to estimate. Estimating \(\tau_{\text{ATT}}\) will involve taking a weighted average over many estimated coefficients rather than relying on a single, less precise, estimated coefficient. In contrast, more specific estimands like \( \tau_{\text{CATT}} (r, t, \boldsymbol{x}) \) will typically be difficult to estimate with precision because their estimators will depend heavily on a small subset of the observed data (though this problem will be mitigated by our fusion penalty, discussed in the next section, which allows us to borrow information from nearby estimators).

Notice that \( \tau_{\text{ATT}} \) can be obtained not just by marginalizing \( \tau_{\text{CATT}} ( \boldsymbol{X}_{i} )\) over the distribution of \(\boldsymbol{X}_i\), but also by marginalizing \(\tau_{\text{ATT}} (  W_i)\) over the distribution of cohort assignments:
\begin{align*}
 \tau_{\text{ATT}}  = ~ & \mathbb{E} \left[  \tau_{\text{CATT}} ( W_i, \boldsymbol{X}_{i}) \mid W_i \neq 0  \right] =   \E \left[ \E \left[  \tau_{\text{CATT}} ( W_i, \boldsymbol{X}_{i}) \mid W_i ,  W_i \neq 0 \right]   \mid W_i \neq 0  \right] 
 \\ = ~ &   \sum_{r \in \mathcal{R}} \tilde{\pi}_r   \cdot  \E[  \tau_{\text{CATT}} ( W_i, \boldsymbol{X}_{i})  \mid W_i = r]  = \sum_{r \in \mathcal{R}} \tilde{\pi}_r  \tau_{\text{ATT}} (r)
 \\ = ~  & \sum_{r \in \mathcal{R}} \tilde{\pi}_r  \frac{1}{T-r+1} \sum_{t=r}^T \tau_{\text{ATT}} (r, t) 
 ,
\end{align*}
which shows that estimand \eqref{att.def} is in the class defined by \eqref{att.estimand.weighted} and matches estimand (3.11) from \citet{callaway2021difference}. This formulation will also be useful for estimation.

\section{Fused Extended Two-Way Fixed Effects}\label{sec.meth}

We will carry out regression on the same design matrix \(\boldsymbol{\tilde{Z}}\) proposed by \citet{wooldridge2021two} to estimate regression \eqref{wooldridge.6.33.model} by OLS, but we will add fusion penalties to leverage our belief that some of the parameters are equal to each other. This design matrix contains \(R\) columns of cohort dummies, \(T - 1\) columns of time dummies for times \(\{2, \ldots, T\}\), \(d_N\) covariates, \(\mathfrak{W} := \sum_{r \in \mathcal{R}} (T - r + 1)\) treatment dummies for each possible base treatment effect \(\tau_{rt}\), and \(d_N(R + T - 1 + \mathfrak{W})\) interactions between covariates and the cohort dummies, time dummies, and treatment dummies. As is typical in penalized regression, we will center \(\boldsymbol{\tilde{y}}\) and the columns of \(\boldsymbol{\tilde{Z}}\) before estimating our regression. It is not necessary to estimate the intercept term \(\eta^*\) in order to estimate the treatment effects.

\subsection{Restrictions in the ETWFE Estimator}

Let \(p_N\) be the number of columns of \(\boldsymbol{\tilde{Z}}\); that is, \(\boldsymbol{\tilde{Z}} \in \mathbb{R}^{NT \times p_N}\). We will assume \(p_N \leq N T\) (see Assumption R2 later), but as we discussed in the introduction, it is clear that \(p_N = R + T - 1 + \mathfrak{W} + d_N(1 + R + T - 1 + \mathfrak{W})\) may be fairly large compared to \(NT\), particularly if \(d_N\) is not very small. This could lead to imprecise treatment effect estimates if  \eqref{wooldridge.6.33.model} is estimated by OLS. \citet[Section 6.5]{wooldridge2021two} suggests alleviating this problem by assuming \textit{restrictions} in the model---assuming some of these parameters equal each other. Some of the possibilities \citeauthor{wooldridge2021two} suggests include assuming that the treatment-covariate interactions \(\boldsymbol{\rho}_{rt}\) are fixed across time or assuming the within-cohort treatment effects \(\tau_{rt}\) are fixed across time. \citet{goodman2021difference} also considers the possibility of treatment effects that are fixed across time. Another weaker assumption \citet{wooldridge2021two} proposes is assuming that different treatment effects do not need to be estimated for every individual time period since treatment starts; instead, these times could be condensed into ``early treated" and ``late treated" effects, for example.

These suggestions seem sensible, but making these simplifying assumptions without strong justification risks re-introducing the bias that was removed by adding these parameters to begin with. In settings where a practitioner believes some of these restrictions hold but has no clear way to choose which particular parameters may be equal, we instead propose a data-driven approach to estimate which, if any, of these parameters may be equal. 

\subsection{Fusion Penalization in FETWFE}

As described in the introduction, we will assign \(\ell_q\) fusion \citep{tibshirani2005sparsity} penalties to the differences between estimated parameters which may be equal, and then carry out a penalized linear regression. For example, for the treatment effects, we propose penalizing the absolute differences between treatment effects within the same cohort at adjacent times, \(|\tau_{r_k,t} - \tau_{r_k,t-1}|\), and we propose the same penalty structure for interaction effects, \(|\rho_{r_k,t,j} - \rho_{r_k,t-1,j}|\) separately for each \(j \in [d_N]\). This leverages our belief that some of these effects may be fixed across time while allowing the optimization problem to learn from the data which restrictions to choose.

We also propose penalizing the absolute differences between initial treatment effects for each cohort, \(|\tau_{r_k,r_k} - \tau_{r_{k-1},r_{k-1}}|\), and likewise we penalize the absolute differences between adjacent initial interaction effects \(|\rho_{r_k,r_k,j} - \rho_{r_{k-1},r_{k-1},j}|\) for each \(j\). This is motivated by the idea that perhaps time since beginning treatment could be the only dimension in which treatment effect changes. It is also natural to penalize the adjacent cohort fixed effects absolute differences \(|\nu_{r_k} - \nu_{r_{k-1}}|\) and time fixed effects absolute differences \(|\gamma_t - \gamma_{t-1}|\) towards each other, with analogous penalties for the interaction terms \(|\zeta_{r_kj} - \zeta_{r_{k-1},j}|\) and \(|\xi_{tj} - \xi_{t-1,j}|\).

Finally, we may also want to regularize the coefficients directly, as is common in regressions with large numbers of parameters to estimate relative to the number of observations. Because of the fusion penalties we employ, regularizing only a handful of base terms allows for this regularization. We can also penalize all of the base covariate coefficients directly, allowing us to screen out irrelevant covariates as \(d_N\) grows.

In summary, we propose estimating a model like \eqref{wooldridge.6.33.model} but with a penalty term
\begin{align}
p(\boldsymbol{\beta}; q) := ~ &  \sum_{k=2}^R | \nu_{r_k} - \nu_{r_{k-1}} |^q  +  |\nu_{r_R}|^q    +   \sum_{t=3}^T | \gamma_t - \gamma_{t-1}  |^q   +  |\gamma_T|^q\nonumber
\\ & + \sum_{j=1}^{d_N} \left( |\kappa_j|^q + \sum_{k =2}^R | \zeta_{r_k,j} - \zeta_{r_{k-1},j} |^q   +    |\zeta_{r_R,j}|^q   \nonumber
+  \sum_{t=3}^T | \xi_{tj} - \xi_{t-1,j} |^q   + |\xi_{Tj}|^q \right)  \nonumber
\\ & + | \tau_{r_1, r_1} |^q + \sum_{k =2}^R \left( |\tau_{r_k, r_k} - \tau_{r_{k-1}, r_{k-1}} |^q +  \sum_{t \geq r_k + 1} | \tau_{r_k,t} - \tau_{r_k,t-1}|^q  \right)  \nonumber
\\ & + \sum_{j =1}^{d_N} \left[ | \rho_{r_1,r_1,j} |^q + \sum_{k =2}^R \left( |\rho_{r_k,r_k,j} - \rho_{r_{k-1}, r_{k-1},j} |^q +  \sum_{t \geq r_k + 1} | \rho_{r_k,t,j} - \rho_{r_k,t-1,j}|^q  \right) \right]  \label{fetwfe.penalty}
,
\end{align}
where \(\boldsymbol{\beta} \in \mathbb{R}^{p_N}\) collects all of the coefficients. (Refer again to Figure \ref{fig:fusion-penalties} for a graphical depiction of the penalty terms on the estimated treatment effects \(\tau_{rt}\).) We can more compactly express this as
\begin{equation}\label{penalty.term}
p(\boldsymbol{\beta}; q) = \lVert \boldsymbol{D}_N \boldsymbol{\beta} \rVert_q^q
,
\end{equation}
where \(\boldsymbol{D}_N \in \mathbb{R}^{p_N \times p_N}\) is a suitably defined differences matrix which we later define explicitly in \eqref{d.expres}.

\begin{remark}\label{d.remark}
Our theory could easily be extended to allow a much broader class of penalties than the one penalty we specify in \eqref{fetwfe.penalty}. Any choice of \(\boldsymbol{D}_N\) that has all finite entries, is invertible, and is block diagonal with blocks of sizes not increasing in \(N\) would work for our theoretical results. Such modifications would lead to a change only in constant factors in our theory. For example, one could choose to only use fusion penalties for the treatment parameters \(\tau_{rt}\) and \(\rho_{rt}\) and directly penalize all of the remaining parameters. For another example, if it is believed that treatment effects are more likely to be equal in time since treatment rather than within cohorts, one could replace our structure of treatment effect penalties with fusion penalties across treatment effects at the same time since treatment; see Figure \ref{fig:modified-fusion-penalties} for a graphical depiction of such a structure. It is important that the chosen \(\boldsymbol{D}_N\) makes Assumption \(S(s_N)\), which we describe later in Section \ref{cons.subsec}, plausible.

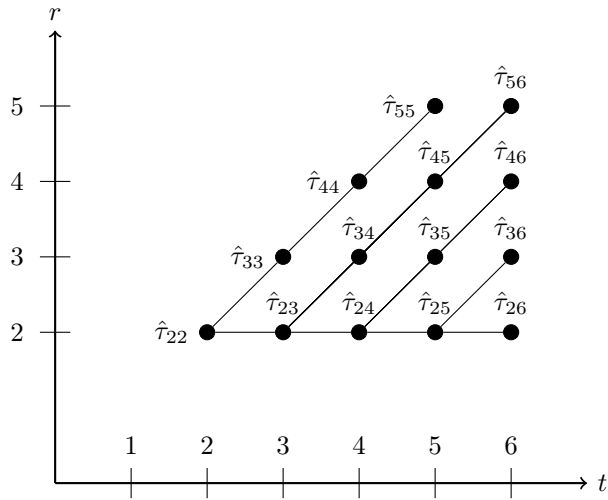
\begin{figure}[ht]
    \centering
\begin{tikzpicture}[every node/.style={font=\small}]
    % Define T and R (Adjust \T as needed for your document)
    \def\T{6} % Number of time periods
    \def\R{5} % Number of cohorts, set equal to T here
    
    % Draw axes
    \draw[->,thick] (0,0) -- (0,\R+1) node[above] {\(r\)};
    \draw[->,thick] (0,0) -- (\T+1,0) node[right] {\(t\)};
    % Draw tick marks and labels on axes
    \foreach \x in {1,...,\T} {
        \draw (\x,-0.2) -- (\x,0.2);
        \node at (\x,0.5) {\x};
    }
    \foreach \y in {2,...,\R} {
        \draw (-0.2,\y) -- (0.2,\y);
        \node at (-0.5,\y) {\y};
    }

    % Draw grid with labeled nodes
    \foreach \r in {2,...,\R} {
        \foreach \t in {\r,...,\T} {
            % Determine label position based on whether t equals r or is greater than r
            \ifnum\t=\r
                \node[draw, circle, inner sep=2pt, fill=black, label=left:\(\hat{\tau}_{\r\t}\)] (node-\r-\t) at (\t,\r) {};
            \else
                \node[draw, circle, inner sep=2pt, fill=black, label=above:\(\hat{\tau}_{\r\t}\)] (node-\r-\t) at (\t,\r) {};
            \fi
        }
    }

    % Draw penalty lines
    % Penalties between treatment effects based on the time since start of treatment
    \foreach \r in {2,...,\R} {
        \foreach \t in {\r,...,\T} {
            \pgfmathtruncatemacro{\diff}{\t - \r}
            \foreach \rOther in {2,...,\R} {
                \pgfmathtruncatemacro{\tOther}{\rOther + \diff}
                \ifnum\tOther>\rOther
                    \ifnum\tOther<\T+1
                        \draw (node-\r-\t) -- (node-\rOther-\tOther);
                    \fi
                \fi
            }
        }
    }

    % Penalties between initial treatment effects of each cohort
    \foreach \r in {3,...,\R} {
        \draw (node-\r-\r) -- (node-\the\numexpr\r-1\relax-\the\numexpr\r-1\relax);
    }

    % Penalties within the first cohort
    \foreach \t in {3,...,\T} {
        \draw (node-2-\t) -- (node-2-\the\numexpr\t-1\relax);
    }
\end{tikzpicture}

    \caption{Another example of a possible valid penalty structure that mainly penalizes treatment effects across cohorts towards each other based on time since the start of treatment, as mentioned in Remark \ref{d.remark}. Compare to Figure \ref{fig:fusion-penalties}. All of our theoretical results for FETWFE would work with such a penalty structure up to a change in constants in the results.}
    \label{fig:modified-fusion-penalties}
\end{figure}

\end{remark}

It is much easier to see that these candidate restrictions are plausible due to the time-structured nature of the parameters than it is to know exactly which restrictions to choose. FETWFE allows practitioners to leverage this structure to improve estimation efficiency without requiring them to choose restrictions by hand, which could introduce bias if too many restrictions are chosen or compromise efficiency if too few restrictions are chosen.

\subsection{Constructing the FETWFE Optimization Problem}

Finally, we take care of some details. First we account for the random effects discussed in Section \ref{setup.sec}. As suggested by \eqref{gls.transform.id}, in practice we will arrange the rows of \(\boldsymbol{\tilde{Z}}\) to contain consecutive blocks of \(T\) rows for each observation and carry out a regression on \(\boldsymbol{\tilde{y}}\) and \(\boldsymbol{\tilde{Z}}\) transformed by 
\begin{equation}\label{g.def}
\boldsymbol{G}_N := \sigma \boldsymbol{I}_N \otimes \boldsymbol{\Omega}^{-1/2} \in \mathbb{R}^{NT \times NT}
,
\end{equation}
 where \(\otimes\) denotes the Kronecker product (recall that \(\boldsymbol{\Omega}\) was defined in Section \ref{setup.sec}). We will also center both the response \(\boldsymbol{G}_N \boldsymbol{\tilde{y}}\) and each column of the covariates \(\boldsymbol{G}_N \boldsymbol{\tilde{Z}}\), as is common in penalized regression. Denote these centered objects by \(\boldsymbol{y}\) and \(\boldsymbol{Z}\). Then our optimization problem is
\begin{equation}\label{opt.prob}
\boldsymbol{\hat{\beta}}_N^{(q)}(\lambda_N) := \underset{ \boldsymbol{\beta} \in \mathbb{R}^{p_N}}{\arg \min} \left\{ \lVert \boldsymbol{y} -  \boldsymbol{Z} \boldsymbol{\beta} \rVert_2^2 + \lambda_N \lVert \boldsymbol{D}_N \boldsymbol{\beta} \rVert_q^q  \right\}
,
\end{equation}
where \(\lambda_N\) is a tuning parameter. For brevity, we will hereafter write \(\boldsymbol{\hat{\beta}}_N^{(q)}(\lambda_N)  = \boldsymbol{\hat{\beta}}^{(q)}\).

\begin{remark}\label{opt.prob.rmk}
Optimization problem \eqref{opt.prob} is convex if and only if \(q \geq 1\), but we will be interested in solving this problem with \(q \in (0, 1)\) as well. In this case, finding the global optimum can in principle be difficult, but many heuristics exist to solve this optimization problem that seem to work well in practice; see, for example \citet[Section 4.1]{Huang2008} and \citet{Breheny:2015aa}. In Sections \ref{synth.exps.sec} and \ref{sec.data.app} we implement FETWFE by solving \eqref{opt.prob} with the \texttt{grpreg} R package \citep{breheny2009penalized}. Our simulation results seem to align with our theory, including the asymptotic normality of our estimators, suggesting that these heuristics seem to work well for FETWFE.
\end{remark}

\subsection{Causal Estimators}\label{caus.est.sec}

As shown in Theorem \ref{te.interp.prop}(a), under (CNAS), (LINS), and (CCTS) 
\[
\tau_{\text{CATT}} (r, t, \boldsymbol{x}) = \tau_{rt}^* + ( \boldsymbol{x} - \E \left[ \boldsymbol{X}_{i} \mid  W_i = r \right] )^\top \boldsymbol{\rho}_{rt}^*  
\]
for any fixed \(\boldsymbol{x}\). Since under these assumptions (and other regularity conditions) FETWFE consistently estimates \( \tau_{rt}^* \) and \( \boldsymbol{\rho}_{rt}^*  \), we can estimate this quantity by
\begin{equation}\label{catt.t.r.est}
\hat{\tau}_{\text{CATT}} (r, t, \boldsymbol{x}) :=    \hat{\tau}_{rt}^{(q)} + \left( \boldsymbol{x} -  \boldsymbol{\overline{X}}_r \right)^\top \hat{\rho}_{rt}^{(q)} , \qquad r \in \mathcal{R}, t \geq r
,
\end{equation}
where \(\hat{\tau}_{rt}^{(q)}\) and \(\hat{\rho}_{rt}^{(q)}\) are the estimators of \(\tau_{rt}\) and \(\rho_{rt}\) from \(\boldsymbol{\hat{\beta}}^{(q)}\) as defined in \eqref{opt.prob} and
\[
\boldsymbol{\overline{X}}_r  = \frac{1}{N_r} \sum_{\{i: W_i = r\}} \boldsymbol{X}_{i}
\]
is the sample mean of the covariates for units observed in cohort \(r\), with \(N_r := \sum_{i=1}^N \mathbbm{1}\{W_i = r\}\) the total number of units in cohort \(r\). The forms of \eqref{catt.estimand.fixed} and \eqref{catt.estimand.weighted} then naturally suggest the estimators
\begin{equation}\label{catt.estimand.fixed.est}
\sum_{r \in \mathcal{R}} \sum_{t=r}^T \psi_{rt}  \hat{\tau}_{\text{CATT}} (r, t, \boldsymbol{x}) ) 
\end{equation}
and
\begin{equation}\label{catt.estimand.weighted.est}
 \sum_{r \in \mathcal{R}} \sum_{t =r}^T \psi_{rt}     \left(  \frac{\hat{\pi}_r(\boldsymbol{x})}{ \sum_{r' \in \mathcal{R}} \hat{\pi}_{r'}(\boldsymbol{x})}   \right) \hat{\tau}_{\text{CATT}} (r, t, \boldsymbol{x}) 
 ,
\end{equation}
where \(\hat{\pi}_r(\boldsymbol{x})\) is an estimator of the propensity scores \(\pi_r(\boldsymbol{x})\). We will discuss what properties this estimator \(\hat{\pi}_r(\boldsymbol{x})\) should have in more detail in Section \ref{cons.subsec}.

Finally we discuss estimating the average treatment effects marginalized over \(\boldsymbol{X}_i\). As we showed in Theorem \ref{te.interp.prop}(a) and (b), because the treatment effects are interacted with the covariates centered with respect to their cohort means, under either (CCTS) or (CTS) and (CIUN) we can estimate the average treatment effects using
\begin{equation}\label{att.est.cov.r.t}
\hat{\tau}_{\text{ATT}} (r,  t)  :=  \hat{\tau}_{rt}^{(q)}
\end{equation}
as proposed by \citet[Sections 5.3 and 6.3]{wooldridge2021two}. In particular, we propose the estimators
\begin{equation}\label{att.estimator.fixed}
\sum_{r \in \mathcal{R}} \sum_{t=r}^T \psi_{rt} \hat{\tau}_{\text{ATT}} (r,  t) 
\end{equation}
and
\begin{equation}\label{att.estimator.weighted}
 \sum_{r \in \mathcal{R}} \sum_{t =r}^T \psi_{rt}     \frac{N_r}{N_\tau}  \hat{\tau}_{\text{ATT}} (r,  t) 
 \end{equation}
for estimands \eqref{att.estimand.fixed} and \eqref{att.estimand.weighted}, respectively. Notice that estimator \eqref{att.estimator.weighted} avoids requiring the specification of an estimator \(\boldsymbol{\hat{\pi}}(\boldsymbol{x})\) of the generalized propensity scores by instead using the observed cohort counts to estimate the marginal cohort assignment probabilities, as we alluded to at the end of Section \ref{att.subsec}.

\section{Theoretical Results and Discussion}\label{sec.theory}

Before we present theoretical guarantees for FETWFE, we will state a sparsity assumption that we will require along with some additional regularity conditions. We will require fewer assumptions to prove consistency alone than we will to prove restriction selection consistency and asymptotic normality. We emphasize that the conceptually most important assumptions are (CCTSB), (CCTSA), and (CTSA), which validate the difference-in-differences approach; (LINS), which justifies a model that is linear in the covariates \(\boldsymbol{X}_i\); and Assumption S(\(s_N\)), defined below, which assumes sparsity in a sense that matches our penalty \eqref{penalty.term}. (See Remark \ref{d.remark} for some comments on how this can be made more flexible.)

\subsection{Consistency}\label{cons.subsec}

\textbf{Assumption S(\(s_N\)):} Assume the vector \(\boldsymbol{D}_N \boldsymbol{\beta}_N^*\) (where the differences matrix \(\boldsymbol{D}_N\) associated with our penalty \eqref{penalty.term} is defined in Equation \ref{d.expres}) is sparse, with indices \(\mathcal{S} \subseteq [p_N]\) nonzero and the remaining indices in \( [p_N] \setminus \mathcal{S}\) equal to 0, where \( |\mathcal{S}| = s_N\).

Assumption S(\(s_N\)) allows us to have sparsity in the appropriate sense: the differences between parameters that we penalized in \eqref{penalty.term} are sparse, and irrelevant covariates have true coefficient equal to 0. That is, all but \(s_N\) of the restrictions that we expect to hold will in fact hold. Crucially, Assumption S(\(s_N\)) does not require us to know which restrictions to select, or even how many restrictions to select---the value of \(s_N\) does not need to be known. Our consistency and selection consistency theorems will allow \(s_N\) to grow with \(N\) subject to certain constraints in later assumptions.

The remaining regularity conditions for Theorem \ref{main.te.cons.thm} are more routine. Before we describe them in detail, we briefly characterize them informally.
\begin{itemize}
\item Assumption (R1) requires the tails of the covariates to not be too heavy, but accommodates a wide variety of covariate distributions.
\item Assumption (R2) requires the eigenvalues of the empirical Gram matrix to be reasonably well-behaved. In particular, it requires \(\boldsymbol{Z}\) to have full column rank, so \(p_N \leq NT\). Otherwise this is a mild assumption as long as (1) none of the covariates are very highly correlated and (2) none of the cohort propensity scores tend too closely to 0. For example, an overlap assumption that \(\pi_r(\boldsymbol{x}) \geq \pi_{\text{low}}\) for some \(\pi_{\text{low}} > 0\) for all \(N \in \mathbb{N}\), \(r \in \{0\} \cup \mathcal{R}\), and \(\boldsymbol{x}\) in the support of \(\boldsymbol{X}_i\) would help the plausibility of this assumption. 
\item In Assumption (R2), the assumptions on the penalty parameter \(\lambda_N\) are not of practical importance because in practice \(\lambda_N\) can be selected from a set of candidate values by a criterion like BIC. (We show in simulation studies in Section \ref{synth.exps.sec} that the theoretical guarantees from this section seem to hold in practice when \(\lambda_N\) is selected in this way.)
\item Assumption (R3) is essentially a technical assumption. It does suggest that FETWFE may not work well if  \(\boldsymbol{D}_N \boldsymbol{\beta}_N^*\) has many very small entries, or a small number of very large entries.
\end{itemize}
In particular, these assumptions allow the number of covariates \(d_N\) and the number of relevant coefficients \(s_N\) to grow to infinity with \(N\) and allow for unbounded distributions of the covariates.

We now state and describe the assumptions.

\textbf{Assumption (R1):} For each \(i \in [N]\), \(\E [ X_{ij}^4 ] < \infty\) for all \(j \in [d_N]\) and \(\E [  u_{(it)}^4 ] < \infty\) for all \(t \in [T] \). Further, for some \(K < \infty\), 
\[
\sup_{1 \leq N \leq \infty} \left\{ \frac{1}{p_N} \sum_{t=1}^T \sum_{j=1}^{p_N} \E \left[ ( \boldsymbol{Z} \boldsymbol{D}_N^{-1})_{(1t)j}^2 \right] \right\}  \leq K
.
\]
We show later (Lemma \ref{lem.d.express}) that \(\boldsymbol{D}_N\) is invertible. The matrix \(\boldsymbol{Z} \boldsymbol{D}_N^{-1}\) is of interest because it turns out this is the design matrix we will actually use for estimation; see the proof of Theorem \ref{first.thm.fetwfe} for details. Because the singular values of \(\boldsymbol{D}_N^{-1}\) are bounded away from 0 and from above for all \(N\) even if \(d_N \to \infty\) (see Lemma \ref{d.sing.val.lem}), this is qualitatively similar to making the same statement about \(\boldsymbol{Z}\) itself. We discuss a closely related point in more detail below.

Before we state Assumption (R2), we will define some needed notation. For any matrix \(\boldsymbol{A} \in \mathbb{R}^{N \times p_N}\), denote the empirical Gram matrix by \(\boldsymbol{\hat{\Sigma}}(\boldsymbol{A})\):
\begin{equation}\label{gram.mat.def}
\boldsymbol{\hat{\Sigma}}(\boldsymbol{A}) :=  \frac{1}{N} \boldsymbol{A}^\top \boldsymbol{A} \in \mathbb{R}^{p_N \times p_N}
.
\end{equation}
Notice that if \(\boldsymbol{A}\) has centered columns then \(\boldsymbol{\hat{\Sigma}}(\boldsymbol{A})\) is the estimated covariance matrix. Let \(e_{1N}\) and \(e_{2N}\) be the smallest and largest eigenvalues of \(\boldsymbol{\hat{\Sigma}}(\boldsymbol{Z})\).(Since \(\boldsymbol{\hat{\Sigma}}(\boldsymbol{Z})\) is random, \(e_{1N}\) and \(e_{2N}\) are as well.)

\textbf{Assumption (R2):} Assume \(e_{1N}\) is almost surely upper-bounded by a constant for all \(N\) and \(h_N' \xrightarrow{a.s.} 0\), where
\begin{equation}\label{h.n.prime.def}
h_N' :=  \sqrt{\frac{p_N + \lambda_N s_N}{N e_{1N} }}
\end{equation}
and \(\xrightarrow{a.s.}\) denotes almost sure convergence. Also, assume \( \lim_{N \to \infty} \lambda_N \sqrt{ \frac{s_N }{N}} = 0 \). 

 \citet{kock2013oracle} points out that \eqref{h.n.prime.def} requires the design matrix to be full rank, so this assumption requires that \(p_N\) grows slower than \(N \), \(\lambda_N\) does not grow too quickly, and \(e_{1N}\) does not vanish too quickly. (In practice, in finite samples \(\lambda_N\) can be selected by cross-validation or a criterion like BIC, similarly to the lasso.) Under the qualitative assumption that the eigenvalues of \(\boldsymbol{\hat{\Sigma}}(\boldsymbol{Z})\) and \(\boldsymbol{\hat{\Sigma}}(\boldsymbol{G}_N\boldsymbol{\tilde{Z}})\) are ``close,"\footnote{This can be made formal by noting that centering the columns of \(\boldsymbol{G}_N\boldsymbol{\tilde{Z}}\) can be characterized as multiplication on the left by a centering matrix \(\boldsymbol{C}\), then bounding the singular values of \(\boldsymbol{C}\). See, for example, Section 8.1 of \citet{ding2024linear} for further details.} (R2) roughly requires that \(\boldsymbol{G}_N\boldsymbol{\tilde{Z}}\) is full rank, and we can provide sufficient conditions for this to hold. First, the singular values of \(\boldsymbol{G}_N\) are bounded from above and below by positive constants for all \(N\) (see Lemma \ref{g.sing.val.lem} in the appendix), so \(\boldsymbol{\tilde{Z}}\) being full rank is enough for \(\boldsymbol{G}_N\boldsymbol{\tilde{Z}}\) to be full rank. For \(\boldsymbol{\tilde{Z}}\) to be full rank, we require at least \(d_N + 1\) observed units in each cohort (see Lemma \ref{rank.cond2} in the appendix), all of the marginal cohort probabilities must be large enough relative to the growth rate of \(d_N\) as \(N \to \infty\), and the distribution of \(\boldsymbol{X}_i\) must be well-behaved as \(N \to \infty\).

\textbf{Assumption (R3):} There exist constants \(0 < b_0 < b_1 < \infty\) such that \(b_0 \leq \min_{j \in \mathcal{S}} \left\{ \left| ( \boldsymbol{D}_N \boldsymbol{\beta}_N^*)_j \right| \right\} \leq  \max_{j \in \mathcal{S}} \left\{ \left| ( \boldsymbol{D}_N \boldsymbol{\beta}_N^*)_j \right| \right\}  \leq b_1\).

The lower bound in Assumption (R3) is essentially a signal strength assumption.
 
 Note that we allow \(s_N \to \infty\) and \(p_N \to \infty\) as long as the above assumptions are satisfied.

 We are now prepared to state our consistency theorem.

\begin{theorem}[Consistency of FETWFE]\label{main.te.cons.thm}

Assume that Assumptions (CNAS), (CCTSB), and (LINS) hold, as well as Assumptions (F1), (F2), S(\(s_N\)), and (R1) - (R3). Let \(q > 0\).

\begin{enumerate}[(a)]

\item Suppose that either Assumptions (CTSA) and (CIUN) hold or Assumption (CCTSA) holds. Then for any set of finite constants \(\{\psi_{rt}\}\),
\[
 \left| \sum_{r \in \mathcal{R}} \sum_{t=r}^T \psi_{rt}  \left( \hat{\tau}_{\text{ATT}} (  r,t )  -  \tau_{\text{ATT}} (r,t )  \right) \right|  = \mathcal{O}_\mathbb{P} \left( \min\{h_N, h_N'\}\right)
 \]
 and
 \[
 \left|  \sum_{r \in \mathcal{R}} \sum_{t =r}^T \psi_{rt} \left( \frac{N_r}{N_\tau}     \hat{\tau}_{\text{ATT}} ( r, t) - \tilde{\pi}_r     \tau_{\text{ATT}} (r, t) \right) \right|  = \mathcal{O}_\mathbb{P} \left( \min\{h_N, h_N'\}\right)
  ,
  \]
where
\begin{equation}\label{h.n.def}
h_N :=  \frac{1}{e_{1N}} \sqrt{\frac{p_N}{N}} 
\end{equation}
and \(h_N'\) is as defined in \eqref{h.n.prime.def}. 

\item
Suppose Assumption (CCTSA) holds and \(d_N = d\) is fixed. For all \(N\) and all \(r \in \{0\} \cup   \mathcal{R}\), assume all of the eigenvalues of \( \Cov ( \boldsymbol{X}_i \mid W_i = r)\) are bounded between \(\lambda_{\text{min}} > 0\) and \(\lambda_{\text{max}} < \infty\). Then for any set of finite constants \(\{\psi_{rt}\}\) and any fixed \(\boldsymbol{x}\) in the support of \(\boldsymbol{X}_i\),
\begin{align*}
\left| \sum_{r \in \mathcal{R}} \sum_{t=r}^T \psi_{rt} \left(    \hat{\tau}_{\text{CATT}}(r, t, \boldsymbol{x}) -  \tau_{\text{CATT}}(r, t, \boldsymbol{x})  \right) \right| =
\mathcal{O}_{\mathbb{P}}( \min\{h_N, h_N'\}  )  
.
\end{align*}
Further, if \(\hat{\pi}_r(\boldsymbol{x})\) is an estimator of \(\pi_r(\boldsymbol{x})\) that satisfies \( | \hat{\pi}_r(\boldsymbol{x}) -  \pi_r(\boldsymbol{x})| =  \mathcal{O}_{\mathbb{P}}(a_N)   \) for a decreasing sequence \(a_N\) for each \(r \in \{0\} \cup \mathcal{R}\) and all \(\boldsymbol{x} \) in the support of \(\boldsymbol{X}_i\), then 
%\begin{align*}
% \left|  \sum_{r \in \mathcal{R}} \sum_{t=r}^T \psi_{rt} \left(   \frac{\hat{\pi}_r(\boldsymbol{x})}{\sum_{r' \in \mathcal{R}} \hat{\pi}_{r'}(\boldsymbol{x})}   \hat{\tau}_{\text{CATT}}(r, t, \boldsymbol{x}) -  \tilde{\pi}_r(\boldsymbol{x}) \tau_{\text{CATT}}(r, t, \boldsymbol{x})  \right) \right| 
%%
%\xrightarrow{p} ~ & 
%0
%.
%\end{align*}
%In particular, if
%\[
% \left|  \frac{\hat{\pi}_r(\boldsymbol{x})}{ \sum_{r' \in \mathcal{R}} \hat{\pi}_{r'}(\boldsymbol{x})}  - \tilde{\pi}_r(\boldsymbol{x})  \right| = \mathcal{O}_{\mathbb{P}}(a_N) 
% \]
%for some sequence \(\{a_N\}\), then
\begin{align*}
& \left|  \sum_{r \in \mathcal{R}} \sum_{t=r}^T \psi_{rt} \left(   \frac{\hat{\pi}_r(\boldsymbol{x})}{\sum_{r' \in \mathcal{R}} \hat{\pi}_{r'}(\boldsymbol{x})}   \hat{\tau}_{\text{CATT}}(r, t, \boldsymbol{x}) -  \tilde{\pi}_r(\boldsymbol{x}) \tau_{\text{CATT}}(r, t, \boldsymbol{x})  \right) \right| 
\\ = ~ & 
\mathcal{O}_{\mathbb{P}}( \min\{h_N, h_N'\} \vee a_N )  
.
\end{align*}

\end{enumerate}

\end{theorem}

\begin{proof} Provided in Section \ref{app.prove.main.res} of the appendix.
\end{proof}

Theorem \ref{main.te.cons.thm} shows that all of our causal estimators are consistent and characterizes their rates of consistency. In particular, it also shows that the classes of marginal parallel trends estimators \eqref{att.estimator.fixed} and \eqref{att.estimator.weighted} are consistent under either (CCTSA) or (CIUN) and (CTSA) for the marginal average treatment effects \eqref{att.estimand.fixed} and \eqref{att.estimand.weighted}. Although we require (CCTSB) to hold regardless of whether (CIUN) holds, (CCTSB) also depends only on the observed outcomes and is testable. In the next section, we show that under mild assumptions FETWFE automatically provides a valid estimate of whether (CIUN) holds (by estimating whether the \(\boldsymbol{\xi}_t^*\) coefficients all equal \(\boldsymbol{0}\)).

To interpret the rate of convergence \( \min\{h_N, h_N'\}\), consider the upper bound \(h_N\) from \eqref{h.n.def}. Suppose the minimum eigenvalue \(e_{1N}\) is bounded away from 0 with high probability, which happens under relatively mild assumptions in our \(p_N \leq NT\) setting (see our more detailed discussion later in Section \ref{add.assum} about Assumption R6). Then \( \mathcal{O}_\mathbb{P} \left( \min\{h_N, h_N'\}\right) =  \mathcal{O}_\mathbb{P} \left( \sqrt{p_N/N} \right)\); see Lemma \ref{conv.claim.lem} in the appendix.

In general we will focus more on \(q \in (0, 1)\) as required by Theorems \ref{te.sel.cons.thm}, \ref{te.oracle.thm}, and \ref{te.asym.norm.thm}, but we note that Theorem \ref{main.te.cons.thm} holds for any \(q > 0\). This is useful because the convex lasso (\(q = 1\)) and ridge (\(q = 2\)) optimization problems, for example, may be more tractable on very large data sets than the nonconvex bridge estimator with \(q < 1\).

To give a concrete example of an estimator \(\boldsymbol{\hat{\pi}}(\boldsymbol{x})\), in the setting of Theorem \ref{main.te.cons.thm}(b) where \(d_N = d\) is fixed, the multinomial logit model is \(1/\sqrt{N}\)-consistent for the generalized propensity scores by standard maximum likelihood theory under correct specification. If this model is overly simplistic or \(d \) is large relative to \(N\), regularized machine learning methods with theoretical convergence guarantees---like deep neural networks \citep[see Lemma 9c]{farrell2021deep}, random forests \citep{gao2022towards}, the \(\ell_1\)-penalized classifiers of \citet{levy2023generalization}, or the group lasso multinomial logit model of \citet{farrell2015robust}---would also work. On the other hand, since the cohorts have a natural ordering, we could add more structure by using an ordinal response estimator like the proportional odds model \citep{mccullagh1980regression}.

It may be possible to relax Assumption (LINS) so that linearity only has to hold approximately, along the lines of e.g. Condition ASTE in \citet{belloni2014inference}. (At a high level, we might expect that if approximate linearity holds so that the error due to misspecification is on the same order as the estimation error of the population least squares model, we can achieve convergence at the same rate as when the linearity assumption holds exactly.)

It may also be possible to relax the fixed \(d\) assumption in part \((b)\) by using high-dimensional central limit theorems like those in \citet{lopes2022central} and \citet{chernozhukov2023nearly}, to name two recent examples among many papers containing such results.

Lastly, we point out that under Theorem \ref{main.te.cons.thm} the generalized propensity scores and treatment effects can be estimated on the same data sets; that is, the full data set can be used for both. Intuitively, this is because these estimators are used for the distinct tasks of estimating the effects \(\tau_{\text{CATT}} (r, t, \boldsymbol{x})\) and estimating the weights used to combine them in a weighted average. For technical details, see the proof of Theorem \ref{main.te.cons.thm} in Appendix \ref{app.prove.main.res}.

In the appendix, we present an extension of Theorem \ref{main.te.cons.thm}, Theorem \ref{main.te.cons.thm.gen}, which proves consistency for a set of estimators that allows  for much broader functions of the estimated cohort membership probabilities than the classes \eqref{att.estimator.weighted} and \eqref{catt.estimand.weighted.est}. This result includes, for example, all of the estimands from \citet{callaway2021difference} as special cases.

\subsection{Restriction Selection Consistency}

Again, we first state some additional needed regularity conditions and then present our selection consistency result, Theorem \ref{te.sel.cons.thm}. Theorem \ref{te.sel.cons.thm} shows that with probability tending towards one FETWFE identifies the correct restrictions, resulting in fusing together parameters that equal each other, screening out irrelevant covariates, and improving estimation efficiency. At a high level, Assumptions (R4) and (R5) add the following requirements:
\begin{itemize}
\item Assumption (R4) is most plausible if the covariates have bounded support and if \(d_N\) is not increasing in \(N\). Intuitively, this means that in finite samples results relying on Assumption (R4) will be most likely to hold when \(NT\) is reasonably large compared to \(p_N\), though we show in the simulation studies in Section \ref{synth.exps.sec} that in practice the conclusion of Theorem \ref{te.sel.cons.thm} seems to hold even when \(N\) is not all that large and \(p_N\) is fairly close to \(NT\).
\item Assumption (R5) is a technical assumption that does not require much more than (R2). (R5) is easier to satisfy if \(d_N\) is increasing only slowly in \(N\) and if the sparsity \(s_N\) is fixed. (In finite samples, we would expect better results if \(p_N\) and the sparsity \(s_N\) are small relative to \(NT\).)
\end{itemize}

\textbf{Assumption (R4):} There exists a positive constant \(e_{\text{max}} < \infty\) such that \( e_{2N} \leq e_{\text{max}}\) almost surely (recall from Section \ref{cons.subsec} that \(e_{2N}\) is the largest eigenvalue of \(\boldsymbol{\hat{\Sigma}}(\boldsymbol{Z})\)).

A sufficient condition for Assumption (R4) is that the distribution of \(\boldsymbol{X}_i\) has bounded support in the sense that \(\lVert \boldsymbol{X}_i \rVert_2 \leq B \) almost surely for some \(B < \infty\).

\textbf{Assumption (R5):}
\[
 \lambda_N \frac{e_{1N}^{2- q}}{ \sqrt{N^q p_N^{2 - q }}} \xrightarrow{a.s.} \infty
 .
\]
We will briefly examine how Assumption (R5) interacts with Assumption (R2). For simplicity, assume the minimum eigenvalue \ \(e_{1N} \geq c\) for some \(c > 0\) with probability tending towards 1, so that we can ignore \(e_{1N}\) asymptotically (using a similar argument to the one in the proof of Lemma \ref{conv.claim.lem} in the appendix). Then Assumption (R5) is equivalent to the deterministic condition
\[
 \lambda_N \frac{1}{ N^{q/2} p_N^{1 - q/2 }}  \to \infty
 .
\]
If we also assume \(s_N = s\) is fixed, we have from Assumption (R2) that
\begin{align*}
\frac{\lambda_N }{N } \to  0 \qquad \text{and} \qquad 
 \lim_{N \to \infty} \frac{ \lambda_N}{N^{1/2}}   =  0
.
\end{align*}
We will require \(q \in (0,1)\), and we see that under these assumptions for (R5) and (R2) to hold simultaneously we require \(N^{q/2}p_N^{1-q/2}\) to grow slower than \(N^{1/2}\); that is, \(p_N\) must grow slower than
\[
\left( \frac{N^{1/2}}{N^{q/2}} \right)^{2/(2-q)} 
%= (N^{1-q})^{1/(2-q)}
 = N^{(1-q)/(2-q)}
.
\]

Even if \(e_{1N}\) is not bounded away from 0 with high probability, Assumption (R5) requires that \(e_{1N}\) does not vanish too quickly. As we discussed in reference to Assumption (R2), this requires that the probability of being in any one cohort does not vanish too quickly as \(N \to \infty\). It also requires that \(\lambda_N\) not decrease too quickly.

Notice that other than these relatively mild regularity conditions on the eigenvalues of \(\boldsymbol{\hat{\Sigma}}(\boldsymbol{Z})\), we do not require a condition like the \textit{irrepresentable condition} \citep{Zhao2006} or \textit{neighborhood stability condition} \citep{meinshausen2006high} that are required for selection consistency of the lasso. However, our restriction selection consistency result will only work for \(q < 1\); we do not prove selection consistency of FETWFE with a lasso (\(q =1\)) penalty.

\begin{theorem}[Selection consistency]\label{te.sel.cons.thm}

 For the FETWFE estimator \(\boldsymbol{\hat{\beta}}^{(q)}\) defined in \eqref{opt.prob}, define
 \[
 \hat{\mathcal{S}} := \left\{j \in [p_N]: ( \boldsymbol{D}_N \boldsymbol{\hat{\beta}}^{(q)})_j \neq 0 \right\}
 \]
 to be the set of nonzero values in \(\boldsymbol{D}_N \boldsymbol{\hat{\beta}}^{(q)}\). Recall from Assumption S(\(s_N\)) that
  \[
\mathcal{S} = \left\{j \in [p_N]: ( \boldsymbol{D}_N \boldsymbol{\beta}_N^*)_j \neq 0 \right\}
.
 \]
Assume that Assumptions (CNAS), (CCTSB), and (LINS) hold, as well as Assumptions (F1), (F2), S(\(s_N\)), and (R1) - (R5). Assume \(q \in (0,1)\). Then \(\lim_{N \to \infty} \mathbb{P} \left(\hat{\mathcal{S}} = \mathcal{S} \right) = 1\).

\end{theorem}

\begin{proof} This is immediate from Theorem \ref{first.thm.fetwfe}(c) in the appendix.
\end{proof}

Theorem \ref{te.sel.cons.thm} shows that FETWFE fulfills its promise to identify the correct restrictions in accordance with our \(\ell_q\) penalty structure, as discussed in Section \ref{sec.meth}. This serves our goal of improving estimation by avoiding unnecessarily estimating parameters that are in fact equal. 

FETWFE also screens out irrelevant covariates and interactions. As mentioned in Section \ref{lins.sec}, if Assumptions (LINS) and (CIUN) both hold than the \(\boldsymbol{\xi}_t^*\) coefficients must all equal 0. If this is true, in large enough samples FETWFE will estimate \(\boldsymbol{\xi}_t^*\) to equal exactly \(\boldsymbol{0}\) with high probability.

The proof of Theorem \ref{te.sel.cons.thm} relies on an extension of Theorem 2 of \citet{kock2013oracle} (Theorem \ref{prop.2i} in the appendix) to show that not only does the bridge estimator exclude false selections, it also selects all of the relevant components of \( \boldsymbol{D}_N \boldsymbol{\beta}_N^*\) with probability tending to 1. Theorem \ref{prop.2i} requires no added assumptions from Theorem 2 of \citet{kock2013oracle}.

One consequence of Theorem \ref{te.sel.cons.thm} is that if \(\tau_{\text{ATT}} (r, t) = 0\) for any \(r\), then 
\[
\lim_{N \to \infty} \mathbb{P} \left( \hat{\tau}_{\text{ATT}} ( r, t) = 0\right) = 1.
\]
The same is true of all of the other causal estimators. This is stronger than convergence in probability to 0, and this stronger notion of convergence turns out to be a key ingredient in the proof of Theorems \ref{te.oracle.thm} and \ref{te.asym.norm.thm}, to come.

\subsection{Asymptotic Normality and Oracle Property}\label{add.assum}

Our asymptotic convergence results (Theorems \ref{te.oracle.thm} and \ref{te.asym.norm.thm} below) require one additional regularity condition, an eigenvalue condition which does not require much more than Assumptions (R2) and (R5).

\textbf{Assumption (R6):} The minimum eigenvalue of 
\begin{equation}\label{min.eigen.eq}
\E \left[  \boldsymbol{\hat{\Sigma}} \left(  \boldsymbol{Z} \boldsymbol{D}_N^{-1} \right)\right] 
\end{equation}
 is greater than or equal to a fixed \(\delta > 0\) for all \(N \) sufficiently large, and
\begin{equation}\label{new.assum.2}
\frac{\lambda_N}{ e_{1N} \sqrt{N } }   \in  o_{\mathbb{P}} \left( 1 \right) 
 .
\end{equation}
Further, for all \(N\) and all \(r \in \{0\} \cup   \mathcal{R}\), all of the eigenvalues of \( \Cov ( \boldsymbol{X}_i \mid W_i = r)\) are bounded between \(\lambda_{\text{min}} > 0\) and \(\lambda_{\text{max}} < \infty\).

The assumption that the minimum eigenvalue of \eqref{min.eigen.eq} is bounded away from 0 is fairly mild since we already require \(\boldsymbol{\hat{\Sigma}}(\boldsymbol{Z})\) to be full rank and for its minimum eigenvalue to not vanish too quickly, and the minimum singular value of \(\boldsymbol{D}_N\) is bounded from below by a positive constant; see Lemma \ref{d.sing.val.lem} in the appendix. Further, this assumption is stronger than necessary for the sake of interpretability; we only require a bound on the minimum eigenvalue of a different matrix \eqref{smaller.pop.eq} for all \(N\). This condition holding for \eqref{new.assum.2} is sufficient but not necessary.

Condition \eqref{new.assum.2} is another assumption that requires \(e_{1N}\) to not vanish too quickly, which is qualitatively similar to Assumptions (R2) and (R5). Holding \(s = s_N\) fixed (as is required in Theorems \ref{te.oracle.thm} and \ref{te.asym.norm.thm}) and ignoring the distinction between almost sure convergence and convergence in probability, \eqref{new.assum.2} is more restrictive on the growth rate of \(\lambda_N\) than (R2) since (R2) requires
\[
\frac{\lambda_N }{N e_{1N} } \to 0
 \]
while \eqref{new.assum.2} requires
\[
\frac{\lambda_N}{ \sqrt{N} e_{1N}}  \to 0
.
\]
However, if \(e_{1N}\) is bounded away from 0 with high probability, using a similar argument to the one in the proof of Lemma \ref{conv.claim.lem} in the appendix, \eqref{new.assum.2} is essentially equivalent to the other requirement in (R2) that \(\lim_{N \to \infty} \lambda_N \sqrt{s_N/N} = 0\) (in the setting of Theorem \ref{te.asym.norm.thm} where \(s_N = s\) is fixed). If the smallest eigenvalue of \(\E \left[ \boldsymbol{\hat{\Sigma}}(\boldsymbol{Z}) \right]\) is bounded away from 0, one can show that the asymptotic distribution for the minimum eigenvalue of \( \boldsymbol{\hat{\Sigma}}(\boldsymbol{Z}) \) has support bounded away from 0 in our setting where \(\boldsymbol{X}_i\) has finite fourth moments and \(p_N \leq NT\) \citep{marchenko1967distribution}. If we also assume that \(\boldsymbol{X}_i\) is subgaussian (a weaker assumption than requiring \(\boldsymbol{X}_i\) to have bounded support), under Theorem 5.39 from \citet{vershynin_2012} \(e_{1N}\) is bounded away from 0 with high probability for finite \(N\) large enough relative to \(p_N\).

\subsubsection{Oracle Property}

Next we show that under Assumptions (R1) - (R6) the estimators \eqref{att.estimator.fixed} and \eqref{att.estimator.weighted} are both oracle estimators \citep{fan2001variable, fan2004nonconcave}, converge at a \(1/\sqrt{N}\) rate (faster than the guarantee from Theorem \ref{main.te.cons.thm}), and are asymptotically Gaussian.

\begin{theorem}[Oracle Property of FETWFE]\label{te.oracle.thm}

Assume that Assumptions (CNAS), (CCTSB), and (LINS) hold, as well as Assumptions (F1), (F2), S(\(s\)) for a fixed \(s\), and (R1) - (R6). Assume \(q \in (0,1)\). Suppose that either Assumptions (CTSA) and (CIUN) hold or Assumption (CCTSA) holds. 

For an arbitrary set of finite constants \(\{\psi_{rt}\}\), if at least one of the \(\tau_{\text{ATT}} (r, t)\) with a nonzero \(\psi_{rt}\) is nonzero, then the sequence of random variables
\begin{align*}
\sqrt{NT} \sum_{r \in \mathcal{R}} \sum_{t=r}^T \psi_{rt}  ( \hat{\tau}_{\text{ATT}} (  r,t )  - \tau_{\text{ATT}} (r,t ) ) 
\end{align*}
converges in distribution to a mean zero Gaussian random variables with variance that depends only on the \(s\) parameters of the model with all restrictions correctly identified.

Further, if \(\hat{\tau}_{\text{ATT}} ( r, t)\) and the probability ratios \(N_r/ N_\tau\) are estimated on two independent data sets of size \(N\), then the sequence of random variables 
\begin{align*}
\sqrt{NT} \sum_{r \in \mathcal{R}} \sum_{t =r}^T \psi_{rt}     \left(   \frac{N_r}{N_\tau}  \hat{\tau}_{\text{ATT}} ( r, t) -  \tilde{\pi}_r \tau_{\text{ATT}} (r, t) \right) 
\end{align*}
also converges in distribution to a mean zero Gaussian random variable with variance that depends only on the \(s\) parameters of the model with all restrictions correctly identified. 

(If all of the \(\tau_{\text{ATT}} (r, t)\) with nonzero \(\psi_{rt}\) terms equal 0, both sequences of random variables converge in probability to 0.)

\end{theorem}

\begin{proof} Provided in Section \ref{app.prove.main.res} of the appendix.
\end{proof}

Theorems \ref{te.oracle.thm} and \ref{te.asym.norm.thm} leverage novel (to the best of our knowledge) extensions of Theorem 2 from \citet{kock2013oracle} that establish finite-sample variance estimators, including when the weights in a linear combination of the estimated coefficients are themselves estimated; see Theorem \ref{prop.ext.2} in the appendix.

Theorem \ref{te.oracle.thm} (and, more directly, Theorem \ref{first.thm.fetwfe}(d) and (f) in the appendix) show that even if \(p_N \to \infty\) FETWFE converges at a \(1/\sqrt{N}\) rate and its asymptotic covariance matrix depends only on the \(s\) parameters of the correct model. That is, FETWFE estimates the model with the same asymptotic efficiency as an OLS-estimated ETWFE estimator that knows the \(p_N - s\) correct restrictions. FETWFE leverages the possible restrictions in the ETWFE estimator to improve efficiency without requiring the practitioner to make assumptions on which restrictions to choose, which could induce asymptotic bias if the wrong restrictions are chosen or compromise asymptotic efficiency if not enough restrictions are chosen. This addresses the need to consider both bias and variance when estimating treatment effects \citep{gelman2021slamming}.

In Remark \ref{eff.rmk} we mentioned that the inclusion of the parameters \eqref{t.1.treated} in our regression allows for unbiased estimation of treatment effects under both (CCTS) and unconfoundedness of the untreated potential outcomes, but under unconfoundedness the parameters equal exactly 0 and are inefficient to include. However, Theorem \ref{te.oracle.thm} shows that when coefficients that equal 0 are included, FETWFE sets these coefficients equal to exactly 0 with high probability and estimates the nonzero coefficients with the same asymptotic efficiency as if those coefficients were omitted from the model. FETWFE is therefore a good model to choose if one wants to maximize robustness without compromising asymptotic efficiency. The same is true if (CIUN) holds and the regression estimands \(\boldsymbol{\xi}_t^*\) all equal \(\boldsymbol{0}\), or if any irrelevant covariates are included in the regression. So practitioners do not need to hesitate to include a possibly relevant covariate that could be helpful for the plausibility of Assumption (CCTSB) or (CCTS) as long as \(\boldsymbol{Z}\) has full column rank.

%Finally, we note that Theorems \ref{te.oracle.thm} and \ref{te.asym.norm.thm} can be extended to analogous results for the estimated conditional average treatment effect estimators \eqref{catt.estimand.fixed.est} and \eqref{catt.estimand.weighted.est} in the fixed \(d_N = d\) setting.

\subsubsection{Asymptotically Valid Confidence Intervals}

Theorem \ref{te.oracle.thm} is conceptual and yields insight on the kind of behavior we can expect from FETWFE in large samples. The following result, which has the same assumptions as Theorem \ref{te.oracle.thm}, is more practical and provides explicit formulas for variance estimators for asymptotically valid confidence intervals of treatment effect estimates.

\begin{theorem}[Asymptotic Confidence Intervals for FETWFE]\label{te.asym.norm.thm}

Assume that Assumptions (CNAS), (CCTSB), and (LINS) hold, as well as Assumptions (F1), (F2), S(\(s\)) for a fixed \(s\), and (R1) - (R6). Assume \(q \in (0,1)\). Suppose that either Assumptions (CTSA) and (CIUN) hold or Assumption (CCTSA) holds. Let \(\{\psi_{rt}\}\) be an arbitrary set of finite constants, and for all of the below results, assume at least one of the \( \tau_{\text{ATT}} (r,t )\) corresponding to a nonzero \(\psi_{rt}\) is nonzero (otherwise, all of the below sequences of random variables converge in probability to 0).

\begin{enumerate}[(a)]

\item For the estimator \eqref{att.estimator.fixed} it holds that
\begin{align*}
\sqrt{ \frac{NT}{\hat{v}_N^{(\text{C})} }}  \sum_{r \in \mathcal{R}} \sum_{t=r}^T \psi_{rt}  ( \hat{\tau}_{\text{ATT}} (  r,t )  - \tau_{\text{ATT}} (r,t ) ) \xrightarrow{d} ~ &  \mathcal{N}(0, 1)
,
\end{align*}
where \(\hat{v}_N^{(\text{C})} \) is a finite-sample variance estimator defined in  \eqref{v.n.r.t.att.const}. (All variance estimators depend only on the known \(\sigma^2\) and the observed data.)

\item Suppose one set of data of size \(N\) is used to estimate the cohort probabilities \(N_r/ N_\tau\) and another independent data set of size \(N\) is used to estimate each \(\hat{\tau}_{\text{ATT}} ( r, t)\). Then for the estimator \eqref{att.estimator.weighted} it holds that
\begin{align*}
 \sqrt{ \frac{NT}{   \hat{v}_N^{\text{C}; \hat{\pi}}  }}  \sum_{r \in \mathcal{R}} \sum_{t =r}^T \psi_{rt}     \left(   \frac{N_r}{N_\tau}  \hat{\tau}_{\text{ATT}} ( r, t) -  \tilde{\pi}_r \tau_{\text{ATT}} (r, t) \right) 
 \xrightarrow{d} ~ &   \mathcal{N}(0, 1)
,
\end{align*}
where \( \hat{v}_N^{\text{C}; \hat{\pi}}  \) is a finite-sample variance estimator defined in  \eqref{v.n.r.t.att.rand}. 

\item Suppose a single set of data of size \(N\) is used to estimate both the cohort probabilities \(N_r/ N_\tau\) and each \(\hat{\tau}_{\text{ATT}} ( r, t)\). Then the estimator \eqref{att.estimator.weighted} is asymptotically subgaussian: the sequence of random variables
\begin{align*}
 \sqrt{ \frac{NT}{   \hat{v}_N^{\text{C, (cons)}; \hat{\pi}}  }}  \sum_{r \in \mathcal{R}} \sum_{t =r}^T \psi_{rt}     \left(   \frac{N_r}{N_\tau}  \hat{\tau}_{\text{ATT}} ( r, t) -  \tilde{\pi}_r \tau_{\text{ATT}} (r, t) \right) 
\end{align*}
converges in distribution to a mean-zero subgaussian random variable with variance at most 1, where \( \hat{v}_N^{\text{C, (cons)}; \hat{\pi}}  \) is a finite-sample conservative variance estimator defined in \eqref{v.n.r.t.att.rand.cons}.

\end{enumerate}

\end{theorem}

\begin{proof} Provided in Section \ref{app.prove.main.res} of the appendix.
\end{proof}

Theorem \ref{te.asym.norm.thm} shows that we can form asymptotically valid \(1 - \alpha\) confidence intervals with estimator \(\hat{T}_N\) and variance estimator \(\hat{v}_N\) with
\begin{equation}\label{conf.int.form}
\text{CI}_N(\hat{v}_N; \alpha) := \left[ \hat{T}_N  - \Phi^{-1} \left(1 - \frac{\alpha}{2} \right) \sqrt{\frac{\hat{v}_N}{NT}}, \hat{T}_N  + \Phi^{-1} \left(1 - \frac{\alpha}{2} \right) \sqrt{\frac{\hat{v}_N}{NT}}  \right]
,
\end{equation}
where \(\Phi(\cdot)\) is the distribution function of the standard normal distribution. 
%(See Lemma \ref{conf.int.formula.lem} in the appendix for details.)

Theorem \ref{te.asym.norm.thm}(b) requires data splitting for asymptotic normality of the estimator \eqref{att.estimator.weighted}. However, the data used for estimating the marginal cohort probabilities do not need to include observed responses. In some domains like public health and medicine where difference-in-differences has many potential uses \citep{rothbard2023tutorial, doi:10.1146/annurev-publhealth-040617-013507}, such \textit{unlabeled} data may be more widely available than \textit{labeled} data \citep[Section 1.3.2]{Candes2018}. In a setting like this, asymptotic normality can be achieved without reducing the amount of data used to estimate the model. (It is straightforward to adjust the variance estimator appropriately if the sample sizes of the two data sets are not equal.)

Theorem \ref{te.asym.norm.thm}(c) does not require sample splitting, at the price of sacrificing asymptotic normality for mere asymptotic subgaussianity and necessitating a conservative variance estimator. Because a subgaussian random variable can be defined as a variable \(X\) satisfying \(\mathbb{P}(|X| \geq x) \leq 2 \exp\{-x^2/K^2\}\) for some \(K > 0\) for all \(x 
\geq 0\) \citep[Proposition 2.5.2(a)]{vershynin2018high}, the tails of such a variable decay at least as quickly as a Gaussian random variable, so we might expect \(95\%\) confidence intervals to have similar asymptotic validity to the confidence intervals from part (b). In practice, we show in the simulation studies in Section \ref{synth.exps.sec} that these confidence intervals do seem to retain validity.

We cannot test the null hypothesis of treatment effects equalling 0 under Theorem \ref{te.asym.norm.thm}, since these test statistics are not asymptotically normal if the estimated coefficients all equal 0; instead they converge in probability to 0 (even after scaling by \(\sqrt{NT}\)).

Theorem \ref{te.asym.norm.thm.gen.cond} in the appendix is a generalization of both Theorems \ref{te.oracle.thm} and \ref{te.asym.norm.thm} for the classes of conditional average treatment effects estimators \eqref{catt.estimand.fixed.est} and \eqref{catt.estimand.weighted.est} in the fixed \(d_N = d\) setting under the assumption that an asymptotically Gaussian estimator of the generalized propensity scores (like the multinomial logit or proportional odds models) is available and an extra split of the data is available to estimate the cohort covariate means.

\section{Simulation Studies}\label{synth.exps.sec}

To test the efficacy of FETWFE under our assumptions, we conduct simulation studies in R using the \texttt{fetwfe} package \citep{Faletto:2025aa} and the \texttt{simulator} package \citep{bien2016simulator}. We choose parameters that bear resemblance to the empirical application from Section III of \citet{goodman2021difference} (\(N = 51, T = 33, R = 12, d_N = 3\)), which we explore in Section \ref{sec.data.app}. We use fewer cohorts for simplicity of presentation, more covariates to increase \(p_N\) in compensation for the smaller \(R\), and more units to avoid realizations where no units are assigned to a cohort by random chance. We generate data with \(N =  120\) units, \(T = 30\) time periods, \(R = 5\) cohorts entering at times \(\{2, \ldots, 6\}\), and \(d_N = 12\) features, which results in a total of \(p_N = 2209\) covariates and \(NT = 3600\) observations. To generate a \(\boldsymbol{\beta}_N^*\) such that \(\boldsymbol{D}_N \boldsymbol{\beta}_N^*\) is sparse, we generate a single random sparse \(\boldsymbol{\theta}_N^* \in \mathbb{R}^{p_N}\), then transform this into a coefficient vector \(\boldsymbol{\beta}_N^* = \boldsymbol{D}_N^{-1} \boldsymbol{\theta}_N^*\) to use across all simulations. We generate \(\boldsymbol{\theta}_N^*\) by taking a \(p_N\)-vector of all 0 entries and setting each entry equal to 2 randomly with probability 0.1. We set the sign of each term randomly, but since each entry of \(\boldsymbol{\beta}_N^*\) is a sum of terms in \(\boldsymbol{\theta}_N^*\), we bias the signs to avoid treatment effects that are too close to 0: individual nonzero terms in \(\boldsymbol{\theta}_N^*\) are positive with probability 0.6 and negative otherwise.

On each of 700 simulations, we randomly generate \(N\) independent realizations of \(\boldsymbol{X}_{(i)} \sim \mathcal{N}(\boldsymbol{0}, \boldsymbol{I}_{d_N})\) for the time-invariant covariates. We randomly assign treatments with probabilities \((\pi_0, \pi_1, \ldots, \pi_R) = (1/(R+1), \ldots, 1/(R+1))\), where
\begin{equation}\label{pi.r.def}
\pi_r := \begin{cases}
\mathbb{P}(W_i = 0), & r = 0,
\\ \mathbb{P}(W_i = r), & r \in \mathcal{R},
\end{cases}
\end{equation}
are the marginal probabilities of treatment assignments. To ensure the model is estimable, in the rare instances where there are no untreated units or one cohort has no observed units, we draw another set of random assignments.

After drawing covariates and treatment assignments, we construct \(\boldsymbol{\tilde{Z}}\). Then we draw \(N\) independent random effects \(c_i \sim \mathcal{N}(0, 5)\) as well as \(NT\) independent noise terms \(u_{(it)} \sim \mathcal{N}(0, 5)\) to generate \(\boldsymbol{\epsilon}_{(i\cdot)} = c_i \boldsymbol{1}_T + \boldsymbol{u}_{(i\cdot)}\) and finally generate \(\boldsymbol{\tilde{y}} = \boldsymbol{\tilde{Z}} \boldsymbol{\beta}_N^*  + \boldsymbol{\epsilon}_{(\cdot \cdot)}\).

Besides FETWFE, we also consider three competitor methods. The first is ETWFE, estimated by OLS on \(\boldsymbol{\tilde{Z}}\). We also consider bridge regression on \(\boldsymbol{Z}\) directly penalizing \(\boldsymbol{\beta}\) rather than penalizing \(\boldsymbol{D}_N \boldsymbol{\beta}\) (BETWFE). For both FETWFE and BETWFE, we follow the experiments of \citet{kock2013oracle} and use \(q = 0.5\), selecting the penalty \(\lambda_N\) over a grid of 100 values (equally spaced on a logarithmic scale) using BIC. We implement bridge regression for both BETWFE and FETWFE using the R \texttt{grpreg} package \citep{breheny2009penalized}. Finally, we also consider a slightly more flexible version of \eqref{bad.reg}, where we estimate the model
\[
\tilde{y}_{it} = \nu_r  + \gamma_t + \boldsymbol{X} \boldsymbol{\kappa} + \sum_{r \in \mathcal{R}} \tau_r \cdot \mathbbm{1}\{W_i = r \} \mathbbm{1}\{t \geq r\}  + \epsilon_{it}
\]
by OLS (TWFE\_COVS). (Notice that adding unit-specific covariates to Model \eqref{bad.reg} induces collinearity. Fitting the model instead on cohort fixed effects, like ETWFE, also matches, e.g., Equation 3.2 in \citet{callaway2021difference}. See also the discussion around Equation 2.5 in \citealt{sant2020doubly}.) The inclusion of covariates might make it seem more plausible that Assumptions (CNAS) and (CCTS) could hold in this model, though \citet[Section 6]{wooldridge2021two} points out that adding time-invariant covariates does not change the estimates of the treatment effects. Separate treatment effects for each cohort, while insufficient to avoid bias and capture the variation in treatment effects, should improve estimation relative to \eqref{bad.reg} and allows for the estimation of cohort-specific treatment effects to compare to the other methods.

\subsection{Estimation Error}\label{sec.exp.est.err}

We start by evaluating the error of each method in estimating the average treatment effect \eqref{att.def}. The assumptions of Theorem \ref{main.te.cons.thm} are satisfied, so we expect FETWFE to estimate the treatment effects more accurately than ETWFE, which will suffer from high variance; BETWFE, which does not assume the correct form of sparsity; and TWFE\_COVS, which is asymptotically biased. We estimate the average treatment effect using the FETWFE estimator \eqref{att.estimator.weighted} with \(\psi_{rt}\) from  \eqref{avg.cohorts.psi.r.t}, as well as analogous estimators using the competitor estimated regression coefficients and observed cohort counts. 

We calculate the squared error of each average treatment effect estimate for each method on each iteration. Boxplots of the results are displayed in Figure \ref{att_mse_fig}. We also provide the means and standard errors for the squared errors of each method in Table \ref{att.mse.tab}. Table \ref{att.mse.t.tab} contains \(p\)-values for paired one-tailed \(t\)-tests of the alternative hypothesis that the squared error for FETWFE is less than the squared error for each competitor method; all results show significantly better performance for FETWFE at the 0.05 significance level. 

These results show that even if most of the restrictions we would naturally consider hold and we are only interested in overall average treatment effect \eqref{att.def}, a simple model like TWFE\_COVS is an unsuitable estimator. This aligns with the conclusions of previous works, as we discussed in the introduction. We see that ETWFE has better estimation error, though in practice it still does not perform as well as FETWFE because it estimates about ten times as many parameters as are needed, leading to imprecise estimates. 

\begin{figure}[htbp]
\begin{center}
\includegraphics[scale=0.7]{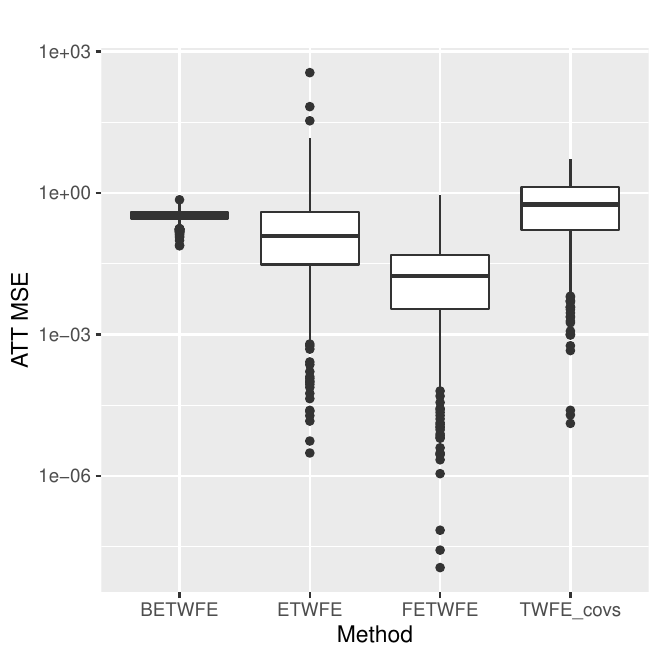}
\caption{Boxplots of squared errors for each treatment effect estimate across all 700 simulations. Vertical axis is on a log scale.}
\label{att_mse_fig}
\end{center}
\end{figure}

%stargazer(att_mse_res$mean_se_df, summary=FALSE)
% Table created by stargazer v.5.2.3 by Marek Hlavac, Social Policy Institute. E-mail: marek.hlavac at gmail.com
% Date and time: Mon, May 15, 2023 - 17:32:57
\begin{table}[!htbp] \centering 
  \caption{Means and standard errors for the squared error of the estimator of the ATT \eqref{att.def} for each method across all 700 simulations in the first simulation study from Section \ref{synth.exps.sec}.} 
  \label{att.mse.tab} 
\begin{tabular}{@{\extracolsep{5pt}} cccc} 
\\[-1.8ex]\hline 
\hline \\[-1.8ex] 
FETWFE & ETWFE & BETWFE & TWFE\_COVS \\ 
\hline \\[-1.8ex] 
0.0399 (0.00263) & 1.19 (0.633) & 0.339 (0.00336) & 0.938 (0.0395) \\ 
\hline \\[-1.8ex] 
\end{tabular} 
\end{table} 

%stargazer(att_mse_res$p_values, summary=FALSE)
% Table created by stargazer v.5.2.3 by Marek Hlavac, Social Policy Institute. E-mail: marek.hlavac at gmail.com
% Date and time: Mon, May 15, 2023 - 17:33:09
\begin{table}[!htbp] \centering 
  \caption{\(p\)-values from paired one-tailed \(t\)-tests of the alternative hypothesis that the FETWFE ATT estimator \eqref{att.estimator.weighted} has a lower squared estimation error for the ATT \eqref{att.def} than each competitor method calculated across all 700 simulations in the first simulation study from Section \ref{synth.exps.sec}. (Low \(p\)-values indicate better predictive performance for FETWFE.)} 
  \label{att.mse.t.tab} 
\begin{tabular}{@{\extracolsep{5pt}} ccc} 
\\[-1.8ex]\hline 
\hline \\[-1.8ex] 
ETWFE & BETWFE & TWFE\_COVS \\ 
\hline \\[-1.8ex] 
0.0343 & 1.24e-320 & 1.68e-86 \\ 
\hline \\[-1.8ex] 
\end{tabular} 
\end{table}

In Appendix \ref{synth.exp.details}, we also present estimation error results for the cohort-specific average treatment effects \( \tau_{\text{ATT}} (r)\) as defined in \eqref{att.cohort} for each of the five cohorts, and estimation error for the coefficients of \(\boldsymbol{\beta}_N^*\) corresponding to the interaction effects, \(\boldsymbol{\rho}^*\). These results show that FETWFE outperforms the competitor methods for these estimands as well.

\subsection{Restriction Selection Consistency}\label{rest.sel.cons}

Next, we examine the extent to which we can trust FETWFE to select the correct restrictions, as we know it does with asymptotically high probability under Theorem \ref{te.sel.cons.thm}. In the same experiment as before, on each simulation we calculate the percentage of restrictions FETWFE successfully identifies among the treatment parameters. In particular, we calculate the percentage of entries of \(\boldsymbol{D} \boldsymbol{\beta}_N^*\) that FETWFE successfully estimates as either zero or nonzero.

On average across the 700 simulations, FETWFE successfully decides whether to use a restriction or not in \(95.0\%\) of the \(p_N\) possible cases, with a standard error of \(0.025\%\). Of the true restrictions (that is, the entries of \(\boldsymbol{D} \boldsymbol{\beta}_N^*\) that equal 0), on average \(97.5\%\) (\(0.021\%\)) are correctly identified. This suggests that FETWFE is effective at identifying restrictions in practice, which contributes to FETWFE's success in estimating treatment effects that we saw in the previous section. (Figure \ref{sel.cons.fig} in Appendix \ref{synth.exp.details} shows a boxplot of the proportion of restrictions successfully identified across all 700 simulations.)

To get a better sense of the asymptotics, we conduct an additional simulation study in the same way as the first one, but with \(T = 5\),\footnote{Notice that increasing \(T\) increases the number of observations while also increasing the number of treatment effects and interactions to estimate. So we might expect larger \(T\) to have a roughly neutral effect on the convergence properties, while greatly increasing the computational cost of simulations with large \(N\).} \(R = 3, d_N = 2\) (resulting in \(p_N = 50\)), and \(N = 1200\) observations (resulting in 300 observations per cohort in expectation). The sparse \(\boldsymbol{\theta}_N^*\) is generated as before, but with entries each equalling 0 with probability 0.5. Again, we conduct 700 simulations. In this second simulation study, FETWFE successfully decides whether to use a restriction or not in \(99.3\%\) (\(0.042\%\)) of the \(p_N\) possible cases across the 700 simulations, and on average FETWFE identifies \(98.5\%\) (\(0.096\%\)) of the true restrictions.

\subsection{Asymptotic Distribution}

Finally we investigate the asymptotic distribution of the FETWFE estimators. On each simulation, we estimate nominal asymptotic \(95\%\) confidence intervals for the cohort average treatment effects using \eqref{conf.int.form} and the variance estimator from Theorem \ref{te.asym.norm.thm}(a). In the first simulation study, we find coverage rates of 0.791 for the \(r = 2\) cohort (standard error 0.015), 0.884 for the \(r = 3\) cohort (standard error 0.012), \(0.854\) (0.013) for the \(r = 4\) cohort, \(0.803\) (0.015) for the \(r = 5\) cohort, and \(0.850\) (0.014) for the \(r = 6\) cohort. We see that the coverage is unsatisfactory for FETWFE in this relatively small sample size, where in expectation 20 units are observed in each cohort.

We also estimate conservative confidence intervals for the average treatment effect with the conservative variance estimator from Theorem \ref{te.asym.norm.thm}(c), yielding a coverage rate of 0.993 (0.001). That is, the conservative confidence intervals seem to achieve the desired coverage level in practice, even though the asymptotic distribution of the test statistic is not known to be Gaussian (merely subgaussian).

Lastly, we examine the non-conservative estimator from Theorem \ref{te.asym.norm.thm}(b): on each simulation study, we generate \(N\) independent unlabeled samples where only cohort assignments (drawn from the same distribution as the \(N\) labeled units) are observed, and we use these to estimate asymptotically exact confidence intervals for the average treatment effect. The resulting coverage is 0.986 (0.001). The coverage may be better for the average treatment effect than the cohort average treatment effects because the number of observations that are relevant for estimating the overall average treatment effect is much larger than that for each of the cohort average treatment effects.

In the second simulation study, described previously in Section \ref{rest.sel.cons}, we find coverage rates of 0.944 (0.009) for the \(r = 2\) cohort, 0.941 (0.009) for the \(r = 3\) cohort, and 0.930 (0.010) for the third cohort. The conservative confidence intervals for the average treatment effects have coverage 0.997 (0.001), and the asymptotically exact confidence intervals have coverage 0.973 (0.002). Overall, the coverage rates are close to the nominal \(95\%\) level in this finite \(N\) setting. This is especially promising because we are using heuristics to solve the FETWFE optimization problem, as discussed in Remark \ref{opt.prob.rmk}.

\section{Empirical Application}\label{sec.data.app}

We use the R \texttt{fetwfe} package \citep{Faletto:2025aa} to analyze the data from \citet{stevenson2006bargaining}, as collected in the \texttt{divorce} data set from \citet{bacondecomp}. \citeauthor{stevenson2006bargaining} estimate the causal effect of unilateral or ``no-fault" divorce laws on suicide rates among women. The data set consists of observations from \(N=51\) states (including the District of Columbia) for \(T = 33\) years, along with demographic covariates. The response variable is the aggregate state suicide rate among women by year (as in the first column of Table 1 of \citealt{stevenson2006bargaining}). See \citet{stevenson2006bargaining} and \citet[Section III]{goodman2021difference} for more details on the data.

We remove states treated at time \(T = 1\) from the data, leaving \(N = 42\) observations. 5 states are never treated, and \(R = 12\) cohorts are treated at various times. Similarly to the analysis in \citet[Section IV(B)]{goodman2021difference}, we consider the time-varying covariates \texttt{murderrate} (state female homicide rate), \texttt{lnpersinc} (natural log of state personal income
per capita), and \texttt{afdcrolls} (state welfare participation rate), including the values from time \(t = 1\) as time-invariant controls. However, \texttt{murderrate} is missing for some observations at \(t = 1\) so is omitted, leaving \(d_N = 2\) covariates. In total, we have \(NT = 1386\) observations and \(p_N = 908\) coefficients to estimate.

We estimate the variances \(\sigma_c^2\) (the variance of the unit random effects) and \(\sigma^2\) (the variance of the added noise) using the estimators from \citet[Section 26.5.1]{pesaran-2015-text} with ridge regression. We use FETWFE with \(q = 0.5\) to estimate the marginal average treatment effects, again leveraging the R \texttt{grpreg} package as in the simulation studies and selecting \(\lambda_N\) by BIC. We estimate that the overall average percent change to the suicide rate for women is \(-3.76\%\), which is similar to the estimate of \(-2.52\%\) (standard error \(1.09\%\)) that \citet{goodman2021difference} calculates when using similar controls. Applying Theorem \ref{te.asym.norm.thm}(c), we estimate a conservative standard error of \(4.58\%\), yielding a conservative \(95\%\) confidence interval of \((-12.73\%, 5.21\%)\) for the average treatment effect.

In Table \ref{cohort.treat.effs} we present point estimates for the estimated average treatment effects for each of the cohorts. (We omit standard errors since the individual cohorts have very small sample sizes.) We estimate that four cohorts have non-zero treatment effects, with the 1970 cohort (consisting of California and Iowa) having by far the largest treatment effect estimate. We emphasize that these individual estimates are likely less reliable than the average treatment effect estimate since most cohorts have a small number of units.

% Table created by stargazer v.5.2.3 by Marek Hlavac, Social Policy Institute. E-mail: marek.hlavac at gmail.com
% Date and time: Thu, Nov 30, 2023 - 16:24:16
\begin{table}[!htbp] \centering 
  \caption{Cohort average treatment effect estimates from the empirical application in Section \ref{sec.data.app}, targeting estimand \( \tau_{\text{ATT}} (r)\) as defined in \eqref{att.cohort}.} 
  \label{cohort.treat.effs} 
\begin{tabular}{@{\extracolsep{5pt}} cc} 
\\[-1.8ex]\hline 
\hline \\[-1.8ex] 
 Cohort & \( \hat{\tau}_{\text{ATT}} (r)\) \\ 
\hline \\[-1.8ex] 
1969 & $0\%$ \\ 
 1970 & $$-$40.142\%$ \\ 
 1971 & $0\%$ \\ 
 1972 & $0\%$ \\ 
 1973 & $$-$3.466\%$ \\ 
 1974 & $0\%$ \\ 
 1975 & $0\%$ \\ 
 1976 & $$-$4.703\%$ \\ 
 1977 & $$-$5.338\%$ \\ 
 1980 & $0\%$ \\ 
 1984 & $0\%$ \\ 
 1985 & $0\%$ \\ 
\hline \\[-1.8ex] 
\end{tabular} 
\end{table}

\section{Conclusion}\label{conc.sec}

In this work, we have proposed the fused extended two-way fixed effects estimator for difference-in-differences with staggered adoptions. Our estimator leverages the fact that the extended two-way fixed effects estimator includes coefficients that are likely to be similar (for example, because of their proximity in time), but it selects appropriate restrictions automatically rather than relying on hand-selected restrictions, which could be biased. FETWFE retains the asymptotic unbiasedness (and normality) of ETWFE but leverages sparsity for large efficiency gains. In particular, we showed that our estimator is consistent, selection consistent, asymptotically Gaussian, and has an oracle property for several classes of heterogeneous marginal and conditional average treatment effects. Finally, we demonstrated the usefulness of FETWFE in practice in both simulation studies and an empirical application. 

\subsection{Future Work}

Throughout the paper we have mentioned places where our work might be relatively straightforward to extend. Here we specify some other possible extensions. It would be useful to develop theory allowing for FETWFE to be estimated when the design matrix \(\boldsymbol{Z}\) is deficient in column rank; for example, when \(p_N > NT\). One could extend our method to this setting by suitably adapting the marginal bridge estimator of \citet{kock2013oracle}, which involves a screening step to remove irrelevant covariates before estimating a regression on the reduced design matrix in a second step. This allows for consistency and asymptotic normality of the second-step estimator under similar assumptions to those that we have used in this work. Alternatively, consistency (but not asymptotic normality) could be established under weaker assumptions in the \(q = 1\) case by suitably adapting theory for lasso estimators like Corollary 2 from \citet{negahban2012unified}.

Another potentially useful extension would be relaxing the parallel trends assumption to allow for heterogeneous trends as described in \citet[Section 7.2]{wooldridge2021two}. This requires adding extra parameters to the ETWFE estimator, but FETWFE is well-poised to adapt to this setting because it can automatically restrict to 0 any parameters that are added unnecessarily.

\section*{Acknowledgement}

I thank Jacob Bien for very helpful discussions that shaped the direction of this work and edits to an early draft.

\section*{Funding}

This work was supported by the University of Southern California.

%This research did not receive any specific grant from funding agencies in the public, commercial, or not-for-profit sectors.

\bibliographystyle{abbrvnat}
\bibliography{mybib2fin}

\newpage

\appendix

\numberwithin{equation}{section}

In Appendix \ref{par.trend.ciuu.app} we investigate in more detail the relatively plausibility of conditional and marginal parallel trends, how Assumption (CIUN) interacts with these assumptions, and how we can achieve consistency of FETWFE for marginal average treatment effects under either form of parallel trends if (CIUN) holds. We present additional results from the simulation studies from Section \ref{synth.exps.sec} in Appendix \ref{synth.exp.details}. We state Theorem \ref{first.thm.fetwfe}, which Theorems \ref{main.te.cons.thm}, \ref{te.sel.cons.thm}, \ref{te.oracle.thm}, and \ref{te.asym.norm.thm} all depend on, in Appendix \ref{first.thm.state.sec}, and we prove it in Appendix \ref{sec.prove.first.thm}. In Appendix \ref{app.ext.sec} we state Theorem \ref{main.te.cons.thm.gen}, which shows the consistency of FETWFE for broader classes of estimators than those considered in Theorem \ref{main.te.cons.thm}. Theorem \ref{te.asym.norm.thm.gen} similarly extends Theorem \ref{te.asym.norm.thm}, and Theorem \ref{te.asym.norm.thm.gen.cond} shows that two classes of conditional treatment effect estimators are asymptotically Gaussian and provides consistent variance estimators for constructing asymptotically valid confidence intervals, analogously to Theorem \ref{te.asym.norm.thm.gen} for marginal treatment effect estimators. We prove the results from Appendix \ref{app.ext.sec}, along with Theorems \ref{main.te.cons.thm}, \ref{te.oracle.thm}, and \ref{te.asym.norm.thm}, in Appendix \ref{app.prove.main.res}. In Appendix \ref{app.extra.proofs} we provide proofs for results that were stated but not proven in Appendix \ref{par.trend.ciuu.app} and supporting results used in Appendix \ref{app.prove.main.res}. We prove supporting results for Theorem \ref{first.thm.fetwfe} in Appendix \ref{sec.main.lemmas}. Finally, in Appendix \ref{tech.lem.sec} we state and prove some technical lemmas.

\section{More on Parallel Trends and Assumption (CIUN)}\label{par.trend.ciuu.app}

In Section \ref{par.trend.app} we discuss the plausibility of marginal and conditional parallel trends in relation to an assumption of marginal or conditional (respectively) independence of the treatment and the potential outcomes. The latter assumptions imply the former, which motivates difference-in-differences estimators like FETWFE as more plausible than identification strategies that rely on stronger untestable assumptions.

We point out that no such strict hierarchy exists to relate marginal and conditional parallel trends in Section \ref{thm.discus.sec}---there exist settings where either assumption can hold without the other holding. Investigating this sheds light on when conditional or marginal parallel trends may be more plausible. We later show (Theorem \ref{te.interp.prop}) that conditional parallel trends is necessary to estimate conditional average treatment effects via FETWFE, but there are estimators that are consistent for marginal average treatment effects under either conditional or marginal parallel trends. So it is natural to wonder what conditions allow an estimator to be consistent for marginal average treatment effects under either assumption, and whether this might be possible for FETWFE. In Section \ref{pt.robust.sec} we explore when this might hold, ultimately investigating Assumption (CIUN) in more detail and showing how it affects the consistency of regression \eqref{wooldridge.6.33.model}. These ideas underlie the proof of Theorem \ref{main.te.cons.thm}.

\subsection{Relationships Between Unconfounded Treatment Assignment, Marginally Randomized Treatment Assignment, and Parallel Trends Assumptions}\label{par.trend.app}

In this section we present Proposition \ref{unconf.ccts.cts.thm.app}, which relates the plausibility of several assumptions on the potential outcomes. Ultimately, this motivates difference-in-differences estimators like FETWFE as having more plausible untestable assumptions than methods that rely on unconfoundedness (like matching) or randomized assignment.

We consider an assumption of unconfoundedness of the untreated potential outcomes, which we define in this setting as
\begin{equation}\label{unconfound.def}
\tilde{y}_{it}(0) \indep W_i \ \mid \ \boldsymbol{X}_i \qquad i \in [N], t \in \{1, \ldots, T\}
.
\end{equation}
This requires that treatment assignment is independent of the untreated potential outcomes conditional on the time-invariant (pre-treatment) covariates \(\boldsymbol{X}_i\). \citet[Section 3.2.2]{abadie2005semiparametric} considered a similar assumption in the \(T = 2\) case. As we discuss in more detail later in this section, it has previously been noted that the assumption of unconfoundedness \citep[Def. 3.6]{imbens_rubin_2015} that is common in causal inference is a stronger assumption than (CCTS).

Sometimes unconfoundedness is taken to be a stronger assumption that all of the potential outcomes are independent of treatment status conditional on covariates; such a condition would imply \eqref{unconfound.def}. \citet[Assumption 2.1]{shaikh2021randomization} is one example of a work that uses such a stronger assumption in a setting with continuous treatment status.

We note that different notions of unconfoundedness have been considered in the difference-in-differences setting. For example, \citet[Section 6.5.4]{imbens2009recent} consider unconfoundedness in a two-period difference-in-differences setting with no covariates, where the lagged outcome is conditioned on rather than covariates. Similarly, \citet{callaway2023policy} consider unconfoundedness conditional on several lagged time-varying variables. \citet{ding2019bracketing} and \citet{lindner2019difference} consider an ignorability condition where the untreated potential outcome is independent of treatment assignment status conditional on covariates and the lagged outcome.

We also refer to an assumption of marginal independence of the untreated potential outcomes and treatment status,
\begin{equation}\label{marg.ind.def}
\tilde{y}_{it}(0) \indep W_i  \qquad i \in [N], t \in \{1, \ldots, T\}.
\end{equation}
Completely randomized assignment, an assumption considered by, for example, \citet[Assumption 1]{athey2022design} and \citet[Case 1]{roth2023parallel}, is sufficient for \eqref{marg.ind.def}.

Lastly, we say that (CCTS) or (CTS) holds in a \textit{transformation-invariant} sense if \eqref{ccts.def} or \eqref{cts} (respectively) hold under any transformation \(h(\tilde{y}_{it}(r))\) of the potential outcomes such that the expectations exist. Some works have expressed concern that parallel trends assumptions are fragile because they may hold only if the response is not transformed, or only under a specific transformation \citep{meyer1995natural, athey2006identification, kahn2020promise}. \citet{roth2023parallel} investigate in detail when marginal parallel trends assumptions are transformation-invariant in the \(T =2\) setting, and they point out that their arguments extend straightforwardly to the arbitrary \(T\) setting as well as to conditional parallel trends assumptions. Proposition \ref{unconf.ccts.cts.thm.app} explicitly draws attention to an implication of their Propositions 3.1 and 3.2 when extended to these more general settings.

\begin{proposition}\label{unconf.ccts.cts.thm.app}
Assume that the conditional expectations in \eqref{ccts.def} and \eqref{cts} exist and are finite. 

\begin{enumerate}[(a)]

\item Unconfoundedness of the untreated potential outcomes is a strictly stronger assumption than conditional parallel trends: \eqref{unconfound.def} is sufficient for Assumption CCTS to hold in a transformation-invariant sense, but (CCTS) does not imply unconfoundedness of the untreated potential outcomes.

\item Marginal independence of treatment status and the untreated potential outcomes is a strictly stronger assumption than marginal parallel trends: \eqref{marg.ind.def} is sufficient for Assumption CTS to hold in a transformation-invariant sense, but (CTS) does not imply \eqref{marg.ind.def}.

\end{enumerate}

\end{proposition}

\begin{proof}

\begin{enumerate}[(a)]

\item

Under unconfoundedness \eqref{unconfound.def}, for any \(i \in [N]\), \(t \in \{1, \ldots, T\}, r \in \mathcal{R}\) almost surely it holds that
\begin{equation}\label{unconf.req}
  \E [ h(\tilde{y}_{(it)} (0) ) \mid  W_i  , \boldsymbol{X}_{i} ]  =  \E [ h(\tilde{y}_{(it)} (0) )\mid   \boldsymbol{X}_{i} ] 
\end{equation}
%Under unconfoundedness, \( \tilde{y}_{(it)} (0)\) is independent of 
for any arbitrary \(h(\cdot)\) such that the above expectation exists. For (CCTS) to hold with transformational invariance, we need that for any \(h(\cdot)\) such that the expectations exist and any \(i \in [N]\), \(t \in \{2, \ldots, T\}, r \in \mathcal{R}\), almost surely
\[
\E [ h(\tilde{y}_{(it)} (0) )- h(\tilde{y}_{(i1)}(0)) \mid  W_i , \boldsymbol{X}_{i} ]  = \E [ h(\tilde{y}_{(it)} (0)) -  h(\tilde{y}_{(i1)}(0)) \mid \boldsymbol{X}_{i} ] 
.
\] 
But this is immediate from \eqref{unconf.req}. We show that the reverse implication does not hold by constructing a counterexample. (CCTS) can be rearranged as
\begin{equation}\label{ccts.step}
\E [ \tilde{y}_{(it)} (0)  \mid  W_i, \boldsymbol{X}_{i} ]  -  \E [ \tilde{y}_{(it)} (0) \mid  \boldsymbol{X}_{i} ]   =  \E [  \tilde{y}_{(i1)}(0) \mid  W_i, \boldsymbol{X}_{i} ]  -  \E [  \tilde{y}_{(i1)}(0) \mid  \boldsymbol{X}_{i} ]  ~ \forall t \in [T]
.
\end{equation}
Suppose that for all \(t \in \{1, 2, \ldots, T\}\), it holds that \( \E [ \tilde{y}_{(it)} (0)  \mid  W_i, \boldsymbol{X}_{i} ]  -  \E [ \tilde{y}_{(it)} (0) \mid  \boldsymbol{X}_{i} ]  = f(W_i, \boldsymbol{X}_i)\) almost surely for an arbitrary \(f(\cdot, \cdot )\) that is not identically 0. Then from \eqref{ccts.step} we see that (CCTS) will hold. However, unconfoundedness does not hold since \eqref{unconf.req} with \(h(\cdot)\) as the identity function is violated.

\item

This proof is similar. Under marginal independence of treatment status and the untreated potential outcomes \eqref{marg.ind.def}, for any \(i \in [N]\), \(t \in \{1, \ldots, T\}, r \in \mathcal{R}\) almost surely it holds that
\begin{equation}\label{marg.ind.req}
  \E [ h(\tilde{y}_{(it)} (0) )  \mid  W_i   ]  =  \E [ h(\tilde{y}_{(it)} (0))  ] 
\end{equation}
%Under unconfoundedness, \( \tilde{y}_{(it)} (0)\) is independent of 
for any \(h(\cdot)\) such that the expectation holds. It is immediate that for any such \(h(\cdot)\) and any \(i \in [N]\), \(t \in \{2, \ldots, T\}, r \in \mathcal{R}\), almost surely \(\E [ h(\tilde{y}_{(it)} (0)) - h( \tilde{y}_{(i1)}(0)) \mid  W_i ]  = \E [ h(\tilde{y}_{(it)} (0) ) - h(\tilde{y}_{(i1)}(0)) ] \), so (CTS) holds with transformational invariance. On the other hand, (CTS) can be rearranged as
\begin{equation}\label{cts.step}
\E [ \tilde{y}_{(it)} (0)  \mid  W_i  ]  -  \E [ \tilde{y}_{(it)} (0)  ]   =  \E [  \tilde{y}_{(i1)}(0) \mid  W_i  ]  -  \E [  \tilde{y}_{(i1)}(0) ]  \qquad \forall t \in [T]
.
\end{equation}
If for all \(t \in \{1, 2, \ldots, T\}\) we have \( \E [ \tilde{y}_{(it)} (0)  \mid  W_i  ]  -  \E [ \tilde{y}_{(it)} (0) ]  = g(W_i )\) almost surely for an arbitrary \(g(\cdot )\) that is not identically 0, then \eqref{cts.step} holds but \eqref{marg.ind.req} with \(h(\cdot)\) as the identity function is violated. So (CTS) holds but marginal independence of treatment status and the untreated potential outcomes does not.

 \end{enumerate}
 
 \end{proof}

Since unconfoundedness often seems more plausible than marginal independence of the potential outcomes and treatment status, Proposition \ref{unconf.ccts.cts.thm.app}(a) and (b) considered together suggest an intuition that conditional parallel trends may often be more plausible than marginal parallel trends. Proposition \ref{unconf.ccts.cts.thm.app} also points out that difference-in-differences estimators are strictly more robust than identification strategies that rely on assumptions of completely randomized assignment or unconfoundedness. 

Proposition \ref{unconf.ccts.cts.thm.app}(a) verifies that (CCTS) is a strictly weaker assumption than unconfoundedness. As we mentioned earlier, this is a fairly straightforward observation that has been noted before; for example, \citet{abadie2005semiparametric, heckman1997matching}, and \citet{heckman1998characterizing} mention that this fact holds in the \(T = 2\) case. Besides the previously mentioned intuition that Proposition \ref{unconf.ccts.cts.thm.app}(a) and (b) provide that conditional parallel trends may be more plausible in many applications than marginal parallel trends, we also draw attention to the facts in Proposition \ref{unconf.ccts.cts.thm.app}(a) to point out explicitly that difference-in-differences identification strategies like FETWFE that rely on (CCTS) are more robust than methods like matching on time-invariant covariates that rely on an unconfoundedness assumption \citep{abadie2006large, imbens_rubin_2015}.

% Proposition \ref{unconf.ccts.cts.thm.app}(a) aligns with works like \citet{chabe2015analysis} which suggest that difference-in-differences with time-invariant covariates is more robust than matching.

% It is also consistent with the works like \citet{chabe2017should} and \citet{daw2018matching}, which show that combining matching on pre-treatment or time-invariant covariates with difference-in-differences estimation can fail to reduce bias, or even induce bias in treatment effect estimation in settings where difference-in-differences is unbiased.

The implications of Proposition \ref{unconf.ccts.cts.thm.app}(b) are analogous, and previous works have also shown that random or ``as good as random" treatment assignment is sufficient for (CTS) \citep{athey2022design, roth2023parallel, ghanem2023selection}.

\subsection{Neither Conditional Nor Marginal Parallel Trends is Strictly More Plausible}\label{thm.discus.sec}

Proposition \ref{unconf.ccts.cts.thm.app} showed that (CTS) and (CCTS) are strictly more plausible than (respectively) marginal independence of the potential outcomes and treatment and unconfoundedness. Theorem \ref{unconf.ccts.cts.thm} shows that (CTS) and (CCTS) cannot be related in this way---neither is strictly more plausible than the other.

\begin{theorem}\label{unconf.ccts.cts.thm}
Assume that the conditional expectations in \eqref{ccts.def} and \eqref{cts} exist and are finite. Neither of conditional parallel trends or marginal parallel trends is a relaxation of the other: Assumption CCTS \eqref{ccts.def} (conditional parallel trends) does not imply Assumption CTS \eqref{cts} (marginal parallel trends), nor does (CTS) imply (CCTS).

\end{theorem}

\begin{proof}
Provided in Appendix \ref{proofs.par.trend.ciuu.app}.

% \ref{app.prove.main.res}.
\end{proof}

We prove Theorem \ref{unconf.ccts.cts.thm} by providing counterexamples where one of the parallel trends assumptions holds but the other does not. We briefly provide some intuition behind the counterexamples here in order to highlight some of the relative strengths and weaknesses of each assumption. 

(CCTS) may hold when the time-invariant covariates have a time-varying effect on the potential outcomes, while (CTS) will not hold. This seems plausible if external conditions change that change the relevance of some features for the response. \citet{ham2022benefits} and \citet{callaway2023treatment} touch on closely related issues, investigating settings where parallel trends does not hold because time-invariant confounders have time-varying effects on the potential outcomes. We also expand on related ideas in Section \ref{pt.robust.sec}.

On the other hand, (CCTS) requires parallel trends to hold on every subgroup of \(\boldsymbol{X}_i\). But (CTS) can hold if this is violated and parallel trends holds merely on average across the distribution of the covariates. That is, for any \(r \in \{0\} \cup \mathcal{R}\), the quantity 
\begin{equation}\label{quantity.subg}
\E [ \tilde{y}_{(it)} (0) - \tilde{y}_{(i1)}(0) \mid  W_i = r, \boldsymbol{X}_{i} ]  -  \E [ \tilde{y}_{(it)} (0) - \tilde{y}_{(i1)}(0) \mid  \boldsymbol{X}_{i} ] 
\end{equation}
may be positive on some subgroups of \(\boldsymbol{X}_i\) if this is canceled out by \eqref{quantity.subg} being negative on other subgroups such that \eqref{cts} holds after marginalizing the left side of \eqref{ccts.def} across the distribution of \(\boldsymbol{X}_i\) conditional on \(W_i = r\) and the right side across the distribution of \(\boldsymbol{X}_i\) conditional on \(W_i = 0\). In this particular sense, (CTS) allows for more variation in the time trends across the distribution of \(\boldsymbol{X}_i\) than (CCTS). We also emphasize that both (CCTS) and (CTS) allow for heterogeneity in the sense that the values of \(\E [ \tilde{y}_{(it)} (0) - \tilde{y}_{(i1)}(0) \mid  W_i, \boldsymbol{X}_{i} ] \) and \( \E [ \tilde{y}_{(it)} (0) - \tilde{y}_{(i1)}(0) \mid  \boldsymbol{X}_{i} ] \) may vary with \(\boldsymbol{X}_i\). 

However, if \eqref{ccts.def} does not hold, \eqref{cts} holding exactly after marginalizing \eqref{ccts.def} appropriately may seem somewhat implausible. Nonetheless, it could be the case that \eqref{ccts.def} is not close to holding, but \eqref{cts} approximately holds. Relaxing (CTS) by allowing small violations has been considered in recent works on sensitivity analyses for difference-in-differences estimation \citep{manski2018right, ban2022generalized, rambachan2023more, ghanem2023selection}. That said, \citet[Section 7.2]{ghanem2023selection} point out that adjusting for covariates (that is, performing a sensitivity analysis under an assumption of conditional parallel trends that may not hold exactly) can result in smaller plausible ranges of marginal average treatment effects. 

Overall, this suggests an intuition that (CCTS) may be somewhat more plausible in practice than (CTS), but since neither assumption nests the other, an ideal estimator would be asymptotically unbiased under either assumption. We discuss this in more detail in the next section.

%Section \ref{pt.robust.sec}. 
%See Appendix \ref{par.trend.app}, next, for further discussion of why (CCTS) may be somewhat more plausible than (CTS) in practice.

\subsection{Parallel Trends, Assumption (CIUN), and Regression \eqref{wooldridge.6.33.model}}\label{pt.robust.sec}

Theorem \ref{unconf.ccts.cts.thm} motivates interest in estimators of marginal average treatment effects that are consistent under either conditional or marginal parallel trends. We have seen that Assumption (CIUN) is helpful for this---under the conditions of Theorem \ref{main.te.cons.thm}, FETWFE consistently estimates heterogeneous marginal average treatment effects under either conditional or marginal parallel trends if (CIUN) holds. In this section we fill in more detail about (CIUN), build some intuition for its role in the consistency of difference-in-differences estimators, and establish the consistency of regression \eqref{wooldridge.6.33.model} under (CIUN).

We first introduce a related assumption that, unlike (CIUN), is not directly testable.

%The punchline of this section is that there is a testable assumption under which regression \eqref{wooldridge.6.33.model} is parallel trends doubly-robust. 

\textbf{Assumption (CIU)} (``conditional independence of trends of the untreated potential outcomes"): Almost surely for all \(r \in \{0\} \cup \mathcal{R}\), \(i \in [N]\), and any \(t \in \{2, \ldots, T\}\),
\begin{align*}
 \E \left[ \tilde{y}_{(i t)}(0) - \tilde{y}_{(i1)}(0) \mid W_i = r \right]     = ~ &    \E \left[ \tilde{y}_{(i t)}(0) - \tilde{y}_{(i1)}(0) \mid W_i = r, \boldsymbol{X}_{i}   \right]  
.
\end{align*}

(CIU) is similar to (CIUN), but it requires the trends in the untreated potential outcomes to be mean-independent of \(\boldsymbol{X}_i\) for units in all cohorts, so it cannot be tested.

%Assumption CIU requires that the trends in the untreated potential outcomes are mean-independent of the covariates conditional on treatment status. 
It is immediate that (CCTS) and (CTS) are equivalent under (CIU).

\begin{proposition}\label{equiv.cond.ciua}
Under (CIU), conditional parallel trends (CCTS) and marginal parallel trends (CTS) are equivalent.
\end{proposition}

\begin{proof}
We can apply (CIU) on each side of (CTS) to immediately yield that for any \(i \in [N]\), \(r \in \mathcal{R}\), and any \(t \geq 2\), 
\begin{align*}
\E [ \tilde{y}_{(it)} (0)   - \tilde{y}_{(i1)} (0)\mid W_i = r ]   = ~ &   \E [ \tilde{y}_{(it)} (0) - \tilde{y}_{(i1)}(0) \mid W_i = 0 ] 
\\ \iff \qquad    \E [ \tilde{y}_{(it)} (0) - \tilde{y}_{(i1)}(0) \mid W_i = r, \boldsymbol{X}_{i} ]  = ~ &   \E [ \tilde{y}_{(it)} (0) - \tilde{y}_{(i1)}(0) \mid W_i = 0, \boldsymbol{X}_{i} ] 
.
\end{align*}

\end{proof}

One implication of Proposition \ref{equiv.cond.ciua} is that any method that consistently estimates the marginal average treatment effects under either (CCTS) or (CTS) is consistent under both assumptions if (CIU) holds. 

Although (CIUN) does not guarantee the equivalence of (CTS) and (CCTS), under (CIUN) conditional and marginal parallel trends are still connected.

\begin{proposition}\label{equiv.cond.ciuu}
Suppose (CIUN) holds. Then conditional parallel trends (CCTS) implies marginal parallel trends (CTS).
\end{proposition}

\begin{proof}
Applying (CIUN) to one side of (CCTS), we have that for any \(i \in [N]\), \(r \in \mathcal{R}\), and any \(t \geq 2\), 
\begin{align*}
 \E [ \tilde{y}_{(it)} (0) - \tilde{y}_{(i1)}(0) \mid W_i = r, \boldsymbol{X}_{i} ]  = ~ &   \E [ \tilde{y}_{(it)} (0) - \tilde{y}_{(i1)}(0) \mid W_i = 0, \boldsymbol{X}_{i} ] 
\\ \iff \qquad \E [ \tilde{y}_{(it)} (0)   - \tilde{y}_{(i1)} (0)\mid W_i = r , \boldsymbol{X}_{i}  ]   = ~ &   \E [ \tilde{y}_{(it)} (0) - \tilde{y}_{(i1)}(0) \mid W_i = 0 ] 
\\ \implies ~ \E \left[ \E [ \tilde{y}_{(it)} (0)   - \tilde{y}_{(i1)} (0)\mid W_i = r , \boldsymbol{X}_{i}  ] \mid W_i = r \right]   
 = ~ &  \E \left[  \E [ \tilde{y}_{(it)} (0) - \tilde{y}_{(i1)}(0) \mid W_i = 0 ]   \mid W_i = r \right] 
\\ \iff  \qquad \E \left[ \tilde{y}_{(it)} (0)   - \tilde{y}_{(i1)} (0)\mid W_i = r \right]   = ~ &   \E [ \tilde{y}_{(it)} (0) - \tilde{y}_{(i1)}(0) \mid W_i = 0 ]   
.
\end{align*}

\end{proof}

Proposition \ref{equiv.cond.ciuu} shows that any method that consistently estimates the marginal average treatment effects is also consistent under (CIUN) and (CCTS).

Notice that Assumptions (CIU) and (CIUN) are closely related to the time-varying effects of the covariates on the potential outcomes that we mentioned in our discussion of Theorem \ref{unconf.ccts.cts.thm} in Section \ref{thm.discus.sec}, which can cause marginal parallel trends to fail to hold even if conditional parallel trends holds. Proposition \ref{equiv.cond.ciuu} shows that (CIUN) rules out such settings.

Next we present a key result that underlies Theorem \ref{first.thm.fetwfe}, and in turn Theorems \ref{main.te.cons.thm}, \ref{te.oracle.thm}, and \ref{te.asym.norm.thm}.

\begin{theorem}[Assumption (CIUN) and the Consistency of Regression \ref{wooldridge.6.33.model}]\label{te.interp.prop}
Assume that (CNAS) and \eqref{treat.eff.def.covs} from (LINS) hold.
\begin{enumerate}[(a)]

\item (Consistency under conditional parallel trends.) The following are equivalent: (i) conditional parallel trends after treatment (CCTSA) holds, and (ii) for all \(r \in \mathcal{R}\) and all \(t \geq r\), for the regression estimands \(\tau_{rt}^*\) and \( \boldsymbol{\rho}_{rt}^*\) defined in \eqref{treat.eff.def.covs} it holds almost surely that the conditional average treatment effects \(\tau_{\text{CATT}} (r, t, \boldsymbol{X}_i)  \) defined in \eqref{catt.t.r.def} equal \(  \tau_{rt}^* +  \left(\boldsymbol{X}_i -  \E \left[ \boldsymbol{X}_{i} \mid  W_i = r \right] \right)^\top \boldsymbol{\rho}_{rt}^*\).

\item (Consistency under marginal parallel trends and CIUN.) If (CIUN) and marginal parallel trends after treatment (CTSA) hold, then for all \(r \in \mathcal{R}\) and all \(t \geq r\) the regression estimand \(\tau_{rt}^*\) is equal to the marginal average treatment effect \(\tau_{\text{ATT}}(r,t)\) from \eqref{att.cohort.time}..

\item (Identifying the bias-inducing covariates under CTSA if CIUN does not hold.) Suppose (CTSA) and \eqref{trend.params} from (LINS) hold. Then for any \(r \in \mathcal{R}\) and \(t \in \{r, \ldots, T\}\), 
\[
\tau_{rt}^*  = \tau_{\text{ATT}}(r, t)  \nonumber  -  \left(  \E[\boldsymbol{X}_i \mid W_i = r]   -  \E \left[  \boldsymbol{X}_i  \mid W_i = 0 \right]   \right)  ^\top \boldsymbol{\xi}_t^* 
.
\]

\end{enumerate}

\end{theorem}

\begin{proof}

 The proofs of parts (b) and (c) are provided in Appendix \ref{proofs.par.trend.ciuu.app}. Part (a) follows from \citet[Sections 6.3 and 6.4]{wooldridge2021two}, but for illustrative purposes we provide a proof in our notation here. For all \(r \in \mathcal{R}, t \geq r\) and any \(\boldsymbol{x}\) in the support of \(\boldsymbol{X}_i\),
\begin{align}
& \tau_{rt}^* + \boldsymbol{\dot{x}}_r^\top \boldsymbol{\rho}_{rt}^*  \nonumber
\\  \stackrel{(a)}{=} ~ &   \E\left[ \tilde{y}_{(i t)}(r)  - \tilde{y}_{(i 1)}(r) \mid W_i = r, \boldsymbol{X}_i  = \boldsymbol{x}\right]  
 - \E \left[  \tilde{y}_{(it)}(0) - \tilde{y}_{(i1)}(0)  \mid  W_i = 0, \boldsymbol{X}_{i} = \boldsymbol{x} \right] \nonumber
\\ \stackrel{(b)}{=} ~ &   \E [ \tilde{y}_{(i t)}(r) \mid W_i = r, \boldsymbol{X}_{i} = \boldsymbol{x} ] -   \bigg(  \E[\tilde{y}_{(i1)}(0) \mid W_i = r,   \boldsymbol{X}_{i} = \boldsymbol{x} ]   \nonumber
\\ &  +  \E [  \tilde{y}_{(i t)}(0) - \tilde{y}_{(i1)}(0) \mid W_i = 0, \boldsymbol{X}_{i} = \boldsymbol{x} ] \bigg) \nonumber
\\ \stackrel{(c)}{=}  ~ &  \E [ \tilde{y}_{(i t)}(r) \mid W_i = r, \boldsymbol{X}_{i} = \boldsymbol{x} ] -   \E [  \tilde{y}_{(i t)}(0) \mid W_i = r, \boldsymbol{X}_{i} = \boldsymbol{x} ] \nonumber
\\ = ~ & \tau_{\text{CATT}} (r, t, \boldsymbol{x}) 
 \label{stagger.tau.deriv}
 ,
\end{align}
where in \((a)\) we used \eqref{treat.eff.def.covs} from (LINS), in \((b)\) we used (CNAS), and \((c)\) follows from rearranging (CCTSA).

\end{proof}

\subsubsection{Discussion of Theorem \ref{te.interp.prop}}

Theorem \ref{te.interp.prop}(a) and (b) show the consistency of regression \eqref{wooldridge.6.33.model}: if (LINS), (CNAS), and (CIUN) hold, any estimator that is consistent for the regression estimands in \eqref{treat.eff.def.covs} is also consistent for the marginal average treatment effects under either conditional parallel trends after treatment or marginal parallel trends after treatment. Theorem \ref{te.interp.prop}(a) and (b) applies to ETWFE immediately under the conditions of Theorem 6.1(ii) in \citet{wooldridge2021two}, which shows that the ETWFE estimator of \(\tau_{rt}^*\) is consistent. It also applies to FETWFE in light of Theorem \ref{first.thm.fetwfe}(b), which is used to prove Theorem \ref{main.te.cons.thm}. 

Theorem \ref{te.interp.prop}(c) shows that it is specifically the covariates \(\boldsymbol{X}_i \mathbbm{1}\{t = t'\}\) for \(t' \in \{2, \ldots, T\}\) (which correspond to regression estimands \(\boldsymbol{\xi}_t^*\)) included in regression \eqref{wooldridge.6.33.model} that induce an asymptotic bias in the estimated marginal average treatment effects \(\tau_{\text{ATT}}(r,t')\) if (CIUN) does not hold. Suppose we assume that \eqref{trend.params} from (LINS) holds and \(\boldsymbol{X}_i\) and \(W_i\) are not mean-independent; in particular, for all \(r \in \mathcal{R}\) we have \(\E[\boldsymbol{X}_i \mid W_i = r] \neq \E[\boldsymbol{X}_i \mid W_i = 0]\)\footnote{Notice from Theorem \ref{te.interp.prop}(c) that \( \tau_{rt}^*  = \tau_{\text{ATT}}(r, t) \) if and only if either \(\E[\boldsymbol{X}_i \mid W_i = r] = \E[\boldsymbol{X}_i \mid W_i = 0]\) or \(\boldsymbol{\xi}_t^* = \boldsymbol{0}\). There is an interesting connection here: it is also true that (CCTS) implies (CTS) under either independence of \(\boldsymbol{X}_i\) and \(W_i\) or (CIUN), which under (LINS) is equivalent to \(\boldsymbol{\xi}_t^* = \boldsymbol{0}\). The latter fact is from Proposition \ref{equiv.cond.ciuu}. Under independence of \(\boldsymbol{X}_i\) and \(W_i\), the distribution of \(\boldsymbol{X}_i \) conditional on \(W_i = r\) is identical to the marginal distribution of \(\boldsymbol{X}_i\) for all \(r \in \{0\} \cup \mathcal{R}\), so we can marginalize each side of \eqref{ccts.def} across the marginal distribution of \(\boldsymbol{X}_i\) to yield \eqref{cts}.}. Then if (CCTSA) and (CIUN) both do not hold, not only does Theorem \ref{te.interp.prop}(a) show that there is no hope of estimating conditional average treatment effects, Theorem \ref{te.interp.prop}(c) shows that in general including the covariates \(\boldsymbol{X}_i \mathbbm{1}\{t = t'\}\) in a regression induces bias in the estimated average treatment effects (since if (CIUN) does not hold and \eqref{trend.params} from (LINS) does then \(\boldsymbol{\xi}_t^* \neq \boldsymbol{0}\) for at least some \(t\)). This is true even if (CTS) holds and the marginal average treatment effects are estimable (by, for example, model (6.15) in \citealt{wooldridge2021two}, which omits these covariates). That is, the covariates \(\boldsymbol{X}_i \mathbbm{1}\{t = t'\}\) for \(t' \in \{2, \ldots, T\}\) are ``bad controls," in the language of \citet{cinelli2022crash}. This observation aligns with the observation that \citet[Section 6.3]{wooldridge2021two} makes that including time-invariant covariates and their interactions with cohort dummies has no effect on the estimated marginal average treatment effects---it is specifically the inclusion of the covariate/time dummy interactions that changes the estimand for each \(\tau_{rt}^*\).

Notice that these covariates correspond exactly to the kind of effect of the covariates on the untreated trends that is ruled out by Assumptions (CIU) and (CIUN), as well as the time-varying effects of the covariates on the potential outcomes that we mentioned in our discussion of Theorem \ref{unconf.ccts.cts.thm} in Section \ref{thm.discus.sec}.
 
 However, if \eqref{trend.params} holds, we see that under (CIUN) the coefficients \(\boldsymbol{\xi}_t^*\) must all equal \(\boldsymbol{0}\). In this case, Theorem \ref{te.interp.prop}(c) shows that no harm is done to asymptotic unbiasedness by including those coefficients in the regression, and any estimator that is consistent for \(\tau_{rt}^*\) from  \eqref{treat.eff.def.covs} is consistent for \(\tau_{\text{ATT}}(r,t)\). (This aligns with part (b), which shows that if (CIUN) holds then we can consistently estimate the marginal average treatment effects under (CTSA) with regression \ref{wooldridge.6.33.model}.)
 
% We will show later that not only does FETWFE enjoy parallel trends double-robustness for the marginal average treatment effects, under the assumptions of 
 
 As we showed in Theorem \ref{te.oracle.thm}, under the required assumptions no harm is done to the asymptotic efficiency of FETWFE by including irrelevant covariates (as we mentioned earlier in Remark \ref{eff.rmk}), like when (CIUN) holds and all of the \(\boldsymbol{\xi}_t^*\) equal \(\boldsymbol{0}\).
 
Theorem \ref{te.interp.prop}(c) has similarities to Theorems 4.1 and 5.1 from \citet{ham2022benefits}, which analyze the \(T = 2\), \(\mathcal{R} = \{2\}\) and arbitrary \(T\), \(\mathcal{R} = \{T\}\) cases, respectively. In particular, their estimand \(\mathbb{E}[\hat{\tau}_{\text{DiD}}^{X} ]\) matches our estimand \(\tau_{rt}^*\) under our assumptions, and they also identify the bias of this estimand in order to investigate when conditional parallel trends might be more plausible than marginal parallel trends. Besides the fact that we consider a more general staggered adoptions (arbitrary \(\mathcal{R}\)) setting, they also consider different structural equations for the potential outcomes (for example, they omit interactions between the covariates and treatment status or cohort membership, though they do allow for time-varying coefficients on the covariates) and they consider a setting with a latent confounder.

\section{Additional Results For Simulation Studies}\label{synth.exp.details}

%\subsection{Implementation Details}

%\subsection{Investigating the cohort average treatment effect estimators}

Table \ref{att.cohort.sim.tab} contains results for the cohort-specific squared estimation errors for the cohort-specific average treatment effects \( \tau_{\text{ATT}} (r)\) as defined in \eqref{att.cohort}. We see that on average FETWFE outperforms every method for most cohorts, except that BETWFE outperforms FETWFE on average for cohort 2. This is because the true treatment effect for Cohort 2 is close to 0 (0.069), so BETWFE's direct penalization of the treatment coefficients directly towards 0 happens to be useful. In addition, this difference in the mean outcomes seems to be driven by outliers---we see from the boxplots in Figure \ref{cohort_2_mse_plot} that in the median case FETWFE still outperforms BETWFE. 

Table \ref{att.cohort.sim.p.tab} contains \(p\)-values for paired one-tailed \(t\)-tests of the alternative hypothesis that FETWFE has lower squared estimation error than the competitor methods, showing that the differences are statistically significant at the \(0.01\) level in all settings except for the Cohort 2 BETWFE estimate and the last two ETWFE estimates, which are not significant because the ETWFE estimator has very high variance (see Table \ref{att.cohort.sim.tab}).

Table \ref{rho.mean.se} contains results for the sum of squared errors of all methods' estimates of the treatment interaction coefficients \(\rho^*\). We see that FETWFE outperforms every method. Table \ref{rho.p.vals} contains \(p\)-values for paired one-tailed \(t\)-tests of the alternative hypothesis that FETWFE has lower squared estimation error than the competitor methods, showing that the differences are statistically significant at the \(0.05\) level.

Figure \ref{sel.cons.fig} depicts a boxplot of the proportion of the \(p_N\) possible restrictions that FETWFE selects correctly in the first simulation study. Figures \ref{cohort_2_mse_plot}, \ref{cohort_3_mse_plot}, \ref{cohort_4_mse_plot}, \ref{cohort_5_mse_plot}, and \ref{cohort_6_mse_plot} display boxplots of these squared errors, similar to Figure \ref{att_mse_fig} for the average treatment effect. These plots also indicate that FETWFE typically outperforms the other methods.

%stargazer(cohort_mse_results, summary=FALSE)
% Table created by stargazer v.5.2.3 by Marek Hlavac, Social Policy Institute. E-mail: marek.hlavac at gmail.com
% Date and time: Sun, May 21, 2023 - 14:14:58
\begin{table}[!htbp] \centering 
  \caption{Means and standard errors for the squared errors of estimates of \( \tau_{\text{ATT}} (r)\) by each method for each cohort in the simulation study from Section \ref{synth.exps.sec}.} 
  \label{att.cohort.sim.tab} 
\begin{tabular}{@{\extracolsep{5pt}} ccccc} 
\\[-1.8ex]\hline 
\hline \\[-1.8ex] 
 & FETWFE & ETWFE & BETWFE & TWFE\_COVS \\ 
\hline \\[-1.8ex] 
Cohort 2 & 0.035 (0.006) & 1.502 (0.312) & 0.008 (0) & 7.265 (0.163) \\ 
Cohort 3 & 0.037 (0.004) & 0.943 (0.19) & 0.425 (0.011) & 1.881 (0.091) \\ 
Cohort 4 & 0.129 (0.009) & 0.856 (0.179) & 3.492 (0.016) & 1.269 (0.062) \\ 
Cohort 5 & 0.093 (0.007) & 4.12 (3.079) & 4.794 (0.011) & 1.685 (0.082) \\ 
Cohort 6 & 0.068 (0.005) & 3.212 (2.016) & 3.759 (0.016) & 0.949 (0.049) \\ 
\hline \\[-1.8ex] 
\end{tabular} 
\end{table} 

%stargazer(cohort_mse_p_values, summary=FALSE)
% Table created by stargazer v.5.2.3 by Marek Hlavac, Social Policy Institute. E-mail: marek.hlavac at gmail.com
% Date and time: Sun, May 21, 2023 - 14:15:42
\begin{table}[!htbp] \centering 
  \caption{\(p\)-values for paired one-tailed \(t\)-tests of the alternative hypothesis that the FETWFE estimate of \( \tau_{\text{ATT}} (r)\) has lower squared error than the competitor methods for each cohort in the simulation study from Section \ref{synth.exps.sec}.} 
  \label{att.cohort.sim.p.tab} 
\begin{tabular}{@{\extracolsep{5pt}} cccc} 
\\[-1.8ex]\hline 
\hline \\[-1.8ex] 
 & ETWFE & BETWFE & TWFE\_COVS \\ 
\hline \\[-1.8ex] 
Cohort 2 & 1.14e-05 & 1 & 8.2e-206 \\ 
Cohort 3 & 8.06e-06 & 2.35e-145 & 1.85e-72 \\ 
Cohort 4 & 0.000119 & 0 & 3.89e-63 \\ 
Cohort 5 & 0.118 & 0 & 1.12e-68 \\ 
Cohort 6 & 0.0795 & 0 & 8.71e-60 \\ 
\hline \\[-1.8ex] 
\end{tabular} 
\end{table} 

% Table created by stargazer v.5.2.3 by Marek Hlavac, Social Policy Institute. E-mail: marek.hlavac at gmail.com
% Date and time: Mon, Aug 28, 2023 - 20:24:25
\begin{table}[!htbp] \centering 
  \caption{Means and standard errors for the summed squared errors of estimates of the treatment interaction coefficients \(\boldsymbol{\rho}^*\) by each method for each cohort in the simulation study from Section \ref{synth.exps.sec}.} 
  \label{rho.mean.se} 
\begin{tabular}{@{\extracolsep{5pt}} ccc} 
\\[-1.8ex]\hline 
\hline \\[-1.8ex] 
FETWFE & ETWFE & BETWFE \\ 
\hline \\[-1.8ex] 
0.529 (0.00827) & 38.1 (18.3) & 4.82 (0.0119) \\
\hline \\[-1.8ex] 
\end{tabular} 
\end{table} 

% Table created by stargazer v.5.2.3 by Marek Hlavac, Social Policy Institute. E-mail: marek.hlavac at gmail.com
% Date and time: Mon, Aug 28, 2023 - 20:25:12
\begin{table}[!htbp] \centering 
  \caption{\(p\)-values for paired one-tailed \(t\)-tests of the alternative hypothesis that the FETWFE estimates of  the treatment interaction coefficients \(\boldsymbol{\rho}^*\) have lower summed squared error than the competitor methods for each cohort in the simulation study from Section \ref{synth.exps.sec}.} 
  \label{rho.p.vals} 
\begin{tabular}{@{\extracolsep{5pt}} cc} 
\\[-1.8ex]\hline 
\hline \\[-1.8ex] 
ETWFE & BETWFE  \\ 
\hline \\[-1.8ex] 
0.02 & 0   \\ 
\hline \\[-1.8ex] 
\end{tabular} 
\end{table} 

\begin{figure}[htbp]
\begin{center}
\includegraphics[scale=0.7]{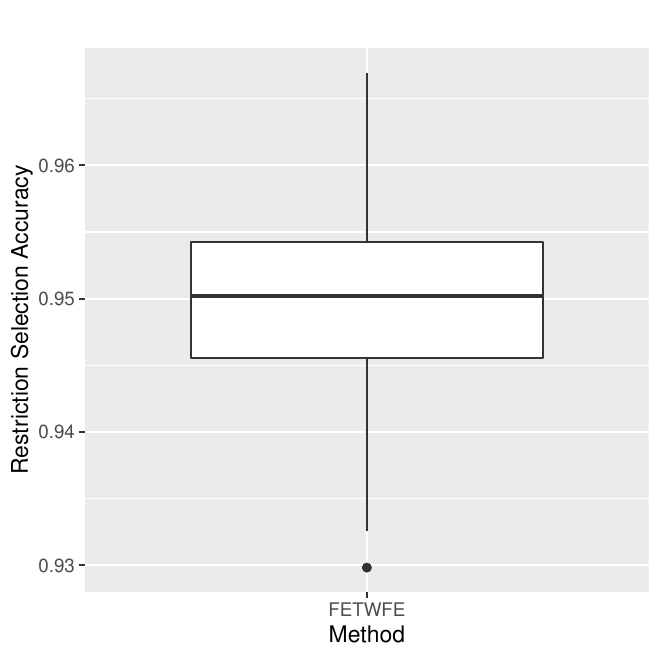}
\caption{Boxplot displaying proportions of correct treatment effect restriction decisions correctly made by FETWFE across each of the 700 simulations from the first simulation study in Section \ref{synth.exps.sec}.}
\label{sel.cons.fig}
\end{center}
\end{figure}

\begin{figure}[htbp]
\begin{center}
\includegraphics[scale=0.7]{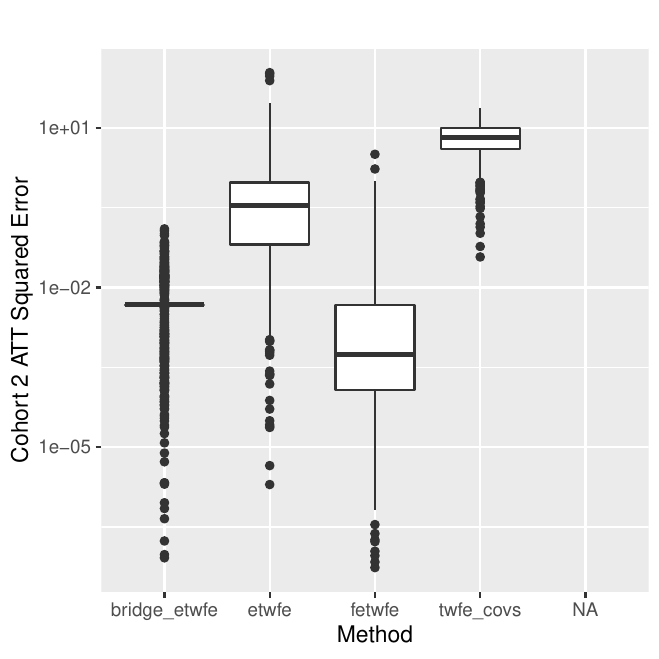}
\caption{Boxplots of squared errors for each method's estimate of \( \tau_{\text{ATT}} (2)\) across all 700 simulations. Vertical axis is on a log scale.}
\label{cohort_2_mse_plot}
\end{center}
\end{figure}

\begin{figure}[htbp]
\begin{center}
\includegraphics[scale=0.7]{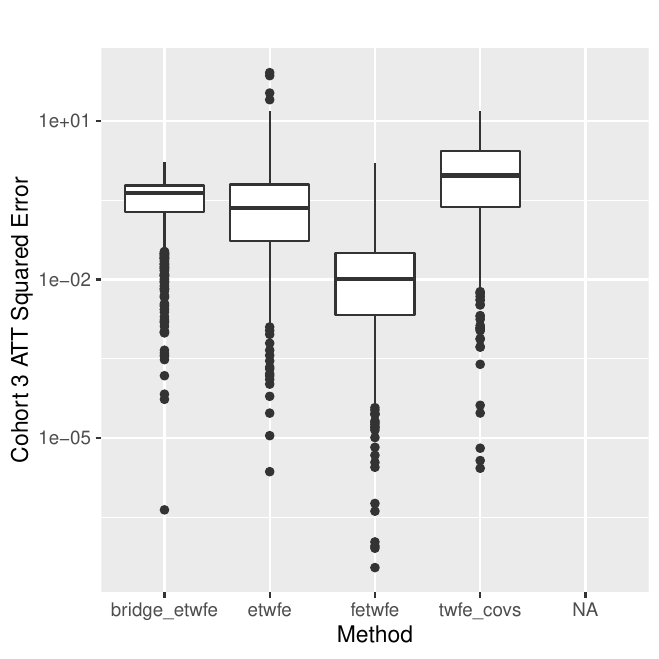}
\caption{Boxplots of squared errors for each method's estimate of \( \tau_{\text{ATT}} (3)\) across all 700 simulations. Vertical axis is on a log scale.}
\label{cohort_3_mse_plot}
\end{center}
\end{figure}

\begin{figure}[htbp]
\begin{center}
\includegraphics[scale=0.7]{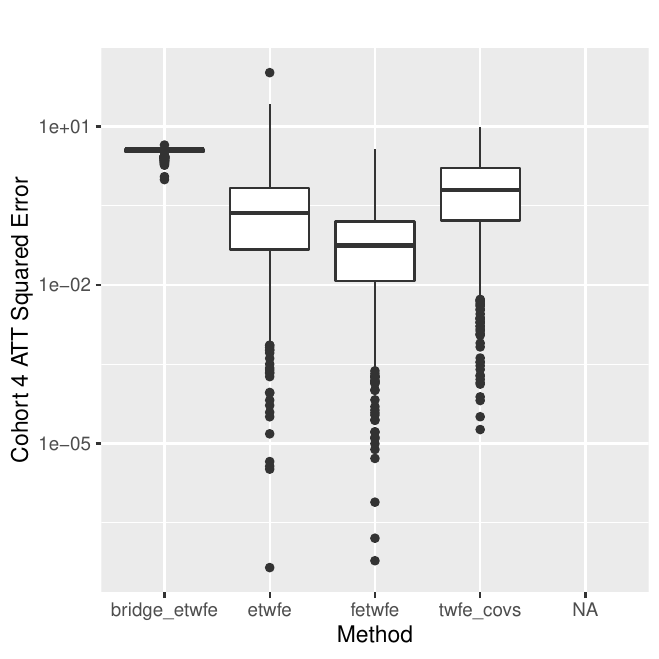}
\caption{Boxplots of squared errors for each method's estimate of \( \tau_{\text{ATT}} (4)\) across all 700 simulations. Vertical axis is on a log scale.}
\label{cohort_4_mse_plot}
\end{center}
\end{figure}

\begin{figure}[htbp]
\begin{center}
\includegraphics[scale=0.7]{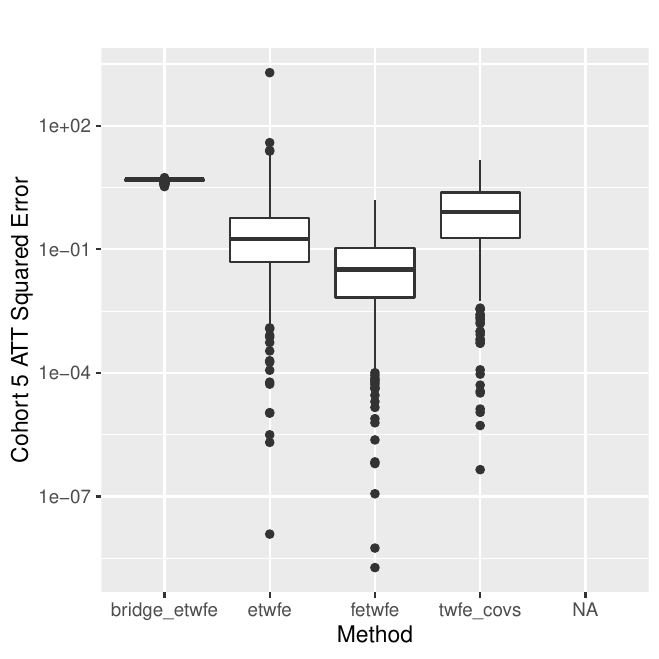}
\caption{Boxplots of squared errors for each method's estimate of \( \tau_{\text{ATT}} (5)\) across all 700 simulations. Vertical axis is on a log scale.}
\label{cohort_5_mse_plot}
\end{center}
\end{figure}

\begin{figure}[htbp]
\begin{center}
\includegraphics[scale=0.7]{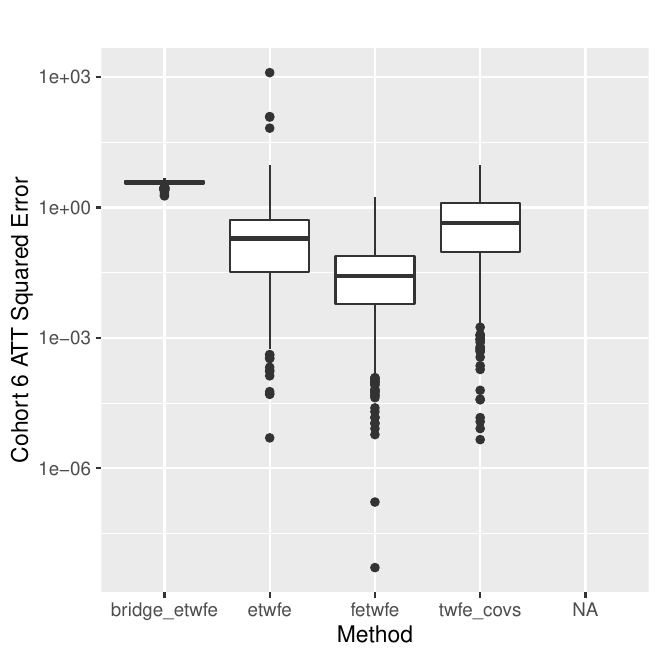}
\caption{Boxplots of squared errors for each method's estimate of \( \tau_{\text{ATT}} (6)\) across all 700 simulations. Vertical axis is on a log scale.}
\label{cohort_6_mse_plot}
\end{center}
\end{figure}

\section{Theorem \ref{first.thm.fetwfe}}\label{first.thm.state.sec}

Theorem \ref{first.thm.fetwfe} connects the bridge regression theory of \citet{kock2013oracle} to the difference-in-differences theory of \citet{wooldridge2021two}. It also makes use of some extensions of the theory of \citet{kock2013oracle} (Theorem \ref{prop.2i}, Theorem \ref{prop.ext.2}). Equipped with Theorem \ref{first.thm.fetwfe}, Theorems \ref{main.te.cons.thm}, \ref{te.sel.cons.thm}, \ref{te.oracle.thm}, and \ref{te.asym.norm.thm} are relatively straightforward to prove.

\begin{theorem}\label{first.thm.fetwfe}

\begin{enumerate}[(a)]

\item (Correct specification.) Suppose Assumptions (CNAS), (CCTSB), and (LINS) hold. Then for all \(i \in [N]\), \(t \in [T]\), \(r \in \mathcal{R}\), 
\begin{align*}
 \E[\tilde{y}_{(it)} \mid W_i = 0, \boldsymbol{X}_i] = ~ & \eta^* +  \gamma_t^*  + \boldsymbol{X}_i^\top(\boldsymbol{\kappa}^* + \boldsymbol{\xi}_t^*)   \qquad \text{and}
\\ \E[\tilde{y}_{(it)} \mid W_i = r, \boldsymbol{X}_i] = ~ &  \eta^* +  \gamma_t^*  +  \nu_r^* + \boldsymbol{X}_i^\top(\boldsymbol{\kappa}^* + \boldsymbol{\xi}_t^* + \boldsymbol{\zeta}_r^*  ) 
+ \mathbbm{1}\{ t \geq r\}  (\tau_{rt}^*  + \boldsymbol{\dot{X}}_{(ir)}^\top  \boldsymbol{\rho}_{rt}^*  )
,
\end{align*}
 where \(\boldsymbol{\dot{X}}_{(ir)} = \boldsymbol{X}_i - \E[ \boldsymbol{X}_i \mid W_i = r]\), for identifiability we define the time-dependent parameters to equal 0 at \(t = 1\), and recall that \(\eta^*, \gamma_t^*\), etc. are the regression estimands defined in (LINS). 
 
% Further, if (CCTS) holds,
%\begin{equation}\label{catt.spec.corr}
%  \tau_{\text{CATT}} (r, t, \boldsymbol{x}) =  \tau_{rt}^* + \boldsymbol{\dot{x}}_r \boldsymbol{\rho}_{rt}^*, \qquad r \in \mathcal{R}, t \geq r 
%  .
%\end{equation}

\item (Estimation consistency.) In addition to the assumptions from part \((a)\), suppose Assumptions (F1), (F2), S(\(s_N\)), and (R1) - (R3) hold. Then for any \(q>0\), \(\lVert \boldsymbol{\hat{\beta}}^{(q)} - \boldsymbol{\beta}_N^* \rVert_2 = \mathcal{O}_{\mathbb{P}} \left( \min\{h_N, h'_N\} \right)\), where \(h_N\) is defined in \eqref{h.n.def} and \(h_N'\) is defined in \eqref{h.n.prime.def}.

\item (Selection consistency.) Define
\begin{equation}\label{theta.def}
\boldsymbol{\theta}_N^* := \boldsymbol{D}_N \boldsymbol{\beta}_N^*
\end{equation}
and \(\boldsymbol{\hat{\theta}}^{(q)} := \boldsymbol{D}_N  \boldsymbol{\hat{\beta}}^{(q)}\). In addition to the previous assumptions, assume that \(q \in (0,1)\) and Assumptions (R4) - (R5) hold. Then  \(\lim_{N \to \infty} \mathbb{P} \left( \boldsymbol{\hat{\theta}}^{(q)}_{\mathcal{S}^c} = \boldsymbol{0} \right) = 1 \). (This also implies that \(\boldsymbol{\hat{\theta}}^{(q)}_{\mathcal{S}^c}  \xrightarrow{p} \boldsymbol{0}\), \( \lim_{N \to \infty} \mathbb{P} \left(  \sqrt{NT} \boldsymbol{\hat{\theta}}^{(q)}_{\mathcal{S}^c} = \boldsymbol{0} \right) = 1 \), and  \( \lim_{N \to \infty} \mathbb{P} \left(  \sqrt{NT}  \boldsymbol{v}_N^\top \boldsymbol{\hat{\theta}}^{(q)}_{\mathcal{S}^c} = 0 \right) = 1 \) for any sequence of finite vectors \(\boldsymbol{v}_N\ \in \mathbb{R}^{p_N}\).)

Further, let \(\hat{\mathcal{S}}_N := \{j: \hat{\theta}^{(q)}_j \neq 0\}\), and similarly let \(\mathcal{S}_N \subseteq [p_N]\) be the set of components of \(\boldsymbol{\theta}_N^*\) that are nonzero. Then
\[
\lim_{N \to \infty} \mathbb{P} \left( \hat{\mathcal{S}}_N = \mathcal{S}_N \right) = 1
.
\]

\item (Asymptotic normality and oracle property.) In addition to the previous assumptions, assume Assumption (R6) holds and Assumption S(\(s\)) holds for a fixed \(s_N = s\), so \(\mathcal{S}_N = \mathcal{S}\) is fixed as well. Let \(\{\boldsymbol{\psi}_N\}_{N=1}^\infty\) be a sequence of real-valued vectors where for each \(\boldsymbol{\psi}_N \in \mathbb{R}^{p_N}\), \((\boldsymbol{\psi}_N)_{\mathcal{S}} = \boldsymbol{\alpha} \in \mathbb{R}^s\) is fixed and finite and \((\boldsymbol{\psi}_N)_{[p_N] \setminus \mathcal{S}} = \boldsymbol{b}_N\), where \(\{\boldsymbol{b}_N\}_{N=1}^\infty\) is any sequence of constants where each \(\boldsymbol{b}_N \in \mathbb{R}^{p_N - s}\) contains all finite entries. Then if \(\boldsymbol{\alpha} \neq \boldsymbol{0}\),
 \[
 \sqrt{ NT  } \boldsymbol{\psi}_N^\top ( \boldsymbol{\hat{\theta}}^{(q)} - \boldsymbol{\theta}_N^*) \xrightarrow{d} \mathcal{N}\left( 0, \sigma^2 \boldsymbol{\alpha}^\top  \left( \Cov \left(\boldsymbol{Z} \boldsymbol{D}_N^{-1}  \right)_{(\cdot \cdot) \mathcal{S}} \right)^{-1} \boldsymbol{\alpha} \right)
 ,
 \]
where \(  \Cov \left(\boldsymbol{Z} \boldsymbol{D}_N^{-1}  \right)_{(\cdot \cdot) \mathcal{S}} \) is the columns and rows of the covariance matrix corresponding to the nonzero components of \(\boldsymbol{\theta}_N^*\).
%\begin{equation}\label{gram.def.sel}
%\boldsymbol{\hat{\Sigma}} ((\boldsymbol{Z} \boldsymbol{D}_N^{-1}) _{(\cdot \cdot) \mathcal{S}}) := \frac{1}{NT} \left( \boldsymbol{Z} \boldsymbol{D}_N^{-1} \right)_{\mathcal{S}}^\top   \left(\boldsymbol{Z} \boldsymbol{D}_N^{-1}\right)_{\mathcal{S}}
%.
%\end{equation}

 \item (Asymptotic normality with feasible variance estimator.) Maintain the assumptions from part \((d)\). For a random matrix \(\boldsymbol{A} \in \mathbb{R}^{NT \times p_N}\) with centered columns, define
 \begin{equation}\label{est.cov.def}
 \widehat{\Cov}(\boldsymbol{A}) := \frac{1}{NT}  \boldsymbol{A} ^\top \boldsymbol{A} 
 \end{equation}
 to be the sample covariance matrix of \(\boldsymbol{A}\). For any \(\mathcal{A}_N \subseteq [p_N]\) such that \(\mathcal{A}_N \neq \emptyset\), consider \( \widehat{\Cov} \left( (\boldsymbol{Z} \boldsymbol{D}_N^{-1} )_{(\cdot \cdot) \mathcal{A}_N} \right)\),
% \(\boldsymbol{\hat{\Sigma}} (( \boldsymbol{Z} \boldsymbol{D}_N^{-1} )_{(\cdot \cdot) \mathcal{A}_N})\),
 the estimated covariance matrix of the columns of \( \boldsymbol{Z} \boldsymbol{D}_N^{-1}\) corresponding to \(\mathcal{A}_N\). For sets \((\mathcal{A}_N)_{N=1}^\infty\), define \(\boldsymbol{\alpha}(\mathcal{A}_N) \in \mathbb{R}^{|\mathcal{A}_N|}\) to be a fixed, finite vector for each unique \(\mathcal{A}_N\), and define the sequence of \(p_N\)-vectors \( \{ \boldsymbol{\psi}_N(\mathcal{A}_N) \}_{N=1}^\infty \) to have components in the entries corresponding to \(\mathcal{A}_N\) equal to \(\boldsymbol{\alpha}(\mathcal{A}_N) \) and remaining entries \(\boldsymbol{b}_n(\mathcal{A}_N) \) all finite (but otherwise arbitrary). Define the sequence of random variables \(U_1(\mathcal{A}_N ) \) to equal
\[
\frac{1}{\sigma}  \sqrt{ \frac{NT}{\boldsymbol{\alpha}(\mathcal{A}_N)^\top \left( \widehat{\Cov} \left( (\boldsymbol{Z} \boldsymbol{D}_N^{-1} )_{(\cdot \cdot) \mathcal{A}_N} \right) \right)^{-1} \boldsymbol{\alpha}(\mathcal{A}_N)}}  \cdot    \left(  \boldsymbol{\psi}_N(\mathcal{A}_N)^\top \boldsymbol{\hat{\theta}}^{(q)} -   \boldsymbol{\psi}_N(\mathcal{S})^\top\boldsymbol{\theta}_N ^* \right)
\]
if \(\mathcal{A}_N \neq \emptyset\) and \( \widehat{\Cov} \left( (\boldsymbol{Z} \boldsymbol{D}_N^{-1} )_{(\cdot \cdot) \mathcal{A}_N} \right) \) is invertible and 0 otherwise. Then if \(\boldsymbol{\alpha}(\mathcal{S}) \neq \boldsymbol{0}\),
\[
U_1(\hat{\mathcal{S}}_N ) \xrightarrow{d} \mathcal{N}(0, 1)
.
\]

\item (Asymptotic normality when \(\boldsymbol{\alpha}\) is estimated independently from \(\boldsymbol{\hat{\theta}}^{(q)}\).) Maintain the assumptions from part \((e)\). For sets \((\mathcal{A}_N)_{N=1}^\infty\), let \(\boldsymbol{\hat{\alpha}}_N(\mathcal{A}_N) \in \mathbb{R}^{|\mathcal{A}_N|}\) be a sequence of random variables for each \(\mathcal{A}_N\) with the properties that for a given data set, \(\mathcal{A} = \mathcal{A}'\) implies \(\boldsymbol{\hat{\alpha}}_N(\mathcal{A})   =   \boldsymbol{\hat{\alpha}}_N(\mathcal{A'})    \) and 
\[
\sqrt{NT}    ( \boldsymbol{\theta}_N^* )^\top   \left(  \boldsymbol{\hat{\alpha}}_N(\mathcal{S})  -   \boldsymbol{\alpha} \right) \xrightarrow{d} \mathcal{N}(\boldsymbol{0}_s,   v_\psi(\boldsymbol{\alpha} ))
\]
for some \(v_\psi(\boldsymbol{\alpha} ) > 0\). Define the sequence of random vectors \( \{ \boldsymbol{\hat{\psi}}_N(\mathcal{A}_N) \}_{N=1}^\infty \) to have components in \(\mathcal{A}_N\) equal to \(\boldsymbol{\hat{\alpha}}(\mathcal{A}_N)\) and components \(\boldsymbol{\hat{b}}_N (\mathcal{A}_N) \) in the components corresponding to \([p_N] \setminus \mathcal{A}_N\) that are all almost surely finite for all \(N\). Define the sequence of random variables
\[
 U_2(\mathcal{A}_N )  :=  \sqrt{ NT}  \cdot   \left(  \boldsymbol{\hat{\psi}}_N(\mathcal{A}_N)^\top \boldsymbol{\hat{\theta}}^{(q)} - \boldsymbol{\psi}_N(\mathcal{S})^\top\boldsymbol{\theta}_{N}^* \right)
  .
\]
Assume \(\boldsymbol{\hat{\alpha}}_N(\mathcal{A}_N)\) and \(\boldsymbol{\hat{\theta}}^{(q)}\) are estimated on independent data sets of size \(N\). Then if \( \boldsymbol{\alpha}  \neq \boldsymbol{0}\),
\[
U_2(\hat{\mathcal{S}}_N ) \xrightarrow{d} \mathcal{N}\left( 0,  \sigma^2   \boldsymbol{\alpha}^\top \left(  \Cov \left( (\boldsymbol{Z} \boldsymbol{D}_N^{-1} )_{(\cdot \cdot) \mathcal{S}} \right)  \right)^{-1}  
\boldsymbol{\alpha} +  v_\psi(\boldsymbol{\alpha} )  \right)
.
\]

\item (Asymptotic normality when \(\boldsymbol{\alpha}\) is estimated independently from \(\boldsymbol{\hat{\theta}}^{(q)}\) with feasible variance estimator.) Maintain the assumptions from part \((f)\). Define the sequence of random variables \(U_3(\mathcal{A}_N ) \) to equal
\[
  \sqrt{ \frac{NT}{ \hat{v}_N (\mathcal{A}_N) }}  \cdot   \left(  \boldsymbol{\hat{\psi}}_N(\mathcal{A}_N)^\top \boldsymbol{\hat{\theta}}^{(q)} - \boldsymbol{\psi}_N(\mathcal{S})^\top\boldsymbol{\theta}_{N}^* \right)
\]
if \(\mathcal{A}_N \neq \emptyset\),  \(\boldsymbol{\hat{\alpha}}(\mathcal{A}_N) \neq \boldsymbol{0}\), and \( \widehat{\Cov} \left( (\boldsymbol{Z} \boldsymbol{D}_N^{-1} )_{(\cdot \cdot) \mathcal{A}_N} \right) \) is invertible and 0 otherwise, where 
\begin{equation}\label{var.est.kock.wooldridge}
\hat{v}_N (\mathcal{A}_N) :=  \sigma^2   \boldsymbol{\hat{\alpha}}(\mathcal{A}_N)^\top \left(  \widehat{\Cov} \left( (\boldsymbol{Z} \boldsymbol{D}_N^{-1} )_{(\cdot \cdot) \mathcal{A}_N} \right)  \right)^{-1}  
\boldsymbol{\hat{\alpha}}(\mathcal{A}_N) + \hat{v}_\psi(\boldsymbol{\alpha} ) 
\end{equation}
 and \(  \hat{v}_\psi(\boldsymbol{\alpha} )  \) is an estimator of \( v_\psi(\boldsymbol{\alpha} )  \) that satisfies \(\hat{v}_\psi(\boldsymbol{\alpha} ) \xrightarrow{p} v_\psi(\boldsymbol{\alpha} ) \). Assume \(\boldsymbol{\hat{\alpha}}_N(\mathcal{A}_N)\) and \(\boldsymbol{\hat{\theta}}^{(q)}\) are estimated on independent data sets of size \(N\). Then if \( \boldsymbol{\alpha} \neq \boldsymbol{0}\),
\[
U_3(\hat{\mathcal{S}}_N ) \xrightarrow{d} \mathcal{N}(0, 1)
.
\]

\item (Asymptotic subgaussianity when \(\boldsymbol{\alpha}\) is estimated on the same data set as \(\boldsymbol{\hat{\theta}}^{(q)}\) with feasible variance estimator.) Maintain the assumptions from part \((g)\). Define the sequence of random variables \(U_4(\mathcal{A}_N ) \) to equal
\[
  \sqrt{ \frac{NT}{ \hat{v}_N^{(\text{cons})} (\mathcal{A}_N) }}  \cdot   \left(  \boldsymbol{\hat{\psi}}_N(\mathcal{A}_N)^\top \boldsymbol{\hat{\theta}}^{(q)} - \boldsymbol{\psi}_N(\mathcal{S})^\top\boldsymbol{\theta}_{N}^* \right)
\]
if \(\mathcal{A}_N \neq \emptyset\),  \(\boldsymbol{\hat{\alpha}}(\mathcal{A}_N) \neq \boldsymbol{0}\), and \( \widehat{\Cov} \left( (\boldsymbol{Z} \boldsymbol{D}_N^{-1} )_{(\cdot \cdot) \mathcal{A}_N} \right) \) is invertible and 0 otherwise, where \(\hat{v}_N^{(\text{cons})} (\mathcal{A}_N) \) is defined as 
\begin{align}
 &  \sigma^2   \boldsymbol{\hat{\alpha}}(\mathcal{A}_N)^\top \left(  \widehat{\Cov} \left( (\boldsymbol{Z} \boldsymbol{D}_N^{-1} )_{(\cdot \cdot) \mathcal{A}_N} \right)  \right)^{-1}  
\boldsymbol{\hat{\alpha}}(\mathcal{A}_N)   \nonumber
\\ ~ & + 2 \sqrt{   \sigma^2   \boldsymbol{\hat{\alpha}}(\mathcal{A}_N)^\top \left(  \widehat{\Cov} \left( (\boldsymbol{Z} \boldsymbol{D}_N^{-1} )_{(\cdot \cdot) \mathcal{A}_N} \right)  \right)^{-1}  
\boldsymbol{\hat{\alpha}}(\mathcal{A}_N)  \hat{v}_\psi(\boldsymbol{\alpha} ) } 
  + \hat{v}_\psi(\boldsymbol{\alpha} )  \label{var.est.kock.wooldridge.subgauss.cons}
  .
\end{align}
Assume \(\boldsymbol{\hat{\alpha}}_N(\mathcal{A}_N)\) and \(\boldsymbol{\hat{\theta}}^{(q)}\) are estimated on the same data set of size \(N\). Then if \( \boldsymbol{\alpha} \neq \boldsymbol{0}\), \(U_4(\hat{\mathcal{S}}_N )\) converges in distribution to a mean-zero subgaussian random variable with variance at most 1.

\end{enumerate}

\end{theorem}

\begin{proof} Provided in Section \ref{sec.prove.first.thm}. 
\end{proof}

\section{Extensions of Main Results}\label{app.ext.sec}

Theorem \ref{main.te.cons.thm} is an immediate consequence of the following more general result that allows for estimands that include broader classes of functions of the marginal cohort membership probabilities or propensity scores.

\begin{theorem}[Consistency of FETWFE; extension of Theorem \ref{main.te.cons.thm}]\label{main.te.cons.thm.gen}

Assume that Assumptions (CNAS), (CCTSB), and (LINS) hold, as well as Assumptions (F1), (F2), S(\(s_N\)), and (R1) - (R3). Let \(q > 0\).

\begin{enumerate}[(a)]

\item Suppose that either Assumptions (CTSA) and (CIUN) hold or Assumption (CCTSA) holds. Then for any set of finite constants \(\{\psi_{rt}\}\),
\[
 \left| \sum_{r \in \mathcal{R}} \sum_{t=r}^T \psi_{rt}  \left( \hat{\tau}_{\text{ATT}} (  r,t )  -  \tau_{\text{ATT}} (r,t )  \right) \right|  = \mathcal{O}_\mathbb{P} \left( \min\{h_N, h_N'\}\right)
 \]
and for any differentiable function \(\boldsymbol{f}: [0,1]^{R +1} \to \mathbb{R}^R\) that is pointwise finite and has a Jacobian \(\nabla \boldsymbol{f} \in \mathbb{R}^{(R + 1) \times R}\) at \((\pi_0, \pi_{r_1}, \ldots, \pi_{r_R})\) that exists and is nonzero and finite, 
 \[
 \left|  \sum_{r \in \mathcal{R}} \sum_{t =r}^T \psi_{rt} \left( f_r(\hat{\pi}_0, \hat{\pi}_{r_1}, \ldots, \hat{\pi}_{r_R})    \hat{\tau}_{\text{ATT}} ( r, t) -  f_r(\pi_0, \pi_{r_1}, \ldots, \pi_{r_R})   \tau_{\text{ATT}} (r, t) \right) \right|  = \mathcal{O}_\mathbb{P} \left( \min\{h_N, h_N'\}\right)
  ,
  \]
where \(\hat{\pi}_r := N_r /N\) for all \(r \in \mathcal{R}\), \(h_N'\) is defined in \eqref{h.n.prime.def}, and \(h_N \) is defined in \eqref{h.n.def}.
%\begin{equation}\label{h.n.def}
%h_N :=  \frac{1}{e_{1N}} \sqrt{\frac{p_N}{N}} 
%%
%\end{equation}

\item
Suppose Assumption (CCTSA) holds and \(d_N = d\) is fixed. For all \(N\) and all \(r \in \{0\} \cup   \mathcal{R}\), assume all of the eigenvalues of \( \Cov ( \boldsymbol{X}_i \mid W_i = r)\) are bounded between \(\lambda_{\text{min}} > 0\) and \(\lambda_{\text{max}} < \infty\). Then for any set of finite constants \(\{\psi_{rt}\}\) and any fixed \(\boldsymbol{x}\) in the support of \(\boldsymbol{X}_i\),
\begin{align*}
\left| \sum_{r \in \mathcal{R}} \sum_{t=r}^T \psi_{rt} \left(    \hat{\tau}_{\text{CATT}}(r, t, \boldsymbol{x}) -  \tau_{\text{CATT}}(r, t, \boldsymbol{x})  \right) \right| =
\mathcal{O}_{\mathbb{P}}( \min\{h_N, h_N'\}  )  
.
\end{align*}
%Further, assume \(\hat{\pi}_r(\boldsymbol{x})\) is an estimator of \(\pi_r(\boldsymbol{x})\) that satisfies \( | \hat{\pi}_r(\boldsymbol{x}) -  \pi_r(\boldsymbol{x})| =  \mathcal{O}_{\mathbb{P}}(a_N)   \) for a decreasing sequence \(a_N\) for each \(r \in \{0\} \cup \mathcal{R}\) and all \(\boldsymbol{x} \) in the support of \(\boldsymbol{X}_i\).

Further, for any \(\boldsymbol{x}\) in the support of \(\boldsymbol{X}_i\), let \(\boldsymbol{\hat{\pi}}(\boldsymbol{x})\) be an estimator for \(\pi_r(\boldsymbol{x})\) satisfying \( | \hat{\pi}_r(\boldsymbol{x}) - \pi_r(\boldsymbol{x}) | = \mathcal{O}_{\mathbb{P}}(a_N)\) for a decreasing sequence \(a_N\) for all \(r \in \{0\} \cup \mathcal{R}\), and let \(\boldsymbol{a}_r \) and \(\boldsymbol{b} _r\in \mathbb{R}^{R + 1}\) be finite vectors satisfying \( \boldsymbol{a}_r^\top \boldsymbol{{\pi}}(\boldsymbol{x})  > 0\) and \( \boldsymbol{b}_r^\top \boldsymbol{{\pi}}(\boldsymbol{x})  > 0\) for any \(\boldsymbol{x}\) in the support of \(\boldsymbol{X}_i\). Then

%\begin{align*}
% \left|  \sum_{r \in \mathcal{R}} \sum_{t=r}^T \psi_{rt} \left(   \frac{\hat{\pi}_r(\boldsymbol{x})}{\sum_{r' \in \mathcal{R}} \hat{\pi}_{r'}(\boldsymbol{x})}   \hat{\tau}_{\text{CATT}}(r, t, \boldsymbol{x}) -  \tilde{\pi}_r(\boldsymbol{x}) \tau_{\text{CATT}}(r, t, \boldsymbol{x})  \right) \right| 
%%
%\xrightarrow{p} ~ & 
%0
%.
%\end{align*}
%In particular, if
%\[
% \left|  \frac{\hat{\pi}_r(\boldsymbol{x})}{ \sum_{r' \in \mathcal{R}} \hat{\pi}_{r'}(\boldsymbol{x})}  - \tilde{\pi}_r(\boldsymbol{x})  \right| = \mathcal{O}_{\mathbb{P}}(a_N) 
% \]
%for some sequence \(\{a_N\}\), then
\begin{align*}
& \left|  \sum_{r \in \mathcal{R}} \sum_{t=r}^T \psi_{rt} \left(  \frac{ \boldsymbol{a}_r^\top \boldsymbol{{\hat{\pi}}}(\boldsymbol{x})  }{ \boldsymbol{b}_r^\top \boldsymbol{{\hat{\pi}}}(\boldsymbol{x})  }  \hat{\tau}_{\text{CATT}}(r, t, \boldsymbol{x}) -    \frac{ \boldsymbol{a}_r^\top \boldsymbol{{\pi}}(\boldsymbol{x})  }{ \boldsymbol{b}_r^\top \boldsymbol{{\pi}}(\boldsymbol{x})  }   \tau_{\text{CATT}}(r, t, \boldsymbol{x})  \right) \right| 
\\ = ~ & 
\mathcal{O}_{\mathbb{P}}( \min\{h_N, h_N'\} \vee a_N )  
.
\end{align*}

\end{enumerate}

\end{theorem}

\begin{proof} Provided in Section \ref{app.prove.main.res}.
\end{proof}

Theorem \ref{te.asym.norm.thm.gen} is a similar generalization of Theorem \ref{te.asym.norm.thm}.

\begin{theorem}[Asymptotic Confidence Intervals for FETWFE; extension of Theorem \ref{te.asym.norm.thm}]\label{te.asym.norm.thm.gen}

Assume that Assumptions (CNAS), (CCTSB), and (LINS) hold, as well as Assumptions (F1), (F2), S(\(s\)) for a fixed \(s\), and (R1) - (R6). Assume \(q \in (0,1)\). Suppose that either Assumptions (CTSA) and (CIUN) hold or Assumption (CCTSA) holds. Let \(\{\psi_{rt}\}\) be an arbitrary set of finite constants, and for all of the below results, assume at least one of the \( \tau_{\text{ATT}} (r,t )\) corresponding to a nonzero \(\psi_{rt}\) is nonzero (otherwise, all of the below sequences of random variables converge in probability to 0).

\begin{enumerate}[(a)]

\item For the estimator \eqref{att.estimator.fixed} it holds that
\begin{align*}
\sqrt{ \frac{NT}{\hat{v}_N^{(\text{C})} }}  \sum_{r \in \mathcal{R}} \sum_{t=r}^T \psi_{rt}  ( \hat{\tau}_{\text{ATT}} (  r,t )  - \tau_{\text{ATT}} (r,t ) ) \xrightarrow{d} ~ &  \mathcal{N}(0, 1)
,
\end{align*}
where \(\hat{v}_N^{(\text{C})} \) is a finite-sample variance estimator defined in  \eqref{v.n.r.t.att.const}. (All variance estimators depend only on the known \(\sigma^2\) and the observed data.)

\item Suppose one set of data of size \(N\) is used to estimate the cohort probabilities \(N_r/ N_\tau\) and another independent data set of size \(N\) is used to estimate each \(\hat{\tau}_{\text{ATT}} ( r, t)\). Then for any differentiable function \(\boldsymbol{f}: [0,1]^{R +1} \to \mathbb{R}^R\) that is pointwise finite and has a Jacobian \(\nabla \boldsymbol{f} \in \mathbb{R}^{(R + 1) \times R}\) at \((\pi_0, \pi_{r_1}, \ldots, \pi_{r_R})\) that exists and is nonzero and finite, 
% \[
% \left|  \sum_{r \in \mathcal{R}} \sum_{t =r}^T \psi_{rt} \left( f_r(\hat{\pi}_0, \hat{\pi}_{r_1}, \ldots, \hat{\pi}_{r_R})    \hat{\tau}_{\text{ATT}} ( r, t) -  f_r(\pi_0, \pi_{r_1}, \ldots, \pi_{r_R})   \tau_{\text{ATT}} (r, t) \right) \right|  = \mathcal{O}_\mathbb{P} \left( \min\{h_N, h_N'\}\right)
%  ,
%  \]
%where \(\hat{\pi}_r := N_r /N\) for all \(r \in \mathcal{R}\), \(h_N'\) is defined in \eqref{h.n.prime.def}, and \(h_N \) is defined in \eqref{h.n.def}.
 it holds that
% \eqref{att.estimator.weighted.gen}
\begin{align*}
 \sqrt{ \frac{NT}{   \hat{v}_N^{\text{C}; \hat{\pi}}  }}  \sum_{r \in \mathcal{R}} \sum_{t =r}^T \psi_{rt}     \left(  f_r \left(   \hat{\pi}_{0},  \hat{\pi}_{r_1},  \ldots,  \hat{\pi}_{r_R} \right)  \hat{\tau}_{\text{ATT}} ( r, t) -   f_r \left(  \pi_{0}, \pi_{r_1}, \ldots,  \pi_{r_R} \right) \tau_{\text{ATT}} (r, t) \right) 
 \xrightarrow{d} ~ &   \mathcal{N}(0, 1)
,
\end{align*}
where \( \hat{v}_N^{\text{C}; \hat{\pi}}  \) is a finite-sample variance estimator defined in  \eqref{v.n.r.t.att.rand}.
%\eqref{v.n.r.t.att.rand}. 

\item Suppose a single set of data of size \(N\) is used to estimate both the cohort probabilities \(N_r/ N_\tau\) and each \(\hat{\tau}_{\text{ATT}} ( r, t)\). Then for any differentiable function \(\boldsymbol{f}: [0,1]^{R +1} \to \mathbb{R}^R\) that is pointwise finite and has a Jacobian \(\nabla \boldsymbol{f} \in \mathbb{R}^{(R + 1) \times R}\) at \((\pi_0, \pi_{r_1}, \ldots, \pi_{r_R})\) that exists and is nonzero and finite, the estimator \eqref{att.estimator.weighted} is asymptotically subgaussian: the sequence of random variables
\begin{align*}
 \sqrt{ \frac{NT}{   \hat{v}_N^{\text{C, (cons)}; \hat{\pi}}  }}  \sum_{r \in \mathcal{R}} \sum_{t =r}^T \psi_{rt}     \left(    f_r \left(  \hat{\pi}_{0}, \hat{\pi}_{r_1}, \ldots,  \hat{\pi}_{r_R} \right)    \hat{\tau}_{\text{ATT}} ( r, t) -   f_r \left(  \pi_{0}, \pi_{r_1}, \ldots,  \pi_{r_R} \right)  \tau_{\text{ATT}} (r, t) \right) 
\end{align*}
converges in distribution to a mean-zero subgaussian random variable with variance at most 1, where \( \hat{v}_N^{\text{C, (cons)}; \hat{\pi}}  \) is a finite-sample conservative variance estimator defined in \eqref{v.n.r.t.att.rand.cons}.

\end{enumerate}

\end{theorem}

\begin{proof} Provided in Section \ref{app.prove.main.res}.
\end{proof}

Theorem \ref{te.asym.norm.thm.gen.cond} shows the asymptotic normality of two classes of conditional treatment effect estimators, analogously to Theorem \ref{te.asym.norm.thm.gen}, and provides consistent variance estimators for constructing asymptotically valid confidence intervals. It also provides sufficient assumptions for the FETWFE estimators of conditional average treatment effects to converge at \(1/\sqrt{N}\) rates.

\begin{theorem}[Asymptotic Confidence Intervals for conditional average treatment effects in FETWFE]\label{te.asym.norm.thm.gen.cond}

Assume that Assumptions (CNAS), (CCTS), and (LINS) hold, as well as Assumptions (F1), (F2), S(\(s\)) for a fixed \(s\), and (R1) - (R6). Assume \(q \in (0,1)\), and assume \(d_N = d\) is fixed. Let \(\{\psi_{rt}\}\) be an arbitrary set of finite constants, and for both of the below results, assume at least one of the \( \tau_{\text{CATT}} (r,t , \boldsymbol{x})\) corresponding to a nonzero \(\psi_{rt}\) is nonzero (otherwise, all of the below sequences of random variables converge in probability to 0).

\begin{enumerate}[(a)]

\item For \(\hat{\tau}_{\text{CATT}} (r, t, \boldsymbol{x}) \) as defined in \eqref{catt.t.r.est}, suppose one set of data of size \(N\) is used to estimate each cohort conditional covariate mean \(\boldsymbol{\overline{X}}_r\) and another independent data set of size \(N\) is used to estimate each \(\hat{\tau}_{rt}^{(q)}\) and \(\hat{\rho}_{rt}^{(q)}\) via the FETWFE regression. Then for any set of finite constants \(\{\psi_{rt}\}_{r \in \mathcal{R}, t \in \{r, \ldots, T\}}\) and any fixed \(\boldsymbol{x}\) in the support of \(\boldsymbol{X}_i\), it holds that
\begin{align*}
\sqrt{ \frac{NT}{ \hat{v}_N^{(\text{C;CATT})} }}  \sum_{r \in \mathcal{R}} \sum_{t=r}^T \psi_{rt}  ( \hat{\tau}_{\text{CATT}} (  r,t , \boldsymbol{x} )  - \tau_{\text{CATT}} (r,t , \boldsymbol{x}) ) \xrightarrow{d} ~ &  \mathcal{N}(0, 1)
,
\end{align*}
where \( \hat{v}_N^{(\text{C;CATT})} \) is a finite-sample variance estimator defined in \eqref{v.n.r.t.catt.const}.

\item Assume we have an estimator \((\hat{\pi}_{0}(\boldsymbol{x}) , \hat{\pi}_{r_1}(\boldsymbol{x}) , \ldots, \hat{\pi}_{r_R}(\boldsymbol{x}) )\) of the generalized propensity scores \(( \pi_0(\boldsymbol{x}) ,  \pi_{r_1}(\boldsymbol{x}) , \ldots,  \pi_{r_R}(\boldsymbol{x}) )\) that satisfies for any fixed \(\boldsymbol{x}\) in the support of \(\boldsymbol{X}_i\)
\[
\sqrt{N} \left(\hat{\pi}_{0}(\boldsymbol{x}) - \pi_0(\boldsymbol{x}) , \hat{\pi}_{r_1}(\boldsymbol{x}) - \pi_{r_1}(\boldsymbol{x}) , \ldots, \hat{\pi}_{r_R}(\boldsymbol{x}) - \pi_{r_R}(\boldsymbol{x})  \right) \xrightarrow{d} \mathcal{N}(0, \boldsymbol{\Sigma}_\pi)
\]
for some finite positive definite \( \boldsymbol{\Sigma}_\pi \in \mathbb{R}^{(R+1) \times (R+ 1)}\). Assume we also have available a consistent estimator \(\boldsymbol{\hat{\Sigma}}_\pi \) that converges entrywise in probability to \( \boldsymbol{\Sigma}_\pi \). Suppose one set of data of size \(N\) is used to estimate each cohort conditional covariate mean \(\boldsymbol{\overline{X}}_r\), another independent data set of size \(N\) is used to estimate each \(\hat{\tau}_{rt}^{(q)}\) and \(\hat{\rho}_{rt}^{(q)}\) via the FETWFE regression, and a third independent data set is used to estimate the generalized propensity scores. Then for any differentiable function \(f: [0,1]^{R +1} \to \mathbb{R}^R\) that is pointwise finite and has a Jacobian \(\nabla \boldsymbol{f} \in \mathbb{R}^{(R + 1) \times R}\) at \((\pi_0(\boldsymbol{x}), \pi_{r_1}(\boldsymbol{x}), \ldots, \pi_{r_R}(\boldsymbol{x}))\) that exists and is nonzero and finite, it holds that
\begin{align*}
&  \sqrt{ \frac{NT}{  \hat{v}_N^{(\text{C;CATT,} \pi)} }}  \sum_{r \in \mathcal{R}} \sum_{t =r}^T \psi_{rt}     \big(  f_r \left(  \hat{\pi}_{0}(\boldsymbol{x}) , \hat{\pi}_{r_1}(\boldsymbol{x}) , \ldots,  \hat{\pi}_{r_R}(\boldsymbol{x})  \right)  \hat{\tau}_{\text{CATT}} ( r, t, \boldsymbol{x})
\\ &  -   f_r \left(  \pi_{0}(\boldsymbol{x}) , \pi_{r_1}(\boldsymbol{x}) , \ldots,  \pi_{r_R}(\boldsymbol{x})  \right) \tau_{\text{CATT}} (r, t, \boldsymbol{x}) \big) 
 \xrightarrow{d}    \mathcal{N}(0, 1)
,
\end{align*}
where \(\hat{v}_N^{(\text{C;CATT,} \pi)}  \) is a finite-sample variance estimator defined in \eqref{v.n.r.t.catt.const.est.pi}.

\end{enumerate}

\end{theorem}

\begin{proof} Provided in Section \ref{app.prove.main.res}.
\end{proof}

It is trivial to extend Theorem \ref{te.asym.norm.thm.gen.cond} to show that these classes of conditional average treatment effect estimators enjoy oracle properties analogous to those of Theorem \ref{te.oracle.thm}. Theorem \ref{te.asym.norm.thm.gen.cond} applies immediately to the classes of conditional average treatment effects estimators \eqref{catt.estimand.fixed.est} and \eqref{catt.estimand.weighted.est} by the same argument we used to prove Theorem \ref{main.te.cons.thm} since the Jacobian \(\nabla \boldsymbol{f} \in \mathbb{R}^{(R+1) \times R}\)
\[
\nabla \boldsymbol{f} := 
\begin{pmatrix}
0 & 0& \cdots & 0
\\ \frac{\sum_{r' \in \mathcal{R} \setminus r_1} \pi_{r'}(\boldsymbol{x}) }{\left(  \sum_{r' \in \mathcal{R}} \pi_{r'}(\boldsymbol{x}) \right)^2} & \frac{-\pi_{r_2}(\boldsymbol{x}) }{\left(  \sum_{r' \in \mathcal{R}} \pi_{r'}(\boldsymbol{x}) \right)^2} & \cdots & \frac{-\pi_{r_R}(\boldsymbol{x}) }{\left(  \sum_{r' \in \mathcal{R}} \pi_{r'}(\boldsymbol{x}) \right)^2}
\\ \frac{-\pi_{r_1}(\boldsymbol{x}) }{\left(  \sum_{r' \in \mathcal{R}} \pi_{r'}(\boldsymbol{x}) \right)^2} &  \frac{\sum_{r' \in \mathcal{R} \setminus r_2} \pi_{r'}(\boldsymbol{x}) }{\left(  \sum_{r' \in \mathcal{R}} \pi_{r'}(\boldsymbol{x}) \right)^2}  & \cdots & \frac{-\pi_{r_R}(\boldsymbol{x}) }{\left(  \sum_{r' \in \mathcal{R}} \pi_{r'}(\boldsymbol{x}) \right)^2}
\\ \vdots & \vdots & \ddots & \vdots
\\  \frac{-\pi_{r_1}(\boldsymbol{x}) }{\left(  \sum_{r' \in \mathcal{R}} \pi_{r'}(\boldsymbol{x}) \right)^2} & \frac{-\pi_{r_2}(\boldsymbol{x}) }{\left(  \sum_{r' \in \mathcal{R}} \pi_{r'}(\boldsymbol{x}) \right)^2} & \cdots & \frac{\sum_{r' \in \mathcal{R} \setminus r_R} \pi_{r'}(\boldsymbol{x}) }{\left(  \sum_{r' \in \mathcal{R}} \pi_{r'}(\boldsymbol{x}) \right)^2} 
\end{pmatrix}
\]
at \((\pi_0(\boldsymbol{x}), \pi_{r_1}(\boldsymbol{x}), \ldots, \pi_{r_R}(\boldsymbol{x}))\) and exists and is nonzero and finite by Assumption (F2).

Other extensions are possible as well. For example, when the data is not split into independent subsamples we can achieve asymptotic subgaussianity with a conservative variance estimator using a similar argument to the one used to prove Theorem \ref{te.asym.norm.thm.gen}(c).

\section{Proofs of Main Results}\label{app.prove.main.res}

%\subsection{Proofs of Theorems \ref{main.te.cons.thm}, \ref{te.oracle.thm}, and \ref{te.asym.norm.thm}}

We present two lemmas we will use in the following proofs.

\begin{lemma}\label{lem.d.express}
The number of treatment effects \(\mathfrak{W}\) is at most \((T-1)^2\). Further, after re-ordering the columns of \(\boldsymbol{\tilde{Z}}\), we can write \(\boldsymbol{D}_N\) as 
\begin{equation}\label{d.expres}
\operatorname{diag} \big( \boldsymbol{D}^{(1)}(R) , \boldsymbol{D}^{(1)}(T - 1) , \boldsymbol{I}_{d_N} , \boldsymbol{I}_{d_N} \otimes \boldsymbol{D}^{(1)}(R), \boldsymbol{I}_{d_N} \otimes \boldsymbol{D}^{(1)}(T - 1) ,  \boldsymbol{D}^{(2)}(\mathcal{R}) , \boldsymbol{I}_{d_N} \otimes \boldsymbol{D}^{(2)}(\mathcal{R}) \big)
,
\end{equation}
where \(\boldsymbol{D}^{(1)}(t) \in \mathbb{R}^{t \times t}\) is defined in \eqref{d.1.expres} and \(\boldsymbol{D}^{(2)}(\mathcal{R}) \in \mathbb{R}^{\mathfrak{W} \times \mathfrak{W}}\) is defined in \eqref{d.2.expres}. Further, \(\boldsymbol{D}_N\) is invertible, and \(\boldsymbol{D}_N^{-1}\) can be written as
\begin{align}
  \operatorname{diag} \Big(  &  \boldsymbol{D}^{(1)}(R)^{-1} , ((\boldsymbol{D}^{(1)}(T - 1))^{-1} ,  \boldsymbol{I}_{d_N} , \boldsymbol{I}_{d_N} \otimes \boldsymbol{D}^{(1)}(R)^{-1} , \boldsymbol{I}_{d_N} \otimes ((\boldsymbol{D}^{(1)}(T - 1))^{-1} , \nonumber
\\ & (\boldsymbol{D}^{(2)}(\mathcal{R}))^{-1} , \boldsymbol{I}_{d_N} \otimes (\boldsymbol{D}^{(2)}(\mathcal{R}))^{-1} \Big)\label{d.inv.exp}
.
\end{align}
Finally, both \(\boldsymbol{D}_N\) and \(\boldsymbol{D}_N^{-1}\) consists of blocks of size at most \((T-1)^2 \times (T-1)^2\). For  \(\boldsymbol{D}_N\) the entries of these blocks are either \(-1\), 0, or 1; and for \(\boldsymbol{D}_N^{-1}\) the entries of these blocks are either 0 or 1. Both matrices equal 0 elsewhere.
\end{lemma}

\begin{lemma}\label{bound.coefs} Under Assumption (R3), \(\max_{j \in p_N} \{|(\boldsymbol{\beta}_N^*)_j|\} \leq (T-1)^2 b_1\).

\end{lemma}

Proofs of both results are provided in Appendix \ref{main.thm.lems}.

\subsection{Proofs of Theorems \ref{main.te.cons.thm.gen} and  \ref{main.te.cons.thm}}

We start by proving Theorem \ref{main.te.cons.thm.gen}, which leads almost immediately to a proof of Theorem \ref{main.te.cons.thm}.

\begin{proof}[Proof of Theorem \ref{main.te.cons.thm.gen}]

\begin{enumerate}[(a)]

\item Observe that under our assumptions our model is correctly specified due to Theorem \ref{first.thm.fetwfe}(a). Recall that due to Theorem \ref{te.interp.prop}, under our assumptions the regression estimands correspond to our causal estimands.

%Then using \eqref{att.cohort.time} and Theorem \ref{te.interp.prop} we have
%\begin{align}
%\tau_{\text{ATT}} (  r, t)  = ~ & \mathbb{E} \left[ \tau_{\text{CATT}} ( r,  X_i, t)  \right] =   \tau_{rt}^* + (\boldsymbol{\rho}_{rt}^* )^\top \E \left[   \boldsymbol{X}_{(it)\cdot}^{(r)} - \E \left[ \boldsymbol{X}_{(it)\cdot} \mid  W_i = r \right] \right] \nonumber
%%
%\\ =  ~ &   \tau_{rt}^* + (\boldsymbol{\rho}_{rt}^* )^\top \E \left[  \E \left[  \boldsymbol{X}_{(it)\cdot} \ - \boldsymbol{X}_{(it)\cdot} \mid  W_i = r \right] \right] = \tau_{rt}^*  \qquad \forall i \in [N], \label{tau.id.spec}
%\end{align}
%where we used that \( \boldsymbol{\dot{X}}_{(it)\cdot}^{(r)}  = \E \left[ \boldsymbol{X}_{(it)\cdot} \mid  W_i = r \right] \). 
Examining our estimator \(\hat{\tau}_{\text{ATT}} ( r, t)\) for the average treatment effect for a unit in cohort \(r\) at time \(t\) from \eqref{att.est.cov.r.t}, we see that we can then apply Theorem \ref{first.thm.fetwfe}(b) to immediately obtain that for every \(r \in \mathcal{R}\) and \(t \geq r\),
\begin{equation}\label{tau.r.t.cons.result}
\left| \hat{\tau}_{\text{ATT}} ( r, t) - \tau_{\text{ATT}} (r, t)  \right| = \left| \hat{\tau}_{rt}^{(q)} - \tau_{rt}^*   \right|  \leq \lVert \boldsymbol{\hat{\beta}}^{(q)} - \boldsymbol{\beta}_N^* \rVert_2 =   \mathcal{O}_\mathbb{P} \left( \min\{h_N, h_N'\}\right).
\end{equation}
%the average treatment effects within cohorts \eqref{att.cohort}
Similarly, for any deterministic fixed linear combination of average treatment effects, using \eqref{tau.r.t.cons.result} and Theorem \ref{te.interp.prop} we have that for our class of estimators \eqref{att.estimator.fixed} and estimands \eqref{att.estimand.fixed},
\begin{align*}
 \left| \sum_{r \in \mathcal{R}} \sum_{t =r}^T \psi_{rt} \hat{\tau}_{\text{ATT}} (  r,t)   -  \sum_{r \in \mathcal{R}} \sum_{t =r}^T  \psi_{rt}  \tau_{\text{ATT}} (r,t)  \right|
%%
% = ~ & \left| \sum_{r \in \mathcal{R}} \sum_{t =r}^T \psi_{rt}  \left( \hat{\tau}_{rt}^{(q)} -     \tau_{rt}^* \right) \right|
%= ~ &  \frac{1}{T-r+1} \left| \sum_{t=r}^T \hat{\tau}_{rt}^{(q)} -  \tau_{rt}^* \right|
%
 \leq ~ &  \sum_{r \in \mathcal{R}} \sum_{t =r}^T | \psi_{rt}|  \left|  \hat{\tau}_{\text{ATT}} ( r, t) - \tau_{\text{ATT}} (r, t) \right|
\\  = ~ &     \mathcal{O}_\mathbb{P} \left( \min\{h_N, h_N'\}\right)
,
\end{align*}
%where \(\boldsymbol{\tau}_r := (\tau_{rr}, \ldots, \tau_{rt}^*)^\top \in \mathbb{R}^{T-r+1}\) and \(\boldsymbol{\hat{\tau}}_r\) is defined analogously. 
where we used the fact that \(T\) is fixed (and upper bounds \(R\)). 

Next we characterize the rate of convergence of each \( f_r(\hat{\pi}_0, \hat{\pi}_{r_1}, \ldots, \hat{\pi}_{r_R})\). Because by assumption \(\boldsymbol{f}(\cdot)\) has a finite Jacobian \(\nabla \boldsymbol{f} \in \mathbb{R}^{(R + 1) \times R}\) at \((\pi_0, \pi_{r_1}, \ldots, \pi_{r_R})\) that exists and is nonzero and finite, we have by the asymptotic normality of the maximum likelihood estimator \(\left( \frac{N_{0}}{N}, \frac{N_{r_1}}{N}, \frac{N_{r_2}}{N}, \ldots, \frac{N_{r_{R}}}{N} \right) \) (Theorem 14.6 in \citealt{wasserman2004all}) and the delta method (Theorem 3.1 in \citealt{van2000asymptotic}) that
\begin{align*}
\sqrt{N} \left( f(\hat{\pi}_0, \hat{\pi}_{r_1}, \ldots, \hat{\pi}_{r_R})  -  f(\pi_0, \pi_{r_1}, \ldots, \pi_{r_R})  \right) \xrightarrow{d} ~ & \mathcal{N}(0, \nabla \boldsymbol{f}^\top \boldsymbol{\Sigma}_M \nabla \boldsymbol{f})
%
%\\ \implies \qquad  \sqrt{N} \boldsymbol{e}_{r'}^\top  \left( \hat{\pi}_{r_1} -  \tilde{\pi}_{r_1},  \hat{\pi}_{r_2} -  \tilde{\pi}_{r_2}, \ldots, \hat{\pi}_{r_R} -  \tilde{\pi}_{r_R}, \right) \xrightarrow{d} ~ & \mathcal{N}(0,  \boldsymbol{e}_{r'}^\top \boldsymbol{\nabla}_\pi^\top \boldsymbol{\Sigma}_M \boldsymbol{\nabla}_\pi  \boldsymbol{e}_{r'})
%
\\ \implies \qquad  \sqrt{N}  \left(  f_{r'} (\hat{\pi}_0, \hat{\pi}_{r_1}, \ldots, \hat{\pi}_{r_R})  -  f_{r'}(\pi_0, \pi_{r_1}, \ldots, \pi_{r_R})  \right) \xrightarrow{d} ~ & \mathcal{N}(0,  \boldsymbol{e}_{r'}^\top  \nabla \boldsymbol{f}^\top \boldsymbol{\Sigma}_M \nabla \boldsymbol{f} \boldsymbol{e}_{r'})
,
\end{align*}
for any \(r' \in [R]\), where 
\begin{equation}\label{asym.cov.multinomi}
\boldsymbol{\Sigma}_M := \begin{pmatrix}
\pi_0(1 - \pi_0) & - \pi_0 \pi_{r_1}  & \cdots & - \pi_0 \pi_{r_R}
\\ - \pi_0 \pi_{r_1} & \pi_{r_1} (1 - \pi_{r_1}) & \cdots & - \pi_{r_1} \pi_{r_R}
\\ \vdots & \vdots & \ddots & \vdots
\\ - \pi_0 \pi_{r_R} & - \pi_{r_1} \pi_{r_R} & \cdots & \pi_{r_R}(1 - \pi_{r_R})
\end{pmatrix} \in \mathbb{R}^{(R +1) \times (R + 1)}
\end{equation}
and \(\boldsymbol{e}_r \in \mathbb{R}^R\) is the selector vector with a 1 in the \(r^{\text{th}}\) entry and zeroes elsewhere. Therefore by Theorem 2.4(i) in \citet{van2000asymptotic} and Proposition 1.8(iii) in \citet{Garcia-Portugues2023}, we have that for all \(r \in [R]\)
\begin{align}
\sqrt{N}  \left|   f_r(\hat{\pi}_0, \hat{\pi}_{r_1}, \ldots, \hat{\pi}_{r_R})  -  f_r(\pi_0, \pi_{r_1}, \ldots, \pi_{r_R})   \right| = ~ & \mathcal{O}_{\mathbb{P}} \left( 1 \right) \nonumber
\\ \implies \qquad  \left|   f_r(\hat{\pi}_0, \hat{\pi}_{r_1}, \ldots, \hat{\pi}_{r_R})  -  f_r(\pi_0, \pi_{r_1}, \ldots, \pi_{r_R})   \right| = ~ & \mathcal{O}_{\mathbb{P}} \left( \frac{1}{ \sqrt{N}} \right) \label{pi.hat.gen.convg.rate}
.
\end{align}

Finally, for our class of estimators \eqref{att.estimator.weighted} and estimands  \eqref{att.estimand.weighted}, we have
\begin{align*}
 &  \left|  \sum_{r \in \mathcal{R}} \sum_{t=r}^T  \psi_{rt} f_r(\hat{\pi}_0, \hat{\pi}_{r_1}, \ldots, \hat{\pi}_{r_R}) \hat{\tau}_{\text{ATT}} (  r,t)  -  \psi_{rt}   f_r(\pi_0, \pi_{r_1}, \ldots, \pi_{r_R})  \tau_{\text{ATT}} (r,t)  \right|
% %
% \\ = ~  &  \left|  \sum_{r \in \mathcal{R}} \sum_{t=r}^T  \psi_{rt} \left( \frac{N_r}{N_\tau}  \hat{\tau}_{rt}^{(q)}  -   f_r(\pi_0, \pi_{r_1}, \ldots, \pi_{r_R})  \tau_{rt}^* \right)  \right|
%%
%\\ = ~ &  \Bigg|  \sum_{r \in \mathcal{R}}   \sum_{t=r}^T   \psi_{rt} \big[   f_r(\hat{\pi}_0, \hat{\pi}_{r_1}, \ldots, \hat{\pi}_{r_R}) \left( \hat{\tau}_{\text{ATT}} (  r,t)  -  \tau_{\text{ATT}} (r,t)  \right) 
%\\ &  + \left(      f_r(\hat{\pi}_0, \hat{\pi}_{r_1}, \ldots, \hat{\pi}_{r_R})  - f_r(\pi_0, \pi_{r_1}, \ldots, \pi_{r_R}) \right) \tau_{\text{ATT}} (r,t)    \big]\Bigg|
%%
\\ = ~ &  \Bigg|  \sum_{r \in \mathcal{R}}   \sum_{t=r}^T   \psi_{rt} \big[  f_r(\pi_0, \pi_{r_1}, \ldots, \pi_{r_R}) \left( \hat{\tau}_{\text{ATT}} (  r,t)  -  \tau_{\text{ATT}} (r,t)  \right) 
\\ &  + \left(   f_r(\hat{\pi}_0, \hat{\pi}_{r_1}, \ldots, \hat{\pi}_{r_R})  -  f_r(\pi_0, \pi_{r_1}, \ldots, \pi_{r_R})   \right) \tau_{\text{ATT}} (r,t)  
\\ & +  ( f_r(\hat{\pi}_0, \hat{\pi}_{r_1}, \ldots, \hat{\pi}_{r_R})  -  f_r(\pi_0, \pi_{r_1}, \ldots, \pi_{r_R})  )\left( \hat{\tau}_{\text{ATT}} (  r,t)  -  \tau_{\text{ATT}} (r,t)  \right)   \big]\Bigg|
\\ \leq ~ &  \sum_{r \in \mathcal{R}}   \sum_{t=r}^T    \psi_{rt}  \big( f_r(\pi_0, \pi_{r_1}, \ldots, \pi_{r_R}) \left| \hat{\tau}_{\text{ATT}} (  r,t)  -  \tau_{\text{ATT}} (r,t)  \right|   +
\\ &   \left|    f_r(\hat{\pi}_0, \hat{\pi}_{r_1}, \ldots, \hat{\pi}_{r_R}) - f_r(\pi_0, \pi_{r_1}, \ldots, \pi_{r_R})   \right| \tau_{\text{ATT}} (r,t)  
\\ & +  | f_r(\hat{\pi}_0, \hat{\pi}_{r_1}, \ldots, \hat{\pi}_{r_R})  -  f_r(\pi_0, \pi_{r_1}, \ldots, \pi_{r_R})  |\left| \hat{\tau}_{\text{ATT}} (  r,t)  -  \tau_{\text{ATT}} (r,t)  \right|  \big)
%
%\\ \leq ~ &    \sum_{r \in \mathcal{R}}\sum_{t=r}^T | \psi_{rt}  | \hat{\pi}_r  \left| \hat{\tau}_{\text{ATT}} (  r,t) -\tau_{\text{ATT}} (r,t)  \right|  +  \sum_{r \in \mathcal{R}} \sum_{t=r}^T  \psi_{rt}   \left| \tilde{\pi}_r   - \hat{\pi}_r\right|   \left|  \tau_{\text{ATT}} (r,t)  \right|
%
%\\ = ~ & \mathcal{O}_{\mathbb{P}}( \min\{h_N, h_N'\}) +  \sum_{r \in \mathcal{R}}  \left| \tilde{\pi}_r   - \frac{N_r}{N_\tau}\right|   \left|  \frac{1}{T-r+1} \sum_{t=r}^T \tau_{rt}^*   \right|
%
\\ \stackrel{(a)}{=} ~ & \mathcal{O}_{\mathbb{P}}( \min\{h_N, h_N'\}) +  \mathcal{O}_{\mathbb{P}}(1/\sqrt{N}) + \mathcal{O}_{\mathbb{P}}( \min\{h_N, h_N'\}) \mathcal{O}_{\mathbb{P}}(1/\sqrt{N})
\\ \stackrel{(b)}{=} ~ & \mathcal{O}_{\mathbb{P}}( \min\{h_N, h_N'\}) 
,
\end{align*}
%\begin{align*}
%\left| \hat{\tau}_{\text{ATT}}  - \tau_{\text{ATT}}  \right| = ~ &  \left|  \sum_{r \in \mathcal{R}}  \frac{1}{T-r+1}  \left( \frac{N_r}{N_\tau}\hat{\tau}_{\text{ATT}} (  r)  -   \tilde{\pi}_r   \tau_{\text{ATT}} (r)  \right)\right|
%%
%\\ = ~ &  \left|  \sum_{r \in \mathcal{R}}  \frac{1}{T-r+1}  \left( \frac{N_r}{N_\tau} \left(\hat{\tau}_{\text{ATT}} (  r) -  \tau_{\text{ATT}} (r)   \right)  - \left(  \tilde{\pi}_r   - \frac{N_r}{N_\tau}\right)\tau_{\text{ATT}} (r)  \right)\right|
%%
%\\ \leq ~ &   \frac{1}{T-r+1}  \left(  \sum_{r \in \mathcal{R}}  \frac{N_r}{N_\tau}  \left| \hat{\tau}_{\text{ATT}} (  r) -\tau_{\text{ATT}} (r)  \right|  +  \sum_{r \in \mathcal{R}}  \left| \tilde{\pi}_r   - \frac{N_r}{N_\tau}\right|   \left|  \tau_{\text{ATT}} (r)  \right| \right)
%%
%\\ = ~ & \mathcal{O}_{\mathbb{P}}( \min\{h_N, h_N'\}) +  \sum_{r \in \mathcal{R}}  \left| \tilde{\pi}_r   - \frac{N_r}{N_\tau}\right|   \left|  \frac{1}{T-r+1} \sum_{t=r}^T \tau_{rt}^*   \right|
%,
%\end{align*}
where in \((a)\) we used that all the \(  f_r(\pi_0, \pi_{r_1}, \ldots, \pi_{r_R}) \) and \(\tau_{\text{ATT}}(r, t)\) are finite, \eqref{tau.r.t.cons.result}, and \eqref{pi.hat.gen.convg.rate}; and in \((b)\) we used that \(1/\sqrt{N} = \mathcal{O}(h_N)\) since \(e_{1N}\) is upper-bounded by a constant and \(p_N\) is nondecreasing. 

\item Again, using Theorem \ref{te.interp.prop},
%
%since our model is correctly specified under Theorem \ref{first.thm.fetwfe}(1),
\[
\tau_{\text{CATT}} ( r, t, \boldsymbol{x})  = \tau_{rt}^* + (\boldsymbol{\rho}_{rt}^*)^\top \left(   \boldsymbol{x} - \E \left[ \boldsymbol{X}_{(1t)\cdot} \mid  W_i = r \right] \right) = \tau_{rt}^* + (\boldsymbol{\rho}_{rt}^*)^\top \left(   \boldsymbol{x} - \boldsymbol{\mu}_r \right)
,
\]
where \(\boldsymbol{\mu}_r := \E \left[ \boldsymbol{X}_{(it)\cdot} \mid  W_i = r \right] \). Using this and \eqref{catt.t.r.est} we have
\begin{align}
&  \left| \hat{\tau}_{\text{CATT}} (r, t, \boldsymbol{x})  - \tau_{\text{CATT}} (r, t, \boldsymbol{x})  \right| \nonumber
\\ = ~ & \left| \hat{\tau}_{rt}^{(q)}  -  \tau_{rt}^*  + (\boldsymbol{x} - \boldsymbol{\overline{X}}_r)^\top  \boldsymbol{\hat{\rho}}_{rt} - \left( \boldsymbol{x} - \boldsymbol{\mu}_r \right)^\top \boldsymbol{\rho}_{rt}^* \right| \nonumber
\\ = ~ & \left| \hat{\tau}_{rt}^{(q)}  -  \tau_{rt}^* + \left(\boldsymbol{x} -  \boldsymbol{\mu}_r \right)^\top \left(  \boldsymbol{\hat{\rho}}_{rt} -  \boldsymbol{\rho}_{rt}^*  \right)   +  \left( \boldsymbol{\mu}_r  - \boldsymbol{\overline{X}}_r \right)^\top \left(  \boldsymbol{\hat{\rho}}_{rt} -  \boldsymbol{\rho}_{rt}^*  \right)  - \left( \boldsymbol{\overline{X}}_r  - \boldsymbol{\mu}_r  \right)^\top \boldsymbol{\rho}_{rt}^*  \right| \nonumber
\\ \leq ~ & \left| \hat{\tau}_{rt}^{(q)}  -  \tau_{rt}^* \right| + \left| \left(\boldsymbol{x} -  \boldsymbol{\mu}_r \right)^\top \left(  \boldsymbol{\hat{\rho}}_{rt} -  \boldsymbol{\rho}_{rt}^*  \right) \right|  +  \left|\left( \boldsymbol{\mu}_r  - \boldsymbol{\overline{X}}_r \right)^\top \left(  \boldsymbol{\hat{\rho}}_{rt} -  \boldsymbol{\rho}_{rt}^*  \right) \right|  + \left| \left( \boldsymbol{\overline{X}}_r  - \boldsymbol{\mu}_r  \right)^\top \boldsymbol{\rho}_{rt}^*  \right| \nonumber
\\ \stackrel{(c)}{\leq} ~ & \left| \hat{\tau}_{rt}^{(q)}  -  \tau_{rt}^* \right| + \left\lVert\boldsymbol{x} -  \boldsymbol{\mu}_r \right\rVert_2 \left\lVert  \boldsymbol{\hat{\rho}}_{rt} -  \boldsymbol{\rho}_{rt}^*  \right\rVert_2  +  \left\lVert \boldsymbol{\mu}_r  - \boldsymbol{\overline{X}}_r \right\rVert_2 \left\lVert  \boldsymbol{\hat{\rho}}_{rt} -  \boldsymbol{\rho}_{rt}^*  \right\rVert_2  + \left\lVert \boldsymbol{\overline{X}}_r  - \boldsymbol{\mu}_r  \right\rVert_2 \left\lVert \boldsymbol{\rho}_{rt}^*  \right\rVert_2 \nonumber
\\ \stackrel{(d)}{=} ~ & \mathcal{O}_{\mathbb{P}}( \min\{h_N, h_N'\})  + \mathcal{O}_{\mathbb{P}}( \min\{h_N, h_N'\})    +  \mathcal{O}_{\mathbb{P}}\left( 1 / \sqrt{N} \right)  \mathcal{O}_{\mathbb{P}}\left( \min\{h_N, h_N'\}  \right) 
 +    \mathcal{O}_{\mathbb{P}}\left( 1 / \sqrt{N} \right) \nonumber
 \\ = ~ &  \mathcal{O}_{\mathbb{P}}( \min\{h_N, h_N'\}) \nonumber
%%
%\\ \stackrel{(c)}{=} ~ & \mathcal{O}_{\mathbb{P}}( \min\{h_N, h_N'\})  + \mathcal{O}_{\mathbb{P}}( \min\{h_N, h_N'\})    +    \mathcal{O}_{\mathbb{P}}\left( \frac{\min\{h_N, h_N'\}  }{\sqrt{N}}\right) 
%+   \mathcal{O}_{\mathbb{P}}\left( \frac{1  }{\sqrt{N}}\right)  \nonumber
%%
%\\ = ~ &  \mathcal{O}_{\mathbb{P}}\left( \min\{h_N, h_N'\} \vee  \frac{1}{\sqrt{N}} \right) \nonumber 
%%
%\\ \stackrel{(d)}{\x} ~ &  \mathcal{O}_{\mathbb{P}}\left( \min\{h_N, h_N'\}  \right)  \label{catt.block.result}
,
\end{align}
 where in \((c)\) we applied Cauchy-Schwarz and \((d)\) follows from Theorem \ref{first.thm.fetwfe}, the fact that \(  \left\lVert \boldsymbol{\rho}_{rt}^*  \right\rVert_2 \) is fixed and bounded since \(d\) is fixed and each coefficient is finite due to Lemma \ref{bound.coefs}, and Lemma \ref{lem.asym.norm.cond.means}:
 
%  the fact that \(\lVert \boldsymbol{\rho}_{rt}^*  \rVert_2\) is fixed and bounded using Lemma \ref{bound.coefs}

%We will require a lemma establishing the asymptotic normality of the cohort sample means.
\begin{lemma}\label{lem.asym.norm.cond.means}
For each \(r \in \{0\} \cup \mathcal{R}\), let \(\boldsymbol{\psi}_{r} \in \mathbb{R}^d \) be constant vectors, with \(\boldsymbol{\psi}_{r} \neq \boldsymbol{0}\) for at least one \(r = r^*\). For all \(N\) and all \(r \in \{0\} \cup   \mathcal{R}\), assume all of the eigenvalues of \( \Cov ( \boldsymbol{X}_i \mid W_i = r)\) are bounded between \(\lambda_{\text{min}} > 0\) and \(\lambda_{\text{max}} < \infty\). Then
\[
 \sqrt{N}   \sum_{r \in \{0\} \cup \mathcal{R}}      \boldsymbol{\psi}_{r}^\top  \left(  \boldsymbol{\overline{X}}_r - \boldsymbol{\mu}_{r}  \right)  \xrightarrow{d}   \mathcal{N}(0, v_R)
 ,
 \]
 where
 \[
v_R :=  \sum_{r \in \{0\} \cup \mathcal{R}}   \frac{\boldsymbol{\psi}_{r}^\top  \Cov ( \boldsymbol{X}_i \mid W_i = r)\boldsymbol{\psi}_{r}}{\mathbb{P}(W_i = r)}    
.
 \]
 Further, for any \(r \in \{0\} \cup \mathcal{R}\),
 \[
  \lVert   \boldsymbol{\overline{X}}_r - \boldsymbol{\mu}_{r}  \rVert_2 = \mathcal{O}_{\mathbb{P}}\left(\frac{1}{\sqrt{N}} \right)
  .
 \]

%  Further,
%\[
% \sum_{r \in \{0\} \cup \mathcal{R}}    \sqrt{\frac{N}{ \hat{v}_R  }}    \boldsymbol{\psi}_{r}^\top  \left(  \boldsymbol{\overline{X}}_r - \boldsymbol{\mu}_{r}  \right)  \xrightarrow{d}   \mathcal{N}(0, 1)
% ,
% \]
%for
% \[
%\hat{v}_R :=  \sum_{r \in \{0\} \cup \mathcal{R}}   \frac{\boldsymbol{\hat{\psi}}_{r}^\top  \widehat{\Cov} ( \boldsymbol{X}_i \mid W_i = r)\boldsymbol{\hat{\psi}}_{r}}{N_r / N}    
%,
% \]
%where \(\boldsymbol{\hat{\psi}}_{r}\) is an estimator that satisfies \(\boldsymbol{\hat{\psi}}_{r} \xrightarrow{p}  \boldsymbol{\psi}_{r} \).

\end{lemma}

\begin{proof}
Provided in Appendix \ref{main.thm.lems}.
\end{proof}

Next, for any deterministic fixed set of constants \(\{\psi_{rt}\}\), the above result gives us
\begin{align}
 &\left| \sum_{r \in \mathcal{R}} \sum_{t =r}^T \psi_{rt} \hat{\tau}_{\text{CATT}} (r, t, \boldsymbol{x})   -  \sum_{r \in \mathcal{R}} \sum_{t =r}^T  \psi_{rt}\tau_{\text{CATT}} (r, t, \boldsymbol{x})   \right| \nonumber
\\ \leq ~ &  \sum_{r \in \mathcal{R}} \sum_{t =r}^T | \psi_{rt} |  \left|   \hat{\tau}_{\text{CATT}} (r, t, \boldsymbol{x})    - \tau_{\text{CATT}} (r, t, \boldsymbol{x})  \right| \nonumber
 \\ = ~ &  \mathcal{O}_{\mathbb{P}}( \min\{h_N, h_N'\}) \label{cov.convg.result}
.
\end{align}
%\[
%\vdots
%\]
%
%
%
%using \eqref{catt.r.est} and \eqref{catt.r.def} we have
%\begin{align*}
% \left| \hat{\tau}_{\text{CATT}} (r, \boldsymbol{x} )  - \tau_{\text{CATT}} (r, \boldsymbol{x})  \right| = ~ &  \frac{1}{T - r + 1}\left| \sum_{t=r}^T  \hat{\tau}_{\text{CATT}} (r, t, \boldsymbol{x})  - \tau_{\text{CATT}} (r, t, \boldsymbol{x})  \right|
%%
% \\ \leq ~ &  \frac{1}{T - r + 1} \sum_{t=r}^T   \left|  \hat{\tau}_{\text{CATT}} (r, t, \boldsymbol{x})  - \tau_{\text{CATT}} (r, t, \boldsymbol{x})  \right|
% %
% \\ = ~ &\mathcal{O}_{\mathbb{P}}\left( \min\{h_N, h_N'\} \right) 
%% %
%% \\ = ~ &  \mathcal{O}_{\mathbb{P}}\left( \min\{h_N, h_N'\} \right) 
% .
%\end{align*}
Finally, we have
\begin{align*}
 &\left| \sum_{r \in \mathcal{R}} \sum_{t =r}^T \psi_{rt} \left(  \frac{\boldsymbol{a}_r^\top \boldsymbol{\hat{\pi}}(\boldsymbol{x})}{\boldsymbol{b}_r^\top \boldsymbol{\hat{\pi}}(\boldsymbol{x})}  \right)  \hat{\tau}_{\text{CATT}} (r, t, \boldsymbol{x})   -  \sum_{r \in \mathcal{R}} \sum_{t =r}^T  \psi_{rt}   \frac{\boldsymbol{a}_r^\top \boldsymbol{{\pi}}(\boldsymbol{x})}{\boldsymbol{b}_r^\top \boldsymbol{{\pi}}(\boldsymbol{x})}  \tau_{\text{CATT}} (r, t, \boldsymbol{x})   \right|
%
%\\ = ~ & \left| \sum_{r \in \mathcal{R}}  \left(  \frac{\hat{\pi}_r(\boldsymbol{x})}{ \sum_{r' \in \mathcal{R}} \hat{\pi}_{r'}(\boldsymbol{x})}   \right)  \hat{\tau}_{\text{CATT}} (r, \boldsymbol{x} )  -  \tilde{\pi}_r(\boldsymbol{x}) \tau_{\text{CATT}} (r, \boldsymbol{x})  \right|
%
\\ = ~ & \Bigg| \sum_{r \in \mathcal{R}} \sum_{t =r}^T \psi_{rt} \Bigg[ \left( \frac{\boldsymbol{a}_r^\top \boldsymbol{\pi}(\boldsymbol{x})}{\boldsymbol{b}_r^\top \boldsymbol{\pi}(\boldsymbol{x})}   \right)  ( \hat{\tau}_{\text{CATT}} (r, \boldsymbol{x} , t)  -  \tau_{\text{CATT}} (r, t, \boldsymbol{x}) ) 
\\ &  +   \left[    \frac{\boldsymbol{a}_r^\top \boldsymbol{\hat{\pi}}(\boldsymbol{x})}{\boldsymbol{b}_r^\top \boldsymbol{\hat{\pi}}(\boldsymbol{x})}   -    \frac{\boldsymbol{a}_r^\top \boldsymbol{{\pi}}(\boldsymbol{x})}{\boldsymbol{b}_r^\top \boldsymbol{{\pi}}(\boldsymbol{x})}   \right]  \tau_{\text{CATT}} (r, t, \boldsymbol{x})   
\\ & + \left[    \frac{\boldsymbol{a}_r^\top \boldsymbol{\hat{\pi}}(\boldsymbol{x})}{\boldsymbol{b}_r^\top \boldsymbol{\hat{\pi}}(\boldsymbol{x})}   -    \frac{\boldsymbol{a}_r^\top \boldsymbol{{\pi}}(\boldsymbol{x})}{\boldsymbol{b}_r^\top \boldsymbol{{\pi}}(\boldsymbol{x})}   \right]   ( \hat{\tau}_{\text{CATT}} (r, \boldsymbol{x} , t)  -  \tau_{\text{CATT}} (r, t, \boldsymbol{x}) )  \Bigg] \Bigg|
\\ \leq ~ & \sum_{r \in \mathcal{R}} \sum_{t =r}^T \psi_{rt} \Bigg[ \frac{\boldsymbol{a}_r^\top \boldsymbol{\pi}(\boldsymbol{x})}{\boldsymbol{b}_r^\top \boldsymbol{\pi}(\boldsymbol{x})}   \left|  \hat{\tau}_{\text{CATT}} (r, \boldsymbol{x} , t)  -  \tau_{\text{CATT}} (r, t, \boldsymbol{x}) \right|
\\ &  +   \left|  \frac{\boldsymbol{a}_r^\top \boldsymbol{\hat{\pi}}(\boldsymbol{x})}{\boldsymbol{b}_r^\top \boldsymbol{\hat{\pi}}(\boldsymbol{x})}   -    \frac{\boldsymbol{a}_r^\top \boldsymbol{{\pi}}(\boldsymbol{x})}{\boldsymbol{b}_r^\top \boldsymbol{{\pi}}(\boldsymbol{x})}   \right|  \tau_{\text{CATT}} (r, t, \boldsymbol{x})   
\\ & + \left|   \frac{\boldsymbol{a}_r^\top \boldsymbol{\hat{\pi}}(\boldsymbol{x})}{\boldsymbol{b}_r^\top \boldsymbol{\hat{\pi}}(\boldsymbol{x})}   -    \frac{\boldsymbol{a}_r^\top \boldsymbol{{\pi}}(\boldsymbol{x})}{\boldsymbol{b}_r^\top \boldsymbol{{\pi}}(\boldsymbol{x})}   \right|  \left| \hat{\tau}_{\text{CATT}} (r, \boldsymbol{x} , t)  -  \tau_{\text{CATT}} (r, t, \boldsymbol{x})\right|  \Bigg]
\\ \stackrel{(e)}{\leq} ~ & \mathcal{O}_{\mathbb{P}}( \min\{h_N, h_N'\})
  +   \mathcal{O}_{\mathbb{P}}( a_N) 
+ \mathcal{O}_{\mathbb{P}}( a_N)   \mathcal{O}_{\mathbb{P}}( \min\{h_N, h_N'\}) 
%%
%\\ \vdots \qquad & \textbf{[FIX THIS]}
%%
%\\ \leq ~ &\sum_{r \in \mathcal{R}} \sum_{t =r}^T | \psi_{rt}|   \left|   \frac{\hat{\pi}_r(\boldsymbol{x})}{ \sum_{r' \in \mathcal{R}} \hat{\pi}_{r'}(\boldsymbol{x})}  \right| \left| \hat{\tau}_{\text{CATT}} (r, \boldsymbol{x} ,t)  -  \tau_{\text{CATT}} (r, t, \boldsymbol{x}) \right| 
%%
%\\ & +  \sum_{r \in \mathcal{R}} \sum_{t =r}^T | \psi_{rt}|  \left| \tilde{\pi}_r(\boldsymbol{x})  -   \frac{\hat{\pi}_r(\boldsymbol{x})}{ \sum_{r' \in \mathcal{R}} \hat{\pi}_{r'}(\boldsymbol{x})} \right| \left| \tau_{\text{CATT}} (r, t, \boldsymbol{x})  \right|
%%
%\\ \leq ~ &\sum_{r \in \mathcal{R}}  \sum_{t =r}^T  | \psi_{rt}|      \left| \hat{\tau}_{\text{CATT}} (r, t, \boldsymbol{x})  -  \tau_{\text{CATT}} (r, t, \boldsymbol{x}) \right| 
%\\ & +  \sum_{r \in \mathcal{R}} \sum_{t =r}^T  | \psi_{rt}|   \left| \tilde{\pi}_r(\boldsymbol{x})  -   \frac{\hat{\pi}_r(\boldsymbol{x})}{ \sum_{r' \in \mathcal{R}} \hat{\pi}_{r'}(\boldsymbol{x})}    \right| \left| \tau_{\text{CATT}} (r, t, \boldsymbol{x})  \right|
%%
%\\ \stackrel{(e)}{=} ~ &\mathcal{O}_{\mathbb{P}}( \min\{h_N, h_N'\})  +   \mathcal{O}_{\mathbb{P}}( a_N) 
%
\\ \stackrel{(f)}{=} ~ & \mathcal{O}_{\mathbb{P}}( \min\{h_N, h_N'\} \vee a_N )  
%\mathcal{O}_{\mathbb{P}}( a_N) 
%
%\\ = ~ &\mathcal{O}_{\mathbb{P}}( \min\{h_N, h_N'\} \vee a_N )  
%
%\\ = ~ & \mathcal{O}_{\mathbb{P}} \left( \min\{h_N, h_N'\}  \right)  +  \sum_{r \in \mathcal{R}}  \sum_{t =r}^T    | \psi_{rt}|  \left| \tilde{\pi}_r(\boldsymbol{x})  -   \frac{\hat{\pi}_r(\boldsymbol{x})}{ \sum_{r' \in \mathcal{R}} \hat{\pi}_{r'}(\boldsymbol{x})}  \right| \left| \tau_{\text{CATT}} (r, t, \boldsymbol{x})  \right|
,
\end{align*}
where in \((f)\) we used Proposition 1.8 from \citet{Garcia-Portugues2023} and in \((e)\) we used \eqref{cov.convg.result}, the fact that \(  | \psi_{rt}|   \) is bounded for any fixed \(\boldsymbol{x}\) for all \(r \in \mathcal{R}\) due to Lemma \ref{bound.coefs}, the fact that  \(    \left| \tau_{\text{CATT}} (r, t, \boldsymbol{x})  \right| \) and \(  \boldsymbol{a}_r^\top \boldsymbol{{\pi}}(\boldsymbol{x}) / \boldsymbol{b}_r^\top \boldsymbol{{\pi}}(\boldsymbol{x})  \) are bounded for any fixed \(\boldsymbol{x}\), and Lemma \ref{tilde.pi.r.o.p.cond.lemma}:

\begin{lemma}\label{tilde.pi.r.o.p.cond.lemma} 
For any \(\boldsymbol{x}\) in the support of \(\boldsymbol{X}_i\), let \(\boldsymbol{\hat{\pi}}(\boldsymbol{x})\) be an estimator for \(\pi_r(\boldsymbol{x})\) satisfying \( | \hat{\pi}_r(\boldsymbol{x}) - \pi_r(\boldsymbol{x}) | = \mathcal{O}_{\mathbb{P}}(a_N)\) for a decreasing sequence \(a_N\) for all \(r \in \{0\} \cup \mathcal{R}\). Let \(\boldsymbol{a} \) and \(\boldsymbol{b} \in \mathbb{R}^{R + 1}\) be finite vectors satisfying \( \boldsymbol{a}^\top \boldsymbol{{\pi}}(\boldsymbol{x})  > 0\) and \( \boldsymbol{b}^\top \boldsymbol{{\pi}}(\boldsymbol{x})  > 0\) for any \(\boldsymbol{x}\) in the support of \(\boldsymbol{X}_i\). Then
%\[
% \left| \tilde{\pi}_r(\boldsymbol{x})   - \frac{\hat{\pi}_r(\boldsymbol{x})}{\sum_{r' \in \mathcal{R}} \hat{\pi}_{r'}(\boldsymbol{x})}\right|  =   \mathcal{O}_{\mathbb{P}}\left( a_N \right)
% .
% \]
 \[
  \left| \frac{\boldsymbol{a}^\top \boldsymbol{\hat{\pi}}(\boldsymbol{x})}{\boldsymbol{b}^\top \boldsymbol{\hat{\pi}}(\boldsymbol{x})} -  \frac{\boldsymbol{a}^\top \boldsymbol{{\pi}}(\boldsymbol{x})}{\boldsymbol{b}^\top \boldsymbol{{\pi}}(\boldsymbol{x})}  \right|   =   \mathcal{O}_{\mathbb{P}}\left( a_N \right)
 .
  \]
\end{lemma}

\begin{proof}
Provided in Appendix \ref{main.thm.lems}.
\end{proof}

The proof is complete.

%The second part of the result then follows from our assumption that
%\[
%  \left| \tilde{\pi}_r(\boldsymbol{x})  -   \frac{\hat{\pi}_r(\boldsymbol{x})}{ \sum_{r' \in \mathcal{R}} \hat{\pi}_{r'}(\boldsymbol{x})}    \right| = \mathcal{O}_{\mathbb{P}}( a_N) 
%\]
%and the fact that  \(  | \psi_{rt}|  \left| \tau_{\text{CATT}} (r, t, \boldsymbol{x})  \right| \) is bounded for any fixed \(\boldsymbol{x}\). The first part follows from this argument: since \(\hat{\pi}_r(\boldsymbol{x}) \xrightarrow{p} \pi_r(\boldsymbol{x})\) for any fixed \(\boldsymbol{x}\) for all \(r \in \{0\} \cup \mathcal{R}\) and \(\pi_r(\boldsymbol{x})  > 0\) for all \(r \in \{0\} \cup \mathcal{R}\) by Assumption (F2), the continuous mapping theorem implies
%\[
%\frac{\hat{\pi}_r(\boldsymbol{x})}{ \sum_{r' \in \mathcal{R}} \hat{\pi}_{r'}(\boldsymbol{x})}  \xrightarrow{p} \frac{\pi_r(\boldsymbol{x})}{\sum_{r' \in \mathcal{R}} \pi_{r'}(\boldsymbol{x})}   = \tilde{\pi}_r(\boldsymbol{x})
%\]
%for all \(r \in \mathcal{R}\). It follows that
%\[
% \sum_{r \in \mathcal{R}}   \left| \tilde{\pi}_r(\boldsymbol{x})  -   \frac{\hat{\pi}_r(\boldsymbol{x})}{ \sum_{r' \in \mathcal{R}} \hat{\pi}_{r'}(\boldsymbol{x})}  \right| \left| \tau_{\text{CATT}} (r, \boldsymbol{x})  \right| \xrightarrow{p} 0.
% \]
%In combination with \eqref{cov.convg.result} and Slutsky's theorem, this proves the result.
\end{enumerate}

\end{proof}

\begin{proof}[Proof of Theorem \ref{main.te.cons.thm}]
In light of Theorem \ref{main.te.cons.thm.gen}, to prove part (a) it is enough to note that
\[
f_r( \hat{\pi}_0,  \hat{\pi}_{r_1}, \ldots,  \hat{\pi_{r_R}}) = 
\begin{pmatrix}
\frac{ \hat{\pi}_{r_1} }{\sum_{r \in \mathcal{R}} \hat{\pi}_r }
\\ \frac{\hat{\pi}_{r_2} }{\sum_{r \in \mathcal{R}} \hat{\pi}_r }
\\ \vdots
\\ \frac{\hat{\pi}_{r_R} }{\sum_{r \in \mathcal{R}} \hat{\pi}_r }
\end{pmatrix}
\]
(where \(\hat{\pi}_r = N_r / N\)) is pointwise (that is, almost surely) finite and show that its Jacobian at \((\pi_0, \pi_{r_1}, \ldots, \pi_{r_R})\) exists and is nonzero and finite. The Jacobian at \((\pi_0, \pi_{r_1}, \ldots, \pi_{r_R})\) is
\[
\begin{pmatrix}
0 & 0 & \cdots & 0
\\  \frac{ \sum_{r \in \mathcal{R} \setminus r_1} \hat{\pi}_r  }{\left(\sum_{r \in \mathcal{R}} \hat{\pi}_r \right)^2}  &  -\frac{  \hat{\pi}_{r_2} }{\left(\sum_{r \in \mathcal{R}} \hat{\pi}_r \right)^2} & \cdots &   -\frac{  \hat{\pi}_{r_R} }{\left(\sum_{r \in \mathcal{R}} \hat{\pi}_r \right)^2}
\\    -\frac{  \hat{\pi}_{r_1} }{\left(\sum_{r \in \mathcal{R}} \hat{\pi}_r \right)^2} &  \frac{ \sum_{r \in \mathcal{R} \setminus r_2} \hat{\pi}_r  }{\left(\sum_{r \in \mathcal{R}} \hat{\pi}_r \right)^2} & \cdots &   -\frac{  \hat{\pi}_{r_R} }{\left(\sum_{r \in \mathcal{R}} \hat{\pi}_r \right)^2}
\\ \vdots & \vdots & \ddots & \vdots
\\    -\frac{  \hat{\pi}_{r_1} }{\left(\sum_{r \in \mathcal{R}} \hat{\pi}_r \right)^2} &    -\frac{  \hat{\pi}_{r_2} }{\left(\sum_{r \in \mathcal{R}} \hat{\pi}_r \right)^2}  & \cdots &  \frac{ \sum_{r \in \mathcal{R} \setminus r_R} \hat{\pi}_r  }{\left(\sum_{r \in \mathcal{R}} \hat{\pi}_r \right)^2}
\end{pmatrix} \in \mathbb{R}^{(R+1) \times R}
;
\]
notice it is nonzero as required.

Part (b) is immediate from Theorem \ref{main.te.cons.thm.gen}(b) using \(a_r = \boldsymbol{e}_r\), the selector vector with a 1 in the \(r^{\text{th}}\) position and zeroes elsewhere, and \(\boldsymbol{b}_r = (1, 1, \ldots, 1)^\top\).
\end{proof}

\subsection{Proofs of Theorems \ref{te.asym.norm.thm.gen}, \ref{te.oracle.thm}, and \ref{te.asym.norm.thm}}

Next we prove Theorem \ref{te.asym.norm.thm.gen}. Then we will use this result to prove Theorem \ref{te.asym.norm.thm}, and finally Theorem \ref{te.oracle.thm}.

%\textbf{[FIX THIS: INTRODUCE GENERALIZATION]}

\begin{proof}[Proof of Theorem \ref{te.asym.norm.thm.gen}]

By Theorem \ref{te.interp.prop}, the average treatment effect \(\tau_{\text{ATT}}(r, t)\) is equal to the regression estimand \(\tau_{rt}^*\).  We will prove Theorem \ref{te.asym.norm.thm.gen} by applying Theorem \ref{first.thm.fetwfe}(e), (g), and (h) with particular choices of vectors \(\boldsymbol{\psi}_N\) and \(\boldsymbol{\hat{\psi}}_N\). 

\begin{enumerate}[(a)]

\item Let \(i(r,t)\) be the index of the row of \(\boldsymbol{D}_N\) corresponding to \(\hat{\tau}_{rt}^{(q)}\) (that is, the column of \(\boldsymbol{\tilde{Z}}\) corresponding to this regression coefficient), so that
\[
\hat{\tau}_{rt}^{(q)} = (\boldsymbol{D}_N^{-1})_{i(r,t) \cdot}  \boldsymbol{\hat{\theta}}^{(q)}
\]
and
\[
\tau_{rt}^* = (\boldsymbol{D}_N^{-1})_{i(r,t) \cdot} \boldsymbol{\theta}_N^*
.
\]
%We will show that \(\boldsymbol{\psi}_N = \left( (\boldsymbol{D}_N^{-1})_{i(r,t), \cdot}\right)^\top \) satisfies the requirements of \(\boldsymbol{\psi}_N\) imposed by Theorem \ref{first.thm.fetwfe}(d), (e), and (f). 
We can see from Lemma \ref{lem.d.express} that each row of \(\boldsymbol{D}_N^{-1}\) contains at most \((T-1)^2\) nonzero entries, all of the nonzero entries equal 1, and the entries corresponding to the \(s\) relevant features will stay fixed even if \(p_N \to \infty\). So the vector \(\boldsymbol{\psi}_N = \left( (\boldsymbol{D}_N^{-1})_{i(r,t), \cdot}\right)^\top \) is a valid choice for Theorem \ref{first.thm.fetwfe}(d) - (h). In fact, any fixed, finite linear combination of rows \((\boldsymbol{D}_N^{-1})_{i(r,t), \cdot}\) is a valid choice. We therefore have from Theorem \ref{first.thm.fetwfe}(e) that if at least one of the nonzero components of \((\boldsymbol{D}_N^{-1})_{i(r,t), \cdot}\) corresponds to one of the \(s\) features in \(\mathcal{S}\),
 \begin{align*}
  \sqrt{ \frac{NT}{  \hat{v}_N^{\text{ATT};r,t} }}  \left( (\boldsymbol{D}_N^{-1})_{i(r,t), \cdot }  \boldsymbol{\hat{\theta}}^{(q)} - ( \boldsymbol{D}_N^{-1})_{i(r,t),  \cdot }  \boldsymbol{\theta}_N^* \right)
  \xrightarrow{d} ~ &  \mathcal{N}(0, 1)
\\ \iff \qquad  \sqrt{ \frac{NT}{ \hat{v}_N^{\text{ATT};r,t}}}     \left(  \hat{\tau}_{\text{ATT}} ( r, t) - \tau_{\text{ATT}} (r, t) \right) 
 \xrightarrow{d} ~ &   \mathcal{N}(0, 1)
,
\end{align*}
where
% and
\[
\hat{v}_N^{\text{ATT};r,t} :=  \sigma^2 (\boldsymbol{D}_N^{-1})_{i(r,t), \hat{\mathcal{S}}}  \left( \widehat{\Cov} \left( (\boldsymbol{Z} \boldsymbol{D}_N^{-1}) _{(\cdot \cdot) \hat{\mathcal{S}}} \right) \right)^{-1}  ( (\boldsymbol{D}_N^{-1})_{i(r,t), \hat{\mathcal{S}}})^\top
,
\]
\(\widehat{\Cov}(\cdot)\) is defined in \eqref{est.cov.def} and \((\boldsymbol{D}_N^{-1})_{i(r,t), \hat{\mathcal{S}}} \in \mathbb{R}^{|\hat{\mathcal{S}}|}\) is a row vector that contains the \(|\hat{\mathcal{S}}|\) components of the \(i(r,t)\) row of \(\boldsymbol{D}_N^{-1}\) corresponding to the selected features. Similarly, for any deterministic fixed set of constants \(\{\psi_{rt}\}\), since
\[
\sum_{r \in \mathcal{R}} \sum_{t=r}^T \psi_{rt} \tau_{rt}^* =   \sum_{r \in \mathcal{R}} \sum_{t=r}^T \psi_{rt} (\boldsymbol{D}_N^{-1})_{i(r,t), \cdot} \boldsymbol{\theta}_N^*
,
\]
%\[
%\tau_{\text{ATT}} (r)  =   \frac{1}{T-r+1} \sum_{t=r}^T \tau_{rt}^* =   \frac{1}{T-r+1} \sum_{t=r}^T  (\boldsymbol{D}_N^{-1})_{i(r,t), \cdot} \boldsymbol{\theta}_N^*
%,
%\]
we can choose \(\boldsymbol{\psi}_N^{(\text{C})} :=  \sum_{r \in \mathcal{R}} \sum_{t=r}^T \psi_{rt}(\boldsymbol{D}_N^{-1})_{i(r,t),  \cdot }^\top \) to establish that
 \begin{align*}
\sqrt{ \frac{NT}{ \hat{v}_N^{(\text{C})} }}  \left(  (\boldsymbol{\psi}_N^{(\text{C})})^\top (\boldsymbol{\hat{\theta}}^{(q)} - \boldsymbol{\theta}_N^*) \right) \xrightarrow{d} ~ &  \mathcal{N}(0, 1)
\\ \iff \qquad   \sqrt{ \frac{NT}{\hat{v}_N^{(\text{C})} }}  \sum_{r \in \mathcal{R}} \sum_{t=r}^T \psi_{rt}  ( \hat{\tau}_{\text{ATT}} (  r,t )  - \tau_{\text{ATT}} (r,t ) ) \xrightarrow{d} ~ &  \mathcal{N}(0, 1)
,
\end{align*}
where
\begin{equation}\label{v.n.r.t.att.const}
\hat{v}_N^{(\text{C})} :=  \sigma^2 \left( (\boldsymbol{\psi}_N^{(\text{C})})_{\hat{\mathcal{S}}} \right)^\top \left( \widehat{\Cov} \left( (\boldsymbol{Z} \boldsymbol{D}_N^{-1}) _{(\cdot \cdot) \hat{\mathcal{S}}} \right) \right)^{-1} (\boldsymbol{\psi}_N^{(\text{C})})_{\hat{\mathcal{S}}}
.
\end{equation}
%\begin{equation}\label{v.n.r.att}
%\hat{v}_N^{\text{ATT};r} :=  \sigma^2 \left( \boldsymbol{\psi}_N^{(\text{ATT};r)}\right)^\top \left( \widehat{\Cov} \left( (\boldsymbol{Z} \boldsymbol{D}_N^{-1}) _{(\cdot \cdot) \hat{\mathcal{S}}} \right) \right)^{-1} \boldsymbol{\psi}_N^{(\text{ATT};r)}
%.
%\end{equation}

\item We will prove the convergence of the results with estimated cohort probabilities by applying Theorem \ref{first.thm.fetwfe}(g). Consider the random vector
\begin{equation}\label{hat.psi.hat.pi.c.def}
\boldsymbol{\hat{\psi}}_N^{(\text{C; } \hat{\pi})}  :=  \sum_{r \in \mathcal{R}} f_r(\hat{\pi}_{r_1}, \hat{\pi}_{r_2}, \ldots, \hat{\pi}_{r_R}) \sum_{t=r}^T   \psi_{rt}(\boldsymbol{D}_N^{-1})_{i(r,t), \cdot } \in \mathbb{R}^{p_N}
\end{equation}
and define
\[
\boldsymbol{\psi}_N^{(\text{C; } \pi)} :=   \sum_{r \in \mathcal{R}}  f_r(\pi_{r_1}, \pi_{r_2}, \ldots, \pi_{r_R})  \sum_{t=r}^T   \psi_{rt}(\boldsymbol{D}_N^{-1})_{i(r,t), \cdot} \in \mathbb{R}^{p_N}
,
\]
so that
 \begin{align*}
& (\boldsymbol{\hat{\psi}}_N^{(\text{C; } \hat{\pi})} )^\top  \boldsymbol{\hat{\theta}}^{(q)} - (  \boldsymbol{\psi}_N^{(\text{C; } \pi)})^\top \boldsymbol{\theta}_N^*
\\ = ~ &  \sum_{r \in \mathcal{R}}  \sum_{t =r}^T  \psi_{rt}  \left(   f_r(\hat{\pi}_{r_1}, \hat{\pi}_{r_2}, \ldots, \hat{\pi}_{r_R})    \hat{\tau}_{\text{ATT}} ( r, t) -   f_r(\pi_{r_1}, \pi_{r_2}, \ldots, \pi_{r_R})  \tau_{\text{ATT}} (r, t)  \right)
.
\end{align*}
%We will show that \((\boldsymbol{\hat{\psi}}_N^{(\text{C; } \hat{\pi})})_{\mathcal{S}} \xrightarrow{p} (\boldsymbol{\psi}_N^{(\text{C; } \pi)} )_{\mathcal{S}}\). By the weak law of large numbers,
%%By Theorem 14.6 in \citet{wasserman2004all},
%%\[
%%\sqrt{N} \left( \frac{N_{r_1}}{N}, \frac{N_{r_2}}{N}, \ldots, \frac{N_{r_{R}}}{N} \right)
%%\]
%%converges in distribution to a mean 0 multivariate Gaussian distribution. This implies that
%\[
%\left( \frac{N_{r_1}}{N}, \frac{N_{r_2}}{N}, \ldots, \frac{N_{r_{R}}}{N} \right) \xrightarrow{p} \left( \pi_{r_1}, \pi_{r_2}, \ldots, \pi_{r_{R}}  \right)
%,
%\]
%where \(\pi_r =  \E[ \pi_r(\boldsymbol{X})] \) are the marginal probabilities of treatment assignments, as defined in \eqref{pi.r.def}. By the continuous mapping theorem, this also implies that
%\[
%\frac{N_\tau}{N} = \frac{1}{N} \sum_{k=1}^R N_{r_k} \xrightarrow{p}  \sum_{k=1}^R \pi_{r_k} =: \pi_\tau
%.
%\]
%Therefore by the continuous mapping theorem and using the fact that under Assumption (F2) \(\pi_r > 0\) for all \(r\), we have that for all \(r \in \mathcal{R}\)
%\begin{equation}\label{conv.cohort.probs}
% \left| \tilde{\pi}_r   - \frac{N_r}{N_\tau}\right|  = \left| \tilde{\pi}_r   - \frac{N_r/N}{N_\tau/N}\right| \xrightarrow{p}  \left| \tilde{\pi}_r   - \frac{\pi_r}{\pi_\tau}\right| =  \left| \tilde{\pi}_r   - \tilde{\pi}_r \right|  = 0
% .
%\end{equation}
%Therefore by the continuous mapping theorem \((\boldsymbol{\hat{\psi}}_N^{(\text{C; } \hat{\pi})})_ {\mathcal{S}} \xrightarrow{p} (\boldsymbol{\psi}_N^{(\text{C; } \pi)})_{\mathcal{S}} \). 

We have assumed that the \((\hat{\pi}_{r_1}, \hat{\pi}_{r_2}, \ldots, \hat{\pi}_{r_R})\) values (and therefore \(\boldsymbol{\hat{\psi}}_N^{(\text{C; } \hat{\pi})}\)) are estimated on an independent data set from the one used to estimate \(\boldsymbol{\hat{\theta}}^{(q)}\). It remains to show that 
\[
\sqrt{NT} ( \boldsymbol{\theta}^* )_{\mathcal{S}}^\top \left((\boldsymbol{\hat{\psi}}_N^{(\text{C; } \hat{\pi})} )_{\mathcal{S}} - (\boldsymbol{\psi}_N^{(\text{C; } \pi)})_{\mathcal{S}}  \right) 
\]
 is asymptotically normal and to identify a consistent estimator for the asymptotic variance. First we establish asymptotic normality. Because by assumption \(\boldsymbol{f}(\cdot)\) has a finite Jacobian \(\nabla \boldsymbol{f} \in \mathbb{R}^{(R + 1) \times R}\) at \((\pi_0, \pi_{r_1}, \ldots, \pi_{r_R})\) that exists and is nonzero and finite, using the asymptotic normality of the maximum likelihood estimator \(\left( \frac{N_{0}}{N}, \frac{N_{r_1}}{N}, \frac{N_{r_2}}{N}, \ldots, \frac{N_{r_{R}}}{N} \right) \) (Theorem 14.6 in \citealt{wasserman2004all}) and the delta method (Theorem 3.1 in \citealt{van2000asymptotic}) that
\begin{align*}
\sqrt{N} \left( \boldsymbol{f}(\hat{\pi}_0, \hat{\pi}_{r_1}, \ldots, \hat{\pi}_{r_R})  -  \boldsymbol{f}(\pi_0, \pi_{r_1}, \ldots, \pi_{r_R})  \right) \xrightarrow{d} ~ & \mathcal{N}(0, \nabla \boldsymbol{f}^\top \boldsymbol{\Sigma}_M \nabla \boldsymbol{f})
%
%\\ \implies \qquad  \sqrt{N} \boldsymbol{e}_{r'}^\top  \left( \hat{\pi}_{r_1} -  \tilde{\pi}_{r_1},  \hat{\pi}_{r_2} -  \tilde{\pi}_{r_2}, \ldots, \hat{\pi}_{r_R} -  \tilde{\pi}_{r_R}, \right) \xrightarrow{d} ~ & \mathcal{N}(0,  \boldsymbol{e}_{r'}^\top \boldsymbol{\nabla}_\pi^\top \boldsymbol{\Sigma}_M \boldsymbol{\nabla}_\pi  \boldsymbol{e}_{r'})
%
%\\ \implies \qquad  \sqrt{NT}  \left(  f_{r'} (\hat{\pi}_0, \hat{\pi}_{r_1}, \ldots, \hat{\pi}_{r_R})  -  f_{r'}(\pi_0, \pi_{r_1}, \ldots, \pi_{r_R})  \right) \xrightarrow{d} ~ & \mathcal{N}(0,  T \boldsymbol{e}_{r'}^\top  \nabla \boldsymbol{f}_{\mathcal{A}_N}^\top \boldsymbol{\Sigma}_M \nabla \boldsymbol{f}_{\mathcal{A}_N} \boldsymbol{e}_{r'})
,
\end{align*}
%for any \(r' \in [R]\), 
where \(\boldsymbol{\Sigma}_M \) was defined in \eqref{asym.cov.multinomi}. Then
\begin{align*}
& \sqrt{NT} ( \boldsymbol{\theta}^* )_{\mathcal{S}}^\top \left((\boldsymbol{\hat{\psi}}_N^{(\text{C; } \hat{\pi})} )_{\mathcal{S}} - (\boldsymbol{\psi}_N^{(\text{C; } \pi)})_{\mathcal{S}}  \right) 
\\ = ~ & \sqrt{NT} ( \boldsymbol{\theta}_N^*)_{\mathcal{S}}^\top \left( \sum_{r \in \mathcal{R}} \left[ f_r(\hat{\pi}_{r_1}, \hat{\pi}_{r_2}, \ldots, \hat{\pi}_{r_R}) - f_r(\pi_{r_1}, \pi_{r_2}, \ldots, \pi_{r_R}) \right] \sum_{t=r}^T   \psi_{rt}(\boldsymbol{D}_N^{-1})_{i(r,t), \mathcal{S}}  \right)
\\ = ~ & \sqrt{NT} ( \boldsymbol{\theta}_N^*)_{\mathcal{S}}^\top \left( \boldsymbol{M}^\top \left[ \boldsymbol{f}(\hat{\pi}_{r_1}, \hat{\pi}_{r_2}, \ldots, \hat{\pi}_{r_R}) - \boldsymbol{f}(\pi_{r_1}, \pi_{r_2}, \ldots, \pi_{r_R}) \right] \right)
\\ \xrightarrow{d} ~ &  \mathcal{N}(0, \boldsymbol{M}^\top \nabla \boldsymbol{f}^\top \boldsymbol{\Sigma}_M \nabla \boldsymbol{f} \boldsymbol{M})
,
\end{align*}
where
\[
\boldsymbol{M} := 
\begin{pmatrix}
\sum_{t=r_1}^T   \psi_{r_1 t}(\boldsymbol{D}_N^{-1})_{i(r_1,t), \mathcal{S}}  
\\ \sum_{t=r_2}^T   \psi_{r_2 t}(\boldsymbol{D}_N^{-1})_{i(r_2,t), \mathcal{S}}  
\\ \vdots
\\ \\ \sum_{t=r_R}^T   \psi_{r_R t}(\boldsymbol{D}_N^{-1})_{i(r_R,t), \mathcal{S}}  
\end{pmatrix} \in \mathbb{R}^{R \times s}
.
\]

% Using the asymptotic normality of the maximum likelihood estimator \(\left( \frac{N_{0}}{N}, \frac{N_{r_1}}{N}, \frac{N_{r_2}}{N}, \ldots, \frac{N_{r_{R}}}{N} \right) \) (Theorem 14.6 in \citealt{wasserman2004all}), by the multivariate delta method (Theorem 3.1 in \citealt{van2000asymptotic}), 
%\begin{align*}
%\sqrt{N}  \left((\boldsymbol{\hat{\psi}}_N^{(\text{C; } \hat{\pi})} )_{\mathcal{S}} - (\boldsymbol{\psi}_N^{(\text{C; } \pi)})_{\mathcal{S}}  \right) \xrightarrow{d} ~ &  \mathcal{N} \left(\boldsymbol{0}_s,  (\nabla \boldsymbol{\psi}_N^{(\text{C; } \pi)})_{\mathcal{S}}^\top \boldsymbol{\Sigma}_M  (\nabla \boldsymbol{\psi}_N^{(\text{C; } \pi)})_{\mathcal{S}}\right)
%%
%\\ \implies \qquad \sqrt{NT}  \left( (\boldsymbol{\hat{\psi}}_N^{(\text{C; } \hat{\pi})})_{\mathcal{S}}  - (\boldsymbol{\psi}_N^{(\text{C; } \pi)})_{\mathcal{S}}  \right) \xrightarrow{d} ~ &  \mathcal{N} \left(\boldsymbol{0}_s,  T (\nabla \boldsymbol{\psi}_N^{(\text{C; } \pi)})_{\mathcal{S}}^\top \boldsymbol{\Sigma}_M  (\nabla \boldsymbol{\psi}_N^{(\text{C; } \pi)})_{\mathcal{S}}\right)
%,
%\end{align*}
%where \(\boldsymbol{\Sigma}_M \) was defined in \eqref{asym.cov.multinomi}, \]
%\[
%\vdots
%\]
Further, we can estimate this covariance matrix by \(\boldsymbol{\hat{V}}(\mathcal{A}_N)\) defined for any \(\mathcal{A}_N \subseteq [p_N]\) by
\begin{equation}\label{cov.mat.def.thm.6.3.gen}
\boldsymbol{\hat{V}} (\mathcal{A}_N) = T \cdot   \boldsymbol{\hat{M}} (\mathcal{A}_N)^\top\nabla \boldsymbol{\hat{f}}^\top \boldsymbol{\hat{\Sigma}}_M    \nabla \boldsymbol{\hat{f}}  \boldsymbol{\hat{M}} (\mathcal{A}_N)
,
\end{equation}
where
\[
\boldsymbol{\hat{M}} (\mathcal{A}_N) := 
\begin{pmatrix}
\sum_{t=r_1}^T   \psi_{r_1 t}(\boldsymbol{D}_N^{-1})_{i(r_1,t), \mathcal{A}_N}  
\\ \sum_{t=r_2}^T   \psi_{r_2 t}(\boldsymbol{D}_N^{-1})_{i(r_2,t), \mathcal{A}_N}  
\\ \vdots
\\ \\ \sum_{t=r_R}^T   \psi_{r_R t}(\boldsymbol{D}_N^{-1})_{i(r_R,t), \mathcal{A}_N}  
\end{pmatrix} \in \mathbb{R}^{R \times |\mathcal{A}_N|}
,
\]
\begin{equation}\label{hat.sigma.mat.def}
\boldsymbol{\hat{\Sigma}}_M :=  \frac{1}{N^2}\begin{pmatrix}
N_0(N - N_0) & - N_0 N_{r_1}  & \cdots & - N_0 N_{r_R}
\\ - N_0 N_{r_1} & N_{r_1} (N - N_{r_1}) & \cdots & - N_{r_1} N_{r_R}
\\ \vdots & \vdots & \ddots & \vdots
\\ - N_0 N_{r_R} & - N_{r_1} N_{r_R} & \cdots & N_{r_R}(N - N_{r_R})
\end{pmatrix} \in \mathbb{R}^{(R +1) \times (R + 1)}
,
\end{equation}
and \(  \nabla \boldsymbol{\hat{f}}\) is defined by replacing each instance of \(\pi_{r}\) in \(\nabla \boldsymbol{f}\) with its estimator \(\hat{\pi}_r\).
%\[
%N \begin{pmatrix} \boldsymbol{0}_{|\mathcal{A}_N|}^\top,
%\\  \frac{\sum_{r' \in \mathcal{R} \setminus r_1} N_{r'}}{\left( \sum_{r' \in \mathcal{R}} N_{r'} \right)^2} \sum_{t=r_1}^T   \psi_{r_1t}(\boldsymbol{D}_N^{-1})_{i(r_1,t), \mathcal{A}_N} - \sum_{r'' \in \mathcal{R} \setminus r_1}  \frac{ N_{r_{r''}}}{\left( \sum_{r' \in \mathcal{R}} N_{r'} \right)^2} \sum_{t=r''}^T   \psi_{r''t}(\boldsymbol{D}_N^{-1})_{i(r'',t), \mathcal{A}_N}   
%\\  \vdots
%\\     \frac{\sum_{r' \in \mathcal{R} \setminus r_R} N_{r'}}{\left( \sum_{r' \in \mathcal{R}} N_{r'} \right)^2} \sum_{t=r_R}^T   \psi_{r_Rt}(\boldsymbol{D}_N^{-1})_{i(r_R,t), \mathcal{A}_N} - \sum_{r'' \in \mathcal{R} \setminus r_R}  \frac{ N_{r_{r''}}}{\left( \sum_{r' \in \mathcal{R}} N_{r'} \right)^2} \sum_{t=r''}^T   \psi_{r''t}(\boldsymbol{D}_N^{-1})_{i(r'',t), \mathcal{A}_N}   
%  \end{pmatrix}
%  .
%\]
%
%\[
%\vdots
%\]
Using the consistency of \(  \boldsymbol{\hat{\theta}}^{(q)} \) and the sample covariance matrix, the continuous mapping theorem yields that
\[
(  \boldsymbol{\hat{\theta}}^{(q)})_{\mathcal{S}} ^\top   \boldsymbol{\hat{V}} (\mathcal{S})   \boldsymbol{\hat{\theta}}^{(q)}_{\mathcal{S}}
\xrightarrow{p} 
T ( \boldsymbol{\theta}_N^* )_{\mathcal{S}}^\top  (\nabla \boldsymbol{\psi}_N^{(\text{C; } \pi)})_{\mathcal{S}}^\top \boldsymbol{\Sigma}_M  (\nabla \boldsymbol{\psi}_N^{(\text{C; } \pi)})_{\mathcal{S}} (  \boldsymbol{\theta}_N^*  )_{\mathcal{S}}
,
\]
%Therefore 
%\[
%\boldsymbol{\hat{\psi}}_N^{(\text{C; } \hat{\pi})}  ( \mathcal{S})  \xrightarrow{p}   \sum_{r \in \mathcal{R}} \tilde{\pi}_r \sum_{t=r}^T   \psi_{rt}(\boldsymbol{D}_N^{-1})_{i(r,t), \mathcal{S}}
%.
%\]
where \(\boldsymbol{\Sigma}_M \) was defined in \eqref{asym.cov.multinomi}. It follows from Theorem \ref{first.thm.fetwfe}(g) that if at least one of the nonzero components of \( \boldsymbol{\psi}_N^{(\text{C; } \pi)}  \) corresponds to one of the \(s\) features in \(\mathcal{S}\), we have
 \begin{align*}
&   \sqrt{ \frac{NT}{  \hat{v}_N^{\text{C}; \hat{\pi}} }}   \left( (\boldsymbol{\hat{\psi}}_N^{(\text{C; } \hat{\pi})} )^\top  \boldsymbol{\hat{\theta}}^{(q)} - (  \boldsymbol{\psi}_N^{(\text{C; } \pi)})^\top \boldsymbol{\theta}_N^* \right)
  \xrightarrow{d}  \mathcal{N}(0, 1)
\\ \iff \qquad  &  \sqrt{ \frac{NT}{   \hat{v}_N^{\text{C}; \hat{\pi}}  }}  \sum_{r \in \mathcal{R}}  \sum_{t =r}^T   \psi_{rt}  \left(   f_r(\hat{\pi}_{r_1}, \hat{\pi}_{r_2}, \ldots, \hat{\pi}_{r_R})    \hat{\tau}_{\text{ATT}} ( r, t) -   f_r(\pi_{r_1}, \pi_{r_2}, \ldots, \pi_{r_R})  \tau_{\text{ATT}} (r, t)  \right)  
\\  \xrightarrow{d} ~ &   \mathcal{N}(0, 1)
,
\end{align*}
where
\begin{equation}\label{v.n.r.t.att.rand}
 \hat{v}_N^{\text{C}; \hat{\pi}} :=   \sigma^2   (\boldsymbol{\hat{\psi}}_N^{(\text{C; } \hat{\pi})})_{ \hat{\mathcal{S}}}^\top \left( \widehat{\Cov} \left( (\boldsymbol{Z} \boldsymbol{D}_N^{-1}) _{(\cdot \cdot) \hat{\mathcal{S}}} \right) \right)^{-1}      (\boldsymbol{\hat{\psi}}_N^{(\text{C; } \hat{\pi})})_{ \hat{\mathcal{S}}}  + (  \boldsymbol{\hat{\theta}}^{(q)})_{\hat{\mathcal{S}}} ^\top  \boldsymbol{\hat{V}} (\hat{\mathcal{S}})  (  \boldsymbol{\hat{\theta}}^{(q)})_{\hat{\mathcal{S}}} 
.
\end{equation}

\item Part (b) follows from the same argument as part (a) but applying Theorem \ref{first.thm.fetwfe}(h) to yield the stated convergence when using the conservative variance estimator
\begin{align}
& \hat{v}_N^{\text{C, (cons)}; \hat{\pi}} \nonumber
\\  :=   ~ & \sigma^2   (\boldsymbol{\hat{\psi}}_N^{(\text{C; } \hat{\pi})})_{ \hat{\mathcal{S}}}^\top \left( \widehat{\Cov} \left( (\boldsymbol{Z} \boldsymbol{D}_N^{-1}) _{(\cdot \cdot) \hat{\mathcal{S}}} \right) \right)^{-1}      (\boldsymbol{\hat{\psi}}_N^{(\text{C; } \hat{\pi})})_{ \hat{\mathcal{S}}}  + (  \boldsymbol{\hat{\theta}}^{(q)})_{\hat{\mathcal{S}}} ^\top  \boldsymbol{\hat{V}} (\hat{\mathcal{S}})  (  \boldsymbol{\hat{\theta}}^{(q)})_{\hat{\mathcal{S}}}  \nonumber
 \\  & + 2 \sqrt{ \sigma^2   (\boldsymbol{\hat{\psi}}_N^{(\text{C; } \hat{\pi})})_{ \hat{\mathcal{S}}}^\top \left( \widehat{\Cov} \left( (\boldsymbol{Z} \boldsymbol{D}_N^{-1}) _{(\cdot \cdot) \hat{\mathcal{S}}} \right) \right)^{-1}      (\boldsymbol{\hat{\psi}}_N^{(\text{C; } \hat{\pi})})_{ \hat{\mathcal{S}}}   \cdot  (  \boldsymbol{\hat{\theta}}^{(q)})_{\hat{\mathcal{S}}} ^\top  \boldsymbol{\hat{V}} (\hat{\mathcal{S}})  (  \boldsymbol{\hat{\theta}}^{(q)})_{\hat{\mathcal{S}}} } 
 .\label{v.n.r.t.att.rand.cons}
\end{align}
(Recall that the random vector \(\boldsymbol{\hat{\psi}}_N^{(\text{C; } \hat{\pi})}\) was defined in Equation \ref{hat.psi.hat.pi.c.def} and the covariance matrix estimator \( \boldsymbol{\hat{V}} (\hat{\mathcal{S}}) \) was defined in Equation \ref{cov.mat.def.thm.6.3.gen}.)

\end{enumerate}

\end{proof}

\begin{proof}[Proof of Theorem \ref{te.asym.norm.thm}]

Part (a) is immediate from Theorem \ref{te.asym.norm.thm.gen}(a). To prove parts (b) and (c), in light of Theorem \ref{te.asym.norm.thm.gen}(b) and (c) it suffices to show that \((N_{r_1}/N_\tau, \ldots, N_{r_R}/N_\tau)\) is almost surely finite and has a Jacobian at \((\pi_0, \pi_{r_1}, \ldots, \pi_{r_R})\) that exists and is nonzero and finite, which we already showed in the proof of Theorem \ref{main.te.cons.thm}.

\end{proof}

\begin{proof}[Proof of Theorem \ref{te.oracle.thm}]
In the case where the treatment effects are nonzero, this result follows from the same arguments that we used in the proof of Theorem \ref{te.asym.norm.thm.gen}, but applying Theorem \ref{first.thm.fetwfe}(d) and (f). The oracle properties persist even in the case that the treatment effects equal 0 because, as noted in Theorem \ref{te.sel.cons.thm}, in this case the random variables scaled by \(\sqrt{NT}\) converge in probability (and distribution) to 0.
\end{proof}

\subsection{Proof of Theorem \ref{te.asym.norm.thm.gen.cond}}

%Finally, we now prove Theorem \ref{te.asym.norm.thm.gen.cond}, 
%
%\begin{proof}[Proof of Theorem \ref{te.asym.norm.thm.gen.cond}]

%
%
%
%
% asymptotic normality (gaussianity) of conditional average treatment effects
% here
%
%
%
%

\begin{enumerate}[(a)]

\item From Theorem \ref{te.interp.prop}, under our assumptions the regression estimands correspond to our causal estimands. Similarly to the proof of Theorem \ref{te.asym.norm.thm.gen}, we will prove the result by applying Theorem \ref{first.thm.fetwfe}(g) and Lemma \ref{lem.asym.norm.cond.means}. Let \(i(r, t, \boldsymbol{x})\) denote the indices of \(d\) rows of \(\boldsymbol{D}_N\) corresponding to \(\hat{\rho}_{rt}^{(q)}\) (that is, the columns of \(\boldsymbol{\tilde{Z}}\) corresponding to the indicator variables for those regression coefficients), so that
\[
 \boldsymbol{\rho}_{rt}^* = (\boldsymbol{D}_N^{-1})_{i(r, t, \boldsymbol{x}), \cdot} \boldsymbol{\theta}_N^*,
\]
and for any \(r \in \mathcal{R}\), \(t \geq r\), and \(\boldsymbol{x}\) in the support of \(\boldsymbol{X}_i\),
\begin{align*}
\tau_{\text{CATT}}(r, t, \boldsymbol{x}) = ~ &  \tau_{rt}^* + \left( \boldsymbol{x} - \E[ \boldsymbol{X}_i \mid W_i = r ] \right)^\top \boldsymbol{\rho}_{rt}^*
\\ = ~ &   (\boldsymbol{D}_N^{-1})_{i(r,t), \cdot} \boldsymbol{\theta}_N^* + \left( \boldsymbol{x} - \E[ \boldsymbol{X}_i \mid W_i = r ] \right)^\top (\boldsymbol{D}_N^{-1})_{i(r, t, \boldsymbol{x}), \cdot} \boldsymbol{\theta}_N^*
%
%\\ = ~ &  \left(  (\boldsymbol{D}_N^{-1})_{i(r,t), \cdot}^\top +  (\boldsymbol{D}_N^{-1})_{i(r, t, \boldsymbol{x}), \cdot} ^\top   \boldsymbol{x} -   (\boldsymbol{D}_N^{-1})_{i(r, t, \boldsymbol{x}), \cdot} ^\top\E[ \boldsymbol{X}_i \mid W_i = r ] \right) \boldsymbol{\theta}_N^*
.
\end{align*}
%and for any fixed finite set of constants \(\{\psi_{rt}\}\),
%\begin{align*}
%& \sum_{r \in \mathcal{R}} \sum_{t = r}^T \psi_{rt} \tau_{\text{CATT}}(r, t, \boldsymbol{x})
%%
%\\ = ~ & \sum_{r \in \mathcal{R}} \sum_{t = r}^T \psi_{rt} \left(  (\boldsymbol{D}_N^{-1})_{i(r,t), \cdot}^\top +  (\boldsymbol{D}_N^{-1})_{i(r, t, \boldsymbol{x}), \cdot} ^\top   \boldsymbol{x} -   (\boldsymbol{D}_N^{-1})_{i(r, t, \boldsymbol{x}), \cdot} ^\top\E[ \boldsymbol{X}_i \mid W_i = r ] \right) \boldsymbol{\theta}_N^*
%%
%\\ = ~ &  (\boldsymbol{\psi}_N^{(\text{C;CATT})})^\top \boldsymbol{\theta}_N^*
%\end{align*}
%By the same arguments as in the previous part, we have that
%\[
% \sqrt{ \frac{NT}{ v_N^{\text{CATT}}}}  \left(  (\boldsymbol{\psi}_N^{(\text{C;CATT})})^\top (\boldsymbol{\hat{\theta}}^{(q)} - \boldsymbol{\theta}_N^*) \right) \xrightarrow{d}   \mathcal{N}(0, 1)
%,
%\]
%where
%\[
%v_N^{\text{CATT}} :=  \sigma^2 (\boldsymbol{\psi}_N^{(\text{C;CATT})})^\top \left( \widehat{\Cov} \left( (\boldsymbol{Z} \boldsymbol{D}_N^{-1}) _{(\cdot \cdot) \hat{\mathcal{S}}} \right) \right)^{-1} \boldsymbol{\psi}_N^{(\text{C;CATT})}
%.
%\]
Define
\[
\boldsymbol{\psi}_N^{(\text{C;CATT})} :=  \sum_{r \in \mathcal{R}} \sum_{t = r}^T \psi_{rt}  \left(  (\boldsymbol{D}_N^{-1})_{i(r,t), \cdot} + \left( \boldsymbol{x} -   \E[ \boldsymbol{X}_i \mid W_i = r ]  \right) ^\top (\boldsymbol{D}_N^{-1})_{i(r, t, \boldsymbol{x}), \cdot}  \right) 
\]
and
\[
\boldsymbol{\hat{\psi}}_N^{(\text{C;CATT})} :=  \sum_{r \in \mathcal{R}} \sum_{t = r}^T \psi_{rt}  \left(  (\boldsymbol{D}_N^{-1})_{i(r,t), \cdot} + ( \boldsymbol{x} -  \boldsymbol{\overline{X}}_r) ^\top (\boldsymbol{D}_N^{-1})_{i(r, t, \boldsymbol{x}), \cdot}  \right) 
,
\]
so
\begin{align*}
 (\boldsymbol{\hat{\psi}}_N^{(\text{C;CATT})} )^\top  \boldsymbol{\hat{\theta}}^{(q)} - ( \boldsymbol{\psi}_N^{(\text{C;CATT})})^\top \boldsymbol{\theta}_N^*  =  \sum_{r \in \mathcal{R}} \sum_{t=r}^T \psi_{rt} \left(    \hat{\tau}_{\text{CATT}}(r, t, \boldsymbol{x}) -  \tau_{\text{CATT}}(r, t, \boldsymbol{x})  \right)
.
\end{align*}
Since \(  \boldsymbol{\hat{\theta}}^{(q)}\) and each \(\boldsymbol{\overline{X}}_r\) are estimated on two independent data sets, \(  \boldsymbol{\hat{\theta}}^{(q)}\)  is independent of \(\boldsymbol{\hat{\psi}}_N^{(\text{C;CATT})}\). So to apply Theorem \ref{first.thm.fetwfe}(g) we need to establish the asymptotic normality of \( \sqrt{NT}  \left( \boldsymbol{\hat{\psi}}_N^{(\text{C;CATT})}   - \boldsymbol{\psi}_N^{(\text{C;CATT})} \right)^\top  \boldsymbol{\theta}_N^*\) and identify a consistent estimator for its asymptotic variance.
%Notice that \(\boldsymbol{\hat{\psi}}_N^{(\text{C;CATT})}  \xrightarrow{p} \boldsymbol{\psi}_N^{(\text{C;CATT})} \), where
%\[
%\boldsymbol{\psi}_N^{(\text{C;CATT})} :=  \sum_{r \in \mathcal{R}} \sum_{t = r}^T \psi_{rt}  \left(  (\boldsymbol{D}_N^{-1})_{i(r,t), \cdot} +  \left( \boldsymbol{x} - \E[ \boldsymbol{X}_i \mid W_i = r ] \right)^\top (\boldsymbol{D}_N^{-1})_{i(r, t, \boldsymbol{x}), \cdot}  \right) 
%.
%\]
%\textbf{[FIX THIS: SOMETHING OFF HERE---MAYBE SHOULD CONVERGE TO A MULTIVARIATE NORMAL? ISN'T \( \boldsymbol{\hat{\psi}}_N^{(\text{C;CATT})}  \) A VECTOR?]} Further, 
Applying Lemma \ref{lem.asym.norm.cond.means} we see that
\begin{align*}
%&   \sqrt{NT}  \sum_{r \in \mathcal{R}} \sum_{t=r}^T \psi_{rt} \left(    \hat{\tau}_{\text{CATT}}(r, t, \boldsymbol{x}) -  \tau_{\text{CATT}}(r, t, \boldsymbol{x})  \right)
%
\\ = ~ & \sqrt{NT}  \left(  \boldsymbol{\hat{\psi}}_N^{(\text{C;CATT})}   - \boldsymbol{\psi}_N^{(\text{C;CATT})}  \right)^\top  \boldsymbol{\theta}_N^* 
\\ = ~ &   \sqrt{NT} \left(  \sum_{r \in \mathcal{R}} \sum_{t = r}^T \psi_{rt}  \left(  (    \E[ \boldsymbol{X}_i \mid W_i = r ]  -  \boldsymbol{\overline{X}}_r) ^\top (\boldsymbol{D}_N^{-1})_{i(r, t, \boldsymbol{x}), \cdot}  \right)  \right)  \boldsymbol{\theta}_N^*
\\ = ~ &   \sqrt{NT} \left(  \sum_{r \in \mathcal{R}} \sum_{t = r}^T  - \psi_{rt}  \left( (  \boldsymbol{D}_N^{-1})_{i(r, t, \boldsymbol{x}), \cdot}  \boldsymbol{\theta}_N^* \right)^\top  \left(   \boldsymbol{\overline{X}}_r -  \E[ \boldsymbol{X}_i \mid W_i = r ]  \right)   \right)  
\\ \xrightarrow{d} ~ & \mathcal{N} \left(0, v_R^{\text{CATT}} \right)
,
\end{align*}
where
\begin{align*}
v_R^{\text{CATT}} :=  ~ & T  \sum_{r \in \mathcal{R}}   \frac{(\boldsymbol{\psi}_{r}^{\text{CATT}})^\top  \Cov ( \boldsymbol{X}_i \mid W_i = r) \boldsymbol{\psi}_{r}^{\text{CATT}}}{\mathbb{P}(W_i = r)} 
\end{align*}
for
\begin{align*}
\boldsymbol{\psi}_{r}^{\text{CATT}} := ~ & \sum_{t = r}^T   \psi_{rt} \left( (  \boldsymbol{D}_N^{-1})_{i(r, t, \boldsymbol{x}), \cdot}  \boldsymbol{\theta}_N^* \right)^\top 
 =  \sum_{t = r}^T   \psi_{rt} \boldsymbol{\rho}_{rt}^* , \qquad r \in \mathcal{R}
.
\end{align*}
Notice that each \(v_R^{\text{CATT}}\) is strictly positive and finite because each \(\Cov ( \boldsymbol{X}_i \mid W_i = r) \) is assumed to be finite and have a strictly positive minimum eigenvalue under Assumption (R6) and each \(\mathbb{P}(W_i = r)\) is strictly positive under Assumption (F2).
%\footnote{Notice that each of these sample proportions and sample conditional covariance matrices can be estimated on either of the two data sets, or on both data sets pooled, because we only need consistency for these terms.} 
Also, by the consistency of \( \boldsymbol{\hat{\theta}}^{(q)} \) and each \(N_r / N\) and \(\widehat{\Cov} ( \boldsymbol{X}_i \mid W_i = r) \), along with the continuous mapping theorem, we have that 
\[
 T \sum_{r \in \mathcal{R}}   \frac{(\boldsymbol{\hat{\psi}}_{r}^{\text{CATT}})^\top  \widehat{\Cov} ( \boldsymbol{X}_i \mid W_i = r) \boldsymbol{\hat{\psi}}_{r}^{\text{CATT}}}{N_r / N} \xrightarrow{p} v_R^{\text{CATT}} 
\]
where
\[
\boldsymbol{\hat{\psi}}_{r}^{\text{CATT}} :=  \sum_{t = r}^T   \psi_{rt}  \boldsymbol{\hat{\rho}}_{rt}^{(q)}   , \qquad r \in \mathcal{R}
.
\]
%Define 
%\[
%\boldsymbol{\hat{\psi}}_N^{(\text{C;CATT})} :=  \sum_{r \in \mathcal{R}} \sum_{t = r}^T \psi_{rt}  \left(  (\boldsymbol{D}_N^{-1})_{i(r,t), \cdot} + ( \boldsymbol{x} -  \boldsymbol{\overline{X}}_r) ^\top (\boldsymbol{D}_N^{-1})_{i(r, t, \boldsymbol{x}), \cdot}  \right) \boldsymbol{\hat{\theta}}
%.
%\]
Then from Theorem \ref{first.thm.fetwfe}(g) we have
\begin{align*}
 \sqrt{ \frac{NT}{ \hat{v}_N^{(\text{C;CATT})}   }}  \left(  (\boldsymbol{\hat{\psi}}_N^{(\text{C;CATT})} )^\top  \boldsymbol{\hat{\theta}}^{(q)} - ( \boldsymbol{\psi}_N^{(\text{C;CATT})})^\top \boldsymbol{\theta}_N^* \right) \xrightarrow{d}   \mathcal{N}(0, 1)
 \\ \iff \qquad  \sqrt{ \frac{NT}{ \hat{v}_N^{(\text{C;CATT})}   }}  \sum_{r \in \mathcal{R}} \sum_{t=r}^T \psi_{rt} \left(    \hat{\tau}_{\text{CATT}}(r, t, \boldsymbol{x}) -  \tau_{\text{CATT}}(r, t, \boldsymbol{x})  \right) \xrightarrow{d}   \mathcal{N}(0, 1)
\end{align*}
for
\begin{align}
\hat{v}_N^{(\text{C;CATT})} :=  ~ &  \sigma^2 (\boldsymbol{\hat{\psi}}_N^{(\text{C;CATT})})^\top \left( \widehat{\Cov} \left( (\boldsymbol{Z} \boldsymbol{D}_N^{-1}) _{(\cdot \cdot) \hat{\mathcal{S}}} \right) \right)^{-1} \boldsymbol{\hat{\psi}}_N^{(\text{C;CATT})}  \nonumber
\\ & + T \sum_{r \in \mathcal{R}}   \frac{(\boldsymbol{\hat{\psi}}_{r}^{\text{CATT}})^\top  \widehat{\Cov} ( \boldsymbol{X}_i \mid W_i = r) \boldsymbol{\hat{\psi}}_{r}^{\text{CATT}}}{N_r / N} \label{v.n.r.t.catt.const}
.
\end{align}

\item 

Observe that
\begin{align*}
& \sum_{r \in \mathcal{R}} \sum_{t = r}^T \psi_{rt}  f_r \left(  \pi_{r_1}(\boldsymbol{x}) , \pi_{r_2}(\boldsymbol{x}) , \ldots,  \pi_{r_R}(\boldsymbol{x})  \right)    \tau_{\text{CATT}}(r, t, \boldsymbol{x})
\\ = ~ & \sum_{r \in \mathcal{R}} \sum_{t = r}^T \psi_{rt}  f_r \left(  \pi_{r_1}(\boldsymbol{x}) , \pi_{r_2}(\boldsymbol{x}) , \ldots,  \pi_{r_R}(\boldsymbol{x})  \right)  \Big(  (\boldsymbol{D}_N^{-1})_{i(r,t), \cdot}^\top 
\\ & +  (\boldsymbol{D}_N^{-1})_{i(r, t, \boldsymbol{x}), \cdot} ^\top    \left( \boldsymbol{x} -  \E[ \boldsymbol{X}_i \mid W_i = r ]  \right) \Big) \boldsymbol{\theta}_N^*
\\ = ~ & \sum_{r \in \mathcal{R}} \sum_{t = r}^T \psi_{rt}  f_r \left(  \pi_{r_1}(\boldsymbol{x}) , \pi_{r_2}(\boldsymbol{x}) , \ldots,  \pi_{r_R}(\boldsymbol{x})  \right)  ( \boldsymbol{\psi}_{rt}^* )^\top\boldsymbol{\theta}_N^*
\\ = ~ &  (\boldsymbol{\psi}_N^{(\text{C;CATT,} \pi)})^\top \boldsymbol{\theta}_N^*
\end{align*}
where
\begin{equation}\label{hat.psi.nostar.rt.def}
 \boldsymbol{\psi}_{rt}^* :=   (\boldsymbol{D}_N^{-1})_{i(r,t), \cdot} +  (\boldsymbol{D}_N^{-1})_{i(r, t, \boldsymbol{x}), \cdot}^\top    \left( \boldsymbol{x} -  \E[ \boldsymbol{X}_i \mid W_i = r ]  \right)  \in \mathbb{R}^p
 \end{equation}
and
\[
\boldsymbol{\psi}_N^{(\text{C;CATT,} \pi)} :=  \sum_{r \in \mathcal{R}} 
%\frac{\pi_r(\boldsymbol{x})}{ \sum_{r' \in \mathcal{R}} \pi_{r'}(\boldsymbol{x})}  
\sum_{t = r}^T \psi_{rt}  
 f_r \left(  \pi_{r_1}(\boldsymbol{x}) , \pi_{r_2}(\boldsymbol{x}) , \ldots,  \pi_{r_R}(\boldsymbol{x})  \right) 
  \boldsymbol{\psi}_{rt}^*  \in \mathbb{R}^p
.
\]
Define the estimators
\begin{equation}\label{hat.psi.star.rt.def}
 \boldsymbol{\hat{\psi}}_{rt}^* :=   (\boldsymbol{D}_N^{-1})_{i(r,t), \cdot} +  (\boldsymbol{D}_N^{-1})_{i(r, t, \boldsymbol{x}), \cdot}^\top    \left( \boldsymbol{x} -  \boldsymbol{\overline{X}}_r  \right) 
 \end{equation}
 and
 \[
\boldsymbol{\hat{\psi}}_N^{(\text{C;CATT,} \pi)} :=  \sum_{r \in \mathcal{R}} 
%\frac{\hat{\pi}_r(\boldsymbol{x})}{ \sum_{r' \in \mathcal{R}} \hat{\pi}_{r'}(\boldsymbol{x})}  
\sum_{t = r}^T \psi_{rt} 
 f_r \left(  \hat{\pi}_{r_1}(\boldsymbol{x}) , \hat{\pi}_{r_2}(\boldsymbol{x}) , \ldots,  \hat{\pi}_{r_R}(\boldsymbol{x})  \right) 
   \boldsymbol{\hat{\psi}}_{rt}^*
\]
and notice that
\begin{align*}
& (\boldsymbol{\hat{\psi}}_N^{(\text{C;CATT,} \pi)})^\top \boldsymbol{\hat{\theta}}^{(q)} - (\boldsymbol{\psi}_N^{(\text{C;CATT,} \pi)} )^\top   \boldsymbol{\theta}_N^*
\\ =  ~ &   \sum_{r \in \mathcal{R}} \sum_{t=r}^T \psi_{rt} \big(   
f_r \left(  \hat{\pi}_{r_1}(\boldsymbol{x}) , \hat{\pi}_{r_2}(\boldsymbol{x}) , \ldots,  \hat{\pi}_{r_R}(\boldsymbol{x})  \right) 
  \hat{\tau}_{\text{CATT}}(r, t, \boldsymbol{x})
  \\ &  -  
   f_r \left(  \pi_{r_1}(\boldsymbol{x}) , \pi_{r_2}(\boldsymbol{x}) , \ldots,  \pi_{r_R} (\boldsymbol{x})  \right) 
   \tau_{\text{CATT}}(r, t, \boldsymbol{x})  \big)
.
\end{align*}
Since by assumption each \(\hat{\pi}_r(\boldsymbol{x})\) and \(\boldsymbol{\overline{X}}_r \) is independent of \( \boldsymbol{\hat{\theta}}^{(q)}\), \(\boldsymbol{\hat{\psi}}_N^{(\text{C;CATT,} \pi)}\) is independent of \( \boldsymbol{\hat{\theta}}^{(q)}\), and the result follows from applying Theorem \ref{first.thm.fetwfe}(g) if we can show that
\[
\sqrt{NT}  (  \boldsymbol{\theta}_N^*)^\top  \left(  \boldsymbol{\hat{\psi}}_N^{(\text{C;CATT,} \pi)} - \boldsymbol{\psi}_N^{(\text{C;CATT,} \pi)} \right)
\]
is asymptotically normal and identify a consistent estimator for the asymptotic variance. Showing this is conceptually straightforward---we can decompose \(\sqrt{NT}  (  \boldsymbol{\theta}_N^*)^\top  \left(  \boldsymbol{\hat{\psi}}_N^{(\text{C;CATT,} \pi)} - \boldsymbol{\psi}_N^{(\text{C;CATT,} \pi)} \right)\) into a term that converges in probability to 0, a term that has an asymptotic distribution depending on \( \sqrt{N}    \left(  \boldsymbol{\overline{X}}_r - \boldsymbol{\mu}_{r}  \right)  \) (asymptotically normal due to Lemma \ref{lem.asym.norm.cond.means}), and a term that depends on the \(\hat{\pi}_r(\boldsymbol{x})\) predictions (asymptotically normal by assumption and application of the delta method). Then the sum converges to a normal distribution because of the independence of each \(\hat{\pi}_r(\boldsymbol{x})\) and \(\boldsymbol{\overline{X}}_r \) and Slutsky's theorem.

However, the calculations are sufficiently tedious that we defer the work to the proof of Lemma \ref{lem.int.asym.norm.cond.est.prob}.

\begin{lemma}\label{lem.int.asym.norm.cond.est.prob} Under the assumptions of Theorem \ref{te.asym.norm.thm.gen.cond}(b),
\[
\sqrt{NT}  (  \boldsymbol{\theta}_N^*)^\top  \left(  \boldsymbol{\hat{\psi}}_N^{(\text{C;CATT,} \pi)} - \boldsymbol{\psi}_N^{(\text{C;CATT,} \pi)} \right) \xrightarrow{d} \mathcal{N} \left(0, v_R^{(\text{C;CATT,} \pi)} \right)
\]
for
\[
 v_R^{(\text{C;CATT,} \pi)} :=  \sum_{r \in \mathcal{R}}   \frac{(\boldsymbol{\psi}_{r}^{\text{CATT; prob}})^\top  \Cov ( \boldsymbol{X}_i \mid W_i = r) \boldsymbol{\psi}_{r}^{\text{CATT; prob}}}{\mathbb{P}(W_i = r)} + T \cdot \boldsymbol{\psi}_\pi^\top  \nabla \boldsymbol{f}^\top  \boldsymbol{\Sigma}_\pi \nabla \boldsymbol{f} \boldsymbol{\psi}_\pi
 ,
\]
\(\boldsymbol{\psi}_{r}^{\text{CATT; prob}}\) is defined in \eqref{psi.catt.prob.def}, and \(\boldsymbol{\Sigma}_\pi \) was defined in the statement of Theorem \ref{te.asym.norm.thm.gen.cond}. Further, \(\hat{v}_R^{(\text{C;CATT,} \pi)}\), defined in \eqref{v_r_c_catt_pi_est_form}, is a consistent estimator for \( v_R^{(\text{C;CATT,} \pi)}\).
% and
%\[
%\sqrt{\frac{NT}{v_R^{(\text{C;CATT,} \pi)}}}  (  \boldsymbol{\theta}_N^*)^\top  \left(  \boldsymbol{\hat{\psi}}_N^{(\text{C;CATT,} \pi)} - \boldsymbol{\psi}_N^{(\text{C;CATT,} \pi)} \right) \xrightarrow{d} \mathcal{N} \left(0, 1  \right)
%.
%\]
\end{lemma}

\begin{proof} Provided in Appendix \ref{main.thm.lems}.
\end{proof}

Then from Lemma \ref{lem.int.asym.norm.cond.est.prob} and Theorem \ref{first.thm.fetwfe}(g) we have
\begin{align*}
&  \sqrt{ \frac{NT}{ \hat{v}_N^{(\text{C;CATT,} \pi)}   }}  \left(  (\boldsymbol{\hat{\psi}}_N^{(\text{C;CATT,} \pi)} )^\top  \boldsymbol{\hat{\theta}}^{(q)} - ( \boldsymbol{\psi}_N^{(\text{C;CATT,} \pi)})^\top \boldsymbol{\theta}_N^* \right) \xrightarrow{d}   \mathcal{N}(0, 1)
 \\ \iff \qquad  &  \sqrt{ \frac{NT}{ \hat{v}_N^{(\text{C;CATT,} \pi)}   }}    \sum_{r \in \mathcal{R}} \sum_{t=r}^T \psi_{rt} \big(   
f_r \left(  \hat{\pi}_{r_1}(\boldsymbol{x}) , \hat{\pi}_{r_2}(\boldsymbol{x}) , \ldots,  \hat{\pi}_{r_R}(\boldsymbol{x})  \right) 
  \hat{\tau}_{\text{CATT}}(r, t, \boldsymbol{x})
  \\ &  -  
   f_r \left(  \pi_{r_1}(\boldsymbol{x}) , \pi_{r_2}(\boldsymbol{x}) , \ldots,  \pi_{r_R} (\boldsymbol{x})  \right) 
   \tau_{\text{CATT}}(r, t, \boldsymbol{x})  \big)  \xrightarrow{d}   \mathcal{N}(0, 1)
.
\end{align*}
for
\begin{align}
\hat{v}_N^{(\text{C;CATT,} \pi)} :=  ~ &  \sigma^2 (\boldsymbol{\hat{\psi}}_N^{(\text{C;CATT,} \pi)})^\top \left( \widehat{\Cov} \left( (\boldsymbol{Z} \boldsymbol{D}_N^{-1}) _{(\cdot \cdot) \hat{\mathcal{S}}} \right) \right)^{-1} \boldsymbol{\hat{\psi}}_N^{(\text{C;CATT,} \pi)}  \nonumber
\\ & + \hat{v}_R^{(\text{C;CATT,} \pi)}\label{v.n.r.t.catt.const.est.pi}
.
\end{align}

\end{enumerate}

%\end{proof}

\section{Proof of Theorem \ref{first.thm.fetwfe}}\label{sec.prove.first.thm}

Before we prove the individual statements, we outline our general approach. Notice that the assumptions leading to equation (6.33) in \citet[Section 6.3]{wooldridge2021two} are satisfied, so we have for any \(i \in [N]\)
\begin{equation}\label{wooldridge.6.33.model.exp}
\E\left[ \boldsymbol{\tilde{y}}_{(i\cdot)} \mid W_i, \boldsymbol{X}_i \right]= \boldsymbol{\tilde{Z}}_{(i \cdot) \cdot} \boldsymbol{\beta}_N^* 
,
\end{equation}
where the subscript \((i\cdot)\) refers to the \(T\) observations corresponding to unit \(i\). (This verifies part (a) of Theorem \ref{first.thm.fetwfe}, though we also provide a proof in our notation below.) We propose estimating the bridge regression \eqref{opt.prob},
\[
\boldsymbol{\hat{\beta}}^{(q)} = \underset{ \boldsymbol{\beta} \in \mathbb{R}^{p_N}}{\arg \min} \left\{ \lVert \boldsymbol{y} -   \boldsymbol{Z} \boldsymbol{\beta} \rVert_2^2 + \lambda_N \lVert \boldsymbol{D}_N \boldsymbol{\beta} \rVert_q^q  \right\}
.
\]
We proved in Lemma \ref{lem.d.express} that the differences matrix \(\boldsymbol{D}_N\) is invertible. We can therefore solve \eqref{opt.prob} through the reparameterization \(\boldsymbol{\hat{\theta}}^{(q)} = \boldsymbol{D}_N \boldsymbol{\hat{\beta}}^{(q)}\), solving for
\[
\boldsymbol{\hat{\theta}}^{(q)} = \underset{ \boldsymbol{\theta} \in \mathbb{R}^{p_N}}{\arg \min} \left\{ \lVert \boldsymbol{y} -  \boldsymbol{Z} \boldsymbol{D}_N^{-1} \boldsymbol{\theta} \rVert_2^2 + \lambda_N \lVert \boldsymbol{\theta} \rVert_q^q  \right\}
\]
using standard bridge regression. Then \(\boldsymbol{\hat{\beta}}^{(q)} = \boldsymbol{D}_N^{-1}\boldsymbol{\hat{\theta}}^{(q)}\) is identical to the solution to \eqref{opt.prob} \citep[Section 3]{tibshirani2011solution}.

Having cast \eqref{opt.prob} as a standard bridge regression optimization problem, our results will follow from Theorem 1, Lemma 3, and Theorem 2(ii) from \citet{kock2013oracle} if we show the assumptions are met in our setting where
\begin{align*}
\boldsymbol{\tilde{Y}}_{iN} = ~ & \boldsymbol{\tilde{y}}_{(i\cdot)},
\\ \boldsymbol{\tilde{X}}_{iN} = ~ & \boldsymbol{\tilde{Z}}_{(i \cdot) \cdot}\boldsymbol{D}_N^{-1},
\\  \beta_0 = ~ & \boldsymbol{\theta}_N^*,
\\ \boldsymbol{c}_{iN} = ~ & c_i \boldsymbol{1}_T,
\\ \boldsymbol{\tilde{\epsilon}}_{iN} = ~ & \boldsymbol{u}_{(i\cdot)},
\\ \boldsymbol{Y}_{iN} = ~ &  \boldsymbol{y}_{(i\cdot)},
\\ \boldsymbol{X}_{iN} = ~ & \boldsymbol{Z}_{(i \cdot) \cdot}\boldsymbol{D}_N^{-1},
\\ \boldsymbol{X}_N = ~ & \boldsymbol{Z} \boldsymbol{D}_N^{-1}, 
\\ \boldsymbol{W}_N = ~ & (\boldsymbol{Z} \boldsymbol{D}_N^{-1})_{(\cdot \cdot)\mathcal{S}},
\\ \boldsymbol{\Sigma}_{1N} = ~ &
%  \frac{1}{NT}  (\boldsymbol{Z} \boldsymbol{D}_N^{-1})_{(\cdot \cdot)\mathcal{S}}^\top  (\boldsymbol{Z} \boldsymbol{D}_N^{-1})_{(\cdot \cdot)\mathcal{S}} = 
  \boldsymbol{\hat{\Sigma}} ((\boldsymbol{Z} \boldsymbol{D}_N^{-1}) _{(\cdot \cdot) \mathcal{S}}), 
\\ \boldsymbol{\epsilon}_{iN} = ~ & \sigma \boldsymbol{\Omega}^{-1/2} \boldsymbol{\epsilon}_{(i\cdot)},
\\ \gamma = ~ & q, \qquad \text{and}
\\ k_N = ~ & s_N
,
\end{align*}
where the left side of each equation uses the notation from \citet{kock2013oracle} and the right side uses the notation in our setting. We focus on the random effects setting of \citet{kock2013oracle} because we cannot perform the forward orthogonal deviations transform since \(\boldsymbol{\tilde{Z}}\) (and \(\boldsymbol{\tilde{Z}}\boldsymbol{D}_N^{-1}\)) contain (linear combinations of) time-invariant covariates and cohort indicators.

Now we are prepared to prove each statement.

\begin{enumerate}[(a)]

\item As mentioned earlier, under our assumptions correct specification is immediate from Equation (6.33) in \citet[Section 6.3]{wooldridge2021two}. We also provide a derivation for the correct specification here in our notation. Observe that for any \(i \in [N]\),
\begin{align}
& \E\left[ \tilde{y}_{(i 1)} \mid W_i, \boldsymbol{X}_i = \boldsymbol{x} \right] \nonumber
%%
%\\  = ~ &  \mathbbm{1}\{W_i = 0\} \E\left[ \tilde{y}_{(i1)} \mid W_i  = 0, \boldsymbol{X}_i = \boldsymbol{x} \right] + \sum_{r \in \mathcal{R}}  \mathbbm{1}\{W_i = r\} \E\left[ \tilde{y}_{(i1)} \mid W_i  = r, \boldsymbol{X}_i = \boldsymbol{x} \right] \nonumber
%
\\ = ~ &  \mathbbm{1}\{W_i = 0\} \E \left[ \tilde{y}_{(i1)}(0) \mid W_i  = 0, \boldsymbol{X}_i = \boldsymbol{x} \right] + \sum_{r \in \mathcal{R}}  \mathbbm{1}\{W_i = r\} \E\left[ \tilde{y}_{(i1)}(r) \mid W_i  = r, \boldsymbol{X}_i = \boldsymbol{x} \right] \nonumber
\\ \stackrel{(a)}{=} ~ &  \mathbbm{1}\{W_i = 0\} \E \left[ \tilde{y}_{(i1)}(0) \mid W_i  = 0, \boldsymbol{X}_i = \boldsymbol{x} \right] + \sum_{r \in \mathcal{R}}  \mathbbm{1}\{W_i = r\} \E\left[ \tilde{y}_{(i1)}(0) \mid W_i  = r, \boldsymbol{X}_i = \boldsymbol{x} \right] \nonumber
\\ \stackrel{(b)}{=} ~ &  \mathbbm{1}\{W_i = 0\}  \left( \eta^*  + \boldsymbol{x}^\top \boldsymbol{\kappa}^* \right) + \sum_{r \in \mathcal{R}}  \mathbbm{1}\{W_i = r\} \left(  \eta^* + \nu_r^* + \boldsymbol{x}^\top \left( \boldsymbol{\kappa}^*   + \boldsymbol{\zeta}_r^*\right)  \right) \nonumber
\\ =~ &     \eta^* +   \boldsymbol{x}^\top \boldsymbol{\kappa}^* + \sum_{r \in \mathcal{R}}  \mathbbm{1}\{W_i = r\} \left( \nu_r^* + \boldsymbol{x}^\top \boldsymbol{\zeta}_r^*\right)  \label{t.0.evs}
,
\end{align}
where in \((a)\) we used Assumption (CNAS) and in \((b)\) we used \eqref{t.1.untreated} and \eqref{t.1.treated} from (LINS). Next, for any \(t \in \{2, \ldots, T\}\) and \(i\) such that \(W_i = 0\),
\begin{align}
& \E\left[ \tilde{y}_{(i t)} \mid W_i = 0, \boldsymbol{X}_i \right]  \nonumber
\\  = ~ &  \E\left[ \tilde{y}_{(it)}(0)  - \tilde{y}_{(i1)}(0)  \mid W_i  = 0, \boldsymbol{X}_i = \boldsymbol{x} \right]  +   \E\left[  \tilde{y}_{(i1)}(0)  \mid W_i  = 0, \boldsymbol{X}_i = \boldsymbol{x} \right]  \nonumber
%%
%\\ \stackrel{(c)}{=} ~ &  \E\left[ \tilde{y}_{(it)}(0)  - \tilde{y}_{(i1)}(0)  \mid \boldsymbol{X}_i = \boldsymbol{x} \right]  +   \E\left[  \tilde{y}_{(i1)}(0)  \mid W_i  = 0, \boldsymbol{X}_i = \boldsymbol{x} \right]  \nonumber
%
\\ \stackrel{(c)}{=} ~ &  \eta^* +    \gamma_t^* + \boldsymbol{x}^\top  (\boldsymbol{\xi}_t^*  +   \boldsymbol{\kappa}^* ) ,  \nonumber % \label{any.t.ev}
\end{align}
where in \((c)\) we used \eqref{trend.params} from (LINS) and \eqref{t.0.evs}. Then for any \(i\) such that \(W_i  = r \in \mathcal{R} \setminus \{2\}\) and \(t \in \{2, \ldots, r - 1\}\),
\begin{align}
& \E\left[ \tilde{y}_{(i t)} \mid W_i = r, \boldsymbol{X}_i  = \boldsymbol{x} \right] \nonumber
\\ = ~ & \E\left[ \tilde{y}_{(i 1)}(r) \mid W_i = r, \boldsymbol{X}_i  = \boldsymbol{x} \right] + \E\left[ \tilde{y}_{(i t)}(r) - \tilde{y}_{(i 1)}(r)  \mid W_i = r, \boldsymbol{X}_i  = \boldsymbol{x}\right]  \nonumber
\\ \stackrel{(e)}{=} ~ & \E\left[ \tilde{y}_{(i 1)}(0) \mid W_i = r, \boldsymbol{X}_i  = \boldsymbol{x} \right] + \E\left[ \tilde{y}_{(i t)}(0) -  \tilde{y}_{(i 1)}(0) \mid  W_i = r, \boldsymbol{X}_i  = \boldsymbol{x}\right] \nonumber  
\\ \stackrel{(f)}{=} ~ & \E\left[ \tilde{y}_{(i 1)}(0) \mid W_i = r, \boldsymbol{X}_i  = \boldsymbol{x} \right] + \E\left[ \tilde{y}_{(i t)}(0) -  \tilde{y}_{(i 1)}(0) \mid  W_i = 0, \boldsymbol{X}_i  = \boldsymbol{x}\right] \nonumber 
\\ \stackrel{(g)}{=} ~ &   \eta^* +   \nu_r^* + \boldsymbol{x}^\top (\boldsymbol{\kappa}^*  + \boldsymbol{\zeta}_r^* +  \boldsymbol{\xi}_t^* )  + \gamma_t^* \nonumber
,
\end{align}
where in \((e)\) we used (CNAS), in \((f)\) we used (CCTSB), and in \((g)\) we used \eqref{t.0.evs} and \eqref{trend.params} from (LINS). Finally, for any \(i\) such that \(W_i = r \in \mathcal{R}\) and \(t \in \{r, \ldots, T\}\),
\begin{align}
& \E\left[ \tilde{y}_{(i t)} \mid W_i = r, \boldsymbol{X}_i  = \boldsymbol{x} \right]  \nonumber
\\ = ~ &  \E\left[ \tilde{y}_{(i 1)}(r) \mid W_i = r, \boldsymbol{X}_i  = \boldsymbol{x} \right]  + \E \left[  \tilde{y}_{(it)}(0) - \tilde{y}_{(i1)}(0) \mid  W_i = 0, \boldsymbol{X}_{i} = \boldsymbol{x} \right]  \nonumber
\\ & 
+ \E\left[ \tilde{y}_{(i t)}(r)  - \tilde{y}_{(i 1)}(r) \mid W_i = r, \boldsymbol{X}_i  = \boldsymbol{x}\right]  
 - \E \left[  \tilde{y}_{(it)}(0) - \tilde{y}_{(i1)}(0)  \mid  W_i = 0, \boldsymbol{X}_{i} = \boldsymbol{x} \right]  \nonumber
\\ \stackrel{(h)}{=} ~ &  \eta^* +   \nu_r^* + \boldsymbol{x}^\top (\boldsymbol{\kappa}^*  + \boldsymbol{\zeta}_r^* +  \boldsymbol{\xi}_t^* )  + \gamma_t^* + \tau_{rt}^* + \boldsymbol{\dot{x}}_r \boldsymbol{\rho}_{rt}^* \label{tau.decomp}
,
\end{align}
where in \((h)\) we used \eqref{t.0.evs} and \eqref{trend.params} and \eqref{treat.eff.def.covs} from (LINS). Therefore regression \eqref{wooldridge.6.33.model} is correctly specified and \eqref{wooldridge.6.33.model.exp} holds.

\item

First we will show that Assumptions (FE1) - (FE3), (RE4), (A1), (A3), (A5), and (A6) from \citet{kock2013oracle} are satisfied. Then we will show that the conclusion of Theorem 1 in \citet{kock2013oracle} about \(\boldsymbol{\theta}_N^*\) leads to our statement about \(\boldsymbol{\beta}_N^*\). We will use three lemmas:

\begin{lemma}\label{d.sing.val.lem}
The smallest and largest singular values of \(\boldsymbol{D}_N\) are bounded as follows: \( \sigma_{\text{max}} (\boldsymbol{D}_N)  \leq   3\) and \(\sigma_{\text{min}}\left( \boldsymbol{D}_N  \right)  \geq  \frac{1}{T\sqrt{2T}}\).
\end{lemma}

\begin{lemma}\label{power.lem}
Under Assumption (R1) it holds that
\[
\E \left[ ( \boldsymbol{\tilde{Z}}_{(i \cdot) \cdot} \boldsymbol{D}_N^{-1})_{(it)j}^4 \right]  \leq \kappa_4 \qquad \forall i \in [N], j \in [p_N], t \in [T]
\]
for some finite \(\kappa_4 > 0\). 

%Further, there exists a finite \(\kappa_2 > 0\) such that 
%\[
%\E \left[ (\boldsymbol{Z}_{(i \cdot) \cdot} \boldsymbol{D}_N^{-1})_{(it)j}^2 \right]  \leq \kappa_2 \qquad \forall i \in [N], j \in [p_N], t \in [T]
%.
%\]
\end{lemma}

\begin{lemma}\label{sig.min.lem}
Under the assumptions of Theorem \ref{first.thm.fetwfe}(b),
\begin{align*}
\rho_{1N} =  \lambda_{\text{min}} \left(  \boldsymbol{\hat{\Sigma}} \left(  \boldsymbol{Z}\boldsymbol{D}_N^{-1} \right)  \right)  \geq ~ & \frac{1}{9} e_{1N}   \qquad \text{and}
\\  
\rho_{2N} = \lambda_{\text{max}} \left(  \boldsymbol{\hat{\Sigma}} \left(  \boldsymbol{Z}\boldsymbol{D}_N^{-1} \right)  \right)  \leq ~ &   e_{2N}  \cdot 2 T^3
.
\end{align*}

\end{lemma}

The proofs of all three of these lemmas are provided in Appendix \ref{sec.main.lemmas}. Now we will verify one assumption at a time.

\begin{itemize}

\item In our notation, Assumption (FE1) from \citet{kock2013oracle} requires that \((\boldsymbol{\tilde{Z}}_{(i \cdot) \cdot} \boldsymbol{D}_N^{-1}, c_i, \boldsymbol{u}_{(i\cdot)})_{i=1}^N\) are iid. This is satisfied by our Assumption (F2) since \(\boldsymbol{\tilde{Z}}_{(i \cdot) \cdot}\) is a deterministic, invertible function of \((W_i, \boldsymbol{X}_i)\), and \(\boldsymbol{D}_N^{-1}\) is also deterministic and invertible conditional on \((W_i, \boldsymbol{X}_i)\).

\item Assumption (FE2) requires that \(\E [ (\boldsymbol{\tilde{Z}} \boldsymbol{D}_N^{-1})_{(it)j}^4 ]\) and \(\E [u_{(it)}^4]\) are finite for all \(i \in [N], j \in [d_N], t \in [T]\). We assumed the latter in Assumption (R1). We proved that the former is satisfied in Lemma \ref{power.lem}.

\item Next, Assumption (FE3) requires that
\begin{align*}
\E \left[ \boldsymbol{u}_{(i\cdot)} \mid  \boldsymbol{\tilde{Z}}_{(i \cdot) \cdot} \boldsymbol{D}_N^{-1} , c_i \right] = ~ & \boldsymbol{0} \qquad \text{and}
\\ \Var \left[ \boldsymbol{u}_{(i\cdot)}  \mid   \boldsymbol{\tilde{Z}}_{(i \cdot) \cdot} \boldsymbol{D}_N^{-1} , c_i\right] = ~ & \sigma^2 \boldsymbol{I}_{T}, \qquad i \in [N]
,
\end{align*}
and Assumption (RE4) requires
\begin{align*}
\E \left[ c_i \mid  \boldsymbol{\tilde{Z}}_{(i \cdot) \cdot} \boldsymbol{D}_N^{-1} \right]  = ~ & 0,
\\  \Var \left[ c_i \mid  \boldsymbol{\tilde{Z}}_{(i \cdot) \cdot} \boldsymbol{D}_N^{-1} \right]  = ~ & \sigma_c^2, 
\end{align*}
and that \(\sigma\) and \(\sigma_c^2\) are known and finite. Our Assumption (F1) is sufficient for this because conditioning on \(W_i, \boldsymbol{X}_i\) is the same as conditioning on \( \boldsymbol{\tilde{Z}}_{(i \cdot) \cdot} \boldsymbol{D}_N^{-1}\) by the following argument. Because \(\boldsymbol{D}_N^{-1}\) is invertible, the \(\sigma\)-algebra generated by \(\boldsymbol{\tilde{Z}}_{(i \cdot) \cdot} \boldsymbol{D}_N^{-1}\), \(\sigma(\boldsymbol{\tilde{Z}}_{(i \cdot) \cdot} \boldsymbol{D}_N^{-1})\), contains the same information as \(\sigma(\boldsymbol{\tilde{Z}}_{(i \cdot) \cdot})\). Similarly, for each \((W_i, \boldsymbol{X}_i)\) there is only one valid \(\boldsymbol{\tilde{Z}}_{(i \cdot) \cdot}\), and every valid \(\boldsymbol{\tilde{Z}}_{(i \cdot) \cdot}\) corresponds to a unique \((W_i, \boldsymbol{X}_i)\)---that is, the mapping from \((W_i, \boldsymbol{X}_i)\) to \(\boldsymbol{\tilde{Z}}_{(i \cdot) \cdot}\) is also invertible, so \(\sigma(W_i, \boldsymbol{X}_i) = \sigma(\boldsymbol{\tilde{Z}}_{(i \cdot) \cdot})\). 

\item Our Assumptions (R1), (R2), and (R3) exactly match assumptions (A1), (A3), and (A5), respectively, from \citet{kock2013oracle} (in our setting where we assume \(T\) is fixed). 

\item Finally, Assumption (A6) in \citet{kock2013oracle} requires that
\[
\frac{p_N + \lambda_N s_N}{N \rho_{1N}} \xrightarrow{a.s.} 0
,
\]
where \(\rho_{1N}\) is the smallest eigenvalue of the random matrix \( \boldsymbol{\hat{\Sigma}} \left( \boldsymbol{Z} \boldsymbol{D}_N^{-1} \right) \). Under Lemma \ref{sig.min.lem} the smallest eigenvalue of \(\boldsymbol{\hat{\Sigma}}(\boldsymbol{Z})\), \(e_{1N}\), is within a constant factor of \(\rho_{1N}\), so our Assumption (R2) suffices to satisfy this condition.

\end{itemize}
 
We have shown that the assumptions of Theorem 1 in \citet{kock2013oracle} are satisfied, so we have \(\lVert \boldsymbol{\hat{\theta}}^{(q)} - \boldsymbol{\theta}_N^* \rVert_2 = \mathcal{O}_{\mathbb{P}} \left( \min\{h_N, h'_N\} \right)\). Finally, part (2) is proven because \(\left\lVert\boldsymbol{\hat{\beta}}^{(q)} - \boldsymbol{\beta}_N^*  \right\rVert_2\) is within a constant factor of \(\lVert \boldsymbol{\hat{\theta}}^{(q)} - \boldsymbol{\theta}_N^* \rVert_2 \):
\begin{align*}
\lVert \boldsymbol{\hat{\beta}}^{(q)} - \boldsymbol{\beta}_N^* \rVert_2 = ~ & \left\lVert  \boldsymbol{D}_N^{-1} \left( \boldsymbol{\hat{\theta}}^{(q)} - \boldsymbol{\theta}_N^*  \right) \right\rVert_2
 \leq   \left\lVert  \boldsymbol{D}_N^{-1} \right \rVert_{\text{op}}   \left\lVert\boldsymbol{\hat{\theta}}^{(q)} - \boldsymbol{\theta}_N^*  \right\rVert_2
 \leq  T\sqrt{2T}  \left\lVert\boldsymbol{\hat{\theta}}^{(q)} - \boldsymbol{\theta}_N^*  \right\rVert_2
,
\end{align*}
where in the last step we applied Lemma \ref{d.sing.val.lem}.

\item In light of the results from part (a), part (c) follows from Theorem \ref{prop.2i} below, which extends Lemma 3 from \citet{kock2013oracle}, under our assumptions. In particular, we will show that our Assumptions (R4) and (R5) suffice for \citet{kock2013oracle}'s Assumptions (A4) and (A7), as well as the needed part of Assumption (A2) for Lemma 3, in \citet{kock2013oracle}. 

\begin{theorem}[Extension of Lemma 3 from \citealt{kock2013oracle}]\label{prop.2i}
Under the assumptions of Lemma 3 from \citet{kock2013oracle},
\[
\lim_{N \to \infty} \mathbb{P} \left( \hat{\mathcal{S}} = \mathcal{S} \right) = 1
.
\]

\end{theorem}

\begin{proof}
Provided in Appendix \ref{sec.main.lemmas}.
\end{proof}

\begin{remark}
It appears that the strategy for the proof of of Theorem \ref{prop.2i} could be used to analogously extend Lemma 2, and Theorem 2(i), in \citet{Huang2008}.
\end{remark}

Again we verify one assumption at a time.

\begin{itemize}

\item Although Lemma 3 from \citet{kock2013oracle} requires their Assumption (A2) in the statement of the theorem, examining the proof we see that a lower bound on the minimum eigenvalue is not needed---the only thing that is needed is an almost sure upper bound on the maximum eigenvalue of \(\boldsymbol{\hat{\Sigma}} \left( (\boldsymbol{Z} \boldsymbol{D}_N^{-1})_{(\cdot \cdot)\mathcal{S}} \right) \), the empirical Gram matrix of the columns of \(\boldsymbol{Z} \boldsymbol{D}_N^{-1}\) corresponding to the relevant features in \(\boldsymbol{\theta}_N^*\). (In the proof, note the use of \(\tau_2\) on p. 140, and note that \(\tau_1\) is not needed.) Our Assumption (R4) is sufficient for this, since we assume the maximum eigenvalue of the empirical Gram matrix for \(\boldsymbol{\tilde{Z}}\), \(e_{2N}\) is finite, which ensures that the maximum eigenvalue of the Gram matrix for the full \(\boldsymbol{Z} \boldsymbol{D}_N^{-1}\) is finite due to Lemma \ref{sig.min.lem}. This maximum eigenvalue upper-bounds the maximum eigenvalue of the submatrix \(\boldsymbol{\hat{\Sigma}} \left( (\boldsymbol{Z} \boldsymbol{D}_N^{-1})_{(\cdot \cdot)\mathcal{S}} \right) \).

%does not depend on Only the part of Assumption (A2) in \citet{kock2013oracle} that matches our Assumption (A2') is needed for the proof of Lemma 3---

\item Our Assumption (R5) matches Assumption (A4) in \citet{kock2013oracle} in our setting where \(T_N = T\) is fixed.

\item Finally, we show that Assumption (A7) in \citet{kock2013oracle} is satisfied under our Assumption (R4). We need to show that
\[
e_{1N} \sqrt{\frac{e_{2N}}{p_N}} = \mathcal{O}_p (1).
\]
But since \(e_{1N} \leq e_{2N}\), almost surely
\[
e_{1N} \sqrt{\frac{e_{2N}}{p_N}} \leq \sqrt{\frac{e_{2N}^3}{p_N}} \leq \sqrt{\frac{e_{\text{max}}^3}{p_N}} = \mathcal{O} \left( \frac{1}{\sqrt{p_N}} \right)
,
\]
so \(e_{1N} \sqrt{e_{2N}/p_N}= \mathcal{O}_p (1)\).

\end{itemize}

So the first part of the result is now immediate. It only remains to show that 
\[
 \lim_{N \to \infty} \mathbb{P} \left(  \sqrt{NT} \boldsymbol{\hat{\theta}}^{(q)}_{\mathcal{S}^c} = \boldsymbol{0} \right) = 1 .
 \]
 From part \((c)\), we have that \(\lim_{n \to \infty} \mathbb{P} \left( \boldsymbol{\hat{\theta}}_{\mathcal{S}^C}^{(q)} = \boldsymbol{0} \right) = 1 \), so using \(\{\sqrt{NT} \boldsymbol{\hat{\theta}}^{(q)}_{\mathcal{S}^c}  = 0\} = \{\boldsymbol{\hat{\theta}}^{(q)}_{\mathcal{S}^c}  = 0 \} \), 
\begin{align*}
%\lim_{N \to \infty} \mathbb{P} \left( \sqrt{N} \left| \boldsymbol{\hat{\theta}}^{(q)}_{\mathcal{S}^c} \right| < \epsilon \right)  \geq ~ &
 \lim_{N \to \infty}  \mathbb{P} \left( \sqrt{NT} \boldsymbol{\hat{\theta}}^{(q)}_{\mathcal{S}^c}  = 0 \right)  = \lim_{N \to \infty} \mathbb{P} \left( \boldsymbol{\hat{\theta}}^{(q)}_{\mathcal{S}^c}  = 0 \right)  = 1.
\end{align*}

\item

%For every \(\epsilon > 0\), there exists a finite \(M > 0\) and a finite \(N^* > 0\) such that
%\begin{align*}
%& \mathbb{P} \left(  \frac{\lVert \hat{\theta} - \theta \rVert_2}{h_N} > M \right) < \epsilon \qquad \forall N > N^*
%%
%\\ \iff \qquad & \mathbb{P} \left(  \frac{\lVert \hat{\theta} - \theta \rVert_2}{h_N} \leq M \right) \geq 1 - \epsilon \qquad \forall N > N^*
%%
%\\ \iff \qquad & \mathbb{P} \left(  \lVert \hat{\theta} - \theta \rVert_2  \leq M h_N \right) \geq 1 - \epsilon \qquad \forall N > N^*
%%
%\\ \iff \qquad & \mathbb{P} \left(  \lVert \hat{\theta} - \theta \rVert_2  \leq  \frac{M}{e_{1N}} \sqrt{\frac{p_N}{N}} \right) \geq 1 - \epsilon \qquad \forall N > N^*
%.
%\end{align*}

To prove parts (d) - (h), we will present an extension of Theorem 2(ii) from \citet{kock2013oracle} that yields our result when in applied in our setting. Then we will show that the needed assumptions are satisfied to conclude the proof. In particular, our extension (a) obtains the asymptotic convergence of linear combinations of the coefficients of both relevant and irrelevant features; (b) establishes the asymptotic convergence of a statistic formed using the estimated Gram matrix on the selected set of coefficients, not the population covariance matrix of the true active set; (c) is like (a) when \(\boldsymbol{\psi}_N\) is estimated rather than known and is independent of \(\hat{\beta}_N\); (d) is like (b) when \(\boldsymbol{\psi}_N\) is estimated rather than known; and (e) is like (d) when \(\hat{\beta}_N\) and \(\boldsymbol{\hat{\psi}}_N\) are estimated on the same data set, so they are dependent.

 \begin{theorem}[Extension of Theorem 2(ii) from \citealt{kock2013oracle}]\label{prop.ext.2} Suppose the assumptions of Theorem 2(ii) from \citet{kock2013oracle} are satisfied.

 \begin{enumerate}[(a)]
 
 \item (Asymptotic normality when coefficients equal to 0 are included.) Let \(\{\boldsymbol{\psi}_N\}_{N=1}^\infty\) be a sequence of real-valued vectors where \(\boldsymbol{\psi}_N \in \mathbb{R}^{p_N}\) has the structure \((\boldsymbol{\alpha}^\top, \boldsymbol{b}_N^\top)^\top\) where \(\boldsymbol{\alpha} \in \mathbb{R}^k\) is a fixed, finite, nonzero vector and \(\{\boldsymbol{b}_N\}_{N=1}^\infty\) is any sequence of constants where each \(\boldsymbol{b}_N \in \mathbb{R}^{p_N - k}\) contains all finite entries. Then
 \[
  \sqrt{ NT}  \boldsymbol{\psi}_N^\top ( \hat{\beta}_N - \beta_0) \xrightarrow{d} \mathcal{N}\left( 0, \sigma^2 \boldsymbol{\alpha}^\top \left( \lim_{N \to \infty}  \E \left[ \Sigma_{1N} \right] \right)^{-1} \boldsymbol{\alpha} \right)
 .
 \]
 
 \item 
 
 (Asymptotic normality when coefficients equal to 0 are included and the covariance matrix and selected set are estimated.) Define a function \(\boldsymbol{\alpha}\) that maps any \(\mathcal{A} \in \mathcal{P}(\mathbb{N})\) to a fixed vector \(\boldsymbol{\alpha}(\mathcal{A}) \in (\mathbb{R} \setminus \{0\})^{|\mathcal{A}|}\). Define the sequence of vectors \( \{ \boldsymbol{\psi}_N(\mathcal{A}_N) \}_{N=1}^\infty \) to have components in \(\mathcal{A}_N\) equal to \(\boldsymbol{\alpha}(\mathcal{A}_N)\) and components \(\{\boldsymbol{b}_N (\mathcal{A}_N)\}_{N=1}^\infty \) in the components corresponding to \([p_N] \setminus \mathcal{A}_N\) that are all finite for all \(N\) (but are otherwise arbitrary). Define the sequence of random variables \(U_b(\mathcal{A}) \) to equal
\[
%U_b(\mathcal{A}) :=  \begin{cases}
\frac{1}{\sigma}  
\sqrt{ \frac{NT}{\boldsymbol{\alpha}(\mathcal{A})^\top \left( \boldsymbol{\hat{\Sigma}}( \boldsymbol{X}_{(\cdot \cdot)\mathcal{A}} ) \right)^{-1} \boldsymbol{\alpha}(\mathcal{A})}}  \cdot  \left( \boldsymbol{\psi}_N(\mathcal{A})^\top  \left( \hat{\beta} -  \beta_0 \right) \right), 
%& 
%\\ 0, & \text{otherwise,}
%\end{cases}
\]
if \(\mathcal{A} \neq \emptyset\) and \( \boldsymbol{\hat{\Sigma}}( \boldsymbol{X}_{(\cdot \cdot) \mathcal{A}})\) is invertible and 0 otherwise, where \(\hat{\beta}\) is the bridge estimator of \citet[Section 3]{kock2013oracle}. Let \(\hat{\mathcal{S}}_N := \{j : \hat{\beta}_{Nj} \neq 0\}\) be the estimated selected set on a data set of size \(N\), and suppose \(\beta_0 \neq \boldsymbol{0} \). Then if \(\boldsymbol{\alpha}(\mathcal{S}) \neq \boldsymbol{0}\), the sequence \(\{U_b(\hat{\mathcal{S}}_N)\}\) converges in distribution to a standard Gaussian random variable.
 
 \item 
 
(Asymptotic normality when coefficients equal to 0 are included and the weights are estimated.) For any \(\mathcal{A} \in \mathcal{P}(\mathbb{N})\), suppose that \(\boldsymbol{\hat{\alpha}}_N(\mathcal{A}_N)\) is a random vector with the property that for a given data set, \(\mathcal{A} = \mathcal{A}'\) implies \(\boldsymbol{\hat{\alpha}}_N(\mathcal{A}) = \boldsymbol{\hat{\alpha}}_N(\mathcal{A}') \). Define the sequence of random vectors \( \{ \boldsymbol{\hat{\psi}}_N(\mathcal{A}_N) \}_{N=1}^\infty \) to have components in \(\mathcal{A}_N\) equal to \(\boldsymbol{\hat{\alpha}}_N(\mathcal{A}_N)\) and components \(\{\boldsymbol{\hat{b}}_N (\mathcal{A}_N)\}_{N=1}^\infty \) in the components corresponding to \([p_N] \setminus \mathcal{A}_N\) that are all almost surely finite for all \(N\) (but are otherwise arbitrary). Assume that for any fixed \(\beta_{0} \in \mathbb{R}^{p_N}\) with the \(k\) entries in the \(\mathcal{S}\) positions finite and the remaining entries all equal to 0,
\[
\sqrt{NT}   \beta_0^\top \left(  \boldsymbol{\hat{\psi}}_N(\mathcal{S})  -   \boldsymbol{\psi}_N(\mathcal{S})\right) \xrightarrow{d} \mathcal{N}(0, v_\psi(\beta_0, \mathcal{S}))
.
\]
Define the sequence of random variables \(U_c(\mathcal{A}) \) to equal
\[
\sqrt{ NT}  \left( \boldsymbol{\hat{\psi}}_N(\mathcal{A})^\top  \hat{\beta}_N - \boldsymbol{\psi}_N(\mathcal{A}) \beta_0 \right)
,
\]
where \(\hat{\beta}_N\) is the bridge estimator of \citet[Section 3]{kock2013oracle}. Assume \(\boldsymbol{\hat{\alpha}}_N(\mathcal{A}_N)\) is independent of \(\hat{\beta}_N\). Let \(\hat{\mathcal{S}}_N := \{j : \hat{\beta}_{Nj} \neq 0\}\) be the estimated selected set on a data set of size \(N\), and suppose \(\beta_0 \neq \boldsymbol{0} \). Then 
\[
U_c(\hat{\mathcal{S}}_N) \xrightarrow{d} \mathcal{N} \left( 0 ,\sigma^2 \boldsymbol{\alpha}^\top \left( \lim_{N \to \infty}  \E \left[ \Sigma_{1N} \right] \right)^{-1} \boldsymbol{\alpha}  +  v_\psi(\beta_0, \mathcal{S}) \right)
.
\]

 \item 
 
(Asymptotic normality when coefficients equal to 0 are included and the weights, covariance matrix, and selected set are estimated.)  Define \(\boldsymbol{\hat{\alpha}}_N(\mathcal{A}_N)\) and the sequence of random vectors \( \{ \boldsymbol{\hat{\psi}}_N(\mathcal{A}_N) \}_{N=1}^\infty \) as in part \((c)\), and assume they have the same properties. Define the sequence of random variables \(U_d(\mathcal{A}) \) to equal
\[
\sqrt{ \frac{NT}{ \hat{v}_N (\mathcal{A}) }}  \left( \boldsymbol{\hat{\psi}}_N(\mathcal{A})^\top  \hat{\beta}_N - \boldsymbol{\psi}_N(\mathcal{A})^\top \beta_0 \right)
\]
if \(\mathcal{A} \neq \emptyset , \boldsymbol{\hat{\Sigma}}( \boldsymbol{X}_{(\cdot \cdot) \mathcal{A}}) \text{ is invertible, and }  \boldsymbol{\hat{\alpha}}_N(\mathcal{A}) \neq \boldsymbol{0}\); and equals 0 otherwise; where \(\hat{\beta}_N\) is the bridge estimator of \citet[Section 3]{kock2013oracle},
\begin{equation}\label{var.est.kock}
\hat{v}_N (\mathcal{A}) :=  \sigma^2 \boldsymbol{\hat{\alpha}}(\mathcal{A})^\top \left( \boldsymbol{\hat{\Sigma}}( \boldsymbol{X}_{(\cdot \cdot)\mathcal{A}} ) \right)^{-1} \boldsymbol{\hat{\alpha}}(\mathcal{A}) + \hat{v}_\psi(\hat{\beta}_N, \mathcal{A})
,
\end{equation}
 and \(\hat{v}_\psi(\hat{\beta}_N, \mathcal{A})\) is an estimator of \(v_\psi(\beta_0, \mathcal{A})\) that satisfies \(\hat{v}_\psi(\hat{\beta}_N, \mathcal{S}) \xrightarrow{p} v_\psi(\beta_0, \mathcal{S})\). Assume \(\boldsymbol{\hat{\alpha}}_N(\mathcal{A}_N)\) is independent of \(\hat{\beta}_N\), and suppose \(\beta_0 \neq \boldsymbol{0} \). Then \(U_d(\hat{\mathcal{S}}_N)\) converges in distribution to a standard Gaussian random variable.

  \item 
  
 (Asymptotic subgaussianity when coefficients equal to 0 are included and the weights, covariance matrix, and selected set are estimated.) Let the sequence of random variables \(U_e(\mathcal{A})\) be defined the same was as \(U_d(\mathcal{A})\), except that we assume that \(\boldsymbol{\hat{\alpha}}_N(\mathcal{A}_N)\) is estimated on the same data as \(\hat{\beta}_N\), so they are not independent, and we use a different variance estimator:
  \[
U_e(\mathcal{A}) := \sqrt{ \frac{NT}{ \hat{v}_N^{(\text{cons})} (\mathcal{A}) }}  \left( \boldsymbol{\hat{\psi}}_N(\mathcal{A})^\top  \hat{\beta}_N - \boldsymbol{\psi}_N(\mathcal{A})^\top \beta_0 \right)
\]
(unless \(\hat{v}_N^{(\text{cons})} (\mathcal{A}) \) is not well-defined or equals 0, in which case \(U_e(\mathcal{A}) = 0\)), where
\begin{align}
\hat{v}_N^{(\text{cons})} (\mathcal{A}) :=  ~ & \sigma^2  \boldsymbol{\hat{\alpha}}(\mathcal{A})^\top \left(\boldsymbol{\hat{\Sigma}}( \boldsymbol{X}_{(\cdot \cdot)\mathcal{A}} )  \right)^{-1} \boldsymbol{\hat{\alpha}}(\mathcal{A})   + \hat{v}_\psi(\hat{\beta}_N, \mathcal{A})  \nonumber
\\ & + 2 \sqrt{\sigma^2  \boldsymbol{\hat{\alpha}}(\mathcal{A})^\top \left(\boldsymbol{\hat{\Sigma}}( \boldsymbol{X}_{(\cdot \cdot)\mathcal{A}} )  \right)^{-1} \boldsymbol{\hat{\alpha}}(\mathcal{A}) \cdot  \hat{v}_\psi(\hat{\beta}_N, \mathcal{A})}
 \label{var.upper.bound.subgaus}
\end{align}
is a conservative variance estimator. Assume again that \(\hat{v}_\psi(\hat{\beta}_N, \mathcal{S}) \xrightarrow{p} v_\psi(\beta_0, \mathcal{S})\) and \(\beta_0 \neq \boldsymbol{0} \). Then \(U_e(\hat{\mathcal{S}}_N)\) converges in distribution to a mean-zero subgaussian random variable with variance at most 1.

 \end{enumerate}
 
 \end{theorem}
 
 \begin{proof}
Provided in Appendix \ref{sec.main.lemmas}.
\end{proof}
 
 \begin{remark}
It appears that the strategies used to prove Theorem \ref{prop.ext.2} could be used to analogously extend Theorem 2(ii) in \citet{Huang2008}.
\end{remark}

It only remains to show that the assumptions of Theorem \ref{prop.ext.2} are satisfied in our setting. First, we note that \citet{kock2013oracle} states on p. 123 that a fixed \(T\) is enough to satisfy the last two assumptions of Theorem 2(ii). We have already shown that Assumptions (FE1), (FE2), (FE3), (RE4), (A1), (A3), (A4), (A5), (A6), and (A7) in \citet{kock2013oracle} are satisfied. Due to Theorem 3(i) from \citet{kock2013oracle}, our fixed \(T\) setting ensures that the uniform integrability assumption of Theorem 2(ii) from \citet{kock2013oracle} is satisfied. We have assumed fixed sparsity in Assumption S(\(s\)), and on p. 122 \citet{kock2013oracle} explains why our fixed \(T\) setting satisfies all of the remaining assumptions except for (A2).

We conclude by showing that the conditions implied by Assumption (A2) in \citet{kock2013oracle} are satisfied. \citet{kock2013oracle} makes an assumption that the minimum eigenvalue of the empirical Gram matrix of the columns of \(\boldsymbol{Z} \boldsymbol{D}_N^{-1}\) corresponding to the nonzero features in \(\boldsymbol{\theta}_N^*\), 
\[
\frac{1}{NT} \left( (\boldsymbol{Z} \boldsymbol{D}_N^{-1})_{(\cdot \cdot) \mathcal{S}}\right)^\top (\boldsymbol{Z} \boldsymbol{D}_N^{-1})_{(\cdot \cdot) \mathcal{S}},
\]
 is bounded away from 0 almost surely. Call this minimum eigenvalue \(e_{1\mathcal{S}N}\). This assumption is used in two places, both in the proof of Theorem 2(ii).

\begin{enumerate}[(a)]

\item  On p. 142, this assumption is used to establish that
\begin{align}
\frac{q}{2}  \frac{1}{e_{1\mathcal{S}N}} \left(\frac{b_0}{2} \right)^{q - 1} \lambda_N \sqrt{\frac{s_N}{NT}} \xrightarrow{p} ~ & 0 \nonumber
\\ \iff \frac{1}{e_{1\mathcal{S}N}}  \frac{\lambda_N}{\sqrt{N}}  \xrightarrow{p} ~ & 0 \nonumber
\\ \iff \frac{1}{e_{1\mathcal{S}N}}  \in ~ & o_{\mathbb{P}} \left( \frac{\sqrt{N}}{\lambda_N} \right)  \nonumber
,
\end{align}
where the first equivalence follows because the removed quantities are constant under our assumptions (specifically, that \(T\) and \(s_N\) are fixed). Note that our assumption \eqref{new.assum.2} ensures that this condition holds for the full empirical Gram matrix due to Lemma \ref{sig.min.lem}. The minimum eigenvalue of the full empirical Gram matrix lower bounds \(e_{1\mathcal{S}N}\), so \eqref{new.assum.2} is sufficient for this condition to hold.

\item The assumption is used again on p. 144 of \citet{kock2013oracle}, where it is sufficient for the minimum eigenvalue of the (population) matrix 
\begin{equation}\label{smaller.pop.eq}
\E \left[  \boldsymbol{\hat{\Sigma}} \left( \left(  \boldsymbol{Z} \boldsymbol{D}_N^{-1} \right)_{(\cdot \cdot) \mathcal{S}} \right) \right] 
\end{equation}
to be bounded away from 0. This is a submatrix of \eqref{min.eigen.eq}, so the minimum eigenvalue of \eqref{min.eigen.eq} is a lower bound for the minimum eigenvalue of \(\E \left[  \boldsymbol{\hat{\Sigma}} \left( \left(  \boldsymbol{Z} \boldsymbol{D}_N^{-1} \right)_{(\cdot \cdot) \mathcal{S}} \right) \right] \). Therefore our assumption in (R6) that the minimum eigenvalue of \eqref{min.eigen.eq} is bounded away from 0 is enough to satisfy this assumption.

\end{enumerate}

%In particular, we expect that since the sparsity \(s_N = s\) is fixed in this context, \(e_{1\mathcal{S}N}\) should be nonvanishing as \(N\to \infty\), and in fact is bounded away from 0 with high probability as \(N \to \infty\) under mild regularity conditions on the distribution of the covariates \(\boldsymbol{X}\). 

%
% can be replaced with an assumption that
%\[
%\frac{1}{e_{1\mathcal{S}N}}  \frac{\lambda_N }{\sqrt{N}}  \xrightarrow{p} 0 
%.
%\]
%(See p. 28 of \citealt{kock2013oracle}, in the proof of Theorem 2.) It seems reasonable to make an assumption like this:
%
%
% If we are comfortable assuming that \(e_{1\mathcal{S}N}\) is bounded away from 0 with high probability, we can interpret this as an assumption that \(\lambda_N\) decays ``quickly enough." The only assumption this could conflict with is our assumption that \(\lambda_N\) does not decay ``too quickly," Assumption (A4):
%\[
% \lambda_N \frac{e_{1N}^{2- q}}{ \sqrt{N^q p_N^{2 - q }}} \xrightarrow{a.s.} \infty
% .
%\]
%We need this assumption in the context of Theorem 2, where we are already assuming \(0 < q < q\). Also, we will later need to apply this assumption in the setting of Corollary 1, where we will assume \(s_N = s\) is fixed. With those assumptions, (A4) simplifies to
%\[
% \lambda_N \frac{e_{1N}^{2- q}}{ \sqrt{N^q p_N^{2 - q }}} \xrightarrow{a.s.} \infty
% .
%\]

So our Assumption (R6), in combination with all of the previous assumptions, is enough to apply Theorem \ref{prop.ext.2}, and the proof is complete.

\end{enumerate}

\section{Proofs of Results Stated in the Appendix}\label{app.extra.proofs}

In Appendix \ref{proofs.par.trend.ciuu.app} we prove results that were stated but not proven in Appendix \ref{par.trend.ciuu.app}. In Appendix \ref{main.thm.lems} we prove supporting results for our main theorems.

\subsection{Proofs of Results from Appendix \ref{par.trend.ciuu.app}}\label{proofs.par.trend.ciuu.app}

\begin{proof}[Proof of Theorem \ref{unconf.ccts.cts.thm}]

First we provide a counterexample to show that (CCTS) does not imply (CTS). This counterexample is inspired by discussions in Sections 2 and 3 of \citet{callaway2023treatment} that touch on closely related issues. For simplicity, let \(T = 2\) and \(\mathcal{R} = \{2\}\). Assume that treatment status is not independent of the covariates, so \( \E[\boldsymbol{X}_i \mid W_i = 2] \neq  \E[\boldsymbol{X}_i \mid W_i = 0]\). Consider a setting where \eqref{trend.params} from Assumption (LINS) holds with \(\boldsymbol{\xi}_2^* \neq \boldsymbol{0}\): 
\[
 \E \left[ \tilde{y}_{(i 2)}(0) - \tilde{y}_{(i1)}(0) \mid W_i = 0, \boldsymbol{X}_{i}   \right] =  \gamma_2^* + \boldsymbol{X}_i^\top \boldsymbol{\xi}_2^* \qquad a.s.
 ,
 \]
 and the same condition holds for the untreated potential outcomes of the treated units:
\[
 \E \left[ \tilde{y}_{(i 2)}(0) - \tilde{y}_{(i1)}(0) \mid W_i = 2, \boldsymbol{X}_{i} \right] =  \gamma_2^* + \boldsymbol{X}_i^\top \boldsymbol{\xi}_2^* \qquad a.s.
 \]
(CCTS) holds since \(\E [ \tilde{y}_{(i2)} (0) - \tilde{y}_{(i1)}(0) \mid  W_i = 2, \boldsymbol{X}_i ]  = \E [ \tilde{y}_{(i2)} (0) - \tilde{y}_{(i1)}(0) \mid  W_i = 0, \boldsymbol{X}_i ]\) almost surely. But for \(r \in \{0, 2\}\),
\begin{align*}
  \E [ \tilde{y}_{(i2)} (0) - \tilde{y}_{(i1)}(0) \mid  W_i = r  ]  
  = ~ &   \E \left[ \E \left[ \tilde{y}_{(i 2)}(0) - \tilde{y}_{(i1)}(0) \mid W_i = r, \boldsymbol{X}_{i} \right]  \mid W_i = r \right]
\\  = ~ &    \gamma_2^* + \E[\boldsymbol{X}_i \mid W_i = r]^\top \boldsymbol{\xi}_2^*
,
\end{align*}
so \(  \E [ \tilde{y}_{(i2)} (0) - \tilde{y}_{(i1)}(0) \mid  W_i = 2  ]   \neq   \E [ \tilde{y}_{(i2)} (0) - \tilde{y}_{(i1)}(0) \mid  W_i = 0  ]  \) since \( \E[\boldsymbol{X}_i \mid W_i = 2] \neq  \E[\boldsymbol{X}_i \mid W_i = 0]\).

Next we show an example where (CTS) holds but (CCTS) does not. Let \(\mathcal{R} = \{2, \ldots, T\}\), and suppose \(\boldsymbol{X}_i\) has a discrete distribution, only taking on values \(\boldsymbol{x}_1\) and \(\boldsymbol{x}_2\). To simplify notation, define \(G_{it} :=  \tilde{y}_{(it)} (0) - \tilde{y}_{(i1)}(0) \) for all \(t \in \{2, \ldots, T\}\). Then (CTS) is equivalent to
\begin{align}
\E [ G_{it} \mid  W_i  ]  =  ~ & \E [ G_{it}  ] ,  \qquad  i \in [N], t \in \{2, \ldots, T\}  \nonumber
\\ \iff \qquad \E [ G_{it} \mid  W_i  = 0 ]  =  ~ & \E [ G_{it}   \mid  W_i  = r] ,  \qquad r \in \mathcal{R}, i \in [N], t \in \{2, \ldots, T\} \label{cts.equiv}
\end{align}
and (CCTS) is equivalent to, for all \(i \in [N]\) and \( t \in \{2, \ldots, T\}\),
\begin{align}
\E [  G_{it} \mid  W_i , \boldsymbol{X}_{i} ]  = ~ &  \E [ G_{it} \mid  \boldsymbol{X}_i ]  \nonumber
\\ \iff \qquad \E [  G_{it} \mid  W_i = 0 , \boldsymbol{X}_{i} = \boldsymbol{x}_j ]  = ~ &  \E [ G_{it} \mid  W_i = r ,  \boldsymbol{X}_{i} = \boldsymbol{x}_j ] ,  \qquad    j \in \{1, 2\}, r \in \mathcal{R}  \label{ccts.equiv}
.
\end{align}
Our goal is to construct an example where \eqref{cts.equiv} holds and \eqref{ccts.equiv} does not hold. Suppose that for some \(a > 0\), for all \(i \in [N]\), \(t \in \{2, \ldots, T\}\)
\begin{align*}
\E [  G_{it} \mid  W_i = 0, \boldsymbol{X}_{i}  = \boldsymbol{x}_1 ]  = ~ &  -a,
\\ \E [  G_{it} \mid  W_i = r, \boldsymbol{X}_{i}  = \boldsymbol{x}_1 ]  = ~ &  -ra, \qquad r \in \mathcal{R},
\\ \E [  G_{it} \mid  W_i = 0, \boldsymbol{X}_{i}  = \boldsymbol{x}_2 ]  = ~ &  a, \qquad \text{and}
\\ \E [  G_{it} \mid  W_i = r, \boldsymbol{X}_{i}  = \boldsymbol{x}_2 ]  = ~ &  ra, \qquad r \in \mathcal{R}
.
\end{align*}
By inspection, one can see that \eqref{ccts.equiv} does not hold. We are done if we show that \eqref{cts.equiv} does. Suppose \(W_i \) and \(\boldsymbol{X}_i\) are independent: let \( \mathbb{P}(  W_i = w, \boldsymbol{X}_{i}  = \boldsymbol{x}_j )  = 1/(2[R+1]) = 1/(2T)\) for \(w \in \{0\} \cup \mathcal{R}\) and \(j \in \{1, 2\}\). We have
 \begin{align*}
&  \E [ G_{it} \mid  W_i = 0  ] 
\\ = ~ &    \frac{1}{\mathbb{P}(W_{i} = 0)} \big( \E [  G_{it} \mid  W_i = 0, \boldsymbol{X}_{i}  = \boldsymbol{x}_1 ]   \mathbb{P}(  W_i = 0, \boldsymbol{X}_{i}  = \boldsymbol{x}_1 )
\\ &  + \E [  G_{it} \mid  W_i = 0, \boldsymbol{X}_{i}  = \boldsymbol{x}_2 ]   \mathbb{P}(  W_i = 0, \boldsymbol{X}_{i}  = \boldsymbol{x}_2 ) \big)
\\ = ~ &    T \left(  -a  \cdot \frac{1}{2T}
  + a \cdot \frac{1}{2T} \right)
\\ = ~ & 0
 \end{align*}
 and for any \(r \in \mathcal{R}\),
 \begin{align*}
&  \E [ G_{it} \mid  W_i = r  ] 
\\ = ~ &    \frac{1}{\mathbb{P}(W_{i} = r)} \big( \E [  G_{it} \mid  W_i = r, \boldsymbol{X}_{i}  = \boldsymbol{x}_1 ]   \mathbb{P}(  W_i = r, \boldsymbol{X}_{i}  = \boldsymbol{x}_1 )
\\ &  + \E [  G_{it} \mid  W_i = r, \boldsymbol{X}_{i}  = \boldsymbol{x}_2 ]   \mathbb{P}(  W_i = r, \boldsymbol{X}_{i}  = \boldsymbol{x}_2 ) \big)
\\ = ~ &    T \left(  -ra  \cdot \frac{1}{2T}
  + ra \cdot \frac{1}{2T} \right)
\\ = ~ & 0
.
 \end{align*}

\end{proof}

\begin{proof}[Proof of Theorem \ref{te.interp.prop}(b) and (c)]

For any \(r \in \mathcal{R}\) and \(t \in \{r, \ldots, T\}\),
\begin{align}
& \tau_{rt}^*  \nonumber
\\ = ~ & \E[ \tau_{rt}^* +   \left(\boldsymbol{X}_i -  \E \left[ \boldsymbol{X}_{i} \mid  W_i = r \right] \right)^\top  \boldsymbol{\rho}_{rt}^* \mid W_i = r]  \nonumber
\\  \stackrel{(a)}{=} ~ &  \E \left[   \E\left[ \tilde{y}_{(i t)}(r)  - \tilde{y}_{(i 1)}(r) \mid W_i = r, \boldsymbol{X}_i  \right]   \mid W_i = r \right]
 - \E \left[ \E \left[  \tilde{y}_{(it)}(0) - \tilde{y}_{(i1)}(0)  \mid  W_i = 0, \boldsymbol{X}_{i}  \right] \mid W_i = r \right] \nonumber
\\ \stackrel{(b)}{=} ~ &  \E \left[ \E [ \tilde{y}_{(i t)}(r) \mid W_i = r, \boldsymbol{X}_{i} ] \mid W_i = r \right]  -   \E \bigg[ \bigg(  \E[\tilde{y}_{(i1)}(0) \mid W_i = r,   \boldsymbol{X}_{i}  ]   \nonumber
\\ &  +  \E [  \tilde{y}_{(i t)}(0) - \tilde{y}_{(i1)}(0) \mid W_i = 0, \boldsymbol{X}_{i}  ] \bigg) \mid W_i = r \bigg] \label{part.c.start}
\\ \stackrel{(c)}{=} ~ &  \E \left[ \E [ \tilde{y}_{(i t)}(r) \mid W_i = r, \boldsymbol{X}_{i} ] \mid W_i = r \right]  -   \E \bigg[ \bigg(  \E[\tilde{y}_{(i1)}(0) \mid W_i = r,   \boldsymbol{X}_{i}  ]   \nonumber
\\ &  +  \E [  \tilde{y}_{(i t)}(0) - \tilde{y}_{(i1)}(0) \mid W_i = 0  ] \bigg) \mid W_i = r \bigg] \nonumber
%%
%\\ \stackrel{(c)}{=} ~ &  \E \left[ \E [ \tilde{y}_{(i t)}(r) \mid W_i = r, \boldsymbol{X}_{i}  ] \mid W_i = r \right]  -   \E \bigg[ \bigg(  \E[\tilde{y}_{(i1)}(0) \mid W_i = r,   \boldsymbol{X}_{i}  ]   \nonumber
%\\ &  + \gamma_t^* + \boldsymbol{X}_i^\top \boldsymbol{\xi}_t^* \bigg) \mid W_i = r \bigg] \nonumber
%%
%\\ = ~ &  \E \left[  \tilde{y}_{(i t)}(r) \mid W_i = r \right]  - \left(  \E[\tilde{y}_{(i1)}(0) \mid W_i = r ]   \nonumber
%  + \gamma_t^* + \E[\boldsymbol{X}_i \mid W_i = r] ^\top \boldsymbol{\xi}_t^* \right) \nonumber
%%
%\\ \stackrel{(d)}{=} ~ &  \E \left[  \tilde{y}_{(i t)}(r) \mid W_i = r \right]  - \Big(  \E[\tilde{y}_{(i1)}(0) \mid W_i = r ]   \nonumber
%\\ &   +  \E \left[ \tilde{y}_{(i t)}(0) - \tilde{y}_{(i1)}(0) \mid W_i = 0, \boldsymbol{X}_{i} =  \E[\boldsymbol{X}_i \mid W_i = r]  \right]   \Big) \nonumber
%%
\\ = ~ &  \E \left[  \tilde{y}_{(i t)}(r) \mid W_i = r \right]  - \Big(  \E[\tilde{y}_{(i1)}(0) \mid W_i = r ]   \nonumber
  +  \E \left[ \tilde{y}_{(i t)}(0) - \tilde{y}_{(i1)}(0) \mid W_i = 0   \right]   \Big) \nonumber
\\ \stackrel{(d)}{=} ~ &  \E \left[  \tilde{y}_{(i t)}(r)  \mid W_i = r \right]  -   \E \left[  \tilde{y}_{(i t)}(0)  \mid W_i = r \right]  \nonumber
\\ = ~ & \tau_{\text{ATT}} (r,  t) 
\nonumber
 ,
\end{align}
where in \((a)\) we used \eqref{treat.eff.def.covs} from (LINS), in \((b)\) we used (CNAS), in \((c)\) we used (CIUN), and in \((d)\) we used (CTSA). This proves part (b).

%To prove part (c), we will prove the contrapositive: under (LINS), (CNAS), and (CTSA), if \( \tau_{rt}^*  = \tau_{\text{ATT}}(r, t) \) for all \(t \in \{2, \ldots, T\}\) then either \(\boldsymbol{X}_i\) is mean-independent of \(W_i\) or \(\boldsymbol{\xi}_t^* = \boldsymbol{0}\) for all \(t \in \{2, \ldots, T\}\). 
Next we prove part (c). Observe that under \eqref{trend.params} and  \eqref{treat.eff.def.covs} from (LINS), for any \(r \in \mathcal{R}\) and \(t \in \{r, \ldots, T\}\), starting from \eqref{part.c.start} we have
\begin{align}
& \tau_{rt}^*  \nonumber
\\ = ~ &  \E \left[ \E [ \tilde{y}_{(i t)}(r) \mid W_i = r, \boldsymbol{X}_{i} ] \mid W_i = r \right]  -   \E \bigg[ \bigg(  \E[\tilde{y}_{(i1)}(0) \mid W_i = r,   \boldsymbol{X}_{i}  ]   \nonumber
\\ &  +  \E [  \tilde{y}_{(i t)}(0) - \tilde{y}_{(i1)}(0) \mid W_i = 0, \boldsymbol{X}_{i}  ] \bigg) \mid W_i = r \bigg] \nonumber
 %
%\\ \stackrel{(c)}{=} ~ &  \E \left[ \E [ \tilde{y}_{(i t)}(r) \mid W_i = r, \boldsymbol{X}_{i} ] \mid W_i = r \right]  -   \E \bigg[ \bigg(  \E[\tilde{y}_{(i1)}(0) \mid W_i = r,   \boldsymbol{X}_{i}  ]   \nonumber
%\\ &  +  \E [  \tilde{y}_{(i t)}(0) - \tilde{y}_{(i1)}(0) \mid W_i = 0  ] \bigg) \mid W_i = r \bigg] \nonumber
%%
\\ \stackrel{(e)}{=} ~ &  \E \left[ \E [ \tilde{y}_{(i t)}(r) \mid W_i = r, \boldsymbol{X}_{i}  ] \mid W_i = r \right]  -   \E \bigg[ \bigg(  \E[\tilde{y}_{(i1)}(0) \mid W_i = r,   \boldsymbol{X}_{i}  ]   \nonumber
\\ &  + \gamma_t^* + \boldsymbol{X}_i^\top \boldsymbol{\xi}_t^* \bigg) \mid W_i = r \bigg] \nonumber
\\ = ~ &  \E \left[  \tilde{y}_{(i t)}(r) \mid W_i = r \right]  - \left(  \E[\tilde{y}_{(i1)}(0) \mid W_i = r ]   \nonumber
  + \gamma_t^* + \E[\boldsymbol{X}_i \mid W_i = r] ^\top \boldsymbol{\xi}_t^* \right) \nonumber
\\ \stackrel{(f)}{=} ~ &  \E \left[  \tilde{y}_{(i t)}(r) \mid W_i = r \right]  - \Big(  \E[\tilde{y}_{(i1)}(0) \mid W_i = r ]   \nonumber
\\ &   +  \E \left[ \tilde{y}_{(i t)}(0) - \tilde{y}_{(i1)}(0) \mid W_i = 0, \boldsymbol{X}_{i} =  \E[\boldsymbol{X}_i \mid W_i = r]  \right]   \Big) \nonumber
\\ =  ~ &  \E \left[  \tilde{y}_{(i t)}(r)  - \tilde{y}_{(i1)}(0)\mid W_i = r \right]   -  \E \left[ \tilde{y}_{(i t)}(0) - \tilde{y}_{(i1)}(0) \mid W_i = 0   \right]   \nonumber
\\ &  +  \E \left[ \tilde{y}_{(i t)}(0) - \tilde{y}_{(i1)}(0) \mid W_i = 0   \right]  -  \E \left[ \tilde{y}_{(i t)}(0) - \tilde{y}_{(i1)}(0) \mid W_i = 0, \boldsymbol{X}_{i} =  \E[\boldsymbol{X}_i \mid W_i = r]  \right]   \nonumber
\\ \stackrel{(g)}{=}  ~ &  \tau_{\text{ATT}}(r, t)  \nonumber  +  \E \left[ \E \left[ \tilde{y}_{(i t)}(0) - \tilde{y}_{(i1)}(0) \mid W_i = 0, \boldsymbol{X}_{i}   \right]   \mid W_i = 0 \right]    - \left(  \gamma_t^* + \E[\boldsymbol{X}_i \mid W_i = r] ^\top \boldsymbol{\xi}_t^*  \right) \nonumber
\\ \stackrel{(h)}{=}  ~ &  \tau_{\text{ATT}}(r, t)  \nonumber  +    \E \left[  \gamma_t^* + \boldsymbol{X}_i^\top \boldsymbol{\xi}_t^*  \mid W_i = 0 \right]      - \left(  \gamma_t^* + \E[\boldsymbol{X}_i \mid W_i = r] ^\top \boldsymbol{\xi}_t^*  \right) \nonumber
\\ =  ~ &  \tau_{\text{ATT}}(r, t)  \nonumber  +   \gamma_t^* +   \E \left[  \boldsymbol{X}_i  \mid W_i = 0 \right] ^\top \boldsymbol{\xi}_t^*      - \left(  \gamma_t^* + \E[\boldsymbol{X}_i \mid W_i = r] ^\top \boldsymbol{\xi}_t^*  \right) \nonumber
\\ =  ~ &  \tau_{\text{ATT}}(r, t)  \nonumber  +  \left(  \E \left[  \boldsymbol{X}_i  \mid W_i = 0 \right]     -  \E[\boldsymbol{X}_i \mid W_i = r]  \right)  ^\top \boldsymbol{\xi}_t^* \nonumber
,
\end{align}
where in (e), (f), and (h) we used \eqref{trend.params} from (LINS) and (g) uses an equality we found under (CTSA) in the proof of part (b) above as well as \eqref{trend.params} from (LINS).

\end{proof}

\subsection{Proofs of Supporting Results for Main Theorems}\label{main.thm.lems}

\begin{proof}[Proof of Lemma \ref{lem.d.express}]
First, \(\mathfrak{W} \leq (T-1)^2\) because treatment effects are interactions between the \(R \leq T-1\) cohort indicators and (at most) \(T-1\) treatment times.

Next, for any \(t \in \mathbb{N}\), define
\begin{equation}\label{d.1.expres}
\boldsymbol{D}^{(1)}(t) := 
\begin{pmatrix}
1 & -1 & 0 & 0 & \cdots & 0 & 0
\\ 0 &  1 & -1 & 0 & \cdots & 0  & 0
\\ 0 & 0 & 1 & -1 & \cdots & 0  & 0
\\ \vdots & \vdots & \vdots & \vdots & \ddots & \vdots & \vdots
\\ 0 & 0 & 0 & 0 & \cdots & 1 & -1
\\ 0 & 0 & 0 & 0 & \cdots & 0 & 1
\end{pmatrix}
 \in \mathbb{R}^{t \times t}
.
\end{equation}
Note that this is the linear transformation we apply to several groups of coefficients in \(\boldsymbol{\beta}\). For example, for the vector of cohort fixed effects \(\boldsymbol{\nu} \in \mathbb{R}^R\),
\[
\boldsymbol{D}^{(1)}(R) \boldsymbol{\nu} = \begin{pmatrix}
 \nu_2 - \nu_1
\\ \vdots
\\ \nu_R - \nu_{R-1}
\\ \nu_R 
\end{pmatrix}
,
\]
which are the transformed coefficients to which we apply an \(\ell_q\) penalty. We similarly penalize \(\boldsymbol{\gamma}\) as well as \(\boldsymbol{\zeta}_{j}\), the coefficients for the interactions of the cohort fixed effects and covariate \(j\) (separately for each of the \(d\) covariates) and \(\boldsymbol{\xi}_j\), the coefficients for the interactions of the time fixed effects and covariate \(j\). So if we rearrange the columns of \(\boldsymbol{\tilde{Z}}\) corresponding to the cohort-covariate interactions so that all of the cohort interactions with the base effect of feature \(j=1\) come first, then the cohort interactions with feature 2, and so on, we can write this \(R d_N \times Rd_N \) block diagonal part of \(\boldsymbol{D}_N\) as \( \boldsymbol{I}_{d_N} \otimes \boldsymbol{D}^{(1)}(R) \), where \(\otimes\) denotes the Kronecker product. We can likewise rearrange the columns of \(\boldsymbol{\tilde{Z}}\) and write the part of \(\boldsymbol{D}_N\) corresponding to the time-covariate interactions as \( \boldsymbol{I}_{d_N} \otimes \boldsymbol{D}^{(1)}(T - 1) \).

We penalize the coefficients \(\boldsymbol{\kappa}\) directly, so this part of \(\boldsymbol{D}\) can be represented simply as \(\boldsymbol{I}_{d_N}\).

Next we consider the matrix \(\boldsymbol{D}^{(2)}(\mathcal{R})\) we use to transform the base treatment effects \(\boldsymbol{\tau}\) for \(\ell_1\) penalization. The matrix \(\boldsymbol{D}^{(2)}(\mathcal{R})\) has a similar structure to \(\boldsymbol{D}^{(1)}(t)\), except that we directly penalize each cohort's first treatment effect, not the last one, and we penalize each cohort's first treatment effect towards the previous cohort's first treatment effect. So denoting \(\mathcal{R} = (r_1, r_2, \ldots, r_R)\), \(\boldsymbol{D}^{(2)}(\mathcal{R})\) has a nearly block diagonal structure
\begin{equation}\label{d.2.expres}
\boldsymbol{D}^{(2)}(\mathcal{R}) := 
\begin{pmatrix}
\boldsymbol{\tilde{D}}^{(1)}(1)^\top & \boldsymbol{0} & \boldsymbol{0} & \cdots & \boldsymbol{0} & \boldsymbol{0}
\\ \boldsymbol{U}(r_1)  &  \boldsymbol{\tilde{D}}^{(1)}(2)^\top  & \boldsymbol{0} & \cdots & \boldsymbol{0} & \boldsymbol{0}
\\ \boldsymbol{0} &  \boldsymbol{U}(r_2)  &  \boldsymbol{\tilde{D}}^{(1)}(3)^\top  & \cdots & \boldsymbol{0} & \boldsymbol{0}
\\ \vdots & \vdots & \vdots & \ddots & \vdots & \vdots
\\ \boldsymbol{0} & \boldsymbol{0} &  \boldsymbol{0} & \cdots &   \boldsymbol{U}(r_R)  &  \boldsymbol{\tilde{D}}^{(1)}(R)^\top 
\end{pmatrix} \in \mathbb{R}^{\mathfrak{W} \times \mathfrak{W}}
,
\end{equation}
where \(\boldsymbol{\tilde{D}}(x) := \boldsymbol{D}^{(1)}(T - r_x + 1)  \) and
\[
 \boldsymbol{U}(w) := \begin{pmatrix}
 -1 & \boldsymbol{0}_{1 \times (T - w)}
 \\ \boldsymbol{0}_{(T - w) \times 1} & \boldsymbol{0}_{(T - w) \times (T - w)}
 \end{pmatrix} \in \mathbb{R}^{(T - w + 1) \times (T - w + 1)}
 .
\]
Then \(\boldsymbol{D}^{(2)}(\mathcal{R})  \boldsymbol{\tau}\) yields the transformed treatment effects to which we apply an \(\ell_q\) penalty, and similarly, for every \(j \in [d_N]\), \(\boldsymbol{D}^{(2)}(\mathcal{R}) \boldsymbol{\rho}_j\) yields the transformed coefficients of the interactions between the \(j^{\text{th}}\) centered covariate and the treatment effects. Putting this all together, after re-ordering the columns appropriately we can write \(\boldsymbol{D}_N\) as \eqref{d.expres}.

Next we show that these matrices are invertible by explicitly presenting their inverses. One can verify by matrix multiplication that
\begin{equation}\label{d.1.inv.exp}
\boldsymbol{D}^{(1)}(t)^{-1} = \begin{pmatrix}
1 & 1 & 1 & \cdots & 1
\\ 0 & 1 & 1 & \cdots & 1
\\ 0 & 0 & 1 & \cdots & 1
\\ \vdots & \vdots & \vdots & \ddots & \vdots
\\ 0 & 0 & 0 & \cdots & 1
\end{pmatrix} \in \mathbb{R}^{t \times t}
,
\end{equation}
Similarly,
\begin{align}
 (\boldsymbol{D}^{(2)}(\mathcal{R}))^{-1} 
 = ~ &  \begin{pmatrix}
 \boldsymbol{\tilde{D}}^{-1}(1)^\top & \boldsymbol{0} & \cdots & \boldsymbol{0} & \boldsymbol{0}
\\ \boldsymbol{\tilde{U}}(r_1, r_2)  & \boldsymbol{\tilde{D}}^{-1}(2)^\top  & \cdots & \boldsymbol{0} & \boldsymbol{0}
\\ \vdots & \vdots & \ddots & \vdots & \vdots
\\  \boldsymbol{\tilde{U}}(r_1, r_R)  &  \boldsymbol{\tilde{U}}(r_2, r_R)  &  \cdots &   \boldsymbol{\tilde{U}}(r_{R-1}, r_R)  &  \boldsymbol{\tilde{D}}^{-1}(R)^\top
\end{pmatrix} \label{d.2.inv.exp}
,
\end{align}
where
\[
\boldsymbol{\tilde{U}}(w_1, w_2) = \begin{pmatrix}
1 & 0 & 0 & \cdots & 0
\\ 1 & 0 & 0 & \cdots & 0
\\ \vdots & \vdots & \vdots & \ddots & \vdots
\\ 1 & 0 & 0 & \cdots & 0
\end{pmatrix} \in \mathbb{R}^{(T - w_2 + 1) \times (T - w_1 + 1) }
.
\]
Finally, using the properties of block diagonal matrices and the fact that \(\boldsymbol{D}^{(1)}(t)\) and \(\boldsymbol{D}^{(2)}(\mathcal{R})\) are invertible, we have from \eqref{d.expres} that \(\boldsymbol{D}_N^{-1} \) is 
\begin{align*}
  \operatorname{diag} \Big( & \boldsymbol{D}^{(1)}(R)^{-1} , ((\boldsymbol{D}^{(1)}(T - 1))^{-1} ,  \boldsymbol{I}_{d_N} , \boldsymbol{I}_{d_N} \otimes \boldsymbol{D}^{(1)}(R)^{-1} , \boldsymbol{I}_{d_N} \otimes ((\boldsymbol{D}^{(1)}(T - 1))^{-1} ,
\\ & (\boldsymbol{D}^{(2)}(\mathcal{R}))^{-1} , \boldsymbol{I}_{d_N} \otimes (\boldsymbol{D}^{(2)}(\mathcal{R}))^{-1} \Big)
,
\end{align*}
matching the expression from \eqref{d.inv.exp}.

Lastly, we show the claimed block diagonal structure of \(\boldsymbol{D}_N\) and \(\boldsymbol{D}_N^{-1}\). Examining the expressions for \(\boldsymbol{D}_N\) and \(\boldsymbol{D}_N^{-1}\), it is clear that they are both block diagonal with blocks of size \(R \leq T - 1\), \(T-1\), and \(\mathfrak{W} \leq (T-1)^2\). We have shown in this proof that all of the blocks contain entries equal to 0, -1, or 1.

\end{proof}

\begin{proof}[Proof of Lemma \ref{bound.coefs}]
By Lemma \ref{lem.d.express}, \(\boldsymbol{D}_N^{-1}\) has a block diagonal structure with blocks of size at most \( (T - 1)^2 \times (T-1)^2\) containing entries that equal either 0 or 1, so \(\lVert  \boldsymbol{D}_N^{-1}  \rVert_\infty \leq (T-1)^2\), where for a matrix \(\lVert \cdot \rVert_\infty\) denotes the maximum absolute row sum. Therefore
\begin{align*}
\lVert \boldsymbol{\beta}_N^* \rVert_\infty = ~ & \left\lVert  \boldsymbol{D}_N^{-1} \boldsymbol{\theta}_N^*  \right\rVert_\infty
 \leq   \left\lVert  \boldsymbol{D}_N^{-1} \right \rVert_{\infty}   \left\lVert  \boldsymbol{\theta}_N^*  \right\rVert_\infty
 \leq  (T-1)^2 b_1
,
\end{align*}
where in the last step we used Assumption (R3).
\end{proof}

\begin{proof}[Proof of Lemma \ref{lem.asym.norm.cond.means}] 

To prove the result we will apply Theorem 5.1.7 in \citet{lehmann1999elements}, which we reproduce here:

\begin{theorem}[Theorem 5.1.7 from \citet{lehmann1999elements}]\label{thm.5.1.7}
Let \(\boldsymbol{a}  = (a_1, \ldots, a_k)\) be a vector of constants, and let \(\{\boldsymbol{U}^{(N)}\}\) and \(\{\boldsymbol{A}^{(N)}\}\) be sequences of random variables, with \(\boldsymbol{U}^{(N)} = (U_1^{(N)}, \ldots, U_k^{(N)})\) and \(\boldsymbol{A}^{(N)} = (A_1^{(N)}, \ldots, A_k^{(N)})\). Suppose the sequence of random variables \(T_N = \sum_{j=1}^k a_j U_j^{(N)} \) converges in distribution to a random variable \(T\), the random variables \(A_j^{(N)} \xrightarrow{p} a_j\) for each \(j \in \{1, \ldots, k\}\), and the random variables \(U_j^{(N)} = \mathcal{O}_{\mathbb{P}}(1)\). Then the sequence of random variables \( \sum_{j=1}^k A_j^{(n)} U_j^{(N)} \) converges in distribution to \(T\).
\end{theorem}

Given the joint probability distribution \(\{(W_i, \boldsymbol{X}_i)\}_{i=1}^N\), consider the iid random variables \(Z_i\) defined by
\begin{align*}
Z_i   := ~ & \sum_{r \in \{0\} \cup \mathcal{R}} \frac{ \mathbbm{1}\{W_i = r\}}{\mathbb{P}(W_i = r)} \boldsymbol{\psi}_{r}^\top \left( \boldsymbol{X}_i - \E[ \boldsymbol{X}_i \mid W_i  = r]  \right)
\\ = ~ & \sum_{r \in \{0\} \cup \mathcal{R}}  \frac{\mathbbm{1}\{W_i = r\}}{\mathbb{P}(W_i = r)}  \boldsymbol{\psi}_{r}^\top  \left( \boldsymbol{X}_i - \boldsymbol{\mu}_{r} \right) 
.
\end{align*}
We have
\begin{align*}
\E [ Z_i] = ~ &  \E [ \E[ Z_i \mid W_i ] ] 
 =  \E \left[  \frac{1}{\mathbb{P}(W_i )}  \boldsymbol{\psi}_{W_i}^\top \left( \E[ \boldsymbol{X}_i \mid W_i ]  - \E[ \boldsymbol{X}_i \mid W_i ] \right) \right]
 =  
   \E \left[ 0 \right]
 =  0
\end{align*}
and
\begin{align}
\Var (Z_i) = ~ & \E[Z_i^2] \nonumber
=  \E \left [ \E\left[ Z_i^2 \mid W_i \right] \right]  \nonumber
\\ = ~ & \E\left[    \frac{1}{(\mathbb{P}(W_i ))^2} \boldsymbol{\psi}_{W_i}^\top  \Cov ( \boldsymbol{X}_i \mid W_i)\boldsymbol{\psi}_{W_i} \right]  \nonumber
\\ = ~ & \sum_{r \in \{0\} \cup \mathcal{R}}    \frac{1}{\mathbb{P}(W_i = r )^2}  \boldsymbol{\psi}_{r}^\top  \Cov ( \boldsymbol{X}_i \mid W_i = r)\boldsymbol{\psi}_{r} \cdot \mathbb{P}(W_i = r)  \nonumber
%
%\\ = ~ & \sum_{r \in \{0\} \cup \mathcal{R}}   \frac{1}{\mathbb{P}(W_i = r )^2}  \mathbb{P}(W_i = r) \boldsymbol{\psi}_{r}^\top  \Cov ( \boldsymbol{X}_i \mid W_i = r)\boldsymbol{\psi}_{r}      \nonumber
%
\\ = ~ & v_R \nonumber % \label{v.r.cond.norm.lem}
.
\end{align}
We can lower-bound \(v_R\) uniformly across choices of \(\boldsymbol{\psi}_r\): since \(\boldsymbol{\psi}_r \neq \boldsymbol{0}\) for at least one \(r = r^*\) and by assumption the smallest eigenvalue of \( \Cov ( \boldsymbol{X}_i \mid W_i = r^*)\) is bounded below by some \(\lambda_{\text{min}} > 0\), we have
\[
v_R \geq  \frac{\boldsymbol{\psi}_{r^*}^\top  \Cov ( \boldsymbol{X}_i \mid W_i = r^*)\boldsymbol{\psi}_{r^*}}{\mathbb{P}(W_i = r^*)}  \geq \frac{ \lambda_{\text{min}}  \lVert \boldsymbol{\psi}_{r^*}\rVert_2^2}{ 1}  > 0
.
\]
Similarly, since the largest eigenvalue of \( \Cov ( \boldsymbol{X}_i \mid W_i = r) \leq \lambda_{\text{max}}\) for some finite \(\lambda_{\text{max}}\) for all \(r \in \{0\} \cup \mathcal{R}\) and \(\min_{r \in \{0\} \cup \mathcal{R}}\{  \mathbb{P}(W_i = r) \} > 0\) by Assumption (R1), we can upper-bound \(v_R\) by a constant. So by the central limit theorem, for any choice of finite \(\{\boldsymbol{\psi}_r\}\) with at least one \(\boldsymbol{\psi}_{r^*} \neq \boldsymbol{0}\),
\begin{align}
\frac{1}{\sqrt{ N }} \sum_{i=1}^N  Z_i \xrightarrow{d}  ~ & \mathcal{N}(0,  v_R) \nonumber
\\ \iff \qquad \frac{1}{\sqrt{ N    }} \sum_{r \in \{0\} \cup \mathcal{R}} \sum_{\{i: W_i = r\}}     \frac{1}{\mathbb{P}(W_i = r)} \boldsymbol{\psi}_{r}^\top  \left( \boldsymbol{X}_i - \boldsymbol{\mu}_{r}  \right)  \xrightarrow{d}  ~ & \mathcal{N}(0,   v_R) \nonumber
\\ \iff \qquad \frac{1}{\sqrt{ N  }} \sum_{r \in \{0\} \cup \mathcal{R}}   \frac{N_r}{\mathbb{P}(W_i = r)}   \boldsymbol{\psi}_{r}^\top  \left(  \frac{1}{N_r} \sum_{\{i: W_i = r\}}  \boldsymbol{X}_i - \boldsymbol{\mu}_{r}  \right)  \xrightarrow{d}  ~ & \mathcal{N}(0,  v_R) \nonumber
\\ \iff \qquad   \sum_{r \in \{0\} \cup \mathcal{R}}   \sqrt{N}  \frac{N_r / N}{\mathbb{P}(W_i = r)}  \boldsymbol{\psi}_{r}^\top  \left(  \boldsymbol{\overline{X}}_r - \boldsymbol{\mu}_{r}  \right)  \xrightarrow{d}  ~ & \mathcal{N}(0,   v_R) \label{dist.result}
.
\end{align}
Now consider the random variables
\begin{align*}
\boldsymbol{A}^{(N)}  :=  ~ & \left( \pi_{0} \cdot  \frac{N}{\sum_{i=1}^N \mathbbm{1}\{W_i = 0\}},    \pi_{r_1} \cdot  \frac{N}{\sum_{i=1}^N \mathbbm{1}\{W_i = r_1\}}, \ldots,    \pi_{r_R} \cdot  \frac{N}{\sum_{i=1}^N \mathbbm{1}\{W_i = r_R\}} \right)
\\ =  ~ & \left(  \pi_{0} \cdot \frac{N}{N_{0}},   \pi_{r_1} \cdot  \frac{N}{N_{r_1}}, \ldots,  \pi_{r_R} \cdot  \frac{N}{N_{r_R}} \right)
.
\end{align*}
By the continuous mapping theorem \(\boldsymbol{A}^{(N)}  \xrightarrow{p}   \left(  1, 1, \ldots,  1 \right)\). Further, let
\begin{align*}
\boldsymbol{U}^{(N)} := ~ & \Bigg( \sum_{i=1}^N \frac{\mathbbm{1}\{W_i = 0\}}{\pi_{0}}    \boldsymbol{\psi}_{0}^\top \left( \boldsymbol{X}_i - \E[ \boldsymbol{X}_i \mid W_i  = 0]  \right) ,  \sum_{i=1}^N \frac{\mathbbm{1}\{W_i = r_1\}}{\pi_{r_1}}    \boldsymbol{\psi}_{r_1}^\top \left( \boldsymbol{X}_i - \E[ \boldsymbol{X}_i \mid W_i  = r_1]  \right),
\\ & \cdots, \sum_{i=1}^N \frac{\mathbbm{1}\{W_i = r_R\}}{\pi_{r_R}}    \boldsymbol{\psi}_{r_R}^\top \left( \boldsymbol{X}_i - \E[ \boldsymbol{X}_i \mid W_i  = r_R]  \right)  \Bigg)
.
\end{align*}
Since from \eqref{dist.result} we have
\[
\sum_{r \in \{0\} \cup \mathcal{R}}   1 \cdot  \frac{1}{\sqrt{N}}U_{r}^{(N)} \xrightarrow{d}  \mathcal{N}(0,   v_R) 
,
\]
we have from Theorem \ref{thm.5.1.7}
\begin{align*}
\sum_{r \in \{0\} \cup \mathcal{R}}   A_{r}^{(N)} \cdot  \frac{1}{\sqrt{N}}U_{r}^{(N)} \xrightarrow{d} ~ &   \mathcal{N}(0,   v_R) 
\\ \iff \qquad   \sum_{r \in \{0\} \cup \mathcal{R}}   \sqrt{\frac{N}{ v_R  }}    \boldsymbol{\psi}_{r}^\top  \left(  \boldsymbol{\overline{X}}_r - \boldsymbol{\mu}_{r}  \right)  \xrightarrow{d}  ~ &  \mathcal{N}(0, 1)
.
\end{align*}
Finally we prove the last claim. It follows from the above that
\[
N \lVert   \boldsymbol{\overline{X}}_r - \boldsymbol{\mu}_{r}  \rVert_2^2
\]
converges in distribution to a generalized chi-squared random variable, so by Theorem 2.4(i) in \citet{van2000asymptotic} and Proposition 1.8(iii) in \citet{Garcia-Portugues2023} we have
\begin{align*}
N \lVert   \boldsymbol{\overline{X}}_r - \boldsymbol{\mu}_{r}  \rVert_2^2 = ~ & \mathcal{O}_{\mathbb{P}}(1)
\\ \implies \qquad  \lVert   \boldsymbol{\overline{X}}_r - \boldsymbol{\mu}_{r}  \rVert_2^2 = ~ & \mathcal{O}_{\mathbb{P}} \left( \frac{1}{N} \right)
\\ \implies \qquad  \lVert   \boldsymbol{\overline{X}}_r - \boldsymbol{\mu}_{r}  \rVert_2 = ~ & \mathcal{O}_{\mathbb{P}} \left( \frac{1}{\sqrt{N}} \right)
.
\end{align*}

%\[
%\vdots
%\]
%
%Since for all \(r \in \{0\} \cup \mathcal{R}\) we have from the weak law of large numbers that \(N_r/N \xrightarrow{p} \mathbb{P}(W_i = r)\), by the continuous mapping theorem
%\[
% \frac{\mathbb{P}(W_i = r)}{N_r/N} \xrightarrow{p} 1
% ,
%\]
%so by the multivariate Slutsky's theorem (Theorem 5.1.6 in \citealt{lehmann1999elements}) the above yields

\end{proof}

\begin{proof}[Proof of Lemma \ref{tilde.pi.r.o.p.cond.lemma}]
%The proof follows a similar structure to the proof of Lemma \ref{tilde.pi.r.o.p.lemma}. 
For any \(r \in \mathcal{R}\),
\begin{align*}
%&  \left| \hat{\pi}_r(\boldsymbol{x})  -  \tilde{\pi}_r(\boldsymbol{x})  \right| 
%%
%\\  = ~ 
& \left| \frac{\boldsymbol{a}^\top \boldsymbol{\hat{\pi}}(\boldsymbol{x})}{\boldsymbol{b}^\top \boldsymbol{\hat{\pi}}(\boldsymbol{x})} -  \frac{\boldsymbol{a}^\top \boldsymbol{{\pi}}(\boldsymbol{x})}{\boldsymbol{b}^\top \boldsymbol{{\pi}}(\boldsymbol{x})}  \right| 
\\  = ~ &  \left|   \frac{\boldsymbol{a}^\top \boldsymbol{\hat{\pi}}(\boldsymbol{x}) \cdot \boldsymbol{b}^\top \boldsymbol{{\pi}}(\boldsymbol{x}) - \boldsymbol{a}^\top \boldsymbol{{\pi}}(\boldsymbol{x}) \boldsymbol{b}^\top \boldsymbol{\hat{\pi}}(\boldsymbol{x})}{\boldsymbol{b}^\top \boldsymbol{\hat{\pi}}(\boldsymbol{x}) \boldsymbol{b}^\top \boldsymbol{{\pi}}(\boldsymbol{x})}  \right|
\\  = ~ &  \left|   \frac{\boldsymbol{a}^\top \boldsymbol{\hat{\pi}}(\boldsymbol{x}) \cdot \boldsymbol{b}^\top \boldsymbol{{\pi}}(\boldsymbol{x}) 
- \boldsymbol{a}^\top \boldsymbol{{\pi}}(\boldsymbol{x}) \cdot \boldsymbol{b}^\top \boldsymbol{{\pi}}(\boldsymbol{x})
+  \boldsymbol{a}^\top \boldsymbol{{\pi}}(\boldsymbol{x}) \cdot \boldsymbol{b}^\top \boldsymbol{{\pi}}(\boldsymbol{x})
- \boldsymbol{a}^\top \boldsymbol{{\pi}}(\boldsymbol{x}) \boldsymbol{b}^\top \boldsymbol{\hat{\pi}}(\boldsymbol{x})}{\boldsymbol{b}^\top \boldsymbol{\hat{\pi}}(\boldsymbol{x}) \boldsymbol{b}^\top \boldsymbol{{\pi}}(\boldsymbol{x})}  \right|
\\  = ~ &  \left|   \frac{\boldsymbol{b}^\top \boldsymbol{{\pi}}(\boldsymbol{x}) \cdot \boldsymbol{a}^\top \left( \boldsymbol{\hat{\pi}}(\boldsymbol{x}) 
- \boldsymbol{{\pi}}(\boldsymbol{x}) \right)
+  \boldsymbol{a}^\top \boldsymbol{{\pi}}(\boldsymbol{x}) \cdot \boldsymbol{b}^\top \left( \boldsymbol{{\pi}}(\boldsymbol{x})
- \boldsymbol{\hat{\pi}}(\boldsymbol{x}) \right)}
{\boldsymbol{b}^\top  \boldsymbol{\hat{\pi}}(\boldsymbol{x}) \boldsymbol{b}^\top \boldsymbol{{\pi}}(\boldsymbol{x})}  \right|
 \\ = ~ &  \frac{1}{ \boldsymbol{b}^\top  \boldsymbol{\hat{\pi}}(\boldsymbol{x})  }  \mathcal{O}_{\mathbb{P}}\left( a_N \right)
 ,
\end{align*}
where in the last step we used our consistency assumption and the finiteness of \(\boldsymbol{a}\) and \(\boldsymbol{b}\). We are done if we can show that
\[
 \frac{1}{ \boldsymbol{b}^\top  \boldsymbol{\hat{\pi}}(\boldsymbol{x})  }    =  \mathcal{O}_{\mathbb{P}}(1)
.
\]
That is, we want to show that for any \(\epsilon > 0\) there exists a finite \(M_\epsilon > 0\) and \(N_\epsilon > 0\) such that for all \(N > N_\epsilon\),
\[
\mathbb{P} \left(  \frac{1}{ \boldsymbol{b}^\top  \boldsymbol{\hat{\pi}}(\boldsymbol{x})  }  > M_\epsilon \right) < \epsilon
\qquad \iff \qquad \mathbb{P} \left(  \boldsymbol{b}^\top  \boldsymbol{\hat{\pi}}(\boldsymbol{x}) <  \frac{1}{M_\epsilon} \right) < \epsilon
.
\]
Fix \(\epsilon > 0\). By assumption, for any \(\boldsymbol{x}\) in the support of \(\boldsymbol{X}_i\) we have \(| \boldsymbol{b}^\top(  \boldsymbol{\hat{\pi}}(\boldsymbol{x}) -  \boldsymbol{\pi}(\boldsymbol{x}) ) | = \mathcal{O}_{\mathbb{P}}(a_N)\), so there exists a finite \(M_\epsilon^{(1)} > 0\) and \(N_\epsilon^{(1)} > 0\) such that for all \(N > N_\epsilon^{(1)}\),
\begin{align*}
\mathbb{P} \left( \frac{ | \boldsymbol{b}^\top(  \boldsymbol{\hat{\pi}}(\boldsymbol{x}) -  \boldsymbol{\pi}(\boldsymbol{x}) ) | }{a_N} > M_\epsilon^{(1)} \right) < \epsilon
\qquad \implies  \qquad & \mathbb{P} \left(   \boldsymbol{b}^\top \boldsymbol{\hat{\pi}}(\boldsymbol{x}) < \boldsymbol{b}^\top \boldsymbol{\pi}(\boldsymbol{x}) -   M_\epsilon^{(1)} a_N \right) < \epsilon
%
%\\ \implies  \qquad & \math{P} \left(  \hat{\pi}_0(\boldsymbol{x})> \pi_0(\boldsymbol{x}) +   M_\epsilon^{(1)} a_N \right) < \epsilon
.
\end{align*}
%Since \(  \mathbb{P} \left(  \hat{\pi}_0(\boldsymbol{x})> 1 \right)  < \epsilon\) is trivial, we can substitute
%\[
%M_\epsilon^{(2)}  := \min \{ M_\epsilon^{(1)} , 1/a_N\}
%.
%\]
%\begin{align*}
%\pi_0(\boldsymbol{x}) +   M_\epsilon^{(1)} a_N   = ~ & 1 -  \frac{1}{M_\epsilon} 
%%
%\\ \iff \qquad \frac{1}{M_\epsilon}  = ~ & 1 - \pi_0(\boldsymbol{x}) -   M_\epsilon^{(1)} 
%%
%\\ \iff \qquad M_\epsilon  = ~ & \frac{1}{1 - \pi_0(\boldsymbol{x}) -   M_\epsilon^{(1)} a_N}
%\end{align*}
Since by assumption \( \boldsymbol{b}^\top \boldsymbol{\pi}(\boldsymbol{x}) > 0\), the result holds for \(N_\epsilon \geq N_\epsilon^{(1)} \) large enough so that for all \(N > N_\epsilon\) it holds that \(\boldsymbol{b}^\top \boldsymbol{\pi}(\boldsymbol{x}) -   M_\epsilon^{(1)} a_N > 0\) (which is assured since \(a_N\) is decreasing), and \(M_\epsilon := 1/( \boldsymbol{b}^\top \boldsymbol{\pi}(\boldsymbol{x}) -   M_\epsilon^{(1)} a_{N_\epsilon}  ) \).
\end{proof}

\begin{proof}[Proof of Lemma \ref{lem.int.asym.norm.cond.est.prob}]
Observe that
\begin{align*}
 & \boldsymbol{\hat{\psi}}_N^{(\text{C;CATT,} \pi)} - \boldsymbol{\psi}_N^{(\text{C;CATT,} \pi)}
 %
% \\ = ~ &  \sum_{r \in \mathcal{R}} f_r \left(  \hat{\pi}_{r_1}(\boldsymbol{x}) , \hat{\pi}_{r_2}(\boldsymbol{x}) , \ldots,  \hat{\pi}_{r_R}(\boldsymbol{x})  \right)   \sum_{t = r}^T \psi_{rt}  \left(  (\boldsymbol{D}_N^{-1})_{i(r,t), \cdot} + \left( \boldsymbol{x} -  \boldsymbol{\overline{X}}_r \right)^\top (\boldsymbol{D}_N^{-1})_{i(r, t, \boldsymbol{x}), \cdot}  \right) 
% %
% \\ &  -   \sum_{r \in \mathcal{R}}  f_r \left(  \pi_{r_1}(\boldsymbol{x}) , \pi_{r_2}(\boldsymbol{x}) , \ldots,  \pi_{r_R} (\boldsymbol{x})  \right)   \sum_{t = r}^T \psi_{rt}  \left(  (\boldsymbol{D}_N^{-1})_{i(r,t), \cdot} + \left( \boldsymbol{x} -   \E[ \boldsymbol{X}_i \mid W_i = r ]\right)^\top (\boldsymbol{D}_N^{-1})_{i(r, t, \boldsymbol{x}), \cdot}  \right) 
%
\\ = ~ &   \sum_{r \in \mathcal{R}}  \sum_{t = r}^T     f_r \left(  \hat{\pi}_{r_1}(\boldsymbol{x}) , \hat{\pi}_{r_2}(\boldsymbol{x}) , \ldots,  \hat{\pi}_{r_R}(\boldsymbol{x})  \right)  \psi_{rt}   \boldsymbol{\hat{\psi}}_{rt}^*  -    \sum_{r \in \mathcal{R}}   f_r \left(  \pi_{r_1}(\boldsymbol{x}) , \pi_{r_2}(\boldsymbol{x}) , \ldots,  \pi_{r_R} (\boldsymbol{x})  \right)   \sum_{t = r}^T  \psi_{rt}   \boldsymbol{\psi}_{rt}^* 
%
%\\ = ~ &   \sum_{r \in \mathcal{R}}   f_r \left(  \pi_{r_1}(\boldsymbol{x}) , \pi_{r_2}(\boldsymbol{x}) , \ldots,  \pi_{r_R} (\boldsymbol{x})  \right)  \sum_{t = r}^T  \psi_{rt}  \boldsymbol{\hat{\psi}}_{rt}^* 
%-     \sum_{r \in \mathcal{R}}    f_r \left(  \pi_{r_1}(\boldsymbol{x}) , \pi_{r_2}(\boldsymbol{x}) , \ldots,  \pi_{r_R} (\boldsymbol{x})  \right)   \sum_{t = r}^T  \psi_{rt}   \boldsymbol{\psi}_{rt}^*
%  % 
%\\ & +   \sum_{r \in \mathcal{R}}   \left( f_r \left(  \hat{\pi}_{r_1}(\boldsymbol{x}) , \hat{\pi}_{r_2}(\boldsymbol{x}) , \ldots,  \hat{\pi}_{r_R}(\boldsymbol{x})  \right)  
%  -   f_r \left(  \pi_{r_1}(\boldsymbol{x}) , \pi_{r_2}(\boldsymbol{x}) , \ldots,  \pi_{r_R} (\boldsymbol{x})  \right)  \right)  \sum_{t = r}^T   \psi_{rt}  \boldsymbol{\hat{\psi}}_{rt}^*
  %
\\ = ~ &    \sum_{r \in \mathcal{R}}   \sum_{t = r}^T   f_r \left(  \pi_{r_1}(\boldsymbol{x}) , \pi_{r_2}(\boldsymbol{x}) , \ldots,  \pi_{r_R} (\boldsymbol{x})  \right)   \psi_{rt}  \left( \boldsymbol{\hat{\psi}}_{rt}^* 
-   \boldsymbol{\psi}_{rt}^*  \right)
\\ & +   \sum_{r \in \mathcal{R}}  \sum_{t = r}^T   \left( f_r \left(  \hat{\pi}_{r_1}(\boldsymbol{x}) , \hat{\pi}_{r_2}(\boldsymbol{x}) , \ldots,  \hat{\pi}_{r_R}(\boldsymbol{x})  \right)  
  -   f_r \left(  \pi_{r_1}(\boldsymbol{x}) , \pi_{r_2}(\boldsymbol{x}) , \ldots,  \pi_{r_R} (\boldsymbol{x})  \right)  \right)   \psi_{rt}   \boldsymbol{\psi}_{rt}^* 
\\ & +    \sum_{r \in \mathcal{R}}   \sum_{t = r}^T \left( f_r \left(  \hat{\pi}_{r_1}(\boldsymbol{x}) , \hat{\pi}_{r_2}(\boldsymbol{x}) , \ldots,  \hat{\pi}_{r_R}(\boldsymbol{x})  \right)  
  -   f_r \left(  \pi_{r_1}(\boldsymbol{x}) , \pi_{r_2}(\boldsymbol{x}) , \ldots,  \pi_{r_R} (\boldsymbol{x})  \right)  \right)   \psi_{rt}  \left( \boldsymbol{\hat{\psi}}_{rt}^* 
-   \boldsymbol{\psi}_{rt}^*  \right)
,
\end{align*}
where \(\boldsymbol{\hat{\psi}}_{rt}^*\) was defined in \eqref{hat.psi.star.rt.def} and \( \boldsymbol{\hat{\psi}}_{rt}^* \) was defined in \eqref{hat.psi.nostar.rt.def}. We will show that after scaling by \(\sqrt{NT}\) and multiplying by \( (  \boldsymbol{\theta}_N^*)^\top\), terms of the first two forms converge in distribution to independent Gaussian random variables, and terms of the third form converge in probability to 0.

\begin{itemize}

\item Applying Lemma \ref{lem.asym.norm.cond.means}, we have
\begin{align*}
& \sqrt{NT}  (  \boldsymbol{\theta}_N^*)^\top  \sum_{r \in \mathcal{R}}  f_r \left(  \pi_{r_1}(\boldsymbol{x}) , \pi_{r_2}(\boldsymbol{x}) , \ldots,  \pi_{r_R} (\boldsymbol{x})  \right) \sum_{t=r}^T \psi_{rt} \left( \boldsymbol{\hat{\psi}}_{rt}^* 
-   \boldsymbol{\psi}_{rt}^*  \right)
\\ = ~ & \sqrt{NT} \sum_{r \in \mathcal{R}}  f_r \left(  \pi_{r_1}(\boldsymbol{x}) , \pi_{r_2}(\boldsymbol{x}) , \ldots,  \pi_{r_R} (\boldsymbol{x})  \right) \sum_{t=r}^T - \psi_{rt}  (  \boldsymbol{\theta}_N^*)^\top  \left( (\boldsymbol{D}_N^{-1})_{i(r, t, \boldsymbol{x}), \cdot}^\top    \left(   \boldsymbol{\overline{X}}_r   -  \boldsymbol{\mu}_r   \right)  
 \right)
 \\ \xrightarrow{d} ~ & \mathcal{N}(0, v_R^{\text{CATT; prob}})
,
\end{align*}
where
\[
v_R^{\text{CATT; prob}} :=  T \cdot \sum_{r \in \mathcal{R}}   \frac{(\boldsymbol{\psi}_{r}^{\text{CATT; prob}})^\top  \Cov ( \boldsymbol{X}_i \mid W_i = r) \boldsymbol{\psi}_{r}^{\text{CATT; prob}}}{\mathbb{P}(W_i = r)} 
\]
for
\begin{align}
\boldsymbol{\psi}_{r}^{\text{CATT; prob}} :=  ~ &  f_r \left(  \pi_{r_1}(\boldsymbol{x}) , \pi_{r_2}(\boldsymbol{x}) , \ldots,  \pi_{r_R} (\boldsymbol{x})  \right) \sum_{t = r}^T   \psi_{rt} (  \boldsymbol{D}_N^{-1})_{i(r, t, \boldsymbol{x}), \cdot}  \boldsymbol{\theta}_N^* \nonumber
\\ = ~ &   f_r \left(  \pi_{r_1}(\boldsymbol{x}) , \pi_{r_2}(\boldsymbol{x}) , \ldots,  \pi_{r_R} (\boldsymbol{x})  \right) \sum_{t = r}^T   \psi_{rt}\boldsymbol{\rho}_{rt}^* , \qquad r \in \mathcal{R} \label{psi.catt.prob.def}
.
\end{align}
\item Next, using the assumed joint asymptotic normality of \((\hat{\pi}_0(\boldsymbol{x}), \hat{\pi}_{r_1}(\boldsymbol{x}), \ldots, \hat{\pi}_{r_R}(\boldsymbol{x}))\) and the delta method (Theorem 3.1 in \citealt{van2000asymptotic}), we have for each \(r \in \mathcal{R}\)
\begin{align*}
& \sqrt{NT}  \big(  f_{r_1} \left(  \hat{\pi}_{r_1}(\boldsymbol{x}) , \hat{\pi}_{r_2}(\boldsymbol{x}) , \ldots,  \hat{\pi}_{r_R} (\boldsymbol{x})  \right)
  -  f_{r_1} \left(  \pi_{r_1}(\boldsymbol{x}) , \pi_{r_2}(\boldsymbol{x}) , \ldots,  \pi_{r_R} (\boldsymbol{x})  \right)
  ,  \ldots,
  \\ &  f_{r_R} \left(  \hat{\pi}_{r_1}(\boldsymbol{x}) , \hat{\pi}_{r_2}(\boldsymbol{x}) , \ldots,  \hat{\pi}_{r_R} (\boldsymbol{x})  \right)
  -    f_{r_R} \left(  \pi_{r_1}(\boldsymbol{x}) , \pi_{r_2}(\boldsymbol{x}) , \ldots,  \pi_{r_R} (\boldsymbol{x})  \right) \big)
\\ \xrightarrow{d} ~ & \mathcal{N}\left(\boldsymbol{0}, T \cdot \nabla \boldsymbol{f}^\top    \boldsymbol{\Sigma}_\pi \nabla \boldsymbol{f} \right)
\end{align*}

Therefore 
\begin{align*}
& \sqrt{NT}  (  \boldsymbol{\theta}_N^*)^\top    \sum_{r \in \mathcal{R}}   \left( f_r \left(  \hat{\pi}_{r_1}(\boldsymbol{x}) , \hat{\pi}_{r_2}(\boldsymbol{x}) , \ldots,  \hat{\pi}_{r_R}(\boldsymbol{x})  \right)  
  -   f_r \left(  \pi_{r_1}(\boldsymbol{x}) , \pi_{r_2}(\boldsymbol{x}) , \ldots,  \pi_{r_R} (\boldsymbol{x})  \right)  \right) \sum_{t=r}^T \psi_{rt} \boldsymbol{\psi}_{rt}^* 
\\ = ~ & \sqrt{NT}   \sum_{r \in \mathcal{R}}   \left( f_r \left(  \hat{\pi}_{r_1}(\boldsymbol{x}) , \hat{\pi}_{r_2}(\boldsymbol{x}) , \ldots,  \hat{\pi}_{r_R}(\boldsymbol{x})  \right)  
  -   f_r \left(  \pi_{r_1}(\boldsymbol{x}) , \pi_{r_2}(\boldsymbol{x}) , \ldots,  \pi_{r_R} (\boldsymbol{x})  \right)  \right) \sum_{t=r}^T \psi_{rt}   (  \boldsymbol{\theta}_N^*)^\top   \boldsymbol{\psi}_{rt}^* 
\\ \xrightarrow{d} ~ & \mathcal{N} \left(0, T \cdot \boldsymbol{\psi}_\pi^\top  \nabla \boldsymbol{f}^\top  \boldsymbol{\Sigma}_\pi \nabla \boldsymbol{f} \boldsymbol{\psi}_\pi \right)
\end{align*}
for
\begin{align*}
\boldsymbol{\psi}_\pi := ~ &  \left( \sum_{t=r_1}^T \psi_{r_1t}   (  \boldsymbol{\theta}_N^*)^\top   \boldsymbol{\psi}_{r_1t}^* , \sum_{t=r_2}^T \psi_{r_2t}   (  \boldsymbol{\theta}_N^*)^\top   \boldsymbol{\psi}_{r_2t}^*, \ldots, \sum_{t=r_R}^T \psi_{r_Rt}   (  \boldsymbol{\theta}_N^*)^\top   \boldsymbol{\psi}_{r_Rt}^* \right) 
\\ = ~ &  \bigg( \sum_{t=r_1}^T \psi_{r_1t} \left( \boldsymbol{\tau}_{r_1t}^* + \left(\boldsymbol{x} - \E [ \boldsymbol{X}_i \mid W_i = r_1] \right)^\top \boldsymbol{\rho}_{r_1t}^* \right), \ldots,
\\ &  \sum_{t=r_R}^T \psi_{r_Rt} \left( \boldsymbol{\tau}_{r_Rt}^* + \left(\boldsymbol{x} - \E [ \boldsymbol{X}_i \mid W_i = r_R] \right)^\top \boldsymbol{\rho}_{r_Rt}^* \right) \bigg) 
.
\end{align*}
\item Finally, the above two results, along with Theorem 2.4(i) in \citet{van2000asymptotic} and Proposition 1.8(iii) in \citet{Garcia-Portugues2023}, yield that for any \(r \in \mathcal{R}\)
\begin{align*}
&  \left( f_r \left(  \hat{\pi}_{r_1}(\boldsymbol{x}) , \hat{\pi}_{r_2}(\boldsymbol{x}) , \ldots,  \hat{\pi}_{r_R}(\boldsymbol{x})  \right)  
  -   f_r \left(  \pi_{r_1}(\boldsymbol{x}) , \pi_{r_2}(\boldsymbol{x}) , \ldots,  \pi_{r_R} (\boldsymbol{x})  \right)  \right)   (  \boldsymbol{\theta}_N^*)^\top \sum_{t =r}^T \left( \boldsymbol{\hat{\psi}}_{rt}^* 
-   \boldsymbol{\psi}_{rt}^*  \right) 
\\ = ~ &  \mathcal{O}_{\mathbb{P}}\left( \frac{1}{\sqrt{N}} \right)   \mathcal{O}_{\mathbb{P}}\left( \frac{1}{\sqrt{N}} \right)  =  \mathcal{O}_{\mathbb{P}}\left( \frac{1}{N} \right) ,
\end{align*}
so
\begin{align*}
& \sqrt{NT} (  \boldsymbol{\theta}_N^*)^\top \sum_{r \in \mathcal{R}} \big( f_r \left(  \hat{\pi}_{r_1}(\boldsymbol{x}) , \hat{\pi}_{r_2}(\boldsymbol{x}) , \ldots,  \hat{\pi}_{r_R}(\boldsymbol{x})  \right)  
 \\ &  -   f_r \left(  \pi_{r_1}(\boldsymbol{x}) , \pi_{r_2}(\boldsymbol{x}) , \ldots,  \pi_{r_R} (\boldsymbol{x})  \right)  \big)   \sum_{t =r}^T \psi_{rt} \left( \boldsymbol{\hat{\psi}}_{rt}^* 
-   \boldsymbol{\psi}_{rt}^*  \right) 
\\ = ~ & \sqrt{NT}\sum_{r \in \mathcal{R}} \sum_{t =r}^T \psi_{rt} \big( f_r \left(  \hat{\pi}_{r_1}(\boldsymbol{x}) , \hat{\pi}_{r_2}(\boldsymbol{x}) , \ldots,  \hat{\pi}_{r_R}(\boldsymbol{x})  \right)  
\\ &   -   f_r \left(  \pi_{r_1}(\boldsymbol{x}) , \pi_{r_2}(\boldsymbol{x}) , \ldots,  \pi_{r_R} (\boldsymbol{x})  \right)  \big)     (  \boldsymbol{\theta}_N^*)^\top  \left( \boldsymbol{\hat{\psi}}_{rt}^* 
-   \boldsymbol{\psi}_{rt}^*  \right) 
\\ = ~ & \sqrt{N}  \mathcal{O}_{\mathbb{P}}\left( \frac{1}{N} \right) 
,
\end{align*}
and
\begin{align*}
& \sqrt{NT} (  \boldsymbol{\theta}_N^*)^\top \sum_{r \in \mathcal{R}} \big( f_r \left(  \hat{\pi}_{r_1}(\boldsymbol{x}) , \hat{\pi}_{r_2}(\boldsymbol{x}) , \ldots,  \hat{\pi}_{r_R}(\boldsymbol{x})  \right)  
 \\ &  -   f_r \left(  \pi_{r_1}(\boldsymbol{x}) , \pi_{r_2}(\boldsymbol{x}) , \ldots,  \pi_{r_R} (\boldsymbol{x})  \right)  \big)   \sum_{t =r}^T \psi_{rt} \left( \boldsymbol{\hat{\psi}}_{rt}^* 
-   \boldsymbol{\psi}_{rt}^*  \right) \xrightarrow{p} 0 
.
\end{align*}
\end{itemize}
Putting this together and using the assumed independence, we have
\[
\sqrt{NT}  (  \boldsymbol{\theta}_N^*)^\top  \left(  \boldsymbol{\hat{\psi}}_N^{(\text{C;CATT,} \pi)} - \boldsymbol{\psi}_N^{(\text{C;CATT,} \pi)} \right) \xrightarrow{d} \mathcal{N} \left(0, v_R^{(\text{C;CATT,} \pi)} \right)
\]
for
\[
 v_R^{(\text{C;CATT,} \pi)} :=  T \cdot \sum_{r \in \mathcal{R}}   \frac{(\boldsymbol{\psi}_{r}^{\text{CATT; prob}})^\top  \Cov ( \boldsymbol{X}_i \mid W_i = r) \boldsymbol{\psi}_{r}^{\text{CATT; prob}}}{\mathbb{P}(W_i = r)} + T \cdot \boldsymbol{\psi}_\pi^\top  \nabla \boldsymbol{f}^\top  \boldsymbol{\Sigma}_\pi \nabla \boldsymbol{f} \boldsymbol{\psi}_\pi
.
\]
Finally, consider the plug-in estimator
\begin{equation}\label{v_r_c_catt_pi_est_form}
 \hat{v}_R^{(\text{C;CATT,} \pi)} :=  T \cdot \sum_{r \in \mathcal{R}}   \frac{(\boldsymbol{\hat{\psi}}_{r}^{\text{CATT; prob}})^\top  \widehat{\Cov} ( \boldsymbol{X}_i \mid W_i = r) \boldsymbol{\hat{\psi}}_{r}^{\text{CATT; prob}}}{N_r / N_\tau} + T \cdot \boldsymbol{\hat{\psi}}_\pi^\top  \nabla \boldsymbol{\hat{f}}^\top  \boldsymbol{\hat{\Sigma}}_\pi \nabla \boldsymbol{\hat{f}} \boldsymbol{\hat{\psi}}_\pi
,
\end{equation}
where
\begin{align*}
\boldsymbol{\hat{\psi}}_{r}^{\text{CATT; prob}} := ~ &  \frac{\hat{\pi}_r(\boldsymbol{x})}{ \sum_{r' \in \mathcal{R}}\hat{\pi}_{r'}(\boldsymbol{x})}  \sum_{t = r}^T   \psi_{rt} \boldsymbol{\hat{\rho}}_{rt}^{(q)} , \qquad r \in \mathcal{R},
\\ \boldsymbol{\hat{\psi}}_\pi := ~ & \bigg( \sum_{t=r_1}^T \psi_{r_1t} \left( \boldsymbol{\hat{\tau}}_{r_1t}^{(q)} + \left(\boldsymbol{x} - \boldsymbol{\overline{X}}_{r_1} \right)^\top \boldsymbol{\hat{\rho}}_{r_1t}^{(q)} \right), \ldots,
\\ &  \sum_{t=r_R}^T \psi_{r_Rt} \left( \boldsymbol{\hat{\tau}}_{r_Rt}^{(q)} + \left(\boldsymbol{x} -  \boldsymbol{\overline{X}}_{r_R}  \right)^\top \boldsymbol{\hat{\rho}}_{r_Rt}^{(q)} \right) \bigg)   , 
%
%\\ \hat{\nabla}_\pi := ~ &
%\begin{pmatrix}
%0 & 0& \cdots & 0
%%
%\\ \frac{\sum_{r' \in \mathcal{R} \setminus r_1} \hat{\pi}_{r'}(\boldsymbol{x}) }{\left(  \sum_{r' \in \mathcal{R}} \hat{\pi}_{r'}(\boldsymbol{x}) \right)^2} & \frac{-\hat{\pi}_{r_2}(\boldsymbol{x}) }{\left(  \sum_{r' \in \mathcal{R}} \hat{\pi}_{r'}(\boldsymbol{x}) \right)^2} & \cdots & \frac{-\hat{\pi}_{r_R}(\boldsymbol{x}) }{\left(  \sum_{r' \in \mathcal{R}} \hat{\pi}_{r'}(\boldsymbol{x}) \right)^2}
%%
%\\ \frac{-\hat{\pi}_{r_1}(\boldsymbol{x}) }{\left(  \sum_{r' \in \mathcal{R}} \hat{\pi}_{r'}(\boldsymbol{x}) \right)^2} &  \frac{\sum_{r' \in \mathcal{R} \setminus r_2} \hat{\pi}_{r'}(\boldsymbol{x}) }{\left(  \sum_{r' \in \mathcal{R}} \hat{\pi}_{r'}(\boldsymbol{x}) \right)^2}  & \cdots & \frac{-\hat{\pi}_{r_R}(\boldsymbol{x}) }{\left(  \sum_{r' \in \mathcal{R}} \hat{\pi}_{r'}(\boldsymbol{x}) \right)^2}
%%
%\\ \vdots & \vdots & \ddots & \vdots
%%
%\\  \frac{-\hat{\pi}_{r_1}(\boldsymbol{x}) }{\left(  \sum_{r' \in \mathcal{R}} \hat{\pi}_{r'}(\boldsymbol{x}) \right)^2} & \frac{-\hat{\pi}_{r_2}(\boldsymbol{x}) }{\left(  \sum_{r' \in \mathcal{R}} \hat{\pi}_{r'}(\boldsymbol{x}) \right)^2} & \cdots & \frac{\sum_{r' \in \mathcal{R} \setminus r_R} \hat{\pi}_{r'}(\boldsymbol{x}) }{\left(  \sum_{r' \in \mathcal{R}} \hat{\pi}_{r'}(\boldsymbol{x}) \right)^2} 
%\end{pmatrix}.
\end{align*}
and \(\nabla \boldsymbol{\hat{f}}\) is defined by replacing each instance of \(\pi_r(\boldsymbol{x})\) in \(\nabla \boldsymbol{f}\) with its estimator \(\hat{\pi}_r(\boldsymbol{x})\). Notice that \(\hat{v}_R^{(\text{C;CATT,} \pi)}  \xrightarrow{p} v_R^{(\text{C;CATT,} \pi)} \) by the consistency of the individual terms and the continuous mapping theorem.
\end{proof}

\section{Proofs of Lemmas From Proof of Theorem \ref{first.thm.fetwfe}}\label{sec.main.lemmas}

\begin{proof}[Proof of Lemma \ref{d.sing.val.lem}]
Notice from \eqref{d.expres} and the expressions in the proof of Lemma \ref{lem.d.express} that the \(\ell_1\) norm of each row or column of \(\boldsymbol{D}_N\) is at most 3, so we have
\[
 \sigma_{\text{max}} (\boldsymbol{D}_N)  \leq   \sqrt{  \lVert \boldsymbol{D}_N \rVert_{1} \lVert \boldsymbol{D}_N \rVert_{\infty} } \leq \sqrt{3 \cdot 3} = 3
 ,
\]
where for matrices \( \lVert \cdot \rVert_1\) denotes the maximum absolute column sum of a matrix and \(\lVert \cdot \rVert_\infty\) denotes the maximum absolute row sum. Now we lower bound \(\sigma_{\text{min}} \left( \boldsymbol{D}_N \right)\). We have that for any matrix \(\boldsymbol{A}\)
\begin{equation}\label{sv.lb}
\sigma_{\text{min}} \left( \boldsymbol{A} \right) \geq \frac{1}{\sqrt{\lVert \boldsymbol{A}^{-1} \rVert_{1} \lVert \boldsymbol{A}^{-1} \rVert_{\infty}}}
.
\end{equation}
Recall our expression for \(\boldsymbol{D}_N^{-1} \) from \eqref{d.inv.exp}.
The singular values of a block diagonal matrix are the union of the singular values of the individual matrices on the diagonal, so it is enough to lower bound the singular values of each of these individual matrices using \eqref{sv.lb} and then take the minimum of all of the lower bounds.

Using the expression for \(\boldsymbol{D}^{(1)}(t)^{-1}\) from \eqref{d.1.inv.exp} and \eqref{sv.lb},
\begin{equation}\label{d.1.sig.min}
\lVert  \boldsymbol{D}^{(1)}(t)^{-1} \rVert_{1} \lVert \boldsymbol{D}^{(1)}(t)^{-1} \rVert_{\infty} \leq t^2 \qquad \implies \qquad  \sigma_{\text{min}} \left(  \boldsymbol{D}^{(1)}(t)^{-1} \right) \geq 1/t
.
\end{equation}

Next, using the expression for \( (\boldsymbol{D}^{(2)}(\mathcal{R}))^{-1} \) from \eqref{d.2.inv.exp}, observe that the maximum column \(\ell_1\) norm of \( (\boldsymbol{D}^{(2)}(\mathcal{R}))^{-1} \) is \(\mathfrak{W}\), the \(\ell_1\) norm of the first column. Similarly, we see that the maximum row \(\ell_1\) norm, in the bottom block of \( (\boldsymbol{D}^{(2)}(\mathcal{R}))^{-1} \), is at most \( R - 1 + T - r_R  + 1  \leq R + T\), because each of the \(R- 1\) blocks of \(\boldsymbol{\tilde{U}}\) matrices has a single 1 in each row and \( \left( \boldsymbol{D}^{(1)}(T - r_R + 1)^{-1}  \right)^\top\) has a row of \(T - r_R + 1\) ones. So
\begin{align*}
 \lVert (\boldsymbol{D}^{(2)}(\mathcal{R}))^{-1}  \rVert_{1} \lVert (\boldsymbol{D}^{(2)}(\mathcal{R}))^{-1} \rVert_{\infty} \leq  ~ & \mathfrak{W}(T+R)
\\ \implies \qquad   \sigma_{\text{min}} \left(  (\boldsymbol{D}^{(2)}(\mathcal{R}))^{-1} \right) \geq ~ & \frac{1}{\sqrt{\mathfrak{W}(T+R)}}
.
\end{align*}
Using the fact that \(\mathfrak{W} \leq (T-1)^2\) from Lemma \ref{lem.d.express}, we have
\begin{equation}\label{d.2.sig.min}
  \sigma_{\text{min}} \left(  (\boldsymbol{D}^{(2)}(\mathcal{R}))^{-1} \right) \geq \frac{1}{(T-1)\sqrt{T+R}} \geq \frac{1}{(T-1)\sqrt{2T-1}}
.
\end{equation}
Finally, using \eqref{d.1.sig.min} and \eqref{d.2.sig.min}, \(T \geq 2\), and the fact that the smallest singular value of the identity matrix is 1, we have
\begin{align*}
\sigma_{\text{min}}\left( \boldsymbol{D}_N  \right) \geq ~ &  \min \left\{ \frac{1}{R}, \frac{1}{T-1},1, \frac{1}{(T-1)\sqrt{2T-1}} \right\}   \geq \frac{1}{(T-1)\sqrt{2T-1}} 
\geq    \frac{1}{T\sqrt{2T}}
.
\end{align*}

\end{proof}

\begin{proof}[Proof of Lemma \ref{power.lem}]
We will first show that \(\boldsymbol{X}_i\) having finite fourth moments under Assumption (R1) is enough for \(\boldsymbol{\tilde{Z}}_{(i \cdot) \cdot} \) to have finite fourth moments as well. We will show that each entry of \( \boldsymbol{\tilde{Z}}_{(i \cdot) \cdot} \boldsymbol{D}_N^{-1}\) is a bounded linear combination of a finite number of elements of \(\boldsymbol{\tilde{Z}}_{(i \cdot) \cdot}\), and likewise for \(\boldsymbol{G}_N  \boldsymbol{\tilde{Z}}_{(i \cdot) \cdot} \boldsymbol{D}_N^{-1}\). Then the finite moments will follow using our Assumption (R1), Cauchy-Schwarz, and H\"{o}lder's inequality.

Besides \(\boldsymbol{X}_i\), the entries of \(\boldsymbol{\tilde{Z}}_i\) consist of indicator variables and interactions of the indicator variables with themselves and \(\boldsymbol{X}_i\). All of these terms are upper-bounded by the greater of 1 or \(\boldsymbol{X}_i\). Using this, our assumption that \(\E[X_{(i)j}^4] \) is finite for all \(j \in [p_N]\), and the fact that \(\boldsymbol{X}_i\) is time-invariant, we have that for any fixed \(N\), \(\max_{j \in [d_N], t \in [T]} \{\E[Z_{(it)j}^4] \}\) is finite.

Now we consider the entries of \(\boldsymbol{\tilde{Z}}_{(i \cdot) \cdot} \boldsymbol{D}_N^{-1} \in \mathbb{R}^{NT \times p_N}\). The columns of \(\boldsymbol{\tilde{Z}}_{(i \cdot) \cdot} \boldsymbol{D}_N^{-1}\) are linear combinations of the columns of \(\boldsymbol{\tilde{Z}}_{(i \cdot) \cdot}\); in particular, the \(j^{\text{th}}\) column of \(\boldsymbol{\tilde{Z}}_{(i \cdot) \cdot} \boldsymbol{D}_N^{-1}\) is \(\boldsymbol{\tilde{Z}}_{(i \cdot) \cdot}  (\boldsymbol{D}_N^{-1})_{\cdot j}  \in \mathbb{R}^{NT}\), where \((\boldsymbol{D}_N^{-1})_{\cdot j} \) is the \(j^{\text{th}}\) column of \(\boldsymbol{D}_N^{-1}\). From \eqref{d.inv.exp} and the expressions in the proof of Lemma \ref{lem.d.express} we see that \(\boldsymbol{D}_N^{-1}\) contains only entries that equal 1 or 0, so for any \(t \in [T]\) and \(j \in [p_N]\),
\[
(\boldsymbol{\tilde{Z}}_{(i \cdot) \cdot} \boldsymbol{D}_N^{-1})_{(1t)j}^4 =  \left[\left(\boldsymbol{\tilde{Z}}_{(i \cdot) \cdot}   (\boldsymbol{D}_N^{-1})_{\cdot j} \right)_{(1t)} \right]^4  =  \left( \sum_{\{j': (\boldsymbol{D}_N^{-1})_{j'j} \neq 0\}}    Z_{(1t)j'}  \right)^4
 .
\]
From \eqref{d.inv.exp} and the expressions in the proof of Lemma \ref{lem.d.express} we have that \(\boldsymbol{D}_N^{-1}\) has a block diagonal structure with blocks of size at most \(\mathfrak{W} \leq (T - 1)^2\) (an upper bound we established in Lemma \ref{lem.d.express}), so this sum contains at most \((T-1)^2\) terms before expansion. By the multinomial theorem, the number of terms after expansion is therefore at most
\[
\binom{4 + (T - 1)^2 - 1}{(T - 1)^2 - 1}
.
\]
 \(T\) is fixed and finite, so the number of terms after expansion is finite as well. Therefore if we can show that each summand after expansion has finite \(4^{\text{th}}\) moment, we have shown that each entry of \(\boldsymbol{\tilde{Z}}_{(i \cdot) \cdot}   \boldsymbol{D}_N^{-1}\) has finite \(4^{\text{th}}\) moment.

This sum when expanded contains terms of the form \(\tilde{Z}_{(1t)j'}^4\), \(\tilde{Z}_{(1t)j'}^2 \tilde{Z}_{(1t)j''}^2 \) for \(j' \neq j''\), \(\tilde{Z}_{(1t)j'}^3 \tilde{Z}_{(1t)j''} \), \(\tilde{Z}_{(1t)j'}^2 \tilde{Z}_{(1t)j''} \tilde{Z}_{(1t)j'''} \) and \(\tilde{Z}_{(1t)j'} \tilde{Z}_{(1t)j''} \tilde{Z}_{(1t)j'''} \tilde{Z}_{(1t)j''''} \) (disregarding constant factors). We demonstrate that each of these terms has finite fourth moment below using Cauchy-Schwarz and H\"{o}lder's inequality:

\begin{itemize}

\item We have already shown that \(\E[\tilde{Z}_{(1t)j'}^4]\) is finite.

\item \(\E[ \tilde{Z}_{(1t)j'}^2 \tilde{Z}_{(1t)j''}^2 ] \leq \sqrt{\E[ \tilde{Z}_{(1t)j'}^4] \E[ \tilde{Z}_{(1t)j''}^4 ] }\), and terms of this form are finite using the finite fourth moments of \(\boldsymbol{\tilde{Z}}_{(i \cdot) \cdot}\).

\item Recall that by H\"{o}lder's inequality, for any \(p, q > 1\) such that \(1/p + 1/q  =1\), if \(X\) and \(Y\) are random variables with \( \left(\E |X|^p\right)^{1/p}  \) and \(  \left(\E|Y|^q \right)^{1/q}\) finite then
\[
\E |XY|\leq \left(\E |X|^p\right)^{1/p} \left(\E|Y|^q \right)^{1/q} %= \lVert X \rVert_p \lVert Y \rVert_q
.
\]
Choosing \(p=4/3\) and \(q = 4\), we have
\begin{align*}
\E[\tilde{Z}_{(1t)j'}^3 \tilde{Z}_{(1t)j''}] \leq ~ & \E[|\tilde{Z}_{(1t)j'}|^3 |\tilde{Z}_{(1t)j''}| ] \leq   \left(\E \left[(|\tilde{Z}_{(1t)j'}|^3)^{4/3} \right] \right)^{3/4} \left(\E \left[ |\tilde{Z}_{(1t)j''}|^4 \right] \right)^{1/4}
\\ = ~ &  \left(\E |\tilde{Z}_{(1t)j'}|^4 \right)^{3/4} \left(\E|\tilde{Z}_{(1t)j''}|^4 \right)^{1/4}
,
\end{align*}
which is finite.

\item \(\E[\tilde{Z}_{(1t)j'}^2 \tilde{Z}_{(1t)j''} \tilde{Z}_{(1t)j'''}]  \leq \sqrt{\E[\tilde{Z}_{(1t)j'}^4]} \sqrt{\E[ \tilde{Z}_{(1t)j''}^2 \tilde{Z}_{(1t)j'''}^2 } ]\), and we already showed that terms of this form are finite.

\item \(\E[\tilde{Z}_{(1t)j'} \tilde{Z}_{(1t)j''} \tilde{Z}_{(1t)j'''} \tilde{Z}_{(1t)j''''}  ] \leq \sqrt{\E[\tilde{Z}_{(1t)j'}^2 \tilde{Z}_{(1t)j''}^2] \E[\tilde{Z}_{(1t)j'''}^2 \tilde{Z}_{(1t)j''''}^2 ]} \), and we already showed that terms of this form are finite.

\end{itemize}

\end{proof}

\begin{proof}[Proof of Lemma \ref{sig.min.lem}] 
Recall that for any matrix \(\boldsymbol{A}\), the eigenvalues of \(\boldsymbol{A}^\top \boldsymbol{A}\) are the squares of the singular values of \(\boldsymbol{A}\). So \(\lambda_{\text{min}} \left( \frac{1}{NT}  \left( \boldsymbol{Z} \boldsymbol{D}_N^{-1} \right)^\top \boldsymbol{Z} \boldsymbol{D}_N^{-1}  \right) =   \sigma_{\text{min}} \left(\boldsymbol{Z} \boldsymbol{D}_N^{-1} \right)^2/(NT)\) and 
\[
\lambda_{\text{max}} \left( \frac{1}{NT}  \left( \boldsymbol{Z} \boldsymbol{D}_N^{-1} \right)^\top \boldsymbol{Z} \boldsymbol{D}_N^{-1}  \right) =    \sigma_{\text{max}} \left(\boldsymbol{Z} \boldsymbol{D}_N^{-1} \right)^2/(NT)
,
\]
and it is enough to bound the singular values of \(\boldsymbol{Z} \boldsymbol{D}_N^{-1} \), which we can do with Lemma \ref{lem.sing.value.prod.bound}:
\begin{lemma}\label{lem.sing.value.prod.bound}
Let \(\boldsymbol{A}\) and \(\boldsymbol{B}\) be conformable real matrices.

\begin{enumerate}[(a)]

\item  \citep{4018462} Suppose \(\boldsymbol{A} \neq \boldsymbol{0}\) and \(\boldsymbol{B}\) has full row rank. Then
\[
\sigma_{\text{min}} \left(\boldsymbol{A}\boldsymbol{B} \right) \geq \sigma_{\text{min}} \left(\boldsymbol{A} \right) \sigma_{\text{min}} \left(\boldsymbol{B} \right) 
.
\]

\item 
\[
 \sigma_{\text{max}} \left(\boldsymbol{A}\boldsymbol{B} \right) \leq \sigma_{\text{max}} \left(\boldsymbol{A} \right) \sigma_{\text{max}} \left(\boldsymbol{B} \right) 
.
\]
\end{enumerate}
\end{lemma}
We bounded the singular values of the deterministic \(\boldsymbol{D}_N^{-1}\) in Lemma \ref{d.sing.val.lem}. Using that and the fact that \(  \frac{1}{NT} \sigma_{\text{min}} \left(\boldsymbol{Z} \right)^2  = \lambda_{\text{min}} \left(  \frac{1}{NT}  \boldsymbol{Z}^\top \boldsymbol{Z} \right) = e_{1N} \), we have
\begin{align*}
\lambda_{\text{min}} \left(  \frac{1}{NT}  \left( \boldsymbol{Z} \boldsymbol{D}_N^{-1} \right)^\top  \boldsymbol{Z} \boldsymbol{D}_N^{-1}  \right)   \geq  ~ &   \frac{1}{NT}  \sigma_{\text{min}} \left(\boldsymbol{Z} \right)^2   \sigma_{\text{min}}  \left( \boldsymbol{D}_N^{-1} \right)^2 
\\ \geq  ~ &   \lambda_{\text{min}} \left(  \frac{1}{NT} \boldsymbol{\tilde{Z}}^\top \boldsymbol{\tilde{Z}} \right)  \left( \frac{1}{3} \right)^2
\\ \geq ~ &   \frac{1}{9} e_{1N} 
,
\end{align*}
where we used that from Lemma \ref{d.sing.val.lem},
\begin{align*}
\sigma_{\text{max}} (\boldsymbol{D}_N)  \leq    3
\qquad \iff \qquad \sigma_{\text{min}} (\boldsymbol{D}_N^{-1})  \geq     \frac{1}{3}
.
\end{align*}
Similarly,
\begin{align*}
\lambda_{\text{max}} \left(  \frac{1}{NT}  \left( \boldsymbol{Z} \boldsymbol{D}_N^{-1} \right)^\top \boldsymbol{Z} \boldsymbol{D}_N^{-1}  \right) \leq  ~ &    \lambda_{\text{max}} \left(  \frac{1}{NT} \boldsymbol{\tilde{Z}}^\top \boldsymbol{\tilde{Z}} \right)   \left(T \sqrt{2T}\right)^2
 \leq  e_{2N}  \cdot 2T^3
,
\end{align*}
where we used that from Lemma \ref{d.sing.val.lem},
\begin{align*}
\sigma_{\text{min}} (\boldsymbol{D}_N)  \geq     \frac{1}{T \sqrt{2T}}
\qquad \iff \qquad \sigma_{\text{max}} (\boldsymbol{D}_N^{-1})  \leq    T \sqrt{2T}
.
\end{align*}
\end{proof}

\begin{proof}[Proof of Theorem \ref{prop.2i}]

For clarity, we use the notation of \citet{kock2013oracle}. To extend Lemma 3 from \citet{kock2013oracle}, we use a similar strategy to the one used to prove Theorem 11.3(d) in \citet{hastie2015statistical}. Since under Assumption (A5) from \citet{kock2013oracle} we have \(\min_{ 1 \leq j \leq k_N}\{|\beta_{10j|}\} \geq b_0\), it is enough to show that with high probability \( \hat{\beta}_{2N} = \boldsymbol{0}\) and \(\lVert \hat{\beta}_N - \beta_N \rVert_\infty  \leq b_0/2\). That is,
\begin{align*}
\left\{ \hat{\beta}_{2N} = \boldsymbol{0} \right\} \cap  \left\{ \left\lVert \hat{\beta}_N - \beta_N \right\rVert_\infty  \leq \frac{b_0}{2} \right\}   \subseteq  ~ & \left\{ \hat{\mathcal{S}} = \mathcal{S}  \right\}
\\ \implies \qquad  \left\{ \hat{\mathcal{S}} \neq \mathcal{S}  \right\} \subseteq ~ & \left\{ \hat{\beta}_{2N} \neq \boldsymbol{0} \right\} \cup  \left\{ \left\lVert \hat{\beta}_N - \beta_N \right\rVert_\infty  > \frac{b_0}{2} \right\} 
\\ \stackrel{(a)}{\implies} \qquad   \mathbb{P} \left( \hat{\mathcal{S}} \neq \mathcal{S}  \right) \leq ~ &  \mathbb{P} \left( \left\{ \hat{\beta}_{2N} \neq \boldsymbol{0} \right\} \cup  \left\{ \left\lVert \hat{\beta}_N - \beta_N \right\rVert_2  > \frac{b_0}{2} \right\}  \right)
\\ \leq ~ &  \mathbb{P} \left( \hat{\beta}_{2N} \neq \boldsymbol{0} \right) + \mathbb{P} \left( \left\lVert \hat{\beta}_N - \beta_N \right\rVert_2  > \frac{b_0}{2}   \right)
,
\end{align*}
where in \((a)\) we used \(\lVert \cdot \rVert_\infty \leq \lVert \cdot \rVert_2\). Lemma 3 from \citet{kock2013oracle} establishes that 
\[
\lim_{N \to \infty} \mathbb{P} \left(  \hat{\beta}_{2N} \neq \boldsymbol{0}  \right) = 0
,
\]
and Theorem 1 in \citet{kock2013oracle} (the assumptions of which are satisfied under the assumptions of Lemma 3 in \citealt{kock2013oracle}) establishes the consistency of \(\hat{\beta}_N\), which is sufficient for
\[
\lim_{N \to \infty} \mathbb{P} \left( \left\lVert \hat{\beta}_N - \beta_N \right\rVert_2  > \frac{b_0}{2}   \right) = 0
.
\]
So
\begin{align*}
0 \leq \lim_{N \to \infty}  \mathbb{P} \left( \hat{\mathcal{S}} \neq \mathcal{S}  \right) \leq ~ &  \lim_{N \to \infty} \mathbb{P} \left( \hat{\beta}_{2N} \neq \boldsymbol{0} \right) +  \lim_{N \to \infty} \mathbb{P} \left( \left\lVert \hat{\beta}_N - \beta_N \right\rVert_2  > \frac{b_0}{2}   \right)  = 0
,
\end{align*}
proving the result.

\end{proof}

\begin{proof}[Proof of Theorem \ref{prop.ext.2}] 
 \begin{enumerate}[(a)]
 
 \item Again, for clarity we use the notation of \citet{kock2013oracle}. Since
 \begin{align*}
\sqrt{N T_N} \boldsymbol{\psi}_N^\top ( \hat{\beta}_N - \beta_0) =  ~ &  \sqrt{N T_N} \boldsymbol{\alpha}^\top ( \hat{\beta}_{1N}- \beta_{10}) + \sqrt{N T_N} \boldsymbol{b}_N^\top ( \hat{\beta}_{2N}- \beta_{20})
 \\  =  ~ &  \sqrt{N T_N}  \boldsymbol{\alpha}^\top ( \hat{\beta}_{1N}- \beta_{10}) + \sqrt{N T_N}  \boldsymbol{b}_N^\top \hat{\beta}_{2N}
 ,
 \end{align*}
 if we can show that 
 \[
\sqrt{N T_N}   \boldsymbol{b}_N^\top \hat{\beta}_{2N} \xrightarrow{p} 0
 \]
 then the result follows from Theorem 2(ii) from \citet{kock2013oracle} (in particular, Equation 7.5 in the proof) and Slutsky's theorem. Theorem 2(i) from \citet{kock2013oracle} establishes that 
\begin{align*}
\lim_{N \to \infty} \mathbb{P} \left( \hat{\beta}_{2N}= \boldsymbol{0} \right) =  1
.
\end{align*}
Since for any \(N\)
\[
\left\{ \hat{\beta}_{2N} = \boldsymbol{0}\right\} \subseteq\left\{  \sqrt{N T_N}  \sum_{j =1}^{p_N - k}  \left| (\boldsymbol{b}_N )_j   (\hat{\beta}_{2N} )_j  \right| = 0 \right\} 
,
\]
we have
\begin{align*}
\mathbb{P} \left(  \hat{\beta}_{2N}= \boldsymbol{0} \right) \leq  \mathbb{P} \left(  \sqrt{N T_N} \sum_{j =1}^{p_N - k}  \left| (\boldsymbol{b}_N )_j   (\hat{\beta}_{2N} )_j \right| = 0 \right)
 ,
\end{align*}
so
\[
1 \geq \lim_{N \to \infty} \mathbb{P} \left(  \sqrt{N T_N}\sum_{j =1}^{p_N - k}   \left| (\boldsymbol{b}_N )_j \hat{\beta}_{2N}   \right|  = 0 \right) \geq  \lim_{N \to \infty} \mathbb{P} \left(  \hat{\beta}_{2N}= \boldsymbol{0} \right) 
 =  1
 ,
\]
establishing
\[
\lim_{N \to \infty} \mathbb{P} \left(  \sqrt{N T_N}\sum_{j =1}^{p_N - k}   \left| (\boldsymbol{b}_N )_j (\hat{\beta}_{2N})_j   \right|  = 0 \right)  = 1
.
\]
Since this notion of convergence is stronger than convergence in probability, it follows that
\begin{align*}
\sqrt{N T_N}\sum_{j =1}^{p_N - k}   (\boldsymbol{b}_N )_j  (\hat{\beta}_{2N} )_j \xrightarrow{p}  ~ & 0
%%
%\\ \iff \qquad (NT_N)^{1/2} s_N^{-1} \boldsymbol{b}_N^\top \hat{\beta}_{2N} \xrightarrow{p} ~ &  0
,
\end{align*}
proving the result.

\item 

Since by the consistency of the sample covariance matrix and the continuous mapping theorem
\[
\frac{1}{\sigma}  
\sqrt{ \frac{T}{\boldsymbol{\alpha}(\mathcal{S})^\top \left( \boldsymbol{\hat{\Sigma}}( \boldsymbol{X}_{(\cdot \cdot)\mathcal{S}} ) \right)^{-1} \boldsymbol{\alpha}(\mathcal{S})}}  \xrightarrow{p} 
\frac{1}{\sigma}  
\sqrt{ \frac{T}{ \boldsymbol{\alpha}(\mathcal{S})^\top \left( \lim_{N \to \infty}  \E \left[ \Sigma_{1N} \right] \right)^{-1} \boldsymbol{\alpha}(\mathcal{S})  }}
,
\]
by part (a) and Slutsky's theorem we have that \(U_b(\mathcal{S} )\) converges in distribution to a standard Gaussian random variable. Further, \(\mathbb{P}( U_b(\hat{\mathcal{S}}_N) - U_b(\mathcal{S})  =   0 ) \geq \mathbb{P}( \hat{\mathcal{S}}_N  = \mathcal{S} ) \), so from Theorem \ref{prop.2i} we have
\[
1 \geq \lim_{N \to \infty}  \mathbb{P}( U_b(\hat{\mathcal{S}}_N) - U_b(\mathcal{S})  =   0 )  \geq \lim_{N \to \infty}  \mathbb{P}( \hat{\mathcal{S}}_N  = \mathcal{S} ) = 1
,
\]
which implies \(U_b(\hat{\mathcal{S}}_N) - U_b(\mathcal{S})  \xrightarrow{p} 0 \). Therefore by Slutsky's theorem,
\[
U_b(\hat{\mathcal{S}}_N) = U_b(\mathcal{S}) + \left( U_b(\hat{\mathcal{S}}_N) - U_b(\mathcal{S}) \right) \xrightarrow{d} \mathcal{N}(0, 1)
.
\]

\item  

We begin by decomposing our sequence of random variables:
\begin{align}
%U_2(\mathcal{A}, \mathcal{A}_2) = ~
U_c(\mathcal{A}) = ~ &  \sqrt{NT}  \left( \boldsymbol{\hat{\psi}}_N(\mathcal{A})^\top \hat{\beta} -   \boldsymbol{\psi}_N(\mathcal{A})^\top \beta_0 \right) \nonumber
\\ = ~ &   \sqrt{NT}  \left( \boldsymbol{\psi}_N(\mathcal{A})^\top \hat{\beta}  -   \boldsymbol{\psi}_N(\mathcal{A})^\top \beta_0  + \left( \boldsymbol{\hat{\psi}}_N(\mathcal{A}) - \boldsymbol{\psi}_N(\mathcal{A}) \right)^\top \hat{\beta}  \right) \nonumber
\\ = ~ &  \sqrt{NT}  \bigg( \boldsymbol{\psi}_N(\mathcal{A})^\top \left( \hat{\beta}  -    \beta_0  \right) + \left( \boldsymbol{\hat{\psi}}_N(\mathcal{A}) - \boldsymbol{\psi}_N(\mathcal{A}) \right)^\top \beta_0 \nonumber
\\ &  + \left( \boldsymbol{\hat{\psi}}_N(\mathcal{A}) - \boldsymbol{\psi}_N(\mathcal{A}) \right)^\top \left( \hat{\beta} - \beta_0  \right) \bigg)  \label{tilde.u3.exp}
.
\end{align}
We examine these terms one at a time. By part (a),
 \[
  \sqrt{ NT}  \boldsymbol{\psi}_N(\mathcal{S})^\top ( \hat{\beta}_N - \beta_0) \xrightarrow{d} \mathcal{N}\left( 0, \sigma^2 \boldsymbol{\alpha}^\top \left( \lim_{N \to \infty}  \E \left[ \Sigma_{1N} \right] \right)^{-1} \boldsymbol{\alpha} \right)
 .
 \]
By assumption, for any \(\beta_{0} \in \mathbb{R}^{p_N}\) with \(k\) entries nonzero and finite and the remaining entries equal to 0,
\[
\sqrt{NT}  \left(  \boldsymbol{\hat{\psi}}_N(\mathcal{S})  -   \boldsymbol{\psi}_N(\mathcal{S})\right)^\top \beta_{0} \xrightarrow{d} \mathcal{N}(0, v_\psi(\beta_0, \mathcal{S}))
.
\]
Finally, similarly to \cite{kock2013oracle} we partition \( \boldsymbol{\hat{\psi}}_N(\mathcal{A})  \) into component \(\boldsymbol{\hat{\psi}}_{1N}(\mathcal{A}) \) that corresponds to the nonzero coefficients \(\beta_{10} \in \mathbb{R}^k\) and component \(\boldsymbol{\hat{\psi}}_{2N}(\mathcal{A}) \) corresponding to \(\beta_{20} = \boldsymbol{0} \in \mathbb{R}^{p_N - k}\), and likewise for \( \boldsymbol{\psi}_N(\mathcal{A}) \). Then we have
\begin{align*}
& \left( \boldsymbol{\hat{\psi}}_N(\mathcal{A}) - \boldsymbol{\psi}_N(\mathcal{A}) \right)^\top \left( \hat{\beta} - \beta_0  \right) 
\\ = ~ & \left( \boldsymbol{\hat{\psi}}_{1N}(\mathcal{A}) - \boldsymbol{\psi}_{1N}(\mathcal{A}) \right)^\top \left( \hat{\beta}_{1N} - \beta_{10}  \right) +  \left( \boldsymbol{\hat{\psi}}_{2N}(\mathcal{A}) - \boldsymbol{\psi}_{2N}(\mathcal{A}) \right)^\top \left( \hat{\beta}_{2N} - \beta_{20}  \right)
\\ = ~ & \left( \boldsymbol{\hat{\psi}}_{1N}(\mathcal{A}) - \boldsymbol{\psi}_{1N}(\mathcal{A}) \right)^\top \left( \hat{\beta}_{1N} - \beta_{10}  \right) +  \left( \boldsymbol{\hat{\psi}}_{2N}(\mathcal{A}) - \boldsymbol{\psi}_{2N}(\mathcal{A}) \right)^\top \left( \hat{\beta}_{2N}  \right)
.
\end{align*}
For the first term, applying Theorem 2.4(i) in \citet{van2000asymptotic} and Proposition 1.8(iii) in \citet{Garcia-Portugues2023} termwise we have that
\begin{align*}
 \left( \boldsymbol{\hat{\psi}}_{1N}(\mathcal{A}) - \boldsymbol{\psi}_{1N}(\mathcal{A}) \right)^\top \left( \hat{\beta}_{1N} - \beta_{10}  \right) = ~ & \mathcal{O}_{\mathbb{P}} \left( \frac{1}{\sqrt{N}} \right)  \mathcal{O}_{\mathbb{P}} \left( \frac{1}{\sqrt{N}} \right)  = \mathcal{O}_{\mathbb{P}} \left( \frac{1}{N} \right)  
 \\ \implies \qquad    \sqrt{NT} \left( \boldsymbol{\hat{\psi}}_{1N}(\mathcal{A}) - \boldsymbol{\psi}_{1N}(\mathcal{A}) \right)^\top \left( \hat{\beta}_{1N} - \beta_{10}  \right) = ~ & \mathcal{O}_{\mathbb{P}} \left( \frac{1}{\sqrt{N}} \right)  
 .
\end{align*}
For the second term, by a similar argument to the one we used in part (a), Theorem 2(i) from \cite{kock2013oracle} gives us that the \(\sqrt{NT} \left( \boldsymbol{\hat{\psi}}_{2N}(\mathcal{A}) - \boldsymbol{\psi}_{2N}(\mathcal{A}) \right)^\top \left( \hat{\beta}_{2N}  \right) \xrightarrow{p} 0\). Putting this together,
\[
\sqrt{NT}   \left( \boldsymbol{\hat{\psi}}_N(\mathcal{A}) - \boldsymbol{\psi}_N(\mathcal{A}) \right)^\top \left( \hat{\beta} - \beta_0  \right)  \xrightarrow{p} 0
.
\]
Putting this all together, using Slutsky's theorem and the independence of \(\hat{\beta}\) and \(\boldsymbol{\hat{\psi}}_N(\mathcal{S})\) we have from \eqref{tilde.u3.exp} that
\[
U_c(\mathcal{S}) \xrightarrow{d} \mathcal{N} \left(0,\sigma^2 \boldsymbol{\alpha}^\top \left( \lim_{N \to \infty}  \E \left[ \Sigma_{1N} \right] \right)^{-1} \boldsymbol{\alpha}  +  v_\psi(\beta_0, \mathcal{S})  \right)
.
\]
Finally, the same argument we used in part (b) yields \(U_c(\hat{\mathcal{S}}) \xrightarrow{d} \mathcal{N}(0,\sigma^2 \boldsymbol{\alpha}^\top \left( \lim_{N \to \infty}  \E \left[ \Sigma_{1N} \right] \right)^{-1} \boldsymbol{\alpha}  +  v_\psi(\beta_0, \mathcal{S})  ) \).

\item  
Using the consistency of the sample covariance matrix, the assumed consistency of \(\boldsymbol{\hat{\alpha}}(\mathcal{S})\), the assumed consistency of \(\hat{v}_\psi(\hat{\beta}_N, \mathcal{S})\), and the continuous mapping theorem,
\begin{align*}
& \frac{1}{\sqrt{\hat{v}_N (\mathcal{S})}}
\\  = ~ &  \frac{1}{\sqrt{\boldsymbol{\hat{\alpha}}(\mathcal{S})^\top \left( \boldsymbol{\hat{\Sigma}}( \boldsymbol{X}_{(\cdot \cdot)\mathcal{S}} ) \right)^{-1} \boldsymbol{\hat{\alpha}}(\mathcal{S}) + \hat{v}_\psi(\hat{\beta}_N, \mathcal{S})}}  
\\ \xrightarrow{p} ~ &   \frac{1}{ \sqrt{\boldsymbol{\alpha}^\top \left( \lim_{N \to \infty}  \E [ \Sigma_{1N} ]  \right)^{-1} \boldsymbol{\alpha} + v_\psi(\beta_0, \mathcal{S}) }} 
 ,
\end{align*}
so part \((c)\) and Slutsky's theorem yield that \(U_d(\mathcal{S}) = \frac{1}{\sqrt{\hat{v}_N (\mathcal{S})}}  U_c(\mathcal{S})
 \xrightarrow{d} \mathcal{N}(0,1) \), where our finite-sample variance estimator \(\hat{v}_N(\mathcal{S})\) was defined in \eqref{var.est.kock}. Finally, the same argument we used in part (b) yields \(U_d(\hat{\mathcal{S}}) \xrightarrow{d} \mathcal{N}(0,1) \).
 
 \item  If \(\hat{\beta}_N\) and \( \boldsymbol{\hat{\psi}}_N(\mathcal{S})  \) are estimated on the same data set, they are not independent. So by the same argument we used in part \((c)\) we still have
 \begin{align*}
 &   \sqrt{NT}  \left( \boldsymbol{\hat{\psi}}_N(\mathcal{S})^\top ( \hat{\beta}_N -   \beta_0 )   + \left(  \boldsymbol{\hat{\psi}}_N(\mathcal{S})  -   \boldsymbol{\psi}_N(\mathcal{S})\right)^\top \beta_0 \right)
\\ \xrightarrow{d} ~ &  \mathcal{N}\left(0,  \sigma^2  \boldsymbol{\alpha}^\top \left( \lim_{N \to \infty}  \E [ \Sigma_{1N} ]  \right)^{-1} \boldsymbol{\alpha} \right) +   \mathcal{N}\left(0, v_\psi(\beta_0, \mathcal{S})\right)
,
\end{align*}
but because the two Gaussian random variables are neither independent nor do we know that they are jointly Gaussian, we cannot say as much about the distribution of their sum. However, due to Lemma \ref{subg.tails} we do know that they are subgaussian.
\begin{lemma}\label{subg.tails}
If \(X\) and \(Y\) are mean-zero subgaussian random variables, then \(X + Y\) is mean zero and subgaussian.
\end{lemma}
\begin{proof} Provided in Appendix \ref{tech.lem.sec}.
\end{proof}

Further, since for any two random variables with finite variance by Cauchy-Schwarz
\begin{align*}
\Var(X + Y) =  ~ & \Var(X) + 2 \Cov(X,Y ) + \Var(Y) 
\\ \leq  ~ &  \Var(X) + 2 \sqrt{\Var(X)  \Var(Y)} + \Var(Y)
,
\end{align*}
we have that the above sequence of random variables converges in distribution to a subgaussian random variable with variance at most
\[
\sigma^2  \boldsymbol{\alpha}^\top \left( \lim_{N \to \infty}  \E [ \Sigma_{1N} ]  \right)^{-1} \boldsymbol{\alpha}  + 2 \sqrt{\sigma^2  \boldsymbol{\alpha}^\top \left( \lim_{N \to \infty}  \E [ \Sigma_{1N} ]  \right)^{-1} \boldsymbol{\alpha}  \cdot v_\psi(\beta_0, \mathcal{S})} + v_\psi(\beta_0, \mathcal{S})
.
\]
Using the assumed consistency of \(\hat{v}_\psi(\hat{\beta}_N, \mathcal{S}) \), we have by the continuous mapping theorem that for \(\hat{v}_N^{(\text{cons})} (\mathcal{A})\) as defined in \eqref{var.upper.bound.subgaus},
\begin{align*}
& \hat{v}_N^{(\text{cons})} (\mathcal{S}) 
\\ =  ~ & \sigma^2  \boldsymbol{\hat{\alpha}}(\mathcal{S})^\top \left(\boldsymbol{\hat{\Sigma}}( \boldsymbol{X}_{(\cdot \cdot)\mathcal{S}} )  \right)^{-1} \boldsymbol{\hat{\alpha}}(\mathcal{S})   + \hat{v}_\psi(\hat{\beta}_N, \mathcal{S})   \nonumber
 + 2 \sqrt{\sigma^2  \boldsymbol{\hat{\alpha}}(\mathcal{S})^\top \left(\boldsymbol{\hat{\Sigma}}( \boldsymbol{X}_{(\cdot \cdot)\mathcal{S}} )  \right)^{-1} \boldsymbol{\hat{\alpha}}(\mathcal{S}) \cdot  \hat{v}_\psi(\hat{\beta}_N, \mathcal{S})}
 \\ \xrightarrow{p} ~ &  \sigma^2  \boldsymbol{\alpha}^\top \left( \lim_{N \to \infty}  \E [ \Sigma_{1N} ]  \right)^{-1} \boldsymbol{\alpha}  + v_\psi(\beta_0, \mathcal{S})
  + 2 \sqrt{\sigma^2  \boldsymbol{\alpha}^\top \left( \lim_{N \to \infty}  \E [ \Sigma_{1N} ]  \right)^{-1} \boldsymbol{\alpha}  \cdot  v_\psi(\beta_0, \mathcal{S})} 
 .
\end{align*}
Therefore by Slutsky's theorem, \(U_e(\mathcal{S})\) converges in distribution to a mean-zero subgaussian random variable with variance at most 1. Finally, once again a similar argument to the one we used in part (b) yields that \(U_e(\hat{\mathcal{S}})\) converges in distribution to the same random variable as \(U_e(\mathcal{S})\).

\end{enumerate}

\end{proof}

\section{Statements and Proofs of Other Results}\label{tech.lem.sec}

%In Section \ref{sec.other.proofs} we present and prove a few other results that were referred to in the paper.
%%our extension of Theorem \ref{unconf.ccts.cts.thm} that shows that unconfoundedness of the untreated potential outcomes is sufficient for conditional and marginal parallel trends assumptions that are insensitive to functional form. 
%
%
%
%
%
%
%
%
%
%\subsection{Other Results}\label{sec.other.proofs}

\begin{proof}[Proof of Lemma \ref{subg.tails}]

The sum \(X + Y\) has mean zero due to linearity of expectation. To show subgaussianity, it is enough to show that \(\E \exp \left\{t(X+Y) \right\}  \leq  \exp \left\{K^2 t^2 \right\}  \) for some finite \(K^2\) for all \(t \in \mathbb{R}\) \citep[Proposition 2.5.2]{vershynin2018high}. Because \(X\) and \(Y\) are subgaussian, we have that \(\E \left[ \exp \left\{t X \right\}  \right]  \leq \exp \left\{K_X^2 t^2 \right\} \) and \(\E \left[ \exp \left\{t Y \right\}  \right]  \leq \exp \left\{K_Y^2 t^2 \right\} \) for all \(t \in \mathbb{R}\). So using these facts and Cauchy-Schwarz, we have for any \(t \in \mathcal{R}\)
\begin{align*}
\E \exp \left\{t(X+Y) \right\} = ~ & \E \left[ \exp \left\{t X \right\}  \exp \left\{t Y \right\}  \right]
\\ \leq ~ &  \sqrt{\E \left[ \exp \left\{2t X \right\}  \right] \E \left[ \exp \left\{2t Y \right\}  \right] }
\\ \leq ~ &  \sqrt{ \exp \left\{K_X ^2 \cdot 4 t^2  \right\} \exp \left\{K_Y^2 \cdot 4  t^2\right\}   }
\\ = ~ & \exp \left\{2t^2 \left( K_X^2  +  K_Y^2  \right)\right\}  
,
\end{align*}
and the result is proven with \(K^2 := 2 \left( K_X^2  +  K_Y^2  \right)\).

\end{proof}

\begin{lemma}\label{g.sing.val.lem}
The smallest and largest singular values of \(\boldsymbol{G}_N\) are bounded as follows: \( \sigma_{\text{max}} (\boldsymbol{G}_N)  \leq  \sqrt{2}\) and \(\sigma_{\text{min}}\left( \boldsymbol{G}_N  \right)  \geq \sigma/\sqrt{T \sigma_c^2 + \sigma^2}  \).
\end{lemma}

\begin{proof}[Proof of Lemma \ref{g.sing.val.lem}]
We will prove the result in a similar way to the proof of Lemma \ref{d.sing.val.lem}. By properties of the Kronecker product, it is enough to bound the singular values of \(\boldsymbol{\Omega}\). We have
\[
 \sigma_{\text{max}} (\boldsymbol{\Omega})  \leq   \sqrt{  \lVert \boldsymbol{\Omega} \rVert_{1} \lVert \boldsymbol{\Omega} \rVert_{\infty} } \leq \sqrt{[T \cdot \sigma_c^2 + \sigma^2]^2} = T \cdot \sigma_c^2 + \sigma^2
 ,
\]
where for matrices \( \lVert \cdot \rVert_1\) denotes the maximum absolute column sum of a matrix and \(\lVert \cdot \rVert_\infty\) denotes the maximum absolute row sum. We will use \eqref{sv.lb} to lower-bound \( \sigma_{\text{min}} (\boldsymbol{\Omega}) \). First we will find \(\boldsymbol{\Omega}^{-1}\). Observe that we can express \(\boldsymbol{\Omega}\) as
\[
\boldsymbol{\Omega} = \sigma^2 \boldsymbol{I}_{T} + \boldsymbol{v} \boldsymbol{v}^\top
\]
where \(\boldsymbol{v} = (\sigma_c, \ldots, \sigma_c) \in \mathbb{R}^{T}\). Therefore we can invert \(\boldsymbol{\Omega}\) using the Sherman-Morrison formula:
\begin{align}
\boldsymbol{\Omega}^{-1} = ~ & \frac{1}{\sigma^2} \boldsymbol{I}_{T} - \frac{ \frac{1}{\sigma^2} \boldsymbol{I}_{T}  \boldsymbol{v} \boldsymbol{v}^\top  \frac{1}{\sigma^2} \boldsymbol{I}_{T}  }{1 + \boldsymbol{v}^\top  \frac{1}{\sigma^2} \boldsymbol{I}_{T} \boldsymbol{v}}
 =  \frac{1}{\sigma^2} \left( \boldsymbol{I}_{T} - \frac{\frac{1}{\sigma^2} \boldsymbol{v} \boldsymbol{v}^\top  }{1+ \frac{1}{\sigma^2} \boldsymbol{v}^\top \boldsymbol{v}} \right)
=  \frac{1}{\sigma^2} \left( \boldsymbol{I}_{T} - \frac{ \boldsymbol{v} \boldsymbol{v}^\top  }{\sigma^2 + T \sigma_c^2} \right) \nonumber
\\ = ~ &  \frac{1}{\sigma^2} \begin{pmatrix}
1 - \frac{\sigma_c^2}{\sigma^2 + T \sigma_c^2} &  - \frac{\sigma_c^2}{\sigma^2 + T \sigma_c^2} & \cdots &  - \frac{\sigma_c^2}{\sigma^2 + T \sigma_c^2}
\\  - \frac{\sigma_c^2}{\sigma^2 + T \sigma_c^2} & 1 - \frac{\sigma_c^2}{\sigma^2 + T \sigma_c^2} & \cdots &  - \frac{\sigma_c^2}{\sigma^2 + T \sigma_c^2}
\\ \vdots & \vdots & \ddots & \vdots
\\   - \frac{\sigma_c^2}{\sigma^2 + T \sigma_c^2} &   - \frac{\sigma_c^2}{\sigma^2 + T \sigma_c^2} & \cdots & 1 - \frac{\sigma_c^2}{\sigma^2 + T \sigma_c^2}
\end{pmatrix} \nonumber % \label{omega.inv.exp}
.
\end{align}

Therefore
\begin{align*}
\sigma_{\text{min}} \left( \boldsymbol{\Omega} \right) \geq ~ &  \frac{1}{\sqrt{\lVert \boldsymbol{\Omega}^{-1} \rVert_{1} \lVert \boldsymbol{\Omega}^{-1} \rVert_{\infty}}} \geq \sigma^2 \left( 1 + T \frac{\sigma_c^2}{\sigma^2 + T \sigma_c^2}\right)^{-1} =  \sigma^2 \left( \frac{\sigma^2 + 2 T\sigma_c^2}{\sigma^2 + T \sigma_c^2} \right)^{-1}
\\  = ~ &  \frac{\sigma^2 \left( \sigma^2 + T \sigma_c^2\right)}{\sigma^2 + 2 T\sigma_c^2} \geq  \frac{\sigma^2 \left( \sigma^2 + T \sigma_c^2\right)}{2 \left(\sigma^2 +  T\sigma_c^2 \right)} = \frac{\sigma^2}{2}.
\end{align*}
Finally, using these results we have \(\sigma_{\text{min}} \left( \boldsymbol{\Omega}^{-1/2} \right)  \geq 1/\sqrt{T \sigma_c^2 + \sigma^2} \) and \(\sigma_{\text{max}} \left( \boldsymbol{\Omega}^{-1/2} \right)  \leq \sqrt{2}/\sigma\). By properties of the Kronecker product, these bounds hold for \(\boldsymbol{G}_N = \sigma \boldsymbol{I}_{T} \otimes \boldsymbol{\Omega}^{-1/2}\) as well after multiplying by \(\sigma\).

\end{proof}

\begin{lemma}\label{rank.cond2}
Suppose one treated cohort has fewer than \(d_N + 1\) units. Then \(\boldsymbol{\tilde{Z}}\) does not have full column rank.
\end{lemma}

%\begin{remark}
%Lemma \ref{rank.cond2} shows that observing \(d_N + 1\) units in each cohort is necessary for \(\boldsymbol{\tilde{Z}}\) to be full rank. From simulation studies, it appears that this condition is also sufficient.
%\end{remark}

\begin{proof}[Proof of Lemma \ref{rank.cond2}]

The treatment effect dummies can be seen as interaction effects between the cohort dummies and (a subset of) the time dummies. In particular, the treatment effect for cohort \(r\) at time \(t\) is equal to the column for the \(r^{\text{th}}\) cohort dummy, but with entries for all times except \(t\) set equal to 0. The \(T - r + 1\) treatment dummies for cohort \(r\) are mutually orthogonal and span a \((T - r + 1)\)-dimensional subspace, and they occupy an \(N_r(T-r+1)\)-dimensional subspace, where \(N_r\) is the number of units in cohort \(r\).

Because the covariates (like the cohort indicators) are time-invariant, the time-covariate interactions are closely related to the treatment effect dummies. The covariate-time interactions for cohort \(r\) lie in the same \(N_r(T-r+1)\)-dimensional subspace as the treatment effects for cohort \(r\) and are either linearly dependent or span a subspace of dimension \(d_N(T - r + 1)\).  \(\boldsymbol{\tilde{Z}}\) is rank deficient in the former setting, so we assume the latter. Then the treatment dummies and covariate-time interactions for cohort \(r\) span a subspace of dimension \((d_N + 1)(T - r + 1)\) within this \(N_r(T-r+1)\)-dimensional subspace. So if \(N_r < d_N + 1\), these columns must be linearly dependent, and \(\boldsymbol{\tilde{Z}}\) does not have full column rank.

\end{proof}

\begin{lemma}\label{conv.claim.lem}
Suppose \(e_{1N}\) is bounded away from 0 with high probability: for some finite \(N^*, c > 0\), for all \(N > N^*\) it holds that \(\mathbb{P} (e_{1N} < c ) <  a_N\) for some vanishing sequence \(\{a_N\}\). Let \(\{X_N\}\) be a sequence of random variables, and suppose that \(X_N  = \mathcal{O}_{\mathbb{P}}( \frac{1}{e_{1N}}\sqrt{p_N/N})\). Then \(X_N = \mathcal{O}_{\mathbb{P}}(\sqrt{p_N/N})\).
\end{lemma}

\begin{proof}[Proof of Lemma \ref{conv.claim.lem}]
By assumption, for every \(\epsilon > 0\) there exists a finite \(M^{(1)}_\epsilon > 0\) and \(N^{(1)}_\epsilon > 0\) such that for all \(N > N^{(1)}_\epsilon\),
\[
\mathbb{P} \left( \left| \frac{e_{1N} X_N}{\sqrt{p_N/N}} \right| > M^{(1)}_\epsilon \right) < \epsilon
.
\]
Now for all \(N >  \max\{N^*, N^{(1)}_{\epsilon/2}\} \),
\begin{align*}
& \mathbb{P} \left( \left| \frac{X_N}{\sqrt{p_N/N}} \right| > \frac{M_{\epsilon/2}^{(1)}}{c} \right) 
\\ = ~ &   \mathbb{P} \left( \left| \frac{X_N}{\sqrt{p_N/N}} \right| > \frac{M_{\epsilon/2}^{(1)}}{c} \mid e_{1N} \geq c \right)   \mathbb{P}(e_{1N} \geq c) 
  + \mathbb{P} \left( \left| \frac{X_N}{\sqrt{p_N/N}} \right| > \frac{M_{\epsilon/2}^{(1)}}{c} \mid e_{1N} < c \right)   \mathbb{P}(e_{1N} < c) 
\\ \leq ~ &   \mathbb{P} \left( \left| \frac{X_N}{\sqrt{p_N/N}} \right| > \frac{M_{\epsilon/2}^{(1)}}{c} \mid e_{1N} \geq c \right)    \mathbb{P}(e_{1N} \geq c) 
+  \mathbb{P}(e_{1N} < c) 
\\ \leq  ~ &   \mathbb{P} \left( \left| \frac{ X_N}{\sqrt{p_N/N}} \right| > \frac{M_{\epsilon/2}^{(1)}}{e_{1N}} \mid e_{1N} \geq c \right)    \mathbb{P}(e_{1N} \geq c) 
 +  \mathbb{P}(e_{1N} < c) 
\\ =  ~ &   \mathbb{P} \left( \left| \frac{e_{1N} X_N}{\sqrt{p_N/N}} \right| > M_{\epsilon/2}^{(1)} \mid e_{1N} \geq c \right)    \mathbb{P}(e_{1N} \geq c) 
 +  \mathbb{P}(e_{1N} < c) 
\\ \leq ~ &   \mathbb{P} \left( \left| \frac{e_{1N} X_N}{\sqrt{p_N/N}} \right| > M_{\epsilon/2}^{(1)} \mid e_{1N} \geq c \right)    \mathbb{P}(e_{1N} \geq c) 
+  \mathbb{P} \left( \left| \frac{e_{1N} X_N}{\sqrt{p_N/N}} \right| > M_{\epsilon/2}^{(1)} \mid e_{1N} < c \right)    \mathbb{P}(e_{1N} < c) 
\\ & +  \mathbb{P}(e_{1N} < c) 
\\ = ~ &   \mathbb{P} \left( \left| \frac{e_{1N} X_N}{\sqrt{p_N/N}} \right| > M_{\epsilon/2}^{(1)}  \right)     +  \mathbb{P}(e_{1N} < c) 
\\ < ~ & \frac{\epsilon}{2} + a_N
.
\end{align*}
Since \(a_N\) is vanishing, there exists a finite \(N^{(2)}_{\epsilon}\) so that \(a_N \leq \epsilon/2\) for all \(N >  N^{(2)}_{\epsilon}\). Then we have established that for any \(\epsilon > 0\), for \(M_\epsilon : = M_{\epsilon/2}^{(1)}/c  \) and \( N_\epsilon := \max\{N^*, N^{(1)}_{\epsilon/2}  , N^{(2)}_{\epsilon}\} \), for all \(N > N_\epsilon\) it holds that
\begin{align*}
 \mathbb{P} \left( \left| \frac{ X_N}{\sqrt{p_N/N}} \right| > M_\epsilon\right) < ~ &  \epsilon
.
\end{align*}

\end{proof}

\section{Other Extensions}
\subsection{Incorporating Direct Ridge Regularization into FETWFE}\label{sec:ridge_extension}

Although FETWFE is a regularized estimation method, in settings where \(p\) is comparable to \(N\) it can be helpful in practice to add additional regularization to the estimated coefficients themselves. One practical way to do this is to add an additional \(\ell_2\) (ridge) penalty directly to the estimated coefficients to enhance numerical stability and improve predictive performance. This idea is similar to the motivation behind the elastic net \citep{Zou2005} which combines the lasso and ridge penalties. In our context, ridge regularization applied to the untransformed coefficients \(\boldsymbol{\beta}\) provides an optional way to stabilize fused extended two‐way fixed effects (FETWFE) estimation in addition to the adaptive fusion penalties.
%\subsubsection{Motivation}
%
%The FETWFE estimator estimates the parameters in the model
%\[
%\min_{\boldsymbol{\beta} \in \mathbb{R}^{p_N}} \left\{ \| \boldsymbol{y} - \boldsymbol{Z}\boldsymbol{\beta} \|_2^2 + \lambda_{1N}\|D_N\boldsymbol{\beta}\|_q^q \right\},
%\]
%which encourages fusion of certain coefficients via the penalty \(\|D_N\boldsymbol{\beta}\|_q^q\). In settings with many parameters or in the presence of multicollinearity, the estimation can benefit from additional stabilization. A
Specifically, adding a ridge penalty to the optimization problem,
\[
\lambda_{2N} \|\boldsymbol{\beta}\|_2^2,
\]
serves this purpose by shrinking the coefficients toward zero. This additional regularization is conceptually similar to the elastic net approach and may reduce the variance of the estimator without sacrificing too much bias.
%Furthermore, because the fusion penalty is applied to the transformed parameters \(D_N\boldsymbol{\beta}\) while the ridge penalty acts on the original coefficients, the two penalties jointly allow one to stabilize estimation while still adaptively enforcing the plausible restrictions.

\subsubsection{Estimation Approach}

To incorporate ridge regularization, consider the following modified optimization problem:
\begin{equation}\label{elastic.net.1}
\min_{\boldsymbol{\beta} \in \mathbb{R}^{p_N}} \left\{ \| \boldsymbol{y} - \boldsymbol{Z}\boldsymbol{\beta} \|_2^2 + \lambda_{1N}\|\boldsymbol{D}_N\boldsymbol{\beta}\|_q^q + \lambda_{2N} \|\boldsymbol{\beta}\|_2^2 \right\}.
\end{equation}
Assuming that \(\boldsymbol{y}\) is centered and the columns of \(\boldsymbol{Z}\) are centered and scaled, his can be rewritten by augmenting the data, in an approach similar to the one mentioned by \citep{Zou2005}. Define the extended response and design matrix as
\[
\tilde{\boldsymbol{y}}^* = \begin{pmatrix} \boldsymbol{y}\\ \boldsymbol{0} \end{pmatrix} \quad\text{and}\quad
\tilde{\boldsymbol{Z}}^* = \begin{pmatrix} \boldsymbol{Z}\\ \sqrt{\lambda_{2N}}\,I_{p_N} \end{pmatrix}.
\]
Then \eqref{elastic.net.1} can be written as
\[
\min_{\boldsymbol{\beta} \in \mathbb{R}^{p_N}} \left\{ \left\|\tilde{\boldsymbol{y}}^* - \tilde{\boldsymbol{Z}}^*\boldsymbol{\beta}\right\|_2^2 + \lambda_{1N}\|D_N\boldsymbol{\beta}\|_q^q \right\}.
\]
After again using the transformation from Appendix \ref{sec.prove.first.thm}
\[
\boldsymbol{\theta} = \boldsymbol{D}_N\boldsymbol{\beta} \quad \Longrightarrow \quad \boldsymbol{\beta} = \boldsymbol{D}_N^{-1}\boldsymbol{\theta},
\]
we can finally write \eqref{elastic.net.1} as
\[
\min_{\boldsymbol{\theta} \in \mathbb{R}^{p_N}} \left\{ \left\|\tilde{\boldsymbol{y}}^* - \begin{pmatrix} \boldsymbol{Z}D_N^{-1}\\ \sqrt{\lambda_{2N}}\,D_N^{-1} \end{pmatrix}\boldsymbol{\theta}\right\|_2^2 + \lambda_{1N}\|\boldsymbol{\theta}\|_q^q \right\}
,
\]
which (given a fixed choice of \(\lambda_{2N}\)) can be solved with a standard bridge regression solver (for example, using packages such as \texttt{grpreg} in R) exactly as one would for the original FETWFE estimator.

Further, if \(\lambda_{2N}\) vanishes appropriately in \(N\) then this ridge-penalized version of FETWFE is likely to retain its consistency and asymptotic normality properties, because the extra ridge rows will have a vanishing influence on the estimation as \(N \to \infty\).

Such an estimator is analogous to the \textit{naive elastic net} of \citep{Zou2005}; it may be wise, by analogy, to scale this estimator by \(1 + \lambda_2\) in order to mitigate the ``double" regularization.

\subsection{The Case with No Covariates}
The fused extended two-way fixed effects (FETWFE) model simplifies when no covariates are included. In this setting, the model requires the estimation of only fixed effects and treatment parameters, and the identification assumptions can be different. In the standard FETWFE model with covariates, the full parameter vector comprises
\[
p = R + (T-1) + d + dR + d(T-1) + \mathfrak{W} + \mathfrak{W} d
\]
parameters to estimate. When no covariates are present (i.e., $d=0$), the parameter vector reduces to a length of
\[
p = R + (T-1) + \mathfrak{W}.
\]
Thus, only the fixed effects and the treatment effects must be estimated, reducing the dimensionality of the model and, in turn, the degrees of freedom required for estimation.

However, estimation without covariates also fundamentally changes the required causal inference assumptions. With covariates, identification of the treatment effects relies on the conditional no anticipation (CNAS) and conditional parallel trends (CCTS) assumptions. These assumptions require that, after conditioning on covariates, units do not anticipate treatment and that the untreated potential outcomes follow parallel trends.

In the absence of covariates, these conditions are replaced by their marginal (unconditional) counterparts, Assumption (CTS) along with marginal no-anticipation. The more key difference is marginal parallel trends. As we discussed in Appendix \ref{par.trend.ciuu.app}, there is some nuance, but overall Assumption (CTS) is generally more stringent than (CCTS).

Omitting covariates does allow us to omit Assumption (LINS). When covariates are included, we require Assumption (LINS) to impose a linear relationship between covariates and potential outcomes. Without covariates, there is no need for such a linearity assumption. Therefore, the identification of treatment effects does not depend on the correctness of a linear specification for covariate effects, which may enhance robustness against model misspecification.

\end{document}